\DocumentMetadata{
	lang		= en-US,
	pdfversion  = 1.7,
	pdfstandard = a-2b,
}

\documentclass[twoside]{mitthesis}% fontset=newtx, fontset=libertine, fontset=libertinus, fontset=lmodern, fontset=newtx-sans-text, fontset=fira-newtxsf, fontset=heros-stix2, fontset=stix2, fontset=lmodern
\usepackage{booktabs}% publication quality tables (https://ctan.org/pkg/booktabs)

\usepackage{dcolumn}% For alignment of numbers on the decimal place (https://ctan.org/pkg/dcolumn) 
  \newcolumntype{d}[1]{D{.}{.}{#1}}% use with dcolumn package. Note: dcolumns are set in math mode.
\usepackage{microtype}% typographic fine-tuning, used in sample thesis committee page, but also acting globally on the text 

 \graphicspath{ {chap3figs/} {chap4figs/} {chap5figs/} {chap5figs-2-to-4-reps/} {chap5figs-2-to-5-reps/} {chap5figs-1-to-n-code/} {chap5figs-appendix/} }%{../art/} }% For details see: https://latexref.xyz/dev/latex2e.html#g_t_005cgraphicspath

\usepackage[style=ext-numeric-comp,giveninits=true,maxbibnames=10,sorting=none]{biblatex}

\hypersetup{%
	pdfsubject={A thesis on working with quantum states, circuits, and error-correcting codes by using the graph formalism},
	pdfkeywords={Massachusetts Institute of Technology, mit,MIT,Andrey,Boris,Khesin,quantum,state,circuit,computing,graph,formalism,ZX,calculus,error,correction,correcting,code,fault,tolerance,stabilizer,Clifford,operator,Pauli,canonical,form,compiler,surface,toric,normal},
	pdfurl={},
	pdfcontactemail={andrey.khesin@gmail.com},
	pdfauthortitle={},
}

\usepackage[ruled,longend,noend]{algorithm2e}
\usepackage{amsmath}
\usepackage{amssymb}
\usepackage{amsthm}
\usepackage{apptools}
\usepackage{array}
\usepackage{bbm}
\usepackage{dsfont}
\usepackage{breqn}
\usepackage{enumerate}
\usepackage{textcomp, gensymb}
\usepackage{graphicx}
\usepackage{hhline}
\usepackage{hyperref}
\usepackage{multirow}
\usepackage{physics}
\usepackage{zx-calculus, quantikz}
\usepackage{stmaryrd}
\usepackage{tabularx}
\usepackage{tikz}
\usepackage{tikzit}

\newcommand{\calP}{\mathcal{P}}

\newcommand{\calC}{\mathcal{C}}
\newcommand{\calS}{\mathcal{S}}

\newcommand{\braopket}[3]{\left<#1\left|#2\right|#3\right>}
\newcommand{\pars}[1]{\left(#1\right)}
\newcommand{\s}{\sigma}
\newcommand{\z}{\zeta}

\newcommand{\CA}{\mathcal{A}}
\newcommand{\CC}{\mathcal{C}}
\newcommand{\CI}{\mathcal{I}}
\newcommand{\CO}{\mathcal{O}}
\newcommand{\CP}{\mathcal{P}}
\newcommand{\CS}{\mathcal{S}}

\renewcommand{\d}{\delta}
\newcommand{\D}{\Delta}
\newcommand{\e}{\epsilon}
\newcommand{\Th}{\Theta}
\newcommand{\w}{\omega}

\newcommand{\Z}{\mathbb{Z}}
\newcommand{\es}{\varnothing}

\newcommand{\set}[1]{\{ #1 \}}
\newcommand{\floor}[1]{\left\lfloor #1 \right\rfloor}
\newcommand{\ceil}[1]{\left\lceil #1 \right\rceil}

\usetikzlibrary{decorations.markings}
\usetikzlibrary{shapes.geometric}
\pgfdeclarelayer{edgelayer}
\pgfdeclarelayer{nodelayer}
\pgfsetlayers{edgelayer,nodelayer,main}
\tikzstyle{none}=[inner sep=0pt]

\tikzstyle{input}=[circle,fill=blue,draw=black,line width=0.8 pt]
\tikzstyle{output}=[circle,fill=black,draw=black,line width=0.8 pt,minimum size=0.1cm,inner sep=0pt]

\tikzstyle{simple}=[-,draw=black,line width=2.000]
\tikzstyle{arrow}=[-,draw=black,postaction={decorate},decoration={markings,mark=at position .5 with {\arrow{>}}},line width=2.000]
\tikzstyle{tick}=[-,draw=black,postaction={decorate},decoration={markings,mark=at position .5 with {\draw (0,-0.1) -- (0,0.1);}},line width=2.000]
\tikzstyle{z-node}=[fill={rgb,255: red,0; green,192; blue,0}, draw=black, shape=circle,minimum size=0.2cm,inner sep=0pt]
\tikzstyle{x-node}=[fill={rgb,255: red,192; green,0; blue,0}, draw=black, shape=circle]
\tikzstyle{h-box}=[fill=yellow, draw=black, shape=rectangle]
\tikzstyle{pivot}=[fill={rgb,255: red,255; green,128; blue,0}, draw=black, shape=circle]
\tikzstyle{h}=[-, draw=blue, dashed,dash pattern=on 2pt off 1pt]

\tikzstyle{z-node-demo}=[fill={rgb,255: red,0; green,192; blue,0}, draw=black, shape=circle]
\tikzstyle{x-node-demo}=[fill={rgb,255: red,192; green,0; blue,0}, draw=black, shape=circle]

\theoremstyle{plain}
\newtheorem{theorem}{Theorem}[section]
\newtheorem{corollary}[theorem]{Corollary}
\newtheorem{lemma}[theorem]{Lemma}
\newtheorem{claim}[theorem]{Claim}
\newtheorem{proposition}[theorem]{Proposition}

\theoremstyle{remark}
\newtheorem{remark}[theorem]{Remark}

\theoremstyle{definition}
\newtheorem{definition}[theorem]{Definition}
\newtheorem{example}[theorem]{Example}
\newtheorem{conjecture}[theorem]{Conjecture}

\renewbibmacro*{name:andothers}{
  \ifboolexpr{
    test {\ifnumequal{\value{listcount}}{\value{liststop}}}
    and
    test \ifmorenames
  }
    {\ifnumgreater{\value{liststop}}{1}
       {\finalandcomma}
       {}
     \andothersdelim\bibstring[\emph]{andothers}}
    {}}

\AtAppendix{
\counterwithin{theorem}{chapter}
}

\begin{document}
%%% edit the following commands to match your thesis %%%%%%%%%%

\title{Quantum Computing from Graphs}

% \Author{Author full name}{Author department}[Author's first PREVIOUS degree][Author's second PREVIOUS degree][...
% Note that third, fourth, fifth, and sixth arguments are optional [] and may be omitted

% note on names: most of the following names are made up; Silas Holman was a physics professor at MIT in the 19th century.

\Author{Andrey Boris Khesin}{Department of Mathematics}

% Use once for each degree fulfilled by thesis
% For two degrees from one department, leave the department argument blank for the second degree {}.
\Degree{Doctor of Philosophy in Mathematics}{Department of Mathematics}

% If there is more than one supervisor, use the \Supervisor command for each.
\Supervisor{Peter Shor}{Professor of Mathematics}

% Professor who formally accepts theses for your department (e.g., the Graduate Officer, Professor Sméagol,...)
% If more than one department, use more than once
\Acceptor{Jonathan Kelner}{Professor of Mathematics}{Graduate Co-Chair, Department of Mathematics} % \\ \> Third title}
%%%  If you need to reduce vertical space, put the acceptor title in the second argument and leave the third blank, {}.

% Usage: \DegreeDate{Month}{year}
% Valid degree months are February, May, June, or September
\DegreeDate{February}{2025}

% Date that final thesis is submitted to department
\ThesisDate{December 20, 2024}

%%%%%%  Choose whether to have a CREATIVE COMMONS License  %%%%%%%%%%%%%%%%%%%%%%%%%%%%%%%%%%%%%%
%
% If you are using a cc license, uncomment the following line and insert details of your cc license here.
%
\CClicense{CC BY-NC-ND 4.0}{https://creativecommons.org/licenses/by-nc-nd/4.0/}
%

%%%%%%%  Solutions for overflowing titlepage  %%%%%%%%%%%%%%%%%%%%%%%%%%%%%%%%%%%%%%%%%%%%%%%%%%%

% If your title page is overflowing (from too many names, degrees, etc.):
%
% (a) you can reduce the 12pt and 18pt skips between various blocks to 6pt with this command:
%
% \Tighten
%
% (b)  you can scale down the Signature block at the bottom with this command:
%
% \SignatureBlockSize{\small}  %or this one \SignatureBlockSize{\footnotesize}
%
% (c) you can put the acceptor name and title onto two lines, rather than three like this:
%
% \Acceptor{Tertius Castor}{Professor and Graduate Officer, Department of Research}{}
%
% (d) you can change the font size of the author name[s] with
%
%	\AuthorNameSize{\normalsize}
%
% (e) and you can omit any previous degrees from the title page, instead mentioning them in the biographical sketch

% Also, if you prefer to keep the text toward the top of the page with most white space at the bottom, you
% can use this command to squash all of the vertical glue (stretchy space) with this command:
%
% \Squash 
%
% This command is useful when the text has not already reach the bottom of the page, since the glue gets squashed automatically
% when the page is too full.

%%%%%%%%%%%%%%%%%%%%%%%%%%%%%%%%%%%%%%%%%%%%%%%%%%%%%%%%%%%%%%%%%%%%%%%%%%%%%%%%%%%%%%%%%%%%%%%%%

%%% Make titlepage
\maketitle

%%%%%%%%% Contents that you need to write follows! %%%%%%%%%%%%%%%%%%%%%%%%%%%%%%%%%%%%%%%%%%%%%%

% \includeonly{acknowledgments,biography,chapter1,chapter2,...,appendixa,...} 
%   for usage of includeonly, see https://latexref.xyz/_005cinclude-_0026-_005cincludeonly.html

%%% Frontmatter (write this material in the mentioned files)  %%%%%%%%%%%%%%%%%%%%%%%%%%%%%%%%%%%

% This page is optional. Edit the file committee_members.tex 
% Sample thesis committee page for mitthesis.cls
% Version 1.01, 2024/10/08
%
% This page is not required by the MIT Libraries, but some departments require it.
%
% Insert between title page and abstract page.
% Format this page in any way that you like.  
% Add supervisor titles, degrees, and departments as appropriate.

%%%%% FORMATTING COMMANDS %%%%%%%%%%%%%%%%%%

%% Format title
\NewDocumentCommand\CommitteePageTitle{m}{
	\vspace*{75pt}%36pt}
	\IfPackageLoadedTF{microtype}
		{\textls*{\Large\textbf{\MakeUppercase{#1}}}}
		{{\Large\textbf{\MakeUppercase{#1}}}}%
	\pdfbookmark[0]{#1}{Committee}%
	\vspace*{10pt}%
}
% \textls* produces additional letter separation (appropriate for capitalized display text),
% PROVIDED THAT \usepackage{microtype} has been loaded in the preamble. 
% The extra space added is 100/1000 em (adjustable, see package documentation).

%% Format committee member subheadings
\NewDocumentCommand\Role{m}{
	\vspace*{50pt}%25pt}
	\IfPackageLoadedTF{microtype}
		{\textls*{\large{\textsc{#1}}}}
		{{\large\textsc{#1}}}%
	\vspace*{12pt}%
}

%%%%%%%%%%%%%%%%%%%%%%%%%%%%%%%%%%%%%%%%%%%%%

\begin{flushright}

\CommitteePageTitle{Thesis committee}

\Role{Thesis Supervisor}

 \textbf{Peter Shor} \\
 {\itshape
 Morss Professor of Applied Mathematics \\
 Department of Mathematics \\
 }

\Role{Thesis Readers}

 \textbf{Isaac Chuang} \\
 {\itshape
   Professor of Physics and EECS \\
   Departments of Physics and EECS \\[18pt]
 }

 \textbf{Aram Harrow}\\
 {\itshape
   Professor of Physics \\
   Department of Physics \\[18pt]
 }

 \textbf{Jonathan Kelner} \\
 {\itshape
   Professor of Applied Mathematics \\
   Department of Mathematics \\
 }

\end{flushright}

\cleardoublepage

% The abstract environment creates all the required headings and footers. 
% You only need to the text of the abstract in the file abstract.tex
\begin{abstract}
	% From mitthesis package
% Version: 1.01, 2023/06/19
% Documentation: https://ctan.org/pkg/mitthesis
%
% The abstract environment creates all the required headers and footnote. 
% You only need to add the text of the abstract itself.
%
% Approximately 500 words or less; try not to use formulas or special characters
% If you don't want an initial indentation, do \noindent at the start of the abstract

While stabilizer tableaus have proven exceptionally useful as a descriptive tool for additive quantum codes, they otherwise offer little guidance for concrete constructions or coding algorithm analysis. We introduce a representation of stabilizer codes as graphs with certain structures. Specifically, the graphs take a semi-bipartite form wherein input nodes map to output nodes, such that output nodes may connect to each other but input nodes may not. Intuitively, the graph's input-output edges represent information propagation of the encoding circuit, while output-output edges represent the code's entanglement structure. We prove that this graph representation is in bijection with tableaus and give an efficient compilation algorithm that transforms tableaus into graphs. We then show that this map is efficiently invertible, which gives a new universal recipe for code construction by way of finding graphs with sufficiently nice properties.

The graph representation gives insight into both code construction and algorithms. To the former, we argue that graphs provide a flexible platform for building codes, particularly at small, non-asymptotic scales. We construct as examples several constant-size codes and several infinite code families. We also leverage graphs in a probabilistic analysis to extend the quantum Gilbert-Varshamov bound into a three-way distance-rate-weight trade-off. To the latter, we show that key coding algorithms, distance approximation, weight reduction, and decoding, are unified as instances of a single optimization game on a graph. Moreover, key code properties such as distance, weight, and encoding circuit depth, are all controlled by the graph degree. We give efficient algorithms for producing simple encoding circuits whose depths scale as twice the degree and for implementing logical diagonal and certain Clifford gates with non-constant but reduced depth. Finally, we construct a simple efficient decoding algorithm and prove a performance guarantee for certain classes of graphs.
These results give evidence that graphs are generically useful for the study of quantum computing and its practical implementations.% use \input rather than \include because we're inside an environment
\end{abstract}

%% acknowledgments.tex

% From mitthesis package
% Version: 1.02, 2024/06/19
% Documentation: https://ctan.org/pkg/mitthesis

\chapter*{Acknowledgements}
\pdfbookmark[0]{Acknowledgements}{acknowledgements}

My first and foremost thanks go out to my advisor, Peter Shor.
Your unwavering support of my interests and ideas helped shape not only this thesis but the researcher I am today.
Your knowledge and your kindness will always be an inspiration.
Thank you for your willingness to explore with me, both inside and outside of the department.
My journey into quantum computing would not have been possible without you, in more ways than one.

Thank you to Isaac Chuang for teaching me my first quantum computing class.
You immediately caused me to become infatuated with the subject, despite the fact that I was lacking a key prerequisite!
Your suggestion of a research problem as a final project led directly to the results in this very thesis.

Thank you to my thesis committee, which also includes Aram Harrow, and Jonathan Kelner. 
I am very grateful for your support and encouragement in the home stretch of my time at MIT, your comments, and your feedback.

Thank you to the MIT Math Department: to Pavel Etingof, John Urschel, Tanya Khovanova, Michel Goemans, Slava Gerovitch, Andr\'e Lee Dixon, Michele Gallarelli, Charlotte Rubel, Theresa Cummings, Rosalee Zammuto, Sandi Miller, and all of the other staff, faculty, and students.
Your help, friendships, advice, and support have been out of this world.
You are all truly the reason why MIT has the best Math Department in the world.

Thank you to Jonathan, my main coauthor and office-mate.
Your excitement has always been infectious and I am eternally grateful that your vision helped our work grow into what it is today.
It's been a privilege working with you.

Thank you to my friends at ET, SigmaCamp, Tech Squares, $\sqrt{}$mathroots, TOPS, as well as my other coauthors and office-mates, Sunny, Norah, Kevin, Alex, Alex, Daniil, and Vasily.
You made it all so much better.
You listened to my riddles and helped me up when I was down.
You filled my days with fun and helped me procrastinate.
You are all awesome.

Thank you to my family: to Mom, Dad, and Nikita; you are my rock.
Though I would get advice from many sources, I would always take end up taking your suggestions over others'.
No amount of thanks can express my gratitude for your love and support.

Finally, thank you to my wife, Polina. You are my everything.
Nothing I can write here will properly express how I grateful am I that you are in my life.
Thank you for the world.

\phantom.

This material is based upon work supported by the Department of Energy,
Office of Science, National Quantum Information Science Research
Centers, Quantum Systems Accelerator, under Grant number DOE
DE-SC0012704, as well as NSF grants CCF-1452616 and CCF-1729369.% acknowledgments.tex (.tex extension is presumed by \include) 

%%% Table of contents and lists of stuff (delete unused lists, i.e., if no tables or figures) %%%%%

\tableofcontents
\listoffigures
\listoftables

%%% Chapters of thesis  %%%%%%%%%%%%%%%%%%%%%%%%%%%%%%%%%%%%%%%%%%%%%%%%%%%%%%%%%%%%%%%%%%%%%%%%%%%

%% If you want to use "double spacing", you should start here...

 % From mitthesis package
% Version: 1.07, 2024/09/26
% Documentation: https://ctan.org/pkg/mitthesis

\chapter{Introduction}

The work done in the past half century on quantum computing has brought large-scale quantum computers closer to reality. Today, quantum computers are just barely crawling out of their proverbial infancy but remain far removed from the public eye.
Quantum computers differ from classical computers and classical supercomputers through the use of qubits rather than bits. The properties of quantum mechanics inherent in qubits, including superposition and entanglement, allow quantum computers to efficiently simulate quantum systems, making certain calculations much more efficient when done on quantum computers~\cite{nielsen2002quantum}.
However, most quantum algorithms which are believed to provide significant speedups to certain computational problems require both many qubits and error correction to succeed in practice~\cite{grover1996fast,shor1999polynomial,brakerski2021cryptographic,farhi2012quantum,georgescu2014quantum,huang2020predicting,harrow2004superdense,nielsen2002quantum}.

As with classical systems, quantum information processors face noise that disrupts information transmission between the sender and receiver. Due to the vulnerability of qubits to this noise, one of the principal challenges in quantum computing is to account for this noise~\cite{kim2023evidence}. To this end, quantum error-correcting codes are developed so that quantum information can be transmitted successfully in the presence of noise~\cite{nielsen2002quantum}.

An important restriction on quantum error-correcting codes stems from the no-cloning theorem: while classical computers can copy bits, quantum mechanics does not allow for the cloning of unknown qubits, and the measurement of a qubit eliminates the information available in the qubit~\cite{nielsen2002quantum}. As such, constructing suitable quantum error-correcting codes presents new challenges compared to their classical counterparts.

As advancements in the experimental scaling of quantum computers have steadily marched forward, the problem of designing good quantum codes for the practical implementation of various quantum algorithm has in turn become of more pressing interest~\cite{kandala2019error,larsen2019deterministic,arute2019quantum,wang2019boson,andersen2020repeated}.
With a few exceptions, stabilizer codes have emerged as one of the most intensely studied families of codes due to their simple description~\cite{kitaev1997quantum,bravyi1998quantum,fowler2012surface,shor1995scheme,gottesman2009introduction,nielsen2002quantum,gottesman1997stabilizer}.

In this thesis, we explore ways of working with stabilizer codes, quantum states, and quantum circuits through the use of graphs.
We present the graph formalism of quantum computing, a framework with all the necessary parts to discuss various aspects of quantum computation.
We highlight ways to transform quantum states, circuits, and codes from one representation to another.
Any representation of a quantum state or code in the stabilizer formalism will be part of a large class of equivalent presentations.
We several examples of representations in the graph formalism which uniquely represent the most well-studied families of quantum states and codes.
We discuss avenues that the graph formalism can be used for universal quantum computation, fault tolerance, decoding, and code search.

Chapter~\ref{chapter:noncliffordgates} is based on work with Kevin Ren~\cite{khesin2021extending}.
We present universal formulas for applying non-Clifford gates to graphs.
Such gates are critical for universal quantum computation and fault tolerance~\cite{aharonov2003-universality,nielsen2002quantum}.
Applying these gates results in linear combinations of quantum states expressed as graphs.
We provide formulas for merging such graphs together and discuss their use and runtime.

Chapter~\ref{chapter:qstatesgraphs} is based on work with Alexander Hu~\cite{hu2022improved}.
We present a canonical form for expressing stabilizer states as graph states with local Clifford operations.
We also showcase an algorithm for turning any stabilizer state expressed this way into its canonical form.
Lastly, we expand on our work from Chapter~\ref{chapter:noncliffordgates} by proving exact conditions for the merging of a linear combination of two stabilizer states into one.

Chapter~\ref{chapter:qcodes-graphs} is based on work with Jonathan Lu and Peter Shor~\cite{khesin2023graphical,khesin2024universal}.
We extend our results from Chapter~\ref{chapter:qstatesgraphs} and derive a canonical form for all stabilizer codes using graphs with input, pivot, and output vertices.
These graphs allow for the efficient calculation of stabilizers and logical operators.
Additionally, we show results on the quality of a randomly selected graph code from a given family.
We give algorithms for converting between various representations of quantum codes as circuits, stabilizer tableaus, and graphs.
Most importantly, we unify the algorithms of distance calculation, stabilizer weight reduction, and decoding in error-correction into a playing a single game on a graph, which we call Quantum Lights Out.
We present results for sufficient conditions on when Quantum Lights Out can be played efficiently.
From this, we prove that we can bound properties of codes by bounding properties of their corresponding graphs, highlighting a promising avenue for future graph search.

Chapter~\ref{chapter:code-equivalence} is based on work with Alexander Li~\cite{khesin2024equivalence}.
We specialize the canonical form for codes from Chapter~\ref{chapter:qcodes-graphs} to CSS codes.
We present results relating graph properties to whether those graphs correspond to CSS codes, as well as what properties the resulting CSS codes might have.
We also analyze the classification of codes under a larger family of equivalence relations.

Appendix~\ref{app:magic-state-stab-rank} contains a discussion on optimizing the stabilizer rank of magic states as well as a description of several decompositions using graphs.
Appendix~\ref{app:proofs-psipq-table} contains a collection of proofs regarding the possible value of a term in the formula for merging two graphs.
Appendix~\ref{app:zx-calculus} contains an introduction to the ZX-calculus, a graphical language for representing and manipulating quantum systems which we use throughout this thesis.
Appendix~\ref{app:compiler} contains an efficient algorithm for turning any Clifford encoder into a graph code in canonical form.
Appendix~\ref{app:recursion} contains the derivation for a particular recursive formula, used to count the number of graphs in canonical form.
Appendix~\ref{app:QGV} contains a proof of the Quantum Gilbert-Varshamov bound, which we include for completeness.
Appendix~\ref{app:3-by-3-toric} contains a discussion of how derive a graph representation of a small toric code.
Appendix~\ref{app:convert-zxcf-to-circuit} contains an illustrated algorithm for how to convert a graph code into a quantum encoding circuit.
 \chapter{Non-Clifford Computation from Graphs}
\label{chapter:noncliffordgates}

\section{Introduction}

Quantum logic gates play a crucial role in quantum informatics.
These gates allow us to perform the elementary steps of quantum computation and produce some desirable final state.
While classical computations can all be carried out using a Turing machine, the physical reality of how classical circuits are implemented makes it useful to describe such circuits in terms of gates.
Specifically, the NAND gate behaves similarly to a transistor and is therefore a good atom out of which to build other gates.
The NAND gate is in fact universal, meaning all Boolean circuits can be constructed solely out of NAND gates.

It is harder to describe the concept of universality in the quantum case because quantum operations are continuous and not discrete.
Notably, if we limit ourselves to using a finite number of universal quantum gates, we will never be able to describe all possible unitary operations that can be performed, as we could only ever construct a countable number of circuits.
Universal quantum computation must be able to approximate unitary quantum operations with arbitrary precision.

As shown by \textcite{aharonov2003-universality}, a small set of gates, such as the Clifford and Toffoli gates, allow for universal quantum computation.
It is thus desirable to not only be able to classically simulate how these gates act on quantum states, but also to do it quickly and in a small amount of space.
Notably, it is worth asking how much time and space will this simulation take as a function of the number of qubits being simulated as well as of the number of operations performed.

There is a classification of certain operations on quantum states called the Clifford hierarchy.
We denote the operations at the lower levels, Pauli gates and Clifford gates, $\calC_1$ and $\calC_2$, respectively.
The Toffoli gate is in a larger set of gates called $\calC_3$.
We wish to be able to efficiently approximate universal circuits by simulating the action of $\calC_3$ gates on quantum states.
We can accomplish this by considering only a particular set of quantum states called stabilizer states, which can be represented efficiently.
While Clifford gates take stabilizer states to stabilizer states, $\calC_3$ gates take stabilizer states to linear combinations of two stabilizer states.
Additionally, being able to implement $\calC_3$ gates is of additional interest since they allow preparing quantum states known as magic states, which allow for the fault tolerant implementation of quantum circuits~\cite{bravyi2005-magicstate}.

In this chapter, we show one way to simulate this action by using a representation of quantum states using graphs, as well as how to reduce the complexity of the representation combining some of the resulting terms.
Informally, we give an explicit formula for expressing a set of $\calC_3$ gates using linear combinations of Pauli and Clifford gates.
By using a representation of stabilizer states with Clifford operations applied to graph states, we are able to apply the afore-mentioned formula for a $\calC_3$ gate to a stabilizer state and express it as a linear combination of several such states.
Next, we describe a procedure to merge these multiple states into just two stabilizer states using an explicit formula.
Our algorithm runs in $O(nd)$ time where $n$ is the number of qubits in the stabilizer state and $d$ is the largest degree of a vertex in the graph used to express the state.
This means that we can simulate the application of $k$ $\calC_3$ gates in time $O(2^knd)$, where $d$ is the largest degree of any graph encountered in the computation.
We believe these results have applications in universal quantum computation and computation of stabilizer rank.

In Section~\ref{sec:noncliffordops-cliffordhierarchy}, Section~\ref{sec:noncliffordops-stabrank}, and Section~\ref{sec:noncliffordops-graphstates} we introduce definitions of the Clifford hierarchy, stabilizer rank, and graph states, respectively.
In Section~\ref{sec:noncliffordops-graphmerging} we state our results and show examples of how to apply them in Section~\ref{sec:noncliffordops-examples}.
The proof of the main result is shown in Section~\ref{sec:noncliffordops-proof} and our graph merging algorithm's runtime is computed in Section~\ref{sec:noncliffordops-runtime}.

\section{The Clifford hierarchy}
\label{sec:noncliffordops-cliffordhierarchy}

We begin by introducing a family of sets of quantum operations known as the \textit{Clifford hierarchy}.
These are key to the study of quantum circuits.

We begin by defining a fundamental set of operations, the Pauli group.
Let $n$ be the number of qubits we are simulating.

\begin{definition}\label{def:pauli}
Let $\calC_1$, the \textit{Pauli group}, be the group generated by matrix products and tensor products of the unitary matrices \{$X$, $Y$, $Z$\}, where $X=\left(\begin{smallmatrix}0 & 1\\ 1 & 0\end{smallmatrix}\right)$, $Y=\left(\begin{smallmatrix}0 & -i\\ i & 0\end{smallmatrix}\right)$, and $Z=\left(\begin{smallmatrix}1 & 0\\ 0 & -1\end{smallmatrix}\right)$.
We also let $I=\left(\begin{smallmatrix}1 & 0\\ 0 & 1\end{smallmatrix}\right)$. $\calC_1$ is the first level of the Clifford hierarchy.
The Pauli group is also often denoted $\calP$.
\end{definition}

The matrices $X$, $Y$, and $Z$ are called Pauli matrices, and have many nice properties.

\begin{remark}\label{rem:pauli-cyclicphase}
The Pauli matrices obey $XY=iZ$, as well as similar expressions obtained by cyclically permuting $X$, $Y$, and $Z$.
This means the Pauli group contains matrices which differ from $X$, $Y$, $Z$, and $I$ by a phase, a scalar multiple by a complex unit, but since the quantum operations represented by those matrices ignore phase, we only need to consider the four matrices above.
Additionally, the matrices $X$, $Y$, and $Z$ are often denoted $\sigma_x$, $\sigma_y$, and $\sigma_z$ or $\sigma_1$, $\sigma_2$, and $\sigma_3$, respectively.
\end{remark}

In general, if we wish to denote a gate $U$ applied to a specific qubit $a$ or pair of qubits $a$ and $b$, we will write $U_a$ or $U_{a,b}$, respectively.
It is implied that the $n$-qubit operation is constructed by taking the tensor product of $U$ on qubit $a$ and identity matrices $I$ applied to all other qubits.
Furthermore, when writing tensor products of several matrices, we will usually omit the $\otimes$ sign for brevity.

Since unitary matrices represent quantum gates, it is useful to examine the conjugation of one matrix by another.
In other words, if a circuit contains the matrix $A$, we want to know whether the matrix $B$ commutes with it, and if it doesn't, what is the value of $C$ in the expression $AB=CA$.
Equivalently, we want to know what $B$ turns into if we push operation $A$ through $B$.
We find that $C=ABA^{-1}$.

It is now natural to examine the normalizer of $\calP$, the set of such matrices $C$ through which we can pull a Pauli matrix and be sure to produce a Pauli matrix.

\begin{definition}\label{def:clifford}
Let $\calC_2$, the Clifford group, be the normalizer of $\calC_1$, the set of operations $U$ where for any Pauli matrix $P$, the matrix $UPU^{-1}$ is Pauli.
This condition is often written as $U\calP U^{-1}=\calP$.
\end{definition}

As Gottesman shows in~\cite{gottesman1998-stabilizers}, $\calC_2$ is generated by the matrices $H$, $S$, and $CX$, which are Hadamard gate, $H=\frac1{\sqrt2}\left(\begin{smallmatrix}1&1\\1&-1\end{smallmatrix}\right)$, the phase gate $S=\left(\begin{smallmatrix}1&0\\0&i\end{smallmatrix}\right)$, and the controlled-$X$ gate, $CX=\left(\begin{smallmatrix}1&0&0&0\\0&1&0&0\\0&0&0&1\\0&0&1&0\end{smallmatrix}\right)$, respectively.
The $CX$ gate applies the $X$ gate to the target qubit if the control qubit is in the $\ket1$ state, and any controlled-Pauli gate can be used as the third generator.
When using subscripts to indicate the qubits to which the gate $CX$ is applied, we write the control qubit first. The above matrix shows $CX_{1,2}$.

\begin{example}
The $\textit{SWAP}$ gate is a $\calC_2$ gate formed by $CX_{1,2}\cdot CX_{2,1}\cdot CX_{1,2}$, which swaps the values of qubits 1 and 2 and is written as $\left(\begin{smallmatrix}1&0&0&0\\0&0&1&0\\0&1&0&0\\0&0&0&1
\end{smallmatrix}\right)$.
\end{example}

\begin{example}
The $CZ$ gate is a $\calC_2$ gate written as $\left(\begin{smallmatrix}1&0&0&0\\0&1&0&0\\0&0&1&0\\0&0&0&-1
\end{smallmatrix}\right)$.
Note that it is symmetric in terms of its control and target.
Stated mathematically, $\textit{SWAP}_{1,2}\cdot CZ_{1,2}\cdot\textit{SWAP}_{1,2}=CZ_{1,2}$.
This means that the diagonal $CZ$ gate commutes with all other diagonal gates, and also that we do not need to specify which of the gate's two qubits is the control.
Note that a controlled-$U$ gate, $CU$, is written in general as $CU=\left(\begin{smallmatrix}I&0\\0&U\end{smallmatrix}\right)$.
\end{example}

\begin{definition}\label{def:local-clifford}
We define a \textit{local Clifford (LC) operation} as any of the matrices generated by $H$ and $S$ through matrix products and tensor products.
The 24 matrices in the set $\langle H,S\rangle$, formed only by multiplying sequences of $H$ and $S$ matrices, form the set of single-qubit local Clifford operations.
\end{definition}

We can further extend this idea and ask ourselves which operators map Pauli gates to the Clifford group under conjugation.
This leads us to the following definition.

\begin{definition}\label{hierarchy}
For $k\in\mathbb{N}$, let $\calC_k$, the \textit{operators on level $k$ of the Clifford hierarchy}, be defined recursively as $\calC_k$ is the set of all unitary operations $U$ where $U\calP U^{-1}=U\calP U^\dagger\subseteq \calC_{k-1}\}$, with the base case being $\calP=\calC_1$. Note that the Clifford group is indeed $\calC_2$ under this definition. Furthermore, we note that for $k>2$, these operators do not form a group.
\end{definition}

As mentioned earlier, we can achieve universal quantum computation by just using Clifford gates and the Toffoli gate, which is in $\calC_3$~\cite{aharonov2003-universality}.

The particular gates in $\calC_3$ that we study are $T=\left(\begin{smallmatrix}1&0\\0&e^{\frac{i\pi}4}\end{smallmatrix}\right)$, $CH$, $CS$, $CCX$, $CCZ$, $\textit{CSWAP}$. These are the $\frac\pi8$ gate, the controlled-$H$ gate, the controlled-$S$ gate, the Toffoli or controlled-controlled-$X$ gate, the controlled-controlled-$Z$ gate, and the controlled-SWAP gate, respectively.
For ease of notation we denote the set of these six $\calC_3$ operators as $\calS$.

We note the following regarding controlled operations.
\begin{remark}
\label{remark:controlled-gate-convention}

To apply a controlled operation $U$, the matrix used is $\left(\begin{smallmatrix}I&0\\0&U\end{smallmatrix}\right)$. The operation $U$ is controlled by qubit 1 and applied to the remaining qubits. This operation is denoted $CU$, which can be nested if $U$ is already a controlled operation. Note that applying a control operation to a gate often raises its level in the Clifford hierarchy, although the exact conditions for when this is true are unclear. For certain special gates $U$, such as $Z$, $S$, and $T$, we use the convention that if the controlled operation $CU_{a,b}$ is applied to qubits $a$ and $b$, then when $b=a$ this operation is equal to $U_a$.
The gates $Z$, $S$, and $T$ are special in this regard since they can be viewed as a controlled global phase.
\end{remark}

\section{Stabilizer rank}
\label{sec:noncliffordops-stabrank}

When doing quantum simulations, it is convenient to start the simulation with a particular state, such as $\ket{0}^{\otimes n}$, and then proceed from there.
In fact there is a specific finite subset of all $n$-qubit states which is convenient to work with.
For this we introduce the stabilizer formalism.

\begin{definition}\label{def:stabstate}
We say that a state $\ket\psi$ is a \textit{stabilizer state} if there exists a set of $n$ independent Pauli operators $U\in\calC_1$ such that $U\ket\psi=\ket\psi$.
We say that the operators $U$ \textit{stabilize} $\ket\psi$.
\end{definition}

\begin{remark}
If $\ket\psi$ is a stabilizer state and $U\in\calC_2$, then $U\ket\psi$ will also be a stabilizer state. Furthermore, any stabilizer state can be expressed as $U\ket{0}^{\otimes n}$ for some $U\in\calC_2$.
\end{remark}

One of the things that make stabilizer states useful is that they can be represented efficiently.
An $n$-qubit state requires $2^n$ components to represent as a state vector, but an $n\times n$ \textit{stabilizer tableau} of Pauli matrices that identify the stabilizer state only requires $O(n^2)$ space.

If we wish to perform universal quantum computation, the natural step after defining stabilizer states is trying to describe the action of $\calC_3$ gates on these states.
Immediately, we see that if $\ket\psi$ is a stabilizer state and $U\in\calC_3$, then $U\ket\psi$ will not always be a stabilizer state, as that is only true for gates in $\calC_2$.
In fact, representing the action of several $\calC_3$ gates on $\ket\psi$ is of great interest as this lets us perform universal quantum computation, magic state distillation, and fault tolerant quantum computation~\cite{aharonov2003-universality,bravyi2005-magicstate}.
For now, we need to introduce new terminology to talk about the resulting states after applying $\calC_3$ gates on stabilizer states.

\begin{definition}\label{def:stabrank}
As defined by~\cite{bravyi2016-stabrank}, let the \textit{stabilizer rank} $\chi(\ket\psi)$ of an $n$-qubit state $\ket\psi$ be the smallest number $k$ of stabilizer states $\ket{\phi_i}$ needed to write $\ket\psi=\sum\limits_{i=1}^kc_i\ket{\phi_i}$.
Note that $\chi(\ket\psi)\geq1$ with equality iff $\ket\psi$ is a stabilizer state. Also we observe that $\chi(\ket\psi)\leq2^n$, since the standard basis vectors are all stabilizer states.
\end{definition}

We will later show that if $\ket\psi$ is a stabilizer state and $U \in \calS$ is one of the $\calC_3$ gates, then we can always write $U\ket\psi$ as a linear combination of no more than 2 stabilizer states.
In other words, $\chi(U\ket\psi)\leq2$.

This is a very important point, as this means that we can represent the action of a $\calC_3$ gate in a circuit as a linear combination of two stabilizer states.
Since matrix multiplication is additive, a second $\calC_3$ gate would act on both parts independently.
This means that we have reduced the problem of finding out how a series of $\calC_3$ gates acts on a stabilizer state, to learning the behaviour of just one.
The cost we pay for this is that our memory system must be able to store a linear combination of stabilizer states.

Presently, this means that applying $m$ $\calC_3$ gates will result in a linear combination of $2^m$ stabilizer states, which we would like to avoid.
What we focus on is detecting when a linear combination of two stabilizer states can be merged into one.
By applying such a merging operation after every application of a $\calC_3$ gate, we can massively reduce the exponential growth in the space required to represent our state.
However, with our current representation of stabilizer states, this merging operation would require us to identify when a linear combination of two stabilizer matrices can be merged and produce a stabilizer matrix of the sum.
We are not currently aware of a method that would allow us to do this, which is why we have to turn to a different representation of stabilizer states.

\section{Graph states}
\label{sec:noncliffordops-graphstates}

A lot of work in quantum cryptography has been done using codes based on graphs.
Similar techniques allow us to adapt graphs for our purposes.

\begin{definition}\label{graph}
A \textit{graph} $G=(V,E)$ consists of a set of vertices $V$ and a set of edges $E\subseteq V\times V$ connecting those vertices.
For now, we will always take $V$ to be the set of $n$ qubits numbered from $1$ to $n$.
By $N(a)$ we denote the neighbourhood of $a$, the set of all vertices $b$ such that $(a,b)\in E$.
We do not allow edges from a vertex to itself, so $a\not\in N(a)$.
\end{definition}

A useful set of stabilizer states that can be described by a graph is the set of graph states.

\begin{definition}\label{def:graphstate}
Let the \textit{graph state} $\ket G$ of an $n$-vertex graph $G$ be the state stabilized by the operators $g_v = X_v\prod\limits_{u\in N(v)}Z_u$, for all vertices $v$.
Note that $\ket G$ is a stabilizer state.
Additionally,
\begin{equation}
\ket G=\prod\limits_{(u,v)\in E}CZ_{u,v}\ket{+}^{\otimes n}\quad\text{ where }\ket+=\frac{\ket0+\ket1}{\sqrt2}.
\end{equation}
\end{definition}

However, not all stabilizer states can be expressed as graph states, which means we need to supplement our graphs with additional information.
As shown by \textcite{vandennest2004graphical}, any stabilizer state can be expressed as a graph state with local Clifford operations performed on the vertices.
This fact will be key to our use of graphs.
Now, any stabilizer state can be denoted by $U\ket G$ where $U$ is a local Clifford operation and $G$ is a graph.
Storing such a graph takes $O(1)$ additional space per qubit, as the number of local Clifford operations is fixed.
Storing the entire graph as an adjacency list takes $O(nd)$ space where $d$ is the average degree of the graph.

Now we examine when and how we can merge a linear combination of two graphs into one.

\section{Graph merging}
\label{sec:noncliffordops-graphmerging}

We begin with the statement of our main results and will then build up to its proof. Let $\calS$ be $\{ T, CS, CH, CCZ, CCX, \textit{CSWAP} \}$, a collection of $\calC_3$ gates.

\begin{theorem}\label{thm:graphsplit}
Let $C\in S$ be a $\calC_3$ operator and let state $\ket\psi$ be $U \ket{G}$ where $U$ is local Clifford and $G$ is an $n$-vertex graph with largest degree $d$.
Then we can find two states $U_1 \ket{G_1}$ and $U_2 \ket{G_2}$ such that $C \ket\psi = U_1 \ket{G_1} + U_2 \ket{G_2}$ in runtime $O(nd)$.
\end{theorem}

\begin{table}[ht]
\begin{center}
\bgroup
\def\arraystretch{1.5}
\begin{tabular}{|c|c|}\hline
    $\mathbf{\calC_3}$ \textbf{Gate} & \textbf{Decomposition}\\\hline
    T   & $\frac12(I+Z)\phantom{I}+\frac12e^{\frac{i \pi}{4}}(I-Z)\phantom{S}$\\\hline
    CS  & $\frac12(I+Z)I+\phantom{e^{\frac{i \pi}{4}}}\frac12(I-Z)S$ \\\hline
    CH  & $\phantom{H}\frac12(I+Z)I+\phantom{e^{\frac{i \pi}{4}}}\frac12(I-Z)H\phantom{S}$ \\\hline
    CCZ & $III-\frac14(I-Z)(I-Z)(I-Z)$ \\\hline
    CCX & $III-\frac14(I-Z)(I-Z)(I-X)$ \\\hline
    \textit{CSWAP} &$III - \frac14(I-Z)\otimes (II-ZZ)\cdot(II-XX)$ \\\hline
\end{tabular}
\egroup
\end{center}
\caption[Decompositions of select non-Clifford gates]{A decomposition of the $\calC_3$ gates we examined into a sum of two products of projectors and $\calC_2$ gates. Note that all tensor products have been removed for brevity except for that in the final entry, to highlight the distinction between matrix product and tensor product.}
\label{table:decomp}
\end{table}

Let us first attempt to apply the gates in $\calS$ as directly as we can.
Table~\ref{table:decomp} allows us to decompose gates in $\calS$ into the sum of two terms, each consisting of the products of Pauli projectors and $\calC_2$ gates.
This highlights the following key fact: if we can describe the action of a Pauli projector on a stabilizer state in terms of graph, we will be able to express a stabilizer decomposition of any gate in $\calS$.
This is nothing more than the outcome of a post-selected measurement, which can easily be described in terms of stabilizer generators.
In the stabilizer formalism, this is done by adding the measurement projector to a set of $n-1$ stabilizer generators that all commute with the measurement.

The goal of constructing the action of a Pauli projection on graphs is the first step in generalizing a graph merging operation, whereby certain linear combinations of two stabilizer states turn into one state.
In a later chapter, we will describe more precisely when two graphs can be merged this way.
With further study, it may be possible to describe conditions for merging and simplifying more complicated linear combinations, such as those of three or more states into fewer.

We will prove Theorem~\ref{thm:graphsplit} by means of the following result, generalizing Proposition~1 of~\cite{hein2004-graphpauli}.

\begin{theorem}\label{thm:graphmerge}
Let $v$ be a vertex, let $A$ be the set $N(v)\cup\{v\}$, and let $B$ be any set containing $v$. Then, for any integer $k$,
\begin{equation}\label{eq:graphmerge}
    \frac{1}{\sqrt{2}} \left(I + i^k \prod_{u \in B} Z_u \right) \ket{G} = H_v Z_v\left(\prod_{u,w \in A}CS_{u,w}^k\right)\left(\prod_{u\in A,\, w \in B}CZ_{u,w}\right) \ket{G}.
\end{equation}
(Note the convention from Remark~\ref{remark:controlled-gate-convention}.)
\end{theorem}

Note that the expression in Theorem~\ref{thm:graphmerge} may contain a $CS$ operation, which is a $\calC_3$ gate that we are not allowed to apply to a graph.
In the expression, the product of such terms iterates over $u\in A$ and $w\in A$.
Thus, for every pair $u\neq w$, we will apply both $CS_{u,w}$ and $CS_{w,u}$ which together produce $CZ_{u,w}$, the toggling of an edge.
Meanwhile, if $u=w$, then by the convention from Remark~\ref{remark:controlled-gate-convention}, $CS_{u,w}=CS_{u,u}=S_u$, which we are allowed to apply.

The proof of Theorem~\ref{thm:graphmerge} is technical and will be given in a later section.

Theorem~\ref{thm:graphmerge} shows us how to apply certain Pauli projectors and operations to graph states.
A general Pauli projector is a term proportional to $I + P$, where $P$ is a Pauli matrix.
Specifically, we need $P$ to be a matrix whose eigenvalues are $\pm1$.
In Equation~(\ref{eq:graphmerge}), we find that the operation on the left side will not be a projector if $k$ is odd.
However, we will find that our formula works in the case of odd $k$ as well.
If one is strictly interested in the case when the operation is a projector, the $i^k$ should be replaced with $(-1)^k$ and the $CS$ product should be replaced by the same product with $CZ$, which is equivalent to $\prod\limits_{u\in A}Z_u$.

The projectors have the nice property that conjugation by a local Clifford operation gives another projector. For example,
\begin{equation}
    H\cdot \frac{(I + Z)}2 \cdot H = \frac{I + X}2.
\end{equation}
Thus, for any local Clifford gate $U$ and projector $R$, we have
\begin{equation}
    R U \ket{\psi} = U R' \ket{\psi},
\end{equation}
where $R'$ is another projector.
This means that if we ever find ourselves applying a projector to a state with local Clifford operations, we can commute the projector through those operations and applied the modified projector to a graph state directly.

With the local Clifford operations out of the way, we merely need to compute the state $P' \ket{\psi}$. Splitting $P' = P_1 P_2 \dots P_k$ into a product of projectors, we will apply $P_k, P_{k-1}, \dots, P_1$ in order to $\ket\psi$.
Now we need to outline how to apply a single projector $P$ to $\ket\psi$.

If $P$ is a product of $Z$ gates, then we can directly apply the merge described in Theorem~\ref{thm:graphmerge}. Otherwise, we convert $Y = iXZ$ to replace all $Y$'s, and we can use the identity
\begin{equation}\label{eq:graphstab}
    X_v\ket{G}=\left(\prod\limits_{u\in N(v)}Z_u\right)\ket{G}
\end{equation}
to convert all $X$ gates to $Z$ gates on other qubits. This identity is valid because $X_v \prod\limits_{u\in N(v)}Z_u$ is a stabilizer for the graph state $\ket{G}$. This operation might cause $Z$ gates to appear on other qubits that still had $X$ gates that needed to be applied to the graph. In this case, we use the anti-commutativity of the Pauli matrices to fix their order, commuting the $Z$ through the $X$ and applying an overall phase of $-1$.
While global phases are usually ignored, they are important when working with linear combinations of stabilizer states when expressed as state vectors or any other representation which is sensitive to phases.

Repeating the above allows us to make sure there are only $Z$ gates left, and we can apply the merging algorithm of Theorem~\ref{thm:graphmerge} to the graph under the resulting projector. The application of $P_k$ is thus another stabilizer state, and then we can apply $P_{k-1}$, $P_1$, and so on. This proves most of the statement of Theorem~\ref{thm:graphsplit}.
The runtime will be analyzed in a later section.

\section{Examples}
\label{sec:noncliffordops-examples}

One important concern is that our algorithm behaves well with circuit identities. For example, applying two $CCX$ gates gives the identity, but since each $CCX$ would turn one stabilizer state into two, two $CCX$ gates should initially leave us with four stabilizer states. A robust algorithm would find a way to merge or otherwise annihilate these four stabilizer states such that only one remains. In the following example, we use the identity $CS \cdot CS = CZ$. Since $CZ$ is a Clifford operation, we should expect the output of the application of two $CS$ gates to be a single stabilizer state.
In the second example, we show that the stabilizer rank of the 2-qubit magic state $TT\ket{++}$ is 2 and not 4.

\subsection{Applying two controlled-S gates}

We showcase the example of applying two $CS$ to a stabilizer state.
We begin with a stabilizer state $\ket\psi = U\ket{G}$. Then

\begin{equation}
    CS \ket\psi = \frac{1}{2} \left( (I+Z)I + (I-Z)S \right) U\ket{G}.
\end{equation}
If we apply $CS$ again then we get
\begin{equation}
    CZ \ket\psi = \pars{\frac{1}{2} \left( (I+Z)I + (I-Z)S \right)}^2 U\ket{G}.
\end{equation}
Now, we use the identity $(I+Z)(I-Z) = 0$ and $(I\pm Z)^2 = 2(I\pm Z)$ to simplify this to
\begin{equation}
     CZ \ket\psi = \frac12\left( (I+Z)I + (I-Z)Z \right) U\ket{G}
\end{equation}
We should have four stabilizer states, but two of them (corresponding to the cross terms) collapsed to zero. We can now perform merges on $(I+Z) U\ket{G}$ and $(I-Z) U\ket{G}$.
Note that in Theorem~\ref{thm:graphmerge}, the two projectors are identical with the sole exception being the value of $k$ differing by 2.
The resulting graphs will thus be identical, since the right-hand sides of Equation~(\ref{eq:graphmerge}) will differ only by a factor of $\prod\limits_{u,w\in A} CS^2_{u,w}=\prod\limits_{u,w\in A} CZ_{u,w}=\prod\limits_{u\in A} Z_u$.
Thus, the output of the algorithm will be two stabilizer states $U_1 \ket{G'}$ and $U_2 \ket{G'}$, where $U_1 = U_2\prod\limits_{u \in A} Z_u$.

Thus, the operators differ by a Pauli, so they form a projector and we can merge $U_1 \ket{G'}$ and $U_2 \ket{G'}$ into a single final state. As expected, two $CS$ gates produce a single stabilizer state.

\subsection{Magic states}
We now show the application of two $T$ gates. Consider $T_1 T_2 \ket{++}$, where $\ket{++}$ is the graph state of the empty graph on 2 vertices. As expected, one application of the merging algorithm leads to two stabilizer states:
\begin{equation}
    T_1 \ket{++} = \frac{1}{\sqrt{2}} H_1 \ket{++} + \frac{1}{\sqrt{2}} e^{\frac{i \pi}{4}} H_1 Z_1 \ket{++}.
\end{equation}
Applying the algorithm again for the second $T$ gate, we get four stabilizer states:
\begin{equation}
    T_2 T_1 \ket{++}= \left(\frac{1}{2} H_1 H_2+
    \frac{1}{2} e^{\frac{i \pi}{4}} H_1 H_2 Z_2 +\frac{1}{2} e^{\frac{i \pi}{4}} H_1 H_2 Z_1 + \frac{i}{2} H_1 H_2 Z_1 Z_2 \right)\ket{++}.
\end{equation}
We now apply our merge on the first and fourth states, and the second and third states, to get
\begin{equation}
     T_2 T_1 \ket{++} = \frac{1}{\sqrt{2}} S_1 H_2 \ket{G} + \frac{1}{\sqrt{2}} e^{\frac{\pi i}{4}} H_1 Z_1 \ket{G},
\end{equation}
where $G$ is the complete graph on two vertices, so $\ket G$ is $CZ\ket{++}$. We obtain a decomposition into 2 stabilizer states, which is optimal since $T_1 T_2 \ket{++}$ is not a stabilizer state.

We return to the discussion of magic states in Appendix~\ref{app:magic-state-stab-rank}, which is connected to results introduced in the next chapter.

\section{Proof of merging formula}
\label{sec:noncliffordops-proof}

In this section, we give a proof of Equation~(\ref{eq:graphmerge}) from Theorem~\ref{thm:graphmerge}.
The outline for the proof structure is to show that the two states on either side of Equation~(\ref{eq:graphmerge}) have the same set of stabilizers and the same global phase.
The proof itself is quite technical.

Let $g_u$ be the canonical stabilizer generator of graph $G$ associated with vertex $u$ as given in Definition~\ref{def:graphstate}. We have that $g_u=X_u\prod\limits_{w\in N(u)} Z_w$.
Let $\sigma$ and $\zeta$ be shorthand for the $CS$ and $CZ$ products in Equation~(\ref{eq:graphmerge}).
We have that $\sigma=\prod\limits_{u,w\in A}CS^k_{u,w}$ and $\zeta=\prod\limits_{u\in A,\, w\in B}CZ_{u,w}$.
Let the unitary $U$ be the operation $U=H_vZ_v\sigma\zeta$.
We wish to show that
\begin{equation}
\frac{1}{\sqrt{2}}\pars{I+i^k\prod_{u\in B}Z_u}\ket G=U\ket G.
\end{equation}
Since any $g_u$ stabilizes $\ket G$, we have that $Ug_uU^\dagger$ stabilizes $U\ket G$.
It this remains to show that this new stabilizer $g'_u=Ug_uU^\dagger$ stabilizes the left-hand side of Equation~(\ref{eq:graphmerge}).

We first consider the case for $g'_v$.
For simplicity, we will denote an expression of the form $\prod\limits_{w\in S}Z_w$ as $Z_S$.
We have the following derivation.
Here, we use the commutation rules that $CZ_{u,w}X_u=X_uZ_wCZ_{u,w}$ as well as that $S^kX=X(iZ)^k$.
We also recall that $A=N(v)\cup\{v\}$.

\begin{align}
g'_v&=Ug_vU^\dagger\\
&=H_vZ_v\s\z X_v Z_{N(v)} \z\s^\dagger Z_vH_v\\
&=H_vZ_v\s Z_v X_v Z_A Z_B Z_{N(v)} Z_v \s^\dagger Z_vH_v\\
&=H_v S_v^k X_v Z_{N(v)}^k Z_v Z_B S_v^{-k} H_v\\
&=H_v X_v (iZ_v)^k Z_{N(v)}^k Z_v Z_B H_v\\
&=Z_v (iX_v)^k Z_{N(v)}^k Z_{B\setminus\set{v}} \\
&= Z_{N(v)}^k Z_{B} (iX_v)^k
\end{align}

Now, we need to check that $g_v$ stabilizes the state on other side of our equation, which we denote as the state $\ket\psi=\frac{1}{\sqrt{2}}\pars{I+i^kZ_B}\ket G$.
We have the following.

\begin{align}
g'_v\ket\psi&=\frac{1}{\sqrt{2}}Z_{N(v)}^k Z_{B} (iX_v)^k\pars{I+i^kZ_B}\ket G\\
&=\frac{1}{\sqrt{2}}Z_{N(v)}^k \pars{Z_{B}+(-1)^ki^kI}(iX_v)^k\ket G\\
&=\frac{1}{\sqrt{2}}Z_{N(v)}^k \pars{i^kZ_{B}+I}Z_{N(v)}^k\ket G\\
&=\frac{1}{\sqrt{2}}\pars{I+i^kZ_B}\ket G=\ket\psi
\end{align}

We now need to show the same is true for an arbitrary vertex $u\neq v$. We introduce two indicator functions $\mathds{1}_A$ and $\mathds{1}_B$ which are equal to 1 if $u\in A$ or $u\in B$, respectively, and are 0 otherwise.
We then have the following.

\begin{align}
g'_u&=Ug_uU^\dagger\\
&=H_vZ_v\s\z X_u Z_{N(u)} \z\s^\dagger Z_vH_v\\
&=H_v\s X_u (-1)^{\mathds{1}_A\mathds{1}_B} Z_A^{\mathds{1}_B} Z_B^{\mathds{1}_A} Z_{N(u)} \s^\dagger H_v\\
&=(-1)^{\mathds{1}_A\mathds{1}_B} H_v (S_u)^{\mathds{1}_Ak} X_u Z_{A\setminus\set{v}}^{\mathds{1}_Ak}Z_A^{\mathds{1}_B} Z_B^{\mathds{1}_A} Z_{N(u)} (S_u)^{-\mathds{1}_Ak} H_v\\
&=(-1)^{\mathds{1}_A\mathds{1}_B} H_v X_u (iZ_u)^{\mathds{1}_Ak}Z_{A\setminus\set{u}}^{\mathds{1}_Ak}Z_B^{\mathds{1}_A} Z_A^{\mathds{1}_B} Z_{N(u)} H_v\\
&=
(-1)^{\mathds{1}_A\mathds{1}_B}
i^{\mathds{1}_Ak} X_u (X_vZ_{A\setminus\set{v}})^{\mathds{1}_Ak}
(X_vZ_{B\setminus\set{v}})^{\mathds{1}_A}
(X_vZ_{A\setminus\set{v}})^{\mathds{1}_B}
(X_vZ_v)^{\mathds{1}_A}Z_{N(u)} \\
&=
i^{\mathds{1}_Ak} 
Z_B^{\mathds{1}_A}
(X_vZ_{N(v)})^{\mathds{1}_Ak+\mathds{1}_B}
X_uZ_{N(u)} \\
&=
(i^kZ_B)^{\mathds{1}_A}
g_v^{\mathds{1}_Ak+\mathds{1}_B}
g_u
\end{align}

All that remains for us to show is that $g'_u\ket\psi=\ket\psi$. We note that $g_u$ and $g_v$ are stabilizers of $\ket G$.

\begin{align}
g'_u\ket\psi&=
\frac{1}{\sqrt{2}}
(i^kZ_B)^{\mathds{1}_A}
g_v^{\mathds{1}_Ak+\mathds{1}_B}
g_u
\pars{I+i^kZ_B}\ket G\\
&=\frac{1}{\sqrt{2}}
(i^kZ_B)^{\mathds{1}_A}
g_v^{\mathds{1}_Ak+\mathds{1}_B}
\pars{I+i^k(-1)^{\mathds{1}_B}Z_B}g_u\ket G\\
&=\frac{1}{\sqrt{2}}
(i^kZ_B)^{\mathds{1}_A}
\pars{I+i^k(-1)^{\mathds{1}_Ak}Z_B}
g_v^{\mathds{1}_Ak+\mathds{1}_B}\ket G\\
&=\frac{1}{\sqrt{2}}
\pars{(i^kZ_B)^{\mathds{1}_A}
+\pars{i^kZ_B}^{1-\mathds{1}_A}}
\ket G\\
&=\frac{1}{\sqrt{2}}
\pars{I+i^kZ_B}
\ket G
=\ket\psi
\end{align}

This proves that both sides of Equation~(\ref{eq:graphmerge}) have the same stabilizers.
We now show that they have the same global phase.

Let $c = \frac{1}{\sqrt{2^n}}$. Let the co-vector $\bra{\mathbf{0}}$ be defined as $\bra{\mathbf{0}} = \bra{0}^{\otimes n}$. We apply $\bra{\mathbf{0}}$ to both sides of Equation~(\ref{eq:graphmerge}).
We also define $\bra{\mathbf{1_v}} = \bra{\mathbf{0}}X_v$.
On the left-hand side of the equation, the product of $Z$ terms acts on the $\bra{\mathbf{0}}$ trivially, since $\ket{0}$ is a $+1$ eigenstate of $Z_u$ for all $u$.
Furthermore, since
$\ket{G} = \prod\limits_{(u,w)\in E} CZ_{u,w} \ket{+}^{\otimes n}$
where $E$ is the edge set of $G$, and since $CZ_{u,w}$ stabilizes $\bra{\mathbf{0}}$ and $\bra{\mathbf{1_v}}$, we obtain the inner products below.
\begin{equation}
    \bra{\mathbf{0}}\prod_{(u,w)\in E} CZ_{u,w}\ket{+}^{\otimes n} = \bra{\mathbf0}\ket{+}^{\otimes n} = c
\end{equation}
\begin{equation}
    \bra{\mathbf{1_v}}\prod_{(u,w)\in E} CZ_{u,w}\ket{\mathbf+} = \bra{\mathbf{1_v}}\ket{+}^{\otimes n} = c
\end{equation}
Thus, $\braket{\mathbf0}{G} = \braket{\mathbf{1_v}}{G} = c$.
Since $B$ contains $v$, we get the following.
\begin{equation}
    \braopket{\mathbf{0}}{\frac{1}{\sqrt{2}} \left( I + i^k \prod_{u\in B} Z_u \right)}{G} = \frac{1 + i^k}{\sqrt{2}} \cdot c
\end{equation}
\begin{equation}
    \braopket{\mathbf{1_v}}{\frac{1}{\sqrt{2}} \left( I + i^k \prod_{u\in B} Z_u \right)}{G} = \frac{1 - i^k}{\sqrt{2}} \cdot c
\end{equation}
For the right-hand side, we note that $\ket{\mathbf0}$ is a $+1$ eigenstate of $Z$, $CZ$, and $CS$. The inner products with the diagonal terms in our equation are given by the following.
\begin{equation}
    \braopket{\mathbf0}{Z_v \prod_{u,w\in A} CS_{u,w}^k \prod_{u\in A,\, w \in B} CZ_{u,w}}{G} = c
\end{equation}
\begin{equation}
    \braopket{\mathbf{1_v}}{Z_v \prod_{u,w\in A} CS_{u,w}^k \prod_{u\in A,\, w \in B} CZ_{u,w}}{G} = i^k \cdot c
\end{equation}
This is since $\ket{\mathbf{1_v}}$ is a $-1$ eigenstate of $Z_v$, a $+i$ eigenstate of $S_v$, and a $+1$ all the other operators in the expression.
Note that there are two $Z_v$'s in the above product, of which one appears as $CZ_{v,v}$.
Similarly, $CS_{v,v}$ is $S_v$.
Thus we have the following.
\begin{align}
    &\braopket{\mathbf0}{H_vZ_v
    \pars{\prod_{u,w\in A} CS_{u,w}^k}
    \pars{\prod_{u\in A,\, w \in B} CZ_{u,w}}}{G}\nonumber\\
    =&\,\frac{\bra{\mathbf0}+\bra{\mathbf1_v}}{\sqrt2}\pars{Z_v \pars{\prod_{u,w\in A} CS_{u,w}^k }\pars{\prod_{u\in A,\, w \in B} CZ_{u,w}}}
    \ket{G} = \frac{1 + i^k}{\sqrt{2}} \cdot c\\
    &\braopket{\mathbf1}{H_vZ_v
    \pars{\prod_{u,w\in A} CS_{u,w}^k}
    \pars{\prod_{u\in A,\, w \in B} CZ_{u,w}}}{G}\nonumber\\
    =&\,\frac{\bra{\mathbf0}-\bra{\mathbf1_v}}{\sqrt2}\pars{Z_v \pars{\prod_{u,w\in A} CS_{u,w}^k }\pars{\prod_{u\in A,\, w \in B} CZ_{u,w}}}
    \ket{G} = \frac{1 - i^k}{\sqrt{2}} \cdot c.
\end{align}
In summary, we proved that applying $\bra{\mathbf0}$ or $\bra{\mathbf{1_v}}$ to both sides gives the same result. If $k$ is 0 or 2, one of these terms can equal 0, not allowing us to confirm that the phases are equal, but the expressions cannot be 0 simultaneously, which proves that the global phases are equal.
This proves the correctness of our algorithm.

\section{Runtime}
\label{sec:noncliffordops-runtime}

Consider the action of a $\calC_3$ gate $C\in\CS$ on stabilizer state $\ket\psi=U \ket{G}$. The full steps of our algorithm are as follows.
\begin{enumerate}[(1)]
    \item Express $C$ as $C=C_1 + C_2$, a sum of operations involving Pauli projectors and $\calC_2$ gates, as shown in Table~\ref{table:decomp}.
    
    \item For each term of the form $I+P$ in $C_1$ that we wish to apply to $\ket\psi$, first compute $P'=U^\dagger PU$. Now, we seek to compute $U(I+P')\ket G$.
    
    \item Convert all $X$'s and $Y$'s in $P'$ to $Z$'s on neighbouring vertices in $G$ through the graph stabilizers of $\ket G$.
    
    \item Apply the merging algorithm given by Equation~(\ref{eq:graphmerge}) to turn the expression $U(I+P')\ket G$ into a single graph state with local Cliffords.

    \item Repeat steps (2)-(4) for each $I+P$ remaining in $C_1$.

    \item Repeat steps (2)-(5) for $C_2$.
\end{enumerate}

The main contributor to the runtime are steps (3) and (4). If $n$ is the number of qubits and $d$ is the average degree of $\ket G$, we claim that our average runtime is $O(nd)$.
The only other step that takes non-constant time is step (2), which requires $O(n)$ time to compute $P'$.
Converting an $X$ or $Y$ to $Z$'s might need to be done for each of $n$ qubits, requiring $d$ updates on per qubit. Hence step (3) has $O(nd)$ runtime.
Lastly, applying the gates in the algorithm takes time $O(|A|\cdot|A|+|A|\cdot|B|)$.
However, on average we have $|A|=d$ and $|B|=n$, so we get $O(d^2+dn)$.
Since $d<n$, we have proven our runtime $O(nd)$.

\section{Conclusion}
In this chapter, we showed how to simulate $\calC_3$ circuits on stabilizer states. We use a decomposition of $\calC_3$ operators into a sum of projectors and a merging algorithm that combines two states. We showed the robustness of our algorithm with circuit identities and we also derived a decomposition of the magic state $TT\ket{++}$ into two stabilizer states.

Future work can examine how our algorithm performs when computing higher-qubit magic states. We can also look to derive more cases where merging two states is possible besides the ones we found.
Additionally, we can look at the implications of our $\calC_3$ circuit simulation algorithm for universal quantum computation.
Another future research direction is to examine the application of $\calC_3$ gates to stabilizer codes in addition to their effect on stabilizer states.
Lastly, we can consider ways to develop the graph formalism further and see what statements we can make about stabilizer states expressed as graphs with local Cliffords.
We investigate this in the next chapter.
 \chapter{Quantum States from Graphs}
\label{chapter:qstatesgraphs}

% TiKZ style file generated by TikZiT. You may edit this file manually,
% but some things (e.g. comments) may be overwritten. To be readable in
% TikZiT, the only non-comment lines must be of the form:
% \tikzstyle{NAME}=[PROPERTY LIST]

% Node styles
\tikzstyle{black dot}=[fill=black, draw=black, shape=circle]
\tikzstyle{white dot}=[fill=none, draw=black, shape=circle]

% Edge styles
\tikzstyle{black edge}=[-]

\section{Introduction}

As we saw in the previous chapter,
the stabilizer formalism has important applications to quantum error correction and fault tolerant quantum computation~\cite{gottesman1998-stabilizers, gottesman1997stabilizer} and to the classical simulation of quantum circuits~\cite{bravyi2019simulation,kerzner2021clifford,bravyi2016-stabrank,aaronson2004improved,gidney2021stim,qassim2021improved,peleg2022lower}.

The general classical simulation of quantum circuits to be exponentially slow, yet simulating stabilizer circuits can be done in polynomial time by the Gottesman-Knill Theorem~\cite{nielsen2002quantum}.
As we discussed in the previous chapter, the addition of certain non-Clifford gates makes stabilizer circuits universal, able to approximate any unitary operation arbitrarily well.
When applying such gates to a stabilizer state, the state becomes a linear combination of multiple stabilizer states, as shown in Section~\ref{sec:noncliffordops-graphmerging}.
Applying non-Clifford gates can be performed with by circuits consisting of Clifford gates, a magic state, and a post-selective measurement.
The circuit then becomes equivalent to a Clifford circuit applied to a certain non-stabilizer state, which is important since Clifford operations can be implemented fault tolerantly~\cite{nielsen2002quantum}.
By expressing this magic state as a linear combination of stabilizer states, any quantum circuit can be simulated~\cite{bravyi2019simulation}.
In order to do so by using as few stabilizer states as possible, it is important to study additive properties of stabilizer states, such as their linear dependence.
This chapter will explore ways in which we can represent this linear dependence using graphs.

In Chapter~\ref{chapter:noncliffordgates}, we saw how graph states can be used to represent the application of gates to stabilizer states.
Graph states, just like the stabilizer formalism, have many applications, such as in measurement based quantum computing~\cite{raussendorf2003measurement}, where they are prepared and inputted to the one-way quantum computer. Graph states are also used to construct codes in quantum error-correction~\cite{elliott2008stabilizer,hein2006entanglement}.

Furthermore, graph states can be extended to represent all stabilizer states by applying local Clifford operators to each of the qubits~\cite{vandennest2004graphical}. 
In fact, only a single layer of either Hadamard or phase gates needs to be applied to a graph state to result in an arbitrary stabilizer state~\cite{elliott2008graphical,elliott2009graphical}. 
By representing stabilizer states in terms of graph states and local Clifford operations, graph states were used to classically simulate stabilizer circuits in~\cite{anders2006fast}. Graph states can simulate certain classes of circuits an order of magnitude faster than stabilizer tableau methods can~\cite{anders2006fast,gidney2021stim}.

Informally, our main results in this chapter can be summarized as follows.
We derive a unique canonical form for expressing any stabilizer state as a graph state with certain local Clifford operations applied to the vertices.
By applying various restrictions to the set of allowable local Clifford operations, we achieve the uniqueness of our canonical form.
Additionally, we give various algorithms for working with stabilizer states in this graph formalism, such as computing inner products and manipulating them using Clifford gates.
We calculate the runtime of our algorithm and show a worst-case runtime of $O(nd^2)$, where $n$ in the number of qubits and $d$ is the degree of the graph.
We also compute exact conditions on when a linear combination of two stabilizer state yields a third, or alternatively, when three stabilizer states are linearly dependent.

In Section~\ref{sec:qstatesgraphs-formalisms}, we give key definitions for this chapter.
In the following three sections, we describe in detail three new contributions to the study of quantum states through graphs.

In Section~\ref{sec:qstatesgraphs-formalisms}, we show that graph states provide us with a new canonical form for stabilizer states, which is unique and phase sensitive.
Also, this canonical form can be directly converted to a canonical circuit that prepares stabilizer states from a computational basis state.
This circuit has only four blocks of Clifford gates, which is fewer than other canonical forms~\cite{garcia2012efficient}. 
This chapter extends upon work done in~\cite{elliott2008graphical,elliott2009graphical,elliott2008stabilizer}, which show how to simplify stabilizer graphs, graph states with at most $2$ layers of Hadamard and phase operators applied to them, to a reduced form.
We show how to simplify \textit{extended graph states}, a generalization of stabilizer graphs where each qubit can have an arbitrary local Clifford operator applied to it, to canonical form, and we also analyze the runtime of this compilation.

In Section~\ref{sec:qstatesgraphs-simulation}, we present a simpler algorithm for graph state simulation.
While the algorithm introduced in~\cite{anders2006fast} for applying controlled-Pauli $Z$ gates to extended graph states has many cases and does not provide insight about the updated state, we develop a formula that easily yields the updated state.
Using our formula, we prove the near-impossibility of improving the theoretical runtime of applying controlled-Phase $Z$ gates. 
Our results give further evidence for the fact that graph state simulation is most useful in circuits where neighbouring qubits have few interactions~\cite{anders2006fast,gidney2021stim}.

In Section~\ref{sec:qstatesgraphs-additiveprops}, we study the additive properties of stabilizer states by fully examining the special case of the linear dependence of three stabilizer states. We extend the work of \textcite{garcia2017geometry} by finding two new cases of linearly dependent triplets. Our characterization enables efficient algorithms for detecting linear dependence between three stabilizer states and for computing stabilizer states that are in the span of two given stabilizer states.
Additionally, we improve on the algorithm we saw in Chapter~\ref{chapter:noncliffordgates} for the merging of two stabilizer states related by a Pauli operator.

Appendix~\ref{app:magic-state-stab-rank} contains a discussion of improving upper bounds on the stabilizer rank of magic states as well as a representation of a few optimal decompositions in the graph formalism.
Appendix~\ref{app:proofs-psipq-table} contains proofs for the graph state transformation rules derived in this chapter.

\section{The stabilizer and graph formalisms}
\label{sec:qstatesgraphs-formalisms}

Let us start by refreshing various key definitions and notations used throughout the chapter.
As in Definition~\ref{def:pauli}, let a \textit{Pauli operator} $P$ on $n$ qubits be of the form $i^k \bigotimes\limits_{i=1}^{n}P_i$ where $k\in \{0,1,2,3\}$ and $P_i\in\{I,X,Y,Z\}$ is a Pauli matrix. The Pauli matrices are defined as $I\coloneq 
\begin{pmatrix}
1&0\\
0&1
\end{pmatrix}$, $X\coloneq
\begin{pmatrix}
0&1\\
1&0
\end{pmatrix}$, $Y\coloneq \begin{pmatrix}
0&-i\\
i&0
\end{pmatrix}$, and $Z\coloneq \begin{pmatrix}
1&0\\
0&-1
\end{pmatrix}$. Let the set of all Pauli operators be $\mathcal{P}$, the \textit{Pauli group}.

Recall that for some gate $U$, we let $CU_{a,b}$ denote the \textit{controlled-U gate} with control qubit $a$ and target qubit $b$. For example, $CX_{a,b}$ denotes the \textit{controlled-X} gate with control qubit $a$ and target $b$, and we define $CY_{a,b}$ and $CZ_{a,b}$ similarly. We place subscripts on single-qubit operators to turn them into $n$-qubit operators where that operator is applied to the particular qubit referred to by the subscript, and $n$ is contextual. For example, $Z_1$ would be the Pauli $Z$ gate on qubit $1$.

As in Definition~\ref{def:clifford}, let a \textit{Clifford operator} $C$ on $n$ qubits be a unitary operator on $2^n$ dimensional state space such that for all Pauli operators $P$ on $n$ qubits, $CPC^{\dagger}\in \mathcal{P}$. Let the set of all Clifford operators be $\mathcal{C}$, the \textit{Clifford group}, which is generated by the \textit{Hadamard gate} $H\coloneq \frac{1}{\sqrt{2}}\begin{pmatrix}
1&1\\
1&-1
\end{pmatrix}$, the \textit{phase gate} $S\coloneq \begin{pmatrix}
1&0\\
0&i
\end{pmatrix}$, and any controlled-Pauli gate~\cite{gottesman1998-stabilizers}.
Additionally, applying certain controlled gates that modify phases in the computational basis to the same pair of qubits, we use the gate without a control. In particular, we have $CZ_{v,v}=Z_v$ and $CS_{v,v}=S_v$.

In Definition~\ref{def:local-clifford} we defined the Clifford operators $C$ acting on a single qubit as \textit{local Clifford operators}, and these operators are generated by $H$ and $S$.

We introduce the shorthand $[n]\coloneq\{1,2,\dots,n\}$.
Recall from Definition~\ref{def:stabstate} that an $n$-qubit state $\ket{\psi}$ is a \textit{stabilizer state} if there exists a set of $n$ commuting independent Pauli operators, $\{g_1,g_2,\dots,g_n\}$, such that for all $i\in [n]$, $g_i^2=I$ and $g_i\ket{\psi}=\ket{\psi}$. We call the operators $g_i$ \textit{stabilizers}. A stabilizer state $\ket{\psi}$ is equivalently defined as a state resulting from the action of a Clifford operator $C$ on a computational basis state~\cite{aaronson2004improved}.

For a graph $G$, we let $E(G)$ refer to the set of edges of $G$ and $V= [n]$, where vertex $i$ and qubit $i$ are synonymous. We assume our graphs are undirected and do not have loop edges or multiple edges between two qubits. Let $N(i)$ be the set of neighbours of $i$ in $G$ not including $i$, where $G$ is contextual.
In this chapter, we will make use of several new graph operations.

\begin{definition}
\label{def:graph-local-complementation-vertex}
The \textit{local complementation about a vertex} (LCV) of a graph $G$ at qubit $v$, $L_v(G)$, be a graph with the same number of vertices as $G$. The edges of $G$ and $L_v(G)$ are identical with the following exception. For each unordered pair of vertices $u,w\in N(v)$, the edge $(u,w)$ is present in $L_v(G)$ if and only if the edge $(u,w)$ is not in $G$.
\end{definition}

\begin{definition}
\label{def:graph-local-complementation-edge}
The \textit{local complementation about an edge} (LCE) of a graph $G$ is defined as follows.
If $(u,v)$ is an edge in graph $G$, the local complementation about the edge $(u,v)$ is constructed from $G$ by toggling several sets of edges.
Specifically, we consider the three sets of vertices $N(u)\setminus(N(v)\cup\set{v})$, $N(v)\setminus(N(u)\cup\set{u})$, and $N(u)\cap N(v)$.
For any pair of vertices $(p,q)$, if $p$ and $q$ each belong to a different set out of the three above, the edge $(p,q)$ is toggled in $G$.
Additionally, the adjacencies of $u$ and $v$ are swapped, meaning a vertex $w$ is connected to $u$ in the resulting graph if and only if it was connected to $v$ in $G$ and vice versa.
We note that the resulting graph is equal to each of $L_u(L_v(L_u(G)))$ and $L_v(L_u(L_v(G)))$.
\end{definition}

As defined in Definition~\ref{def:graphstate}, the \textit{graph state} of a graph $G$, $\ket{G}$, is the stabilizer state with stabilizers $g_v\coloneq X_v\prod\limits_{u\in N(v)}Z_u$ for $v\in[n]$~\cite{vandennest2004graphical}. 
An equivalent definition of graph states is
\begin{equation}
    \ket{G}\coloneq \left(\prod\limits_{(u,v)\in E(G)}\hspace{-2ex}CZ_{u,v}\right)\ket{+}^{\otimes n},
\end{equation}
where $\ket{+}\coloneq \frac{1}{\sqrt{2}}(\ket{0}+\ket{1})$~\cite{vandennest2004graphical}. When $\ket{G}$ is expressed as a state vector, the global phase is fixed by assuming the amplitude of $\ket{0}^{\otimes n}$ is positive and real.

We define our terminology for stabilizer states represented as applications of local Clifford operators to graph states, which is enabled by a theorem proved in~\cite{vandennest2004graphical}.
\begin{definition}
An \textit{extended graph state} is a stabilizer state written as a graph state with local Clifford operators applied to it in the form $C\ket{G}$ where $C$ is a tensor product of local Clifford operators.
\end{definition}
This definition formalizes the application of Clifford operations to graph states used in Chapter~\ref{chapter:noncliffordgates}.
Lastly, we let the \textit{support} of a quantum state $\ket{\psi}$ be the number of non-zero amplitudes it has when written as a state vector, and let the \textit{support set} be the set of vectors corresponding to the computational basis states with non-zero amplitudes in $\ket{\psi}$. 

These definitions enable us to begin examining stabilizer states from the perspective of graph states.

\section{Canonical forms for stabilizer states}
\label{sec:qstatesgraphs-canonforms}

When expressing a stabilizer state as a set of stabilizers, we have many choices for how to pick and order our set of stabilizer generators.
This means that the task of deciding whether two stabilizer states are equal by considering their tableaus can be difficult to do by hand.
Some attempts have been made to find a more canonical form for stabilizer presentation to make this task easier.

\subsection{Canonical generator matrix}

The binary representation of the stabilizer formalism~\cite{nielsen2002quantum} associates a binary vector with each Pauli operator generator. 
The generators of an $n$-qubit stabilizer state are stored in an $n\times 2n$ generator matrix. 
The rows of the generator matrix are linearly independent, and a shifted inner product can be defined so that it is $0$ for all pairs of distinct rows of the generator matrix.
Swapping rows corresponds to swapping generators, adding a row to another corresponds to multiplying generators, and switching columns corresponds to swapping qubits. 
These operations can transform a generator matrix into a \textit{canonical form},
\begin{align}
\left(
    \begin{array}{cc|cc}
    I & A & B & 0\\
    0 & 0 & A^T & I
    \end{array}\right),
\end{align}
where $B$ is a symmetric matrix.
However, this canonical form is not unique because of the freedom in choosing how to swap qubits. 
Furthermore, this canonical form can be converted into the \textit{reduced form} for extended graph states~\cite{elliott2008graphical}.
\begin{definition}
Let an extended graph state $C\ket{G}$ be in \textit{reduced form} if there exist $n$-tuples $c\coloneq (c_1,\dots c_n)$ and $z\coloneq (z_1,\dots z_n)$ with $c_v\in \{ I,S,H\}$ and $z_v\in \{I,Z\}$ such that $C=\bigotimes\limits_{v=1}^nc_vz_v$, and for all $(u,v)\in E(G)$, either $c_v\neq H$ or $c_u\neq H$.
\end{definition}
The reduced form provides an elegant graphical representation of stabilizer states, but multiple extended graph states in reduced form can refer to the same quantum state.
In this chapter we extend the conditions of the reduced form to find a description of a truly canonical form, meaning that each stabilizer state has a \textit{unique} representation in this form.

\subsection{A unique canonical form}

To define our canonical form, we extend the definition of the reduced form.
\begin{definition}
\label{def:state-canonical-form}

Let an extended graph state in reduced form be in \textit{canonical form} if for all $(u,v)\in E(G)$ such that $c_v=H$, we have $v<u$.
\end{definition}

In subsequent chapters, we will refer to extended graph states expressed in this canonical form as diagrams in HK form, or HK diagrams.

\begin{figure}[h]
    \centering
    \begin{tikzpicture}
	    \begin{pgfonlayer}{nodelayer}
    		\node [label={[label distance=-0.1cm]30: $SZ$ },style=white dot] (0) at (0.5, 0.5) {$3$};
    		\node [style=white dot] (1) at (-2, 1.5) {$4$};
    		\node [label={[label distance=-0.1cm]30: $Z$ },style=white dot] (2) at (2.75, 2.75) {$5$};
    		\node [label={[label distance=-0.1cm]30: $S$ },style=white dot] (3) at (-0.1, 3.5) {$6$};
    		\node [label={[label distance=-0.1cm]15: $H$ },style=white dot] (4) at (-1, -1.25) {$1$};
    		\node [label={[label distance=-0.1cm]30: $H$ },style=white dot] (5) at (2.75, -0.25) {$2$};
    		\node [style=none] (6) at (-1.25, 3.75) {};
    		\node [label={[label distance=-0.1cm]30: $HZ$ },style=white dot] (7) at (-2.0, 3.9) {$7$};
    	\end{pgfonlayer}
    	\begin{pgfonlayer}{edgelayer}
    		\draw (1) to (0);
    		\draw (0) to (3);
    		\draw (5) to (0);
    		\draw (2) to (5);
    		\draw (4) to (0);
    		\draw (4) to (3);
	    \end{pgfonlayer}
    \end{tikzpicture}
    \caption[A state in HK form]{An illustration of the stabilizer state $\ket\psi$ in canonical form. Here, $\ket{\psi}=H_1H_2S_3Z_3Z_5S_6H_7Z_7\ket{G}$, where $G$ is a 7-vertex graph with the edges (1,~3), (1,~6), (2,~3), (2,~5), (3,~4), and (3,~6). The vertices are ordered by their vertical position in the figure. Note that no $H$ operations are on vertices adjacent to lower-numbered vertices.}
    \label{fig:statecanonicalform-example}
\end{figure}
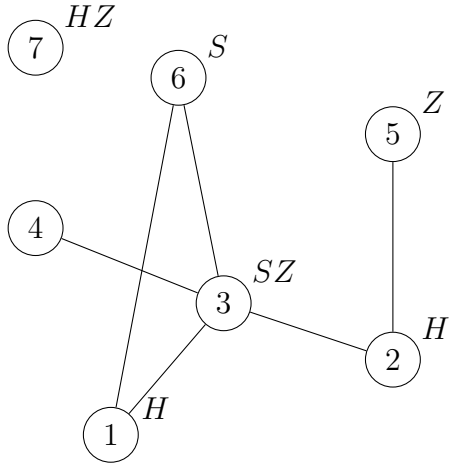

An example of a state in canonical form is shown in Figure~\ref{fig:statecanonicalform-example}.
The following result relates the number of $H$'s to the support and helps us prove the canonical form is unique.

\begin{lemma}\label{lemma:state-amplitude}
Let $\ket{\psi}=\pars{\bigotimes\limits_{v=1}^nc_vz_v}\ket{G}$ be in reduced form. Let $k$ be the number of $c_i$ that are equal to $H$. Then the support of $\ket{\psi}$ is $2^{n-k}$.
\end{lemma}
\begin{proof}
Let $A\coloneq \{v\in[n]|c_v=H\}$, where $k=|A|$. 
We use the identity $H_v\cdot CZ_{u,v}=CX_{u,v}\cdot H_v$. 
We also define single-qubit Pauli operators $p_i$ as $p_i\coloneq c_iz_ic_i^{\dagger}$. Then
\begin{multline}
\label{eq:state-amplitude-proof}
\pars{\bigotimes_{v=1}^nc_vz_v}\ket{G}=
\pars{\bigotimes_{v=1}^np_v}
\pars{\prod_{v\in [n]\setminus A}(c_v)_v}
\pars{\prod_{v\in A}H_v}\ket{G}\\
=\pars{\bigotimes_{v=1}^np_v}
\pars{\prod_{v\in [n]\setminus A}(c_v)_v}
\pars{\prod_{\substack{(u,v)\in E(G)\\u,v\not\in A}}CZ_{u,v}}
\pars{\prod_{v\in A}
\pars{H_v\prod_{u\in N(v)}CZ_{u,v}}} 
\ket{+}^{\otimes n}\\
=\pars{\bigotimes_{v=1}^np_v}
\pars{\prod_{v\in [n]\setminus A}(c_v)_v}
\pars{\prod_{\substack{(u,v)\in E(G)\\u,v\not\in A}}CZ_{u,v}}
\pars{\prod_{\substack{v\in A\\ u\in N(v)}}CX_{u,v}}
\pars{\prod_{v\in A}H_v}\ket{+}^{\otimes n}.
\end{multline}

$\prod\limits_{i\in A}H_i\ket{+}^{\otimes n}$ has support $2^{n-k}$ because it consists of a tensor product of $k$ $\ket{0}$'s and $n-k$ $\ket{+}$'s. 
The rest of Equation~(\ref{eq:state-amplitude-proof}) is a product of phase operators, Pauli operators, and controlled-Pauli operators, which do not change the support of $\ket{\psi}$.
\end{proof}
The main advantage of our canonical form is that it uniquely represents a stabilizer state.
\begin{theorem}
\label{thm:canonicalform-state}
If $\ket{\psi}=\ket{\psi'}$ up to global phase, and $\ket{\psi}\coloneq \bigotimes\limits_{v=1}^nc_vz_v\ket{G}$ and $\ket{\psi'}\coloneq \bigotimes\limits_{v=1}^nc_v'z_v'\ket{G'}$ are both in canonical form, then $G=G'$, $c=c'$, and $z=z'$.
The representation of a state in canonical form is unique.
\end{theorem}
\begin{proof}
Let $A\coloneq \{v\in[n]|c_v=H\}$ and $A'\coloneq \{v\in [n]|c_v'=H\}$. The supports of $\ket{\psi}$ and $\ket{\psi'}$ are equal, so by Lemma~\ref{lemma:state-amplitude}, $|A|=|A'|$.
Let $k\coloneq |A|=|A'|$. Now let $A=\{a_1,a_2,\dots a_k\}$ and $A'=\{a_1',a_2',\dots,a_k'\}$ where $a_1<a_2<\dots<a_k$, $a_1'<a_2'<\dots<a_k'$.
Having defined these quantities, we will need several important definitions to continue the proof.

\begin{definition}
For a binary string $s$ and a subset $B\subseteq [n]$ of the same size, let $\ket{s}\bra{s}_{B}$ be a $n$-qubit projector onto the subspace of $n$-qubit state space spanned by the basis of computational basis states that agree with $s$ on the qubits in $B$. 
We can think of $\ket{s}\bra{s}_{B}$ as stretching the bits in $s$ out in a $n$ dimensional vector to occupy the slots corresponding to qubits in $B$, and
we let $s_v$ denote the bit in $s$ in the slot corresponding to qubit $v$. 
\end{definition}

\begin{definition}
For a single-qubit state $\ket{\varphi}$ and a subset $B\subseteq [n]$, let $\ket{\varphi}_B$ be a tensor product of $|B|$ $\ket{\varphi}$'s, in the registers corresponding to the qubits in $B$.
\end{definition}
Now, suppose for the sake of contradiction $a_1<a_1'$. We will apply projectors of the form $\ket{s}\bra{s}_{[n]\setminus A}$ and $\ket{s}\bra{s}_{[n]\setminus A'}$ to $\ket{\psi}$ and $\ket{\psi'}$ to derive a contradiction. 
Letting $Q\coloneq \ket{s}\bra{s}_{[n]\setminus A}$, we write

\begin{align}
    Q\pars{\bigotimes_{v=1}^nc_vz_v}\ket{G}
    =\pars{\bigotimes_{v=1}^nc_vz_v}
    \pars{\prod_{(u,v)\in E(G)}CZ_{u,v}}
    Q\ket{+}^{\otimes n}\label{eq:canonical-projection-1}\\
    =\frac{1}{\sqrt{2^{n-k}}}
    \pars{\bigotimes_{v=1}^nc_vz_v}
    \pars{\prod_{(u,v)\in E(G)}CZ_{u,v}}
    \ket{+}_A\otimes \ket{s} \label{eq:canonical-projection-2}\\
    =\frac{1}{\sqrt{2^{n-|A|}}}
    \pars{\bigotimes_{v=1}^nc_vz_v}
    \pars{\prod_{(u,v)\in E(G)}Z_u^{s_v}}
    \ket{+}_A\otimes \ket{s}\label{eq:canonical-projection-3}\\
    =\frac{1}{\sqrt{2^{n-|A|}}}
    \pars{\bigotimes_{v=1}^np_v}
    \pars{\prod_{v\in [n]\setminus A}(c_v)_v}
    \ket{0}_{A}\otimes \ket{s}\label{eq:canonical-projection-4}.
\end{align}

Here, Equation~(\ref{eq:canonical-projection-1}) uses the fact that $Q$ commutes with diagonal operators such as $S$, $Z$, and $CZ$.
Equation~(\ref{eq:canonical-projection-2}) applies $Q$ to the state vector of supported equally on all computational basis states, which sets the bits in $[n]\setminus A$ to match $s$.
Equation~(\ref{eq:canonical-projection-3}) follows from assuming without loss of generality that $v\not\in A$ due to there being no edges in $G$ 
between qubits in $A$, so at least one of $u$ and $v$ is not in $A$.
Equation~(\ref{eq:canonical-projection-4}) is derived by conjugating the $Z$ operators to form Pauli operators $p_i$.
We then observe the following.

\begin{lemma}
\label{lemma:projbasisstate}
For all binary strings $s$ of length $n-k$, $\ket{s}\bra{s}_{[n]\setminus A}\ket{\psi}$ and $\ket{s}\bra{s}_{[n]\setminus A'}\ket{\psi'}$ are computational basis states.
\end{lemma}

\begin{proof}
By Equation~(\ref{eq:canonical-projection-4}), $\ket{s}\bra{s}_{[n]\setminus A}\ket{\psi}$ is a single computational basis state because it consists of Pauli 
and phase operators applied to a computational basis state, and similarly $\ket{s}\bra{s}_{[n]\setminus A'}\ket{\psi'}$ is as well.
\end{proof}
We consider $Q_1\coloneq \ket{s}\bra{s}_{[n]\setminus A'}$, where $s_{a_1}=1$ and the rest of the $s_v$ equal $0$, note that by our assumption that $a_1<a_1'$, we have $a_1\in\pars{[b]\setminus A'}$.
If we also set $Q_2\coloneq \ket{0}\bra{0}_{[n]\setminus A'}$, we have the following claim.

\begin{claim}
\label{claim:projectionbasis}
$Q_1\ket{\psi'}$ and $Q_2\ket{\psi'}$ are computational basis states that differ only in qubit $a_1$.
\end{claim}

\begin{proof}
By Equation~(\ref{eq:canonical-projection-3}),
\begin{align}
    Q_1\ket{\psi'}&=
    \frac{1}{\sqrt{2^{n-k}}}
    \pars{\bigotimes_{v=1}^nc_v'z_v'}
    \pars{\prod_{(u,v)\in E(G')}Z_u^{s_v}}
    \ket{+}_{A'}\otimes \ket{s}\nonumber\\
    &=\frac{1}{\sqrt{2^{n-k}}}
    \pars{\bigotimes_{v=1}^nc_v'z_v'}
    \pars{\prod_{u\in N_{G'}(a_1)}Z_u}
    \ket{+}_{A'}\otimes \ket{s}\nonumber\\
    &=\frac{1}{\sqrt{2^{n-k}}}
    \pars{\bigotimes_{v=1}^nc_v'z_v'}
    \ket{+}_{A'}\otimes \ket{s}\nonumber\\
    &=\frac{1}{\sqrt{2^{n-k}}}
    \pars{\bigotimes_{v=1}^np_v'}
    \pars{\prod_{v\in [n]\setminus A'}(c_v')_v}
    \ket{0}_{A'}\otimes \ket{s}.
\end{align}
Here, the third step follows by the fact that no neighbours of $a_1$ in $G'$ can be in $A'$ since $a_1<a_1'$ and $\ket{\psi'}$ is in canonical form, so vertices where $c_v=H$ cannot be connected to lower numbered vertices such as $a_1$.
Thus, $Z_u$ must act on the register $\ket s$, which is equal to $\ket0$ on all qubits other than $a_1$.
Again, the final expression must be a computational basis state as the state is only transformed by Pauli and phase operators.

We also have
\begin{align}
    Q_2\ket{\psi'}
    &=\frac{1}{\sqrt{2^{n-k}}}
    \pars{\bigotimes_{v=1}^nc_v'z_v'}
    \pars{\prod_{(u,v)\in E(G')}Z_u^{0}}
    \ket{+}_{A'}\otimes \ket{0}_{[n]\setminus A'}\nonumber\\
    &=\frac{1}{\sqrt{2^{n-k}}}
    \pars{\bigotimes_{v=1}^nc_v'z_v'}
    \ket{+}_{A'}\otimes \ket{0}_{[n]\setminus A'}\nonumber\\
    &=\frac{1}{\sqrt{2^{n-k}}}
    \pars{\bigotimes_{v=1}^np_v'}
    \pars{\prod_{v\in [n]\setminus A'}(c_v')_v}
    \ket{0}_{A'}\otimes \ket{0}_{[n]\setminus A'}.
\end{align}
Once again, this must be a computational basis state.
Since the support of $s$ only contains a single qubit, we have proven our claim that $Q_1\ket{\psi'}$ and $Q_2\ket{\psi'}$ are basis-states differing only in qubit $a_1$.
\end{proof}

Let $s'$ be the unique binary string such that 
    $\ket{s'}\bra{s'}_{[n]\setminus A}Q_1\ket{\psi'}=Q_1\ket{\psi'}$. Equivalently, $s'$ is the string of bits of the basis state $Q_1\ket{\psi'}$ on $[n]\setminus A$.
By Claim~\ref{claim:projectionbasis}, $\ket{s'}\bra{s'}_{[n]\setminus A}Q_2\ket{\psi'}=Q_2\ket{\psi'}$ since $a_1\notin[n]\setminus A$. 
Then, the support of $\ket{s'}\bra{s'}_{[n]\setminus A}\ket{\psi'}$ is at least 2, with at least one term coming from each of $Q_1$ and $Q_2$. 
However, the support of $\ket{s'}\bra{s'}_{[n]\setminus A}\ket{\psi}$ is $1$ by Lemma~\ref{lemma:projbasisstate}. Since $\ket{\psi}=\ket{\psi'}$, this produces the desired contradiction. 

This has merely shown that $a_1=a_1'$.
We cancel $H_{a_1}$ from both sides and repeat this procedure.
Each time we do this, we reduce $k$ by 1 until $k$ is 0. 
Now, for all $v$, we have $c_vz_v\in\{I,S,Z,SZ\}.$ If $c_v\neq c_v'$ or $z_v\neq z_v'$ for some $v$, then the amplitudes of $\ket{0}^{\otimes(v-1)}\otimes \ket{1}\otimes\ket{0}^{\otimes(n-j)}$ in $\ket{\psi}$ and $\ket{\psi'}$ would differ by some positive power of $i$. Therefore, $c=c'$ and $z=z'$, and we have $\ket{G}=\ket{G'}$. If $(u,v)\in E(G)$ but $(u,v)\not\in E(G')$ or vice versa, the amplitudes of $\ket{1}_{\{i,j\}}\otimes \ket{0}_{[n]\setminus\{i,j\}}$ would differ. Since $\ket\psi=\ket{\psi'}$, we also have that $G=G'$.
This completes the proof of the theorem stating the uniqueness of our canonical form.
\end{proof}

We show that any stabilizer state can be expressed in our canonical form by a counting argument. 
In~\cite{aaronson2004improved}, it is proven that the number of $n$-qubit stabilizer states is

\begin{equation}
    2^n\prod_{k=1}^n(2^{k}+1).
\end{equation}

Since our canonical form is unique, it suffices to show the following result.

\begin{theorem}
There are $2^n\prod\limits_{k=1}^n(2^{k}+1)$ $n$-qubit extended graph states in canonical form.
\end{theorem}
\begin{proof}
We wish to count all possible $c,z,$ and $G$ such that $c_v\in \{I,S,H\}$ and $z_v\in \{I,Z\}$ for all $v\in [n]$ and whenever $c_v=H$, all the edges in $G$ incident to $v$ connect to higher numbered qubits.
For each qubit $k$, we choose $c_k,z_k$, and all the edges of the form $(v,k)$ where $v<k$.
If we decide to select none of the $(v,k)$ to add to $E(G)$, then there are no restrictions on $c_k$ and $z_k$, yielding $6$ possibilities. 
Otherwise, the only restriction is $c_k\neq H$, and there are $2^{k-1}-1$ possible choices for the edges, yielding $4(2^{k-1}-1)$ possibilities. 
In total, there are $2^{k+1}+2$ ways to choose $c_k,z_k$, and all the edges of the form $(v,k)$ where $v<k$, and doing so for each $1\le k\le n$ yields all possible extended graph state in canonical form.
Thus, there are $\prod\limits_{k=1}^n(2^{k+1}+2)=2^n\prod\limits_{k=1}^n(2^{k}+1)$ $n$-qubit extended graph states in canonical form.
\end{proof}
\subsection{Simplifying extended graph states}
\label{sec:qstatesgraphs-simplification-algorithm}

We demonstrate how to simplify extended graph states to canonical form. 
As in~\cite{elliott2008graphical}, we repeatedly apply two transformation rules.
The first rule relates $\ket{L_v(G)}$ to $\ket{G}$.

\begin{theorem}[\textcite{hein2004-graphpauli}]
\label{thm:vertex-local-complementation}

For any graph state $\ket{G}$ and qubit $v$,
\begin{equation}
\label{eq:vertex-local-complementation}
    \ket{G}=H_vS_v^{\dagger}H_v\pars{\prod_{u\in N(v)}S_u}\ket{L_v(G)}.
\end{equation}
\end{theorem}

The second transformation, discovered by~\cite{elliott2008graphical}, allows us to simplify extended graph states to reduced form by eliminating pairs of $H$'s applied to connected qubits and also allows us to simplify to canonical form by sliding remaining $H$'s down to lower-numbered qubits.

\begin{theorem}[\textcite{elliott2008graphical}]
\label{thm:edge-local-complementation}
Let $(u,v)$ be an edge in $E(G)$. Let $A=N(u)\cup \{u\}$ and $B=N(v)\cup \{v\}$. Then,
\begin{equation}
\label{eq:edge-local-complementation}
    H_uH_v\ket{G}=Z_uZ_v\pars{\prod_{p\in A,q\in B}CZ_{p,q}}\ket{G}.
\end{equation}
\end{theorem}

\begin{remark}
Note that the transformation applied to the graph $G$ in Equation~(\ref{eq:edge-local-complementation}) is precisely a local complementation about the edge $(u,v)$.
Specifically, this refers to the edges toggled with $CZ$ gates and not the factors of $Z_w$ from a term of the form $CZ_{w,w}$ or the other local Clifford operations in the expression.

Additionally, as discussed in Section~\ref{sec:noncliffordops-graphmerging}, the action on a graph of a product of the form $\prod\limits_{p,q\in R}CS_{p,q}$ will precisely toggle every edge in the set $R$, in addition to apply several local operations.
This means that if we set $R=N(v)$ to be a neighbourhood of a vertex $v$, we can use the $CS$ product to perform a local complementation about the vertex $v$.
This allows us to express Equation~(\ref{eq:vertex-local-complementation}) in terms of $CS$ gates.
\end{remark}

The runtime of the graph state simplification algorithm, though cubic in the worst case, can be much quicker.

\begin{theorem}
\label{thm:canonical-algorithm}
There exists an algorithm to simplify an arbitrary extended graph state, $\ket{\psi}\coloneq \pars{\bigotimes\limits_{v=1}^nC_v}\ket{G}$, into canonical form, that runs in $O(nd^2)$, where $d$ is the maximum degree in $G$ encountered during the calculation.
\end{theorem}
\begin{proof}
Multiplying both sides of Equation~(\ref{eq:vertex-local-complementation}) by $S_vH_v$, we obtain 

\begin{equation}
\label{eq:simplification-sh}
S_vH_v\ket{G}=H_v\pars{\prod_{u,w\in N(v)}CS_{u,w}}\ket{G}.
\end{equation}

Equivalently,
\begin{equation}
\label{eq:simplification-hsh}
H_vS_vH_v\ket{G}=\pars{\prod_{u,w\in N(v)}CS_{u,w}}\ket{G}.
\end{equation}

Multiplying both sides of Equation~(\ref{eq:simplification-sh}) by $H_xS_xH_xS_x^{\dagger}$, we obtain
\begin{align}
\label{eq:simplification-hs}
    H_vS_v\ket{G}&=H_vS_vH_vS_v^\dagger H_v\pars{\prod_{u,w\in N(v)}CS_{u,w}}\ket{G}\nonumber\\
    &=\frac{1+i}{\sqrt{2}}
    S_v^{3}
    \pars{\prod_{u\in N(v)}Z_u}
    \pars{\prod_{u,w\in N(v)}CS_{u,w}}
    \ket{G}.\nonumber\\
    &=\frac{1+i}{\sqrt{2}}
    \pars{\prod_{u\in N(v)\cup\set{v}}S^3_u}
    \ket{L_v(G)}.
\end{align}
Since $C_v\in \langle H,S\rangle$, each $C_v$ is equivalent to a product of $H$'s and $S$'s. Since $HH=I$ and $SS=Z$ are both Pauli operators, $C_v$ is equivalent to a global phase and a Pauli operator applied to an alternating product of $H$'s and $S$'s, which we define as $D_v$. 
Thus we can write $\ket{\psi}=\alpha P \pars{\bigotimes\limits_{v=1}^n D_v} \ket{G}$ for some phase $\alpha$ and some Pauli operator $P$. 
In what follows, we do not mention Pauli operators or global phases because we can automatically keep track of them by conjugating them through and updating $\alpha$ and $P$ accordingly. 
We define the following useful monovariant and describe an algorithm to decrease it.

\begin{definition}
Let $M$ be the sum of the total number of $H$'s among all $D_v$ for $1\le v\le n$ as well as the number of vertices $v$ such that $D_v$ ends in $SH$.
\end{definition}

Our goal will be to reduce $M$ as much as possible until all terms $D_v$ consist of at most one Clifford operation.

\begin{lemma}
\label{lemma:local-clifford-elimination}
 If $D_v$ has length at least 2, we can update $D_v$ and all $D_u$ for $u\in N(v)$ so that $M$ decreases.
\end{lemma}
\begin{proof}
If $D_v$ ends in $HS$, we apply Equation~(\ref{eq:simplification-hs}) on qubit $v$, which effectively removes $HS$ from $D_v$ and appends $S$ onto the ends of $D_v$ and all $D_u$ where $u\in N(v)$. This move does not increase the number of terms that end in $SH$ but does decrease the number of $H$'s in $D_v$.
If this transformation cannot be performed, then $D_v$ ends in $SH$ since it has length at least 2.
In this case, we apply Equation~(\ref{eq:simplification-sh}) to vertex $v$. This replaces the $SH$ in $D_v$ with $H$, which decreases $M$ by 1. This also adds $S$ gates to $N(v)$, which does not increase $M$.
This proves the claim.
\end{proof}

We can apply the updates in Lemma~\ref{lemma:local-clifford-elimination} a finite number of times until all products $D_v$ have length at most 1, in which case $D_v\in \{I,S,H\}$ for all $v$. Rearranging Equation~(\ref{eq:edge-local-complementation}) and assuming $u>v$, we have
\begin{equation}
\label{eq:edge-local-complementation-hsliding}
    H_v\ket{G}=H_uZ_vZ_u\pars{\prod\limits_{p\in A,q\in B}CZ_{p,q}}\ket{G},
\end{equation}
which we repeatedly apply on qubits $u$ with $D_u$ ending in $H$ whenever $u$ has a neighbour $v$ in $G$ with $v<u$.
If $v$ already contained a local $H$ operation, the two $H$ gates cancel out.
This procedure must terminate after a finite number of steps since the total index of all $H$ gates can only decrease.
After this terminates, we conjugate $P$ through the remaining local Clifford operations. Using the fact that $X_v\ket{G}=\pars{\prod\limits_{u\in N(v)}Z_u} \ket{G}$, and $Y=-iZX$, we can turn all $Y$'s and $X$'s into $Z$'s and phases. Now, we have transformed $\ket{\psi}$ into our canonical form.

Every time we apply Theorem~\ref{thm:edge-local-complementation}, Equation~(\ref{eq:simplification-sh}), or Equation~(\ref{eq:simplification-hs}), we perform $O(d^2)$ edge toggles where $d$ is the maximum degree of $G$. 
Initially, $M=O(n)$ because any local Clifford operator can be represented with a finite number of $H$'s. Then, shortening the lengths of all the $D_v$ terms to 1 takes $O(nd^2)$ operations. 
Next, we only need to apply Equation~(\ref{eq:edge-local-complementation-hsliding}) at most $n-1$ times by applying it for $v=n,n-1,n-2,\dots ,2$ in that order. Since the $H$'s move to lower numbered qubits or are eliminated, the only way an $H$ could still exist on a qubit $i$ after the algorithm passes through $i$ the first time is if right before the algorithm passes through $v$, all $u\in N(v)$ satisfy $u>i$. 
In that case, $N(i)$ cannot change once $v<i$, because $i$ is not connected to any lower-numbered qubits.
Therefore, after $v$ reaches 2, none of the $H$'s can be moved to lower numbered qubits. 
Thus, moving the $H$'s to the lowest possible numbered qubits takes $O(nd^2)$ operations. Removing all $D_v$ that end in $SH$ using Equation~(\ref{eq:simplification-sh}) takes $O(nd^2)$ operations, and simplifying the Pauli operators takes $O(nd)$ operations, so the total runtime is $O(nd^2)$.
\end{proof}

\section{Graph state stabilizer simulation}
\label{sec:qstatesgraphs-simulation}

\subsection{Algorithm}
Graph state simulators of stabilizer circuits are advantageous in that local Clifford gates such as $S$ and $H$ can be applied trivially in $O(1)$ time. 
The bottleneck of a graph state simulator is the application of controlled-Pauli gates, such as $CZ$ gates, which currently can be done in $O(d^2)$ time where $d$ is the maximum degree of the graph encountered during the calculation. 

To apply a gate $CZ_{u,v}$ to an extended graph state $\ket{\psi}=\pars{\bigotimes\limits_{v=1}^nC_v}\ket{G}$, we use the identity \begin{equation*}
    CZ_{u,v}=\frac{1}{2}\left((I+Z_u)+(I-Z_u)Z_v\right).
\end{equation*}
If we conjugate this expression through the $C_i$, it suffices to apply operators of the form $\frac{1}{2}((I+P_u)I_v+(I-P_u)Q_v)$ to graph states $\ket{G}$ where $P_u$ and $Q_v$ are Hermitian Pauli operators. This motivates the following definition.

\begin{definition}
\label{def:psi-pq-ops}

For Hermitian Pauli operators $P$ and $Q$, let $\ket{\psi_{PQ}}$ be the extended graph state obtained from simplifying the expression
\begin{equation}
\label{eq:psi-pq}
    \ket{\psi_{PQ}}=\frac{1}{2}\left( (I+P_u)+(I-P_u)Q_v \right)\ket{G}.
\end{equation}
For example, $\ket{\psi_{ZZ}}$ is $CZ_{u,v}\ket{G}$, so updating $G$ takes $O(1)$ time.
The values of $u$ and $v$ are contextual for $\psi_{PQ}$.
\end{definition}

\begin{table}
\centering
\begin{tabular}{|c|c|c|}
\hline
$\mathbf{(P,Q)}$ & $\left|\boldsymbol{\psi}_{\mathbf{PQ}}\right\rangle$       & \textbf{Runtime} \\\hhline{|=|=|=|}
$(Z,Z)$                  &
$CZ_{u,v}\ket{G}$            & \multirow{3}{*}{$O(d)$}           \\\cline{1-2}
$(Z,X)$                  &
$\pars{\prod\limits_{w\in N_v}CZ_{u,w}}\ket{G}$            &                              \\ \cline{1-2}
$(Y,Z)$                  &
$S_vZ_v\pars{\prod\limits_{w\in M_u}CZ_{v,w}}\ket{G}$            &                              \\ \hline
$(X,X)$, $(u,v)\in E$                  &
$H_uH_vCZ_{u,v}
\pars{\prod\limits_{p\in M_u,\, q\in M_v}CZ_{p,q}}\ket{G}$ & \multirow{6}{*}{$O(d^2)$}\\ \cline{1-2}
$(X,X)$, $(u,v)\notin E$                  & 
$\pars{\prod\limits_{p\in N_u,\, q\in N_v}CZ_{p,q}}\ket{G}$                 &  \\\cline{1-2}
$(Y,X)$, $(u,v)\in E$                  &
$\frac{1-i}{\sqrt{2}}
\pars{\prod\limits_{w\in M_u}S_w}H_u
\pars{\prod\limits_{w\in M_u\, \Delta\, N_v}CZ_{u,w}}
\ket{L_u(G)}$ &                        \\ \cline{1-2}
$(Y,X)$, $(u,v)\notin E$              & 
$\pars{\prod\limits_{w\in M_u\, \Delta\, N_v}Z_w}
\pars{\prod\limits_{p,\, q\in M_u\Delta N_v}CS_{p,q}}
\pars{\prod\limits_{p,\, q\in M_u}CS_{p,q}}\ket{G}$                 &                              \\ \cline{1-2}
$(Y,Y)$, $(u,v)\in E$                  &
$-i\pars{\prod\limits_{p,q\in M_u}CS_{p,q}}
\pars{\prod\limits_{p,q\in M_v}CS_{p,q}}\ket{G}$             &                 \\ \cline{1-2}

$(Y,Y)$, $(u,v)\notin E$             &
$\frac{1-i}{\sqrt{2}}
\pars{\prod\limits_{w\in M_u}S_w} H_u
\pars{\prod\limits_{w\in M_2}CZ_{u,w}}\ket{L_u(G)}$                 &                              \\ \hline
\end{tabular}
\caption[Formulas for Pauli projector operations]{A table of formulas for $\ket{\psi_{PQ}}$ as defined in Definition~\ref{def:psi-pq-ops}, where $d=\max(\deg(u),\deg(v))$.
With this data, we can compute $\ket{\psi_{PQ}}$ for all possible unordered pairs $(P,Q)$ since $\ket{\psi_{PQ}}$ and $\ket{\psi_{QP}}$ are equal with the roles of qubits $u$ and $v$ flipped, and changing the sign of $P$ changes $\ket{\psi_{PQ}}$ by $Q$.
Note that $\ket{\psi_{PQ}}$ is not necessarily in canonical form.
Also, we define the shorthands $N_w\coloneq N(w)$ and $M_w\coloneq N_w\cup \{w\}$.
The operand $\Delta$ denotes the symmetric difference of two sets.
Note that for $(P,Q)\in \{(Z,Z),(Z,X),(Y,Z)\}$, $\ket{\psi_{PQ}}$ consists of $O(d)$ $CZ$ operators applied to $\ket{G}$, whereas for $(P,Q)\in \{(X,X),(Y,X),(Y,Y)\}$, $\ket{\psi_{PQ}}$ consists of $O(d^2)$ $CZ$ operators and $O(d)$ local Clifford operators applied to $\ket{G}$, hence the $O(d^2)$ update time.
Additionally, we split the latter expressions into cases based on whether or not the graph $G$ contains the edge $(u,v)$. The proofs of these statements are found in Appendix~\ref{app:proofs-psipq-table}.}
\label{table:psipq-expressions}
\end{table}

Our expressions for $\ket{\psi_{PQ}}$ and the update times based on the expressions are depicted in Table~\ref{table:psipq-expressions}. 

The formulas for $\ket{\psi_{ZX}}$ and $\ket{\psi_{XX}}$ were derived by~\textcite{elliott2009graphical}. The rest of the formulas are computed by applying Theorem~\ref{thm:graphmerge} and Theorem~\ref{thm:graphmerging-special}. Since all the proofs are similar, they can be found in Appendix~\ref{app:proofs-psipq-table}.

\subsection{Discussion}

To apply a $CZ$ gate to qubits $u$ and $v$ of the extended graph state $\ket{\psi}=\pars{\prod\limits_{w=1}^nC_w}\ket{G}$, the original graph state simulation algorithm described in~\cite{anders2006fast,kerzner2021clifford} performs local complementations on $u$, $v$, or neighbouring qubits of $u$ and $v$, changing $C_u$ and $C_v$ until they are both diagonal. 
Local complementations run in $\Omega(d^2)$, where $d$ is the degree of the vertex at which it was applied. 
When applying a $CZ$ gate using our algorithm, if $P=\pm Z$ or $Q=\pm Z$, then it takes $O(d)$ time and runs faster than the original algorithm would have. 
For example, when $C_u=HSH=\sqrt{X}$ and $C_v=I$, the original algorithm would perform a local complementation at qubit $u$, whereas our algorithm would update $\ket{\psi}$ more efficiently, based on the expression in Table~\ref{table:psipq-expressions} for $\ket{\psi_{YZ}}$.

In order to perform $CZ$ updates in under quadratic time, we must find efficient update rules for $\ket{\psi_{PQ}}$ for all pairs $(P,Q)\in \{(X,X),(Y,Y),(X,Y)\}$. 
We show that the runtime for such update rules cannot by improved further.

\begin{theorem}
\label{thm:minimum-edge-toggling}
There exists a family of extended graph states such that applying a $CZ$ gate requires $\Omega(n^2)$ edges of $G$ to be toggled.
\end{theorem}

\begin{proof}
Let $A$ and $B$ be disjoint sets such that $A,B\subset [n]$, $|A|=\Omega(n)$, and $|B|=\Omega(n)$. We select vertices $u\in A$, $v\in B$. The following graphs are used in this proof.

\begin{definition}
Let $W_{p,q}$, where $p\in A$ and $q\in B$, be the graph consisting solely of the following edge set $E$.

\begin{enumerate}
\item $(p,q)\in E$
\item For all $w\in A\setminus\{p\}$, $(p,w)\in E$
\item For all $w\in B\setminus\{q\}$, $(q,w)\in E$
\end{enumerate}

Let graph $K_{A,B}$ be the complete bipartite graph with edges between each vertex in $A$ and each vertex in $B$.
Let $K_A$ (respectively, $K_B$) be the graph union of $K$ together with all the edges between vertices in $A$ ($B$).
Let $K_{A,w}$ ($K_{B,w}$) be the graph where all the vertices in $A$ ($B$) are connected to each other, and vertex $w$ is connected to every other vertex.
Let $G$ be the graph that has all possible edges, except those between vertex $v$ and the vertices in $A$.
\end{definition}

With these definitions, we can state the following lemma.

\begin{lemma}
\label{lemma:orbit-under-local-comp}
Let $R$ be the set of all graphs $G'$ that are local Clifford equivalent to $\ket{W_{u,v}}$. Then

\begin{equation}
\label{eq:orbit-under-local-comp}
R=\{K_{A,B}, K_A, K_B\}\cup \{K_{A,w}|w\in B\}\cup \{K_{B,w}|w\in A\}\cup \{W_{p,q}| p\in A,q\in B\}.
\end{equation} 
\end{lemma}
\begin{proof}
By a theorem proved in~\cite{vandennest2004graphical}, the local Clifford equivalence of the graph states $\ket{G'}$ and $\ket{W_{u,v}}$ is equivalent to the existence of a sequence of local complementation operations taking $G'$ to $W_{u,v}$. 
If we let $\mathcal{G}$ be the orbit of $W_{u,v}$ under local complementation, then $R=\mathcal{G}$.
$\mathcal{G}$ is depicted in Figure~\ref{fig:orbit-under-local-comp}.
The rest of the proof details traversing $\mathcal{G}$.
%use align here, spread out
For all $i\in A$ and $j\in B$,
\begin{itemize}
\item We consider all the edges in $\mathcal{G}$ emanating from $W_{i,j}$. For all $k\in [n]\setminus\{i,j\}$, \begin{equation}
    L_i(W_{i,j})=K_{A,j}\qquad
    L_j(W_{i,j})=K_{B,i}\qquad
    L_k(W_{i,j})=W_{i,j}
\end{equation}
\item We consider all the edges in $\mathcal{G}$ emanating from $K_{A,j}$. The case for $K_{B,i}$ is analogous. For all $k\in B\setminus \{ j\}$,
\begin{equation}
    L_i(K_{A,j})=W_{i,j}\qquad
    L_j(K_{A,j})=K_B\qquad
    L_k(K_{A,j})=K_{A,j}
\end{equation}
\item We consider all the elements of $\mathcal{G}$ reachable from $K_{A,B}$, $K_A$, or $K_B$ in a single local complementation.
\begin{align}
    L_i(K_A)&=K_{B,i}
    \phantom{KK_{A,B}}\qquad
    L_j(K_A)=K_{A,B}\nonumber\\
    L_i(K_B)&=K
    \phantom{K_{B,i}K_{A,B}}\qquad
    L_j(K_B)=K_{A,j}\\
    L_i(K_{A,B})&=K_B
    \phantom{K_{B,i}K}\qquad
    L_j(K_{A,B})=K_A\nonumber
\end{align}
\item The graph $W_{u,v}$ is in the same orbit as $W_{i,j}$ in $\mathcal{G}$.
\begin{equation}
    L_j(L_v(L_i(L_u(W_{u,v}))))=W_{i,j}
\end{equation}
\end{itemize}
This completes the proof.
\end{proof}

\begin{figure}
\centering
\begin{tikzpicture}[node distance={20mm}, thick, main/.style = {draw, circle,minimum size=3em}] 
\node[main] (1) {$W_{i,j}$}; 
\node[main] (2) [above right of=1] {$K_{A,j}$}; 
\node[main] (3) [below right of=1] {$K_{B,i}$}; 
\node[main] (4) [right of=2] {$K_B$}; 
\node[main] (5) [right of=3] {$K_A$}; 
\node[main] (6) [below right of=4] {$K_{A,B}$}; 
\draw (1) -- (2); 
\draw (1) -- (3); 
\draw (2) -- (4); 
\draw (3) -- (5); 
\draw (4) -- (6); 
\draw (5) -- (6); 
\draw (6) -- node[midway, above right, sloped, pos=0.6] {$i$} (4); 
\draw (6) -- node[midway, below right, sloped, pos=0.6] {$j$} (5); 
\draw (4) -- node[midway, above, sloped, pos=0.5] {$j$} (2); 
\draw (2) -- node[midway, above left, sloped, pos=0.4] {$i$} (1);
\draw (1) -- node[midway, below left, sloped, pos=0.5] {$j$} (3); 
\draw (3) -- node[midway, below, sloped, pos=0.6] {$i$} (5); 
\end{tikzpicture} 
\caption[Orbits under local complementation]{A depiction of $\mathcal{G}$ in the proof of Lemma~\ref{lemma:orbit-under-local-comp}, with undirected edges labelled with the vertex that local complementation is applied to and loop edges omitted. To generate $\mathcal{G}$ in its entirely, let $i$ and $j$ range over all vertices in $A$ and in $B$ respectively.}
\label{fig:orbit-under-local-comp}
\end{figure}
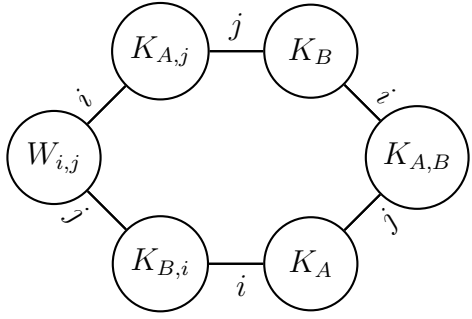

We show that $|E(G)\, \Delta \, E(G')|=\Omega(n^2)$ for any $G'\in R$.
Suppose without loss of generality that $\frac{1}{2}n\le |A|\le cn$ where $c$ is some fixed constant less than $1$. Then $K_A$ has $\binom{n}{2}-\binom{|B|}{2}$ edges, which is the most number of edges out of all graphs in $R$. 
\begin{equation}
    |E(G)\, \Delta \, E(G')|
    \ge |E(G)|-|E(G')|
    \ge \binom{n}{2}-|A|-\left(\binom{n}{2}-\binom{|B|}{2}\right)
    \ge \binom{|B|}{2}-|A|.
\end{equation}
Thus, $|E(G)\, \Delta \, E(G')|=\Omega(n^2)$ since $|B|\ge (1-c)n=\Omega(n)$.
\end{proof}

\section{Additive properties of stabilizer states}
\label{sec:qstatesgraphs-additiveprops}

\subsection{Graph merging}
We now consider the case of two states related by a Pauli operator. The case when the Pauli operator acts on a single qubit was explored in~\cite{hein2004-graphpauli}, and the case when the Pauli operator acts on multiple qubits was explored in~\cite{elliott2009graphical} as well as in Chapter~\ref{chapter:noncliffordgates}.
We restate Theorem~\ref{thm:graphmerge} here. In~\cite{elliott2009graphical} a related theorem is proven but without the case where $k$ is odd.

Let $A=N(1)\cup \{1\}$, and let $B$ be a set including $1$. Let $k$ be an integer. Then
\begin{equation}
\frac{1}{\sqrt{2}}\left(I+i^k\prod_{u\in B}Z_u\right)\ket{G}
=H_vZ_v\pars{\prod_{u,w\in A}CS_{u,w}^k}
\pars{\prod_{u\in A,\, w\in B}CZ_{u,w}}\ket{G}.
\end{equation}

Here we provide an alternative formula for when $k$ is odd that is more concise than previous formulas.
\begin{theorem}
\label{thm:graphmerging-special}
Let $k=2m+1$.
Let $A$ be an arbitrary set. Then
\begin{equation}
\frac{1}{\sqrt{2}}(I+i^{2m+1}\prod_{u\in A} Z_u)\ket{G}=
\frac{1+i^{2m+1}}{\sqrt{2}} \pars{\prod_{u\in A}Z_u^{m+1}}
\pars{\prod_{u,v\in A}CS_{u,v}}
\ket{G}.
\end{equation}
\begin{proof}
Let $\ket{z}$ be some computational basis state, and let $r$ be the number of values $i$ in $A$ where the $i^{\text{th}}$ bit in $z$ is $1$.
Let indicator functions $f$ and $g$ be defined as $f(r)=1$ when $r\equiv 2 \pmod{4}$ or $r\equiv 3\pmod{4}$ and $0$ otherwise and $g(r)=1$ when $r$ is odd and $0$ otherwise. Then,
\begin{align}
&\bra{z}\frac{1}{\sqrt{2}}(I+i^{2m+1}\prod_{p\in A} Z_p)\ket{G}\nonumber\\
&=\frac{1}{\sqrt{2}}\pars{\langle z|G\rangle + i^{2m+1}(-1)^r \langle z|G\rangle}\nonumber\\
&=\frac{1}{\sqrt{2}}\pars{1+ i^{2m+1}(-1)^r}\braket{z}{G},
\end{align}
and
\begin{align}
&\bra{z}\frac{1+i^{2m+1}}{\sqrt{2}}
\prod_{p\in A}Z_p^{m+1}\prod_{p,q\in A}CS_{p,q}\ket{G}\nonumber\\
=&\,\frac{1+i^{2m+1}}{\sqrt{2}}
(-1)^{f(r)}(-1)^{(m+1)r}i^r\braket{z}{G}\nonumber\\
=&\,\frac{1+i^{2m+1}}{\sqrt{2}}
i^{g(r)}(-1)^{(m+1)r}\braket{z}{G}.
\end{align}
The two expressions are equal for all parities of $m$ and $r$.
\end{proof}
\end{theorem}
The merging formulas, Theorem~\ref{thm:graphmerge} and Theorem~\ref{thm:graphmerging-special}, can be used to compute measurements of Pauli operators on extended graph states, by conjugating Pauli projectors through the local Clifford operators. These formulas can also be used to prove the correctness of the expressions for $\ket{\psi_{PQ}}$ in Table~\ref{table:psipq-expressions}, which we do in Appendix~\ref{app:proofs-psipq-table}.

Considering ways to merge stabilizer states that are not related by a Pauli operator, an interesting formula arises when $u$ and $v$ are not connected in Theorem~\ref{thm:edge-local-complementation}. 

\begin{theorem}
\label{thm:hsplitting}
Let $u$ and $v$ be two vertices of graph $G$ such that $(u,v)\notin E(G)$. Let sets $M_u$ and $M_v$ be defined as $M_u=N(u)\cup\{u\}$ and $M_v=N(v)\cup\{v\}$. Then
\begin{equation}
\label{eq:hsplitting}
H_uH_v\ket{G}=
Z_uZ_v\ket{G}+
\pars{\prod\limits_{w\in N(u)}Z_w}
\pars{\prod\limits_{w\in N(v)}Z_w}
\pars{\prod_{p\in M_u,\, q\in M_v}CZ_{p,q}}\ket{G}.
\end{equation}
\end{theorem}
\begin{proof}
We start by defining the states $\ket\psi$, $\ket{\phi_1}$, and $\ket{\phi_2}$ as $\ket\psi=H_uH_v\ket{G}$, $\ket{\phi_1}=Z_uZ_v\ket{G}$, and $\ket{\phi_2}=\pars{\prod\limits_{w\in N(u)}Z_w}
\pars{\prod\limits_{w\in N(v)}Z_w}
\pars{\prod\limits_{p\in A,\, q\in B}CZ_{p,q}}\ket{G}$.
For any bit string $s$, we seek to show that $\braket{s}{\psi}=\braket{s}{\phi_1}+\braket{s}{\phi_2}$.

We start by computing $\braket{s}{\psi}$.
Let $s_i$ be the bit in the $i^{\text{th}}$ position of $s$ and let $s_R$ be the number of 1's in $s$ at the indices denoted by set $R$, or equivalently, the sum of $s_i$ over $i\in R$.
We proceed by applying the $H$ operations in $\ket\psi$ to $\bra s$ and producing a sum of four states, one for each possibility for the values of $s_u$ and $s_v$.
We denote these states as $s$ with $X_u$ and $X_v$ applied with appropriate exponents.
We also make use of the stabilizers of a graph.
\begin{align}
\braket{s}{\psi}&=\braopket{s}{H_uH_v}{G}\\
&=\frac1{2}\sum_{i,j\in\set{0,1}}(-1)^{s_ui+s_vj}\braopket{s}{X_u^{s_u+i}X_v^{s_v+j}}{G}\\
&=\frac1{2}\sum_{i,j\in\set{0,1}}(-1)^{s_ui+s_vj}\braopket{s}{\pars{\prod_{w\in N(u)}Z_w^{s_u+i}}\pars{\prod_{w\in N(v)}Z_w^{s_v+j}}}{G}\\
&=\frac1{2}\sum_{i,j\in\set{0,1}}(-1)^{s_ui+s_vj+(s_u+i)s_{N(u)}+(s_v+j)s_{N(v)}}\braket{s}{G}\\
&=\frac1{2}\sum_{i,j\in\set{0,1}}(-1)^
{(s_u+s_{N(u)})i
+(s_v+s_{N(v)})j
+s_us_{N(u)}
+s_vs_{N(v)}}\braket{s}{G}
\end{align}
If either $s_u+s_{N(u)}$ or $s_v+s_{N(v)}$ are odd, the entire expression will be 0.
If that is not the case, the expression will be $2(-1)^{s_u+s_v}\braket{s}{G}$.

We also see that $\braket{s}{\phi_1}=(-1)^{s_u+s_v}\braket{s}{G}$.
All that remains is to compute $\braket{s}{\phi_2}$.

\begin{align}
\braket{s}{\phi_2}&=\braopket{s}{
\pars{\prod\limits_{w\in N(u)}Z_w}
\pars{\prod\limits_{w\in N(v)}Z_w}
\pars{\prod\limits_{p\in M_u,\, q\in M_v}CZ_{p,q}}}{G}\\
&=
(-1)^{s_{N(u)}+s_{N(v)}+s_{M_u}s_{M_v}}\braket{s}{G}\\
&=
(-1)^{s_{N(u)}+s_{N(v)}+(s_u+s_{N(u)})(s_v+s_{N(v)})}\braket{s}{G}
\end{align}

If $s_u+s_{N(u)}$ is odd, we have the following.

\begin{align}
\braket{s}{\phi_1}+\braket{s}{\phi_2}&=
(-1)^{s_u+s_v}\braket{s}{G}+
(-1)^{s_{N(u)}+s_{N(v)}+(s_u+s_{N(u)})(s_v+s_{N(v)})}\braket{s}{G}\\
&=
(-1)^{s_u+s_v}\braket{s}{G}+
(-1)^{s_{N(u)}+s_v}\braket{s}{G}\\
&=
(-1)^{s_v}\pars{(-1)^{s_u}+
(-1)^{s_{N(u)}}}\braket{s}{G}=0
\end{align}

Analogously, the same is true if $s_v+s_{N(v)}$ is odd.
If both $s_u+s_{N(u)}$ and $s_v+s_{N(v)}$ are even, we simplify our expression further.

\begin{align}
\braket{s}{\phi_1}+\braket{s}{\phi_2}&=
(-1)^{s_u+s_v}\braket{s}{G}+
(-1)^{s_{N(u)}+s_{N(v)}+(s_u+s_{N(u)})(s_v+s_{N(v)})}\braket{s}{G}\\
&=
(-1)^{s_u+s_v}\braket{s}{G}+
(-1)^{s_{N(u)}+s_{N(v)}}\braket{s}{G}
=(-1)^{s_u+s_v}\braket{s}{G}
\end{align}

This shows that $\braket{s}{\psi}$ is always equal to $\braket{s}{\phi_1}+\braket{s}{\phi_2}$, which completes our proof.
\end{proof}

\subsection{Linearly dependent triplets}
We now turn our attention to characterizing linearly dependent triplets of stabilizer states.
The following theorem shows that there are three types.

\begin{theorem}
\label{thm:stab-triples}
Let $\mathcal{S}\coloneq \{\ket{\psi_1},\ket{\psi_2},\ket{\psi_3}\}$ be a set of linearly dependent stabilizer states that are not all parallel. Then, up to global phase, one of the three cases must be true

\begin{enumerate}
    \item For some stabilizer state $\ket{\phi}$ and some Pauli operator $P$, \begin{equation}
        \mathcal{S}=\{\ket{\phi},P\ket{\phi},\frac{I+P}{\sqrt{2}}\ket{\phi}\}.
    \end{equation}
    \item 
    \label{case:stab-triples-case2}
    For some Clifford operator $C$, $1\le v\le n$, an extended graph state in reduced form $\ket{\psi}$ such that $v$ is the only vertex $u$ such that $c_u\neq H$, and $z_u=I$ whenever $c_u=H$, 
    \begin{equation}
       \label{eq:stab-triples-case2}
       \mathcal{S}=\{C\ket{0^n},C\ket{\psi},C\left(S_x\ket{\psi}\right)\}.
    \end{equation}
    \item \label{case:stab-triples-case3}
    For some Clifford operator $C$, $1\le u<v\le n$, an extended graph state in reduced form $\ket{\psi}$ such that $u$ and $v$ are the only two vertices $w$ such that $c_w\neq H$, and $z_w=I$ whenever $c_w=H$, 
    \begin{equation}
       \label{eq:stab-triples-case3}
       \mathcal{S}=\{C\ket{0}^{\otimes n},C\ket{\psi}, C\left(Z_uZ_vCZ_{u,v}\ket{\psi}\right)\}.
    \end{equation}
\end{enumerate}
\end{theorem}
\begin{proof}
Let $U$ be a Clifford operator such that $\ket{\psi_1}=U\ket{0}^{\otimes n}$.
Let $\ket{\psi}\coloneq U^{\dagger}\ket{\psi_2}$ and $\ket{\varphi}\coloneq U^{\dagger}\ket{\psi_3}$. Any stabilizer state can be represented up to global phase as \begin{equation*}
    \frac{1}{\sqrt{|V|}}\sum_{x\in V}i^{l(x)}(-1)^{q(x)}\ket{x},
\end{equation*} where $V$ is an affine subspace of $\mathbb{F}_2^n$, $\ell(x)$ is a linear binary function on $n$ bits, and $q(x)$ is a quadratic binary function on $n$ bits. Let $(V,\ell(x),q(x))$ be the corresponding triple for $\ket{\psi}$. Without loss of generality let the first non-zero amplitudes in $\ket{\psi}$ and $\ket{\varphi}$ be positive real numbers. The linear dependence of the state vectors in $\mathcal{S}$ is equivalent to the existence of $\alpha,\beta \in \mathbb{C}\setminus \{0\}$ such that \begin{equation*}
    \frac{1}{\sqrt{|V|}}\ket{0}^{\otimes n}+\alpha\ket{\psi}=\beta\ket{\varphi}.
\end{equation*}
Note that $|V|$ is a power of $2$. If $|V|=1$, $\ket{\psi}$ is a non-zero computational basis state. Since the non-zero amplitudes in stabilizer states differ from each other by powers of $i$, $\alpha$ must be a power of $i$,  $\alpha\ket{\psi}$ and $\ket{0^n}$ are Pauli related, and $\ket{\varphi}= \frac{\ket{0^n}+\alpha\ket{\psi}}{\sqrt{2}}$.

From now on assume $|V|>1$. Then $0^n\in V$ or else $\ket{\varphi}$ would have $|V|+1\neq 2^m$ non-zero amplitudes and could not be a stabilizer state.
Also note the support set of $\ket{\varphi}$ is either $V$ or $V\setminus\{0^n\}$, and $|V|-1\neq 2^m\quad\forall m\ge2$. Therefore, the only case when the support set of $\ket{\varphi}$ is $V\setminus\{0^n\}$ is if $|V|=2$ and $\alpha=-1$, in which case $\ket{\varphi}$ is a computational basis state, related by a Pauli operator to $\ket{0}^{\otimes n}$.

From now on the support set of $\ket{\varphi}$ is $V$. Then, $\beta=1+\alpha$ and by comparing non-zero amplitudes of the left and right hand sides, $\frac{\alpha}{1+\alpha}=i^k$ for some $k\in \{1,2,3\}$. 

\begin{claim}
\label{claim:stab-triples-subspace}
If $|V|\ge 8$ and $k=2$, it is not possible for $\ket{\varphi}$ to be a stabilizer state.
\end{claim}

\begin{proof}
Suppose $\ket{\varphi}$ was a stabilizer state. We consider the stabilizer state $\ket{\phi}$ with support set $V$ and quadratic and linear functions equal to the difference of the quadratic and linear functions of $\ket{\varphi}$ and $\ket{\psi}$. 
The un-normalized amplitudes of $\ket{\phi}$ are equal to the ratios of the amplitudes of $\ket{\varphi}$ and $\ket{\psi}$, which are $i^k$ for all non-zero computational basis states and $1$ for $\ket{0^n}$. We use the following proposition to derive a contradiction.

If we write $\ket{\phi}=C\ket{0^n}$ for some Clifford operator $C$. Then, for some Pauli operator $P'$, $\bra{\phi}P\ket{\phi}=\bra{0^n}C^{\dagger}PC\ket{0^n}=\bra{0^n}P'\ket{0^n}\in \{0,1,i,-1,-i\}$.

That $0^n\in V$ implies $V$ is a subspace of $\mathbb{F}_2^n$. Let $e\coloneq e_1e_2\dots e_n$ be a basis vector of $V$. Let $P\coloneq \bigotimes\limits_{i=1}^nX_i^{e_i}$. Then $P\ket{\phi}=\frac{1}{\sqrt{|V|}}\left(\ket{e}-\sum\limits_{x\in V\setminus\{e\}}\ket{x}\right)$, so $\bra{\phi}P\ket{\phi}=\frac{|V|-4}{|V|}\not\in \{0,1,i,-1,-i\}$, which is impossible.
This completes the proof of the claim.
\end{proof}

\begin{claim}
\label{claim:stab-triples-phase}
If $|V|\ge 4$ and $k\in \{1,3\}$, it is not possible for $\ket{\varphi}$ to be a stabilizer state.
\end{claim}
\begin{proof}
As in Claim~\ref{claim:stab-triples-subspace}, define $\ket{\phi}$ equal to the stabilizer state whose un-normalized amplitudes are the ratios of the amplitudes of $\ket{\varphi}$ and $\ket{\psi}$, in which case $\ket{\phi}\propto \ket{0^n}\pm i \sum\limits_{x\in V\setminus\{0^n\}}\ket{x}$. 
It is known that in a stabilizer state with its first non-zero amplitude positive and real, the number of pure imaginary amplitudes must be $0$ or half of the support. $\ket{\phi}$ does not satisfy this condition, which proves the claim.
\end{proof}

By Claim~\ref{claim:stab-triples-subspace} and Claim~\ref{claim:stab-triples-phase}, the remaining cases either satisfy $|V|=2$ or $|V|=4$ and $k=2$. 
If $|V|=2$, then $\ket{\psi}$ and $\ket{\varphi}$ are of the form $\ket{\psi}=\frac{\ket{0^n}+i^h\ket{s}}{\sqrt{2}}$ and $\ket{\varphi}=\frac{\ket{0^n}+i^{k+h}\ket{s}}{\sqrt{2}}$ for some $h$ and computational basis state $\ket{s}$.
Also, $\ket{0^n}=-\sqrt{2}\alpha \ket{\psi}+\sqrt{2}(1+\alpha)\ket{\varphi}$. 
If $k=2$, then $\alpha=-\frac{1}{2}$ and $-2\alpha\ket{\psi}$ and $2(1+\alpha)\ket{\varphi}$ are Pauli-related stabilizer states such that their sum divided by $\sqrt{2}$ is $\ket{0^n}$. 
If $k=1$, then $\alpha = \frac{i-1}{2}$. If we express $\ket{\psi}$ in reduced form, then $n-1$ of the $c_u$ are equal to $H$ by Lemma~\ref{lemma:state-amplitude}, and we can let $v$ be the unique index $u$ such that $c_u\neq H$. By Proposition~\ref{prop:stab-state-inner-products}, since $\bra{0^n}\psi\rangle\neq0$, for each $c_u=H$, we have $z_u=I$. Note that $s_v=1$ by Lemma~\ref{lemma:projbasisstate}, so $S_v\ket{\psi}=\ket{\varphi}$, and we have
\begin{equation}
    \ket{0^n}=\frac{1-i}{\sqrt{2}}\ket{\psi}+\frac{1+i}{\sqrt{2}}S_v\ket{\psi},
\end{equation}
which corresponds to Case~\ref{case:stab-triples-case2}. If $k=3$, then similar arguments yield the same result with the roles of $\ket{\psi}$ and $\ket{\varphi}$ swapped.

If $|V|=4$ and $k=2$, then $\alpha=-\frac{1}{2}$ and $\beta=\frac{1}{2}$. If we express $\ket{\psi}$ in reduced form, then $n-2$ of the $c_i$ are equal to $H$ by Lemma~\ref{lemma:state-amplitude}, and we can let $u$ and $v$ be the indices $w$ such that $c_w\neq H$. By Proposition~\ref{prop:stab-state-inner-products}, since $\bra{0^n}\psi\rangle\neq0$, for each $c_w=H$, we have $z_w=I$. By Lemma~\ref{lemma:projbasisstate}, we can write the computational basis states in $\ket{\varphi}$ and $\ket{\psi}$ as $\ket{i}_u\otimes \ket{j}_v\otimes  \ket{s_{ij}}$ for $i,j\in \{0,1\}$ and binary strings $s_{ij}$ of length $n-2$. We compute
\begin{equation}
    \bra{i}_x\otimes \bra{j}_y\otimes \bra{s_{ij}} Z_xZ_yCZ_{x,y}\ket{\psi}
    =(-1)^{1-(1-i)(1-j)}\bra{i}_x\otimes \bra{j}_y\otimes \bra{s_{ij}}\hspace{0.2ex}\ket{\psi},
\end{equation}
so we have
\begin{equation}
    \ket{0}^{\otimes n}=\ket{\psi}+Z_xZ_yCZ_{x,y}\ket{\psi},
\end{equation}
which corresponds to Case~\ref{case:stab-triples-case3}.
\end{proof}

\begin{example}
We show small illustrative examples of each of the three cases in Theorem~\ref{thm:stab-triples}.
Each of the stabilizer states is in canonical form with vertex $1$ being the lowest node in the diagram and vertex $3$ being the highest.
\begin{align}
    \resizebox{2cm}{!}{
\begin{tikzpicture}[baseline={([yshift=-.5ex]current bounding box.center)}]
	\begin{pgfonlayer}{nodelayer}
		\node [style=white dot] (0) at (-0.25, 1) {$H$};
		\node [style=white dot] (1) at (-0.75, -0.5) {$H$};
		\node [style=white dot] (2) at (0.75, 0) {$H$};
		\node [style=none] (3) at (1, 1.5) {};
		\node [style=none] (4) at (2, 0.25) {};
		\node [style=none] (5) at (1, -1) {};
		\node [style=none] (6) at (-1.5, 1.5) {};
		\node [style=none] (7) at (-1.5, -1) {};
	\end{pgfonlayer}
	\begin{pgfonlayer}{edgelayer}
		\draw (6.center) to (7.center);
		\draw (3.center) to (4.center);
		\draw (4.center) to (5.center);
	\end{pgfonlayer}
\end{tikzpicture}}&=
\frac{1}{\sqrt{2}}
\resizebox{2cm}{!}{
\begin{tikzpicture}[baseline={([yshift=-.5ex]current bounding box.center)}]
	\begin{pgfonlayer}{nodelayer}
		\node [style=white dot] (0) at (-0.25, 1) {$H$};
		\node [style=white dot] (1) at (-0.75, -0.5) {$I$};
		\node [style=white dot] (2) at (0.75, 0) {$H$};
		\node [style=none] (3) at (1, 1.5) {};
		\node [style=none] (4) at (2, 0.25) {};
		\node [style=none] (5) at (1, -1) {};
		\node [style=none] (6) at (-1.5, 1.5) {};
		\node [style=none] (7) at (-1.5, -1) {};
	\end{pgfonlayer}
	\begin{pgfonlayer}{edgelayer}
		\draw (6.center) to (7.center);
		\draw (3.center) to (4.center);
		\draw (4.center) to (5.center);
	\end{pgfonlayer}
\end{tikzpicture}}+
\frac{1}{\sqrt{2}}
\resizebox{2cm}{!}{
\begin{tikzpicture}[baseline={([yshift=-.5ex]current bounding box.center)}]
	\begin{pgfonlayer}{nodelayer}
		\node [style=white dot] (0) at (-0.25, 1) {$H$};
		\node [style=white dot] (1) at (-0.75, -0.5) {$Z$};
		\node [style=white dot] (2) at (0.75, 0) {$H$};
		\node [style=none] (3) at (1, 1.5) {};
		\node [style=none] (4) at (2, 0.25) {};
		\node [style=none] (5) at (1, -1) {};
		\node [style=none] (6) at (-1.5, 1.5) {};
		\node [style=none] (7) at (-1.5, -1) {};
	\end{pgfonlayer}
	\begin{pgfonlayer}{edgelayer}
		\draw (6.center) to (7.center);
		\draw (3.center) to (4.center);
		\draw (4.center) to (5.center);
	\end{pgfonlayer}
\end{tikzpicture}}\\
\resizebox{2cm}{!}{
\begin{tikzpicture}[baseline={([yshift=-.5ex]current bounding box.center)}]
	\begin{pgfonlayer}{nodelayer}
		\node [style=white dot] (0) at (-0.25, 1) {$H$};
		\node [style=white dot] (1) at (-0.75, -0.5) {$H$};
		\node [style=white dot] (2) at (0.75, 0) {$H$};
		\node [style=none] (3) at (1, 1.5) {};
		\node [style=none] (4) at (2, 0.25) {};
		\node [style=none] (5) at (1, -1) {};
		\node [style=none] (6) at (-1.5, 1.5) {};
		\node [style=none] (7) at (-1.5, -1) {};
	\end{pgfonlayer}
	\begin{pgfonlayer}{edgelayer}
		\draw (6.center) to (7.center);
		\draw (3.center) to (4.center);
		\draw (4.center) to (5.center);
	\end{pgfonlayer}
\end{tikzpicture}}&=
\frac{1-i}{\sqrt{2}}
\resizebox{2cm}{!}{
\begin{tikzpicture}[baseline={([yshift=-.5ex]current bounding box.center)}]
	\begin{pgfonlayer}{nodelayer}
		\node [style=white dot] (0) at (-0.25, 1) {$I$};
		\node [style=white dot] (1) at (-0.75, -0.5) {$H$};
		\node [style=white dot] (2) at (0.75, 0) {$H$};
		\node [style=none] (3) at (1, 1.5) {};
		\node [style=none] (4) at (2, 0.25) {};
		\node [style=none] (5) at (1, -1) {};
		\node [style=none] (6) at (-1.5, 1.5) {};
		\node [style=none] (7) at (-1.5, -1) {};
	\end{pgfonlayer}
	\begin{pgfonlayer}{edgelayer}
		\draw (6.center) to (7.center);
		\draw (3.center) to (4.center);
		\draw (4.center) to (5.center);
		\draw (1) to (0);
		\draw (0) to (2);
	\end{pgfonlayer}
\end{tikzpicture}}+
\frac{1+i}{\sqrt{2}}
\resizebox{2cm}{!}{
\begin{tikzpicture}[baseline={([yshift=-.5ex]current bounding box.center)}]
	\begin{pgfonlayer}{nodelayer}
		\node [style=white dot] (0) at (-0.25, 1) {$S$};
		\node [style=white dot] (1) at (-0.75, -0.5) {$H$};
		\node [style=white dot] (2) at (0.75, 0) {$H$};
		\node [style=none] (3) at (1, 1.5) {};
		\node [style=none] (4) at (2, 0.25) {};
		\node [style=none] (5) at (1, -1) {};
		\node [style=none] (6) at (-1.5, 1.5) {};
		\node [style=none] (7) at (-1.5, -1) {};
	\end{pgfonlayer}
	\begin{pgfonlayer}{edgelayer}
		\draw (6.center) to (7.center);
		\draw (3.center) to (4.center);
		\draw (4.center) to (5.center);
		\draw (1) to (0);
		\draw (0) to (2);
	\end{pgfonlayer}
\end{tikzpicture}}\\
\resizebox{2cm}{!}{
\begin{tikzpicture}[baseline={([yshift=-.5ex]current bounding box.center)}]
	\begin{pgfonlayer}{nodelayer}
		\node [style=white dot] (0) at (-0.25, 1) {$H$};
		\node [style=white dot] (1) at (-0.75, -0.5) {$H$};
		\node [style=white dot] (2) at (0.75, 0) {$H$};
		\node [style=none] (3) at (1, 1.5) {};
		\node [style=none] (4) at (2, 0.25) {};
		\node [style=none] (5) at (1, -1) {};
		\node [style=none] (6) at (-1.5, 1.5) {};
		\node [style=none] (7) at (-1.5, -1) {};
	\end{pgfonlayer}
	\begin{pgfonlayer}{edgelayer}
		\draw (6.center) to (7.center);
		\draw (3.center) to (4.center);
		\draw (4.center) to (5.center);
	\end{pgfonlayer}
\end{tikzpicture}}&=
\resizebox{2cm}{!}{
\begin{tikzpicture}[baseline={([yshift=-.5ex]current bounding box.center)}]
	\begin{pgfonlayer}{nodelayer}
		\node [style=white dot] (0) at (-0.25, 1) {$I$};
		\node [style=white dot] (1) at (-0.75, -0.5) {$H$};
		\node [style=white dot] (2) at (0.75, 0) {$I$};
		\node [style=none] (3) at (1, 1.5) {};
		\node [style=none] (4) at (2, 0.25) {};
		\node [style=none] (5) at (1, -1) {};
		\node [style=none] (6) at (-1.5, 1.5) {};
		\node [style=none] (7) at (-1.5, -1) {};
	\end{pgfonlayer}
	\begin{pgfonlayer}{edgelayer}
		\draw (6.center) to (7.center);
		\draw (3.center) to (4.center);
		\draw (4.center) to (5.center);
		\draw (1) to (0);
		\draw (1) to (2);
	\end{pgfonlayer}
\end{tikzpicture}}+
\resizebox{2cm}{!}{
\begin{tikzpicture}[baseline={([yshift=-.5ex]current bounding box.center)}]
	\begin{pgfonlayer}{nodelayer}
		\node [style=white dot] (0) at (-0.25, 1) {$Z$};
		\node [style=white dot] (1) at (-0.75, -0.5) {$H$};
		\node [style=white dot] (2) at (0.75, 0) {$Z$};
		\node [style=none] (3) at (1, 1.5) {};
		\node [style=none] (4) at (2, 0.25) {};
		\node [style=none] (5) at (1, -1) {};
		\node [style=none] (6) at (-1.5, 1.5) {};
		\node [style=none] (7) at (-1.5, -1) {};
	\end{pgfonlayer}
	\begin{pgfonlayer}{edgelayer}
		\draw (6.center) to (7.center);
		\draw (3.center) to (4.center);
		\draw (4.center) to (5.center);
		\draw (1) to (0);
		\draw (1) to (2);
		\draw (0) to (2);
	\end{pgfonlayer}
\end{tikzpicture}}
\end{align}
\end{example}

We take a closer look at Case~\ref{case:stab-triples-case2} and Case~\ref{case:stab-triples-case3} of Theorem~\ref{thm:stab-triples} by considering inner products, revealing the symmetries in non-Pauli-related triplets of linearly dependent stabilizer states.
\begin{theorem}
If two stabilizer states $\ket{\psi_1}$ and $\ket{\psi_2}$ satisfy $\bra{\psi_1}\psi_2\rangle\in \{\frac{i-1}{2},-\frac{1}{2}\}$, then $\ket{\psi_3}$, defined as $\ket{\psi_3}\coloneq -(\ket{\psi_1}+\ket{\psi_2})$, is a stabilizer state and satisfies $\bra{\psi_2}\psi_3\rangle=\bra{\psi_3}\psi_1\rangle=\bra{\psi_1}\psi_2\rangle$.
\end{theorem}
\begin{proof}
Let $\ket{\psi_1}=C\ket{0^n}$ and $\ket{\psi}\coloneq C^{\dagger}\ket{\psi_2}$ for some Clifford operator $C$. 
If $\bra{0^n}\psi \rangle =\frac{i-1}{2}$, then $\ket{\psi}$ is of the form $\frac{i-1}{2}\ket{0^n}+i^k\frac{i-1}{2}\ket{s}$ for some non-zero computational basis state $\ket{s}$ and integer $k$. Then, $\ket{\psi_3}$, which is equal to $C(-\frac{i+1}{2}\ket{0^n}-i^k\frac{i-1}{2}\ket{s})$, is a stabilizer state and 
satisfies $\bra{\psi_2}\psi_3\rangle=\bra{\psi_3}\psi_1\rangle=\frac{i-1}{2}$.
Likewise, if $\bra{0^n}\psi \rangle =-\frac{1}{2}$, then $\ket{\psi}$ is of the form $-\frac{1}{2}(\ket{0^n}+i^{k_1}\ket{s_1}+i^{k_2}\ket{s_2}+i^{k_3}\ket{s_3})$ for some distinct computational basis states $\ket{s_1},\ket{s_2},\ket{s_3}$ and some integers $k_1,k_2,k_3$, so $\ket{\psi_3}$ similarly is a stabilizer state and satisfies
$\bra{\psi_2}\psi_3\rangle=\bra{\psi_3}\psi_1\rangle=-\frac{1}{2}$.
\end{proof}
\subsection{Inner product algorithm}
We now turn our attention to computing inner products between extended graph states. 
Our inner product algorithm has cubic worst-case runtime, same as the current best algorithm based on generator matrices~\cite{garcia2012efficient}. 
Our algorithm is more direct in implementation due to the correspondence between an extended graph state and the gates required to produce it and is also naturally global phase sensitive.
Our algorithm uses the following observation.

\begin{proposition}
\label{prop:stab-state-inner-products}
Let $\ket{\psi}\coloneq\pars{\bigotimes\limits_{i=v}^nc_vz_v}\ket{G}$ be in reduced form, and let $A\coloneq \{v|c_v=H\}$. Then
\begin{equation}\bra{0}^{\otimes n}\ket{\psi}=\begin{cases} 
0 & \text{if } \exists v\in A, z_v=Z\\
\frac{1}{\sqrt{2^{n-|A|}}} &\text{otherwise}
\end{cases}.
\end{equation}
\end{proposition}

\begin{proof}
Note that
\begin{equation}
    \bra{0}^{\otimes n}
    \pars{\bigotimes\limits_{v=1}^nc_vz_v}\ket{G} =\bra{0}^{\otimes n}
    \pars{\prod\limits_{v\in A}H_v(z_v)_v}\ket{+}^{\otimes n}
    =\frac{\bra{0}_A
    \pars{\prod\limits_{v\in A}(x_v)_v}\ket{0}^{\otimes n}}{\sqrt{2^{n-|A|}}},
\end{equation}
where $x_v=X$ when $z_v=Z$ and $x_v=I$ when $z_v=I$. If $x_v=X$ for some $v\in A$, then $\bra{0}^{\otimes n}\ket{\psi}=0$ and otherwise $\bra{0}^{\otimes n}\ket{\psi}=\frac{1}{\sqrt{2^{n-|A|}}}$.
\end{proof}
We present our algorithm in the proof of the following theorem.
\begin{theorem}
\label{thm:graph-state-inner-product-algorithm}
Let $\ket{\psi}\coloneq \pars{\bigotimes\limits_{v=1}^nC_v}\ket{G}$ and $\ket{\psi'}\coloneq\pars{\bigotimes\limits_{v=1}^nC_v'}\ket{G'}$ be two extended graph states. Then $\langle \psi \ket{\psi'}$ can be computed in $O(nd^2)$ time, where $d$ is the maximum degree in $G$ and $G'$ encountered during the calculation.
\end{theorem}
\begin{proof}
First we apply $C_v^{\dagger}$ to $C_v'$ for each $v$. It suffices to take the inner product of $\ket{G}$ and $\pars{\bigotimes\limits_{v=1}^nD_v}\ket{G'}$ for local Clifford operators $D_i$. We do so by taking the inner product of $\ket{0}^{\otimes n}$ and $\ket{\varphi}\coloneq
\bra{0}^{\otimes n}
\pars{\bigotimes\limits_{v=1}^nH_v}
\pars{\prod\limits_{(u,v)\in E(G)}CZ_{u,v}}
\pars{\bigotimes\limits_{v=1}^nD_v}
\ket{G}$. 
We first simplify the layer of $CZ$ operators. 
\begin{definition}
\label{def:stab-inner-star-operation}
A \textit{star operation} on qubit $v$ is a product of operators $CZ_{u,v}$ for various $u$'s.
\end{definition}

For each $v\in [n]$, we apply star operations of the form $\prod\limits_{u\in N(v),\, v<u}CZ_{u,v}$ to $\pars{\bigotimes\limits_{w=1}^nD_w}\ket{G'}$.
We perform updates in $O(d^2)$ time as follows. If $D_v$ takes $Z$ to $\pm Z$ upon conjugation, then for each neighbour $u$ of $v$ we apply $CZ_{u,v}$ by the method described in Section~\ref{sec:qstatesgraphs-simulation}, conjugating it through $D_v$ and $D_u$ and either applying a normal $CZ$ gate to $G'$ or the appropriate operation from Table~\ref{table:psipq-expressions}.
After updating the local Clifford operations $D_w$ and the graph $G'$, $D_v$ still takes $Z$ to $\pm Z$ upon conjugation since $D_v$ is changed by a diagonal Clifford, so we can repeat the same update process for all qubits $u$ in the star operation.
If $D_v$ takes $Z$ to $\pm X$ upon conjugation, we apply Equation~(\ref{eq:edge-local-complementation}) to qubit $v$ and some neighbour $u$ of $v$, which changes $D_v$ to $D_vH$.
We proceed as before because $D_v$ now takes $Z$ to $\pm Z$ upon conjugation.
If qubit $v$ does not have a neighbour, then the application of the $X$ operator to qubit $v$ does not change $\ket{G'}$, so applying $CZ_{u,v}$ becomes equivalent to applying some Pauli operator on qubit $u$, which can be done in constant time.
If $D_v$ takes $Z$ to $\pm Y$ upon conjugation, then we apply Equation~(\ref{eq:vertex-local-complementation}) to qubit $v$, which changes $D_v$ to $D_vHS^{\dagger}H$. We proceed as in the first case since $D_v$ takes $Z$ to $\pm Z$.

Next, we append $H$ to $D_v$ for all $v$ and simplify $\pars{\bigotimes\limits_{v=1}^nD_v}\ket{G'}$ to reduced form, following the algorithm described in Theorem~\ref{thm:canonical-algorithm}. The inner product $\bra{\psi}\psi'\rangle$ is equal to the product of the result of Proposition~\ref{prop:stab-state-inner-products} applied to $\ket{\varphi}$ and the global phase factors produced during the calculation.
The total runtime of the algorithm is $O(nd^2)$.
\end{proof}

\section{Conclusion}
In this chapter, we explored various operations on Clifford states through the lens of the graph formalism.
We created a canonical form for expressing extended graph states in a concise and unique way that improves upon previous reduced forms~\cite{nielsen2002quantum,elliott2008graphical,elliott2009graphical}.
We developed efficient simplification and inner product algorithms, and the connections between stabilizer states and the properties of their corresponding graphs when expressed in canonical will be explored in later chapters.

We applied our merging formulas to discover new rules that describe the action of controlled-$Z$ gates on arbitrary extended graph states. 
Our transformation rules enable us to improve the algorithm for applying controlled-Pauli operators to graph states~\cite{anders2006fast,kerzner2021clifford} and also to prove that under certain assumptions, it is impossible to update extended graph states in under quadratic time in the number of qubits upon the application of a controlled-Pauli gate.
Therefore, in order to improve graph state simulation, we should consider algorithms that do not simply apply one gate at a time. Whenever multiple $CZ$ gates can be applied consecutively, we can potentially apply star operations following the method described in the proof of Theorem~\ref{thm:graph-state-inner-product-algorithm} to spread out the $O(d^2)$ update time over multiple $CZ$ gates, improving performance.
Future work to improve graph state simulation could study the relationship between the circuit and the runtime, as well as design more efficient algorithms for simulating certain types of circuits.
We discuss the simulation of quantum circuits and error-correcting codes in later chapters.
 \chapter{Quantum Codes from Graphs}
\label{chapter:qcodes-graphs}

\section{Introduction}
\label{sec:intro}

As we saw in previous chapters, stabilizer codes are the quantum analogue of linear codes for classical computing. More precisely, a $\llbracket  n, k\rrbracket$ stabilizer code encodes $k$ logical qubits to $n$ physical qubits by embedding $2^k$ logical degrees of freedom into a subspace defined as the simultaneous $+1$ eigenspace of $n-k$ $n$-qubit Pauli operators, known as the stabilizers. Stabilizers comprise a considerably large class of quantum codes, including all CSS codes~\cite{steane1996multiple,calderbank1996good}. 

Generally, a stabilizer code is specified by the \textit{stabilizer tableau}, the $n-k$ Pauli operator strings which define the code space. This description is simple and elegant, and has led to remarkable results such as the Gottesman-Knill theorem~\cite{aaronson2004improved}, which efficiently simulates any Clifford operation classically, and the quantum Gilbert-Varshamov bound~\cite{nielsen2002quantum}, which proves by a probabilistic technique that asymptotically good stabilizer codes exist. At the same time, the stabilizer tableau representation is prohibitive in many other ways. Without further structure, it is unclear as to how one might construct desirable codes by simply writing down a clever collection of $n-k$ Pauli strings, or how one might devise and analyze coding algorithms for a given tableau. For this reason, much of the constructive effort in quantum coding theory has relied on developing sophisticated CSS-type product operations that generate CSS codes from nice classical codes, and relying on classical algorithm analysis for inspiration~\cite{tillich2013quantum,breuckmann2021balanced,panteleev2021degenerate,calderbank1996good}.

In this chapter, we argue that a different approach based on graphs to represent stabilizer codes may also prove useful for the development and analysis of quantum codes. The employment of graphs in code construction has a rich history in coding theory.
Classically, Tanner graphs, which encode a classical linear code into a bipartite graph, built the bridge between expander graphs and classical codes and significantly contributed to modern classical coding~\cite{sipser1996expander}.
In the quantum realm, graph reasoning, specifically the topological structure of Euclidean lattices with periodic boundary conditions, motivated the development of the toric and other surface codes~\cite{kitaev2003fault,bravyi2024high}. The toric code was even shown to be efficiently decoded by way of classical graph algorithms, in particular Edmonds' minimum weight perfect matching~\cite{edmonds1965maximum,edmonds1965paths}. Outside of quantum coding, tensor networks have played a key role in quantum simulation and algorithm design.

We have already seen what graphs can do to represent certain non-Clifford operations and stabilizer states.
It is natural to believe that these representations can be extended to quantum circuits and codes.
The critical contributions of graphs in both classical and quantum coding in many specific constructive cases naturally lead to the questions of (a) what the most general graph techniques on stabilizer codes are and (b) how useful such a representation is for a variety of important coding tasks.
This graph-based approach can likely be extended for studying non-Clifford logical operations on stabilizer codes as well.

\subsection{Contributions, scope, and conjectures}

The following chapters take a first step in answering both such questions. Our first contribution, in the direction of the first question, is the development of a simple graph representation of all stabilizer codes. The general structure of the graph is known as an \textit{encoder-respecting form}, which has the structure of a semi-bipartite graph that maps $k$ input-representing nodes to $n$ output-representing nodes. By semi-bipartite, we mean that inputs may not be connected but outputs generically may. Informally, our result may be stated as the following.
\begin{theorem}[Informal] \label{thm:informal}
    Every stabilizer code is essentially uniquely represented by an encoder-respecting graph satisfying some rules, and conversely every such graph essentially uniquely represents a stabilizer code. There is an efficient compilation algorithm that maps tableaus to graphs in $O(n^3)$ time and vice versa in $O(n^2)$ time.
\end{theorem}
The precise version of Theorem~\ref{thm:informal} is given in Section~\ref{sec:zxcf} and Section~\ref{sec:inversion}.
These results establish that our graph representation is sufficiently general to universally represent all stabilizer codes.
We use the compilation algorithm from tableaus to graphs in part for the study of graph properties of well-known codes.
As examples, we compile the famous 5-qubit, 7-qubit, and 9-qubit quantum codes and show that each take simple geometric forms respectively of a cone, cube, and tree-star.
Conversely, Theorem~\ref{thm:informal} enables us to essentially transform any semi-bipartite graph into a stabilizer code, in such a way that (a) any stabilizer code can be constructed this way and (b) the properties of the code are tied to the properties of the graph. As a consequence, we can leverage the tools of graph theory and graph algorithms to study stabilizer codes. We show that the key properties of a stabilizer code are controlled in a unified way, namely by the corresponding graph's degree.

\begin{theorem}[Informal]
    Let $G$ be a graph and $C(G)$ the corresponding code. The following hold.
    \begin{enumerate}[(1)]
        \item The distance of $C(G)$ is bounded above by the min-degree of $G$. 
        \item The stabilizer tableau produced by the algorithm in Theorem~\ref{thm:informal} has weights bounded by a quadratic function of the degree of certain nodes in $G$. 
        \item There exists an efficient algorithm which produces an encoding circuit for $C(G)$ whose depth is bounded by (up to a small additive constant) twice the maximum degree of $G$. 
        \item Every constant-depth diagonal gate and the $\sqrt{X}$ can be implemented logically in $C(G)$ with a circuit of depth, respectively, about (up to an additive constant) twice the max-degree and the max-degree.
    \end{enumerate}
\end{theorem}
The formal versions of these statements are given in Section~\ref{sec:inversion}.

To further strengthen the connection between graphs and codes, we next define a simple, one-player optimization-type game on a graph which we call quantum lights out (QLO), in Section~\ref{sec:inversion}. We prove that QLO unifies many important tasks in quantum coding, which opens a path towards the study of coding algorithms via their correspondence to approximately optimal QLO strategies.

\begin{theorem}[Informal]
    Many coding algorithms such as distance approximation, stabilizer weight reduction, and decoding are all approximately optimal strategies for QLO games.
\end{theorem}

The above results demonstrate that graphs in many senses unify various aspects of stabilizer coding in a simple manner. Overall, we give four pieces of evidence that graphs serve as a useful representation. In the first, which we discussed above, we show that graphs enable efficient algorithms which produce encoding circuits and logical-operators whose depths and weights are controlled by the degree.
Specifically, we give a procedure for transforming a graph into an encoding circuit for a code, such that the circuit's depth is linear in the degree of the graph.
Similarly, we show that the logical operators of the code have a weight linear in the graph's degree.
Later on, when we study decoding, we are able to bound the depths directly by the distance of the code for graphs which satisfy a certain property.

In the second aspect, we explore the geometric and topological properties of graphs as tools for building codes. 
Noticing that simple geometric solids realize celebrated quantum codes, we construct codes based on other such solids, including platonic solids and topological transformations of them. For simpler constructions, we are able to calculate the distance by direct computation. As an example, we produce a non-CSS stabilizer code from a dodecahedron, which has parameters $\llbracket 16, 4, 3\rrbracket$. We also construct a $\llbracket 54, 6, 5\rrbracket$ code by lifting the icosahedron to a covering space. For more complicated constructions, we use more indirect analytical techniques based on decoding. The main advantage of graphs as a constructive tool is their flexibility: we show that graphs enable us to more easily produce codes that have a desired numerical distance, by finding graphs which have a similar degree. Graphs also can be chosen to immediately satisfy desired experimental constraints based on locality, geometric structure, etc. Thus, the graph representation may prove helpful tools for code construction at constant-size scales fit for experimental use.

Our third avenue is decoding. In general, we believe that the graph representation serves as an excellent guide for the design of decoding algorithms. Historical evidence of minimum-weight perfect matching on the surface code suggests that graphs are indeed helpful. In this chapter, we take a first step in that direction by designing a decoder that employs a simple greedy strategy on the QLO game. We give a general condition on graphs, which we call its \textit{sensitivity}, that sufficiently characterizes a graph's amenity to be decoded well greedily.
The sensitivity is the size of the largest intersections of vertex neighbourhoods in the graph.
For concreteness, we also design a simple family of $\llbracket \frac{m 2^m}{m+1}, \frac{2^m}{m+1}\rrbracket$ codes, which we call the hypercube codes. Their distance is provably at most $m$, but we show that the decoder can correct at least $m/4$ errors on it, so that the distance is also at least $m/2$ (up to a $\pm 1$ additive constant). We conjecture, however, that a more careful analysis can prove that the distance is actually $m$. Assuming this conjecture, a small case of the hypercube code family is $\llbracket 112, 16, 7\rrbracket$. Moreover, our results imply a lower bound on the distance in terms of the degree, which in turn lead to bounds on encoding circuit depth, logical operator circuit depth, and stabilizer weight entirely in terms of distance.

\begin{theorem}[Informal]
    Many coding algorithms such as distance approximation, stabilizer weight reduction, and decoding are all approximately optimal strategies for QLO games.
\end{theorem}

\begin{theorem}[Informal]
    If the graph of a code on $n$ qubits has degree $d$ vertices whose neighbourhoods share at most $B$ vertices, there exists efficient greedy decoder as a QLO game instance that will correct at least $\frac{d}{2B}$ errors, putting a lower bound on the code's distance.
\end{theorem}

Finally, we exploit the structure of graphs to give a more fine-grained argument about random codes. By randomly connecting edges of our graphs, we give a large class of codes for which almost all members satisfy a trade-off between rate, distance, and stabilizer weight. This extends the result of the quantum Gilbert-Varshamov bound, which shows by a coarse-grained probabilistic stabilizer tableau analysis a similar result with a rate-distance trade-off but always large stabilizer weight. Such a result further exemplifies that graphs enable greater flexibility in code analysis.

\begin{theorem}[Informal]
    There exists a large family of random graphs which form codes on $n$ qubits with rate $R$, distance $d$ and stabilizer weight no greater than $4R d^2 \log^2 n$.
    This is an extension over the quantum Gilbert-Varshamov bound for arbitrary random codes, which have linear stabilizer weight.
\end{theorem}

We present our results at a very general scope, namely at the level of all stabilizer codes. The universality of these results therefore has limitations that do not exist for highly specific quantum low-density parity check (qLDPC) constructions.
In particular, even when optimized as much as possible with known techniques, code properties such as check weight, encoding circuit depth, and logical operator circuit depth, will scale with the degree and therefore the code parameters, though not necessarily by a great deal (logarithmic, for example).
Thus, our intention is not to use the graph representation as a means for continued study into asymptotically good qLDPC codes with very large constants~\cite{bergamaschi2024approaching,rakovszky2024physics,ashikhmin2001asymptotically,panteleev2022asymptotically}. Rather, we consider our results to be complementary to qLDPC in the sense that we give general insight into stabilizer codes and new flexible ways to construct codes that may be tailored specifically for a given experimental implementation.
As an example, a future quantum long-term memory device may seek a particular large numerical distance threshold but care little for comparably large check weights if the measurement error is sufficiently small. This philosophy underpins our explicit construction of codes which have desirable properties at small scales.
We believe also that further applications of graph theoretical techniques may further improve the techniques used this chapter to build increasingly practical recipes for code constructions that flexibly adapt to experimental constraints.

More generally, we argue that graph theory provides a promising path forward \textit{generally} across avenues of quantum coding theory. Recent breakthroughs in qLDPC codes, for example, have relied on graphs in their construction to generalize the use of surface code lattices~\cite{bravyi2024high}. Perhaps further generalizations of these specific types of constructive techniques, as well as less universal but more fine-grained (i.e. finding certain graph constructions which map to codes in such a way that some properties scale favourably or not at all with parameters) refinements of the universal graph representation will lead to new constructions of and insights into quantum codes.
Additionally, the ability to perform logical non-Clifford operations is of great interest when trying to perform universal, fault-tolerant quantum computation.
Future work can also adapt techniques from Chapter~\ref{chapter:noncliffordgates} to study the effects of both logical and physical non-Clifford gates on codes by using linear combinations of graphs.

\subsection{Related work}

Graph forms have been exploited in many ways in the quantum computation literature, most generally via tensor network representations. One simple case of a graph representation is a graph state. In Chapter~\ref{chapter:qstatesgraphs}, we discussed a generalized graph state representation of stabilizer \textit{states}, i.e. a stabilizer code with 0 logical qubits. Our reduction from tableau to graph utilizes insights from Chapter~\ref{chapter:qstatesgraphs}, and in the special case of codes with no inputs we find a transformation that reduces our tableau-to-graph map to the canonical form from Chapter~\ref{chapter:qstatesgraphs}. We remark that \textcite{mcelvanney2022complete}, building upon results from Chapter~\ref{chapter:qstatesgraphs} which were published previously, derived a different unique canonical form for stabilizer states that are particularly convenient for measurement-based quantum computation. Independently, \textcite{yu2007graphical} also utilized graph theory in quantum coding, albeit at a much more general level (including non-stabilizer codes) for constructive and classification purposes.

Parts of this chapter, particularly the reduction from 
stabilizer tableaus to graphs shown in Section~\ref{sec:zxcf}, use rules and notation from the ZX-calculus.
The ZX-calculus is a graphical language used for expressing quantum operations in a fundamental way.
It represents quantum states, circuits, codes, and operations as tensor networks containing just two types of tensors.
The ZX-calculus is imbued with a handful of rewrite rules which are useful for manipulating such diagrams.
An introduction to the ZX-calculus and a list of rewrite rules is included in Appendix~\ref{app:zx-calculus}.

Recent works have also explored the use of the ZX-calculus to diagrammatically express quantum codes. A recent paper of \textcite{kissinger2022phase} also constructs diagrammatic forms for \textit{CSS codes} specifically. The ZX normal forms used in that work are substantially different for ours and serve a different purpose. Primarily, they are used to establish a correspondence between CSS codes and ZX diagrams by non-uniquely mapping CSS codes to such diagrams and mapping \textit{phase-free} ZX diagrams to CSS codes.
The tableau-to-graph map in Section~\ref{sec:zxcf} extends the conceptual idea of diagrammatic forms inspired by the ZX-calculus to all stabilizer codes, but takes a significantly different approach and utilizes the result for constructive and algorithmic purposes.
We will examine extensions of our results with a focus on CSS codes in a subsequent chapter.

More generally, ZX-calculus approaches to quantum error correction have been considered in various ways. Two relatively recent works which use the ZX-calculus exemplify these different approaches. In the first, \textcite{wu2023zx} represent specifically graph codes with ZX-calculus diagrams. This approach, while not general to stabilizer or CSS codes, has shown promise in code analysis. The second, due to \textcite{chancellor2016graphical}, uses the ZX-calculus specifically to quantize classical codes in the spirit of CSS code construction. Both of these works give evidence that graphs enable new insights into quantum coding, both in their construction and analysis, and thus motivate our completely general theory of graph representations.
Moreover, aside from our tableau-to-graph reduction algorithm which does use elements of the ZX-calculus, our results build everything from just the standard definition of graphs. We emphasize that the graphical diagrams of the ZX-calculus require additional structure. Therefore, graph representations further simplify related previous works that also devise similar representations based on the ZX-calculus. The specific motivations, applications, and analysis in such works are also distinct from those pertinent to our results.

\section{From tableaus to graphs}
\label{sec:zxcf}

We begin by constructing the general graph form for a code and compiling any stabilizer tableau into such a graph. More precisely, we will compile an equivalence class of encoding Clifford circuits corresponding to an equivalence class of stabilizer tableaus, and at the end of this section we discuss how to compile starting with a tableau~\cite{gottesman2009introduction}. As a starting point, we appeal to the ZX-calculus, a graph language for vectors that has become of great interest in quantum information research~\cite{coecke2008interacting,peham2022equivalence,cowtan2022quantum,van2021constructing,east2022aklt,kissinger2022phase}. The ZX-calculus produces visual diagrams that represent quantum states, circuits, and more. As with any formal logical system, the ZX-calculus has a set of rules which may be iteratively applied to transform diagrams into equivalent diagrams. These rules are representations of identities in quantum circuits. The ZX-calculus has become of more interest than ever in fault tolerant quantum computation and quantum compiler theory because it can explicitly visualize quantum properties such as entanglement in an intuitive manner. It has recently been applied to a host of quantum computation problems, including lattice surgery~\cite{de2020zx} and quantum optimization~\cite{kissinger2020reducing}. Importantly, the ZX-calculus is, for stabilizer tableaus, complete (equalities of tableaus can be derived from corresponding ZX diagrams), sound (vice versa), and universal (every quantum operation can be expressed in the ZX-calculus)~\cite{coecke2008interacting,backens2014zx}. 

We are consequently motivated to leverage the ZX-calculus as an intermediate step in the development of a graph representation. Ultimately, our graphs will be simple graphs and not the more complicated and structured diagrams in the ZX-calculus. The diagrams that we map our circuits into are a subset of ZX diagrams that are equivalent by a simple transformation to graphs. To more clearly elucidate the transformation from tableau to graph, we construct the map from a composition of simpler maps. In particular, we first convert the tableau into a circuit. Next we show that up to an equivalence relation, we can efficiently biject the circuits into a class of ZX diagrams of a certain form that satisfies four rules. We thus denote these diagrams \textit{ZX canonical forms} (ZXCFs). Lastly, we remove unnecessary structure from ZXCFs to obtain a graph representation of the code. This is pictorially represented by the following maps. \begin{align} \label{eq:map_sequence}
    \text{Tableau} \longleftrightarrow \text{Encoder} \longleftrightarrow \text{ZXCF} \stackrel{\text{LE}}{\longleftrightarrow} \text{Graph} .
\end{align}
In absence of the last step, everything is completely invertible and has an interpretation as a unique compiler. For the last step, the map between ZXCFs and graphs is invertible modulo \textit{local equivalence}; that is, modulo local (1-qubit) Clifford operations on the outputs.
An immediate consequence for this ZXCF compiler is a diagrammatic method of testing equality between stabilizer codes.

We begin the formalism by noting that properties of a code depend only on the codespace, and thus there is a unitary degree of freedom on the logical space. In particular, encoding circuits for a given stabilizer tableau will produce the same stabilizer code if and only if they differ by an image. We therefore make the following definition.
\begin{definition}
    Two circuits are equivalent as \textit{Clifford encoders} if they have the same image over all possible input states. Equivalently, $\CC_1$ and $\CC_2$ are equivalent encoders if there exists a unitary $U$ acting on the inputs of the circuit such that $\CC_1 = \CC_2 U$ or $\CC_2 = \CC_1 U$. 
\end{definition}
Any reference we make to encoding circuits will either implicitly or explicitly refer to them up to this equivalence, and indeed they all map to the same stabilizer code.

\subsection{ZX canonical form construction}
\label{subsec:zx-construction}

In the intermediate steps using ZX-calculus rules, we follow the standard notation of ZX-calculus graph diagrammatics, specified by \textcite{backens2014zx}.
An introduction to the ZX-calculus, its notations, key definitions, and rewrite rules is included in Appendix~\ref{app:zx-calculus}.

A ZX diagram is a graph with additional structure. In particular, nodes may be either red or green, nodes may be decorated with local Clifford operators, and edges may be Hadamarded. Additionally, nodes may have a ``free edge'', an edge connected only to one node and dangling on the other side.
Green nodes are associated with $Z$ operators, and red with $X$.
Each node may be associated with a local Clifford operator (which may be expressed as a phase which is a multiple of $\pi/2$), and each edge may have a Hadamard gate on it. We colour an edge blue if it has a Hadamard gate on it~\cite{duncan2020graph} (alternatively, we can instead use an edge with a yellow box, such as in~\cite{backens2014zx}).
The circuit takes $k$ input qubits to $n$ output qubits, for $k \leq n$.

The specific structure of our graphs will be those in the following form.
\begin{definition}
    \label{def:respect}
    An \textit{encoder-respecting form} $\mathcal{D}$ has only $Z$ (green) nodes and is structured as a \textit{semi-bipartite graph}. A semi-bipartite graph has a left and right cluster, such that left-cluster nodes may have edges only to right-cluster nodes, but right-cluster nodes have no such restrictions. In our case, the left and right cluster are denoted as the input and output clusters, respectively.
    The input cluster $\CI$ has $k$ nodes associated with 
    the $k$ input qubits of the corresponding encoder, and the output cluster $\widetilde{\CO}$ has $n$ nodes. 
    Each input node has a free (not connected to any other nodes) edge, the input edge. Similarly, each output node also has a free output edge. The output edges are numbered from $1$ to $n$, in order from left to right on a stabilizer tableau or top to bottom on Clifford circuit output wires. For convenience, we refer to the node connected to an output edge numbered $v$ as the output node or qubit numbered $v$.
\end{definition}

The design of the encoder-respecting form graphically illustrates how information propagates from input to output (which edges connect $\CI$ to $\widetilde{\CO}$) as well as the entanglement structure (which edges connect $\widetilde{\CO}$ to $\widetilde{\CO}$) of the underlying encoder. We emphasize that this idea is only for intuition, as there exists equivalence transformation rules on ZX diagrams that appear to change the connectivity in nontrivial ways but nonetheless represent the same tensor.
Note that the structure of $\mathcal{D}$ gives rise to a natural binary ``partial adjacency matrix'' $M_{\mathcal{D}}$ of size $k \times n$, which describes the edges between $\CI$ and $\widetilde{\CO}$ akin to the standard graph adjacency matrix.

Although this structure appears intuitively appealing, it is insufficient because there exists several equivalence transformations in the ZX-calculus that yield an equivalent (but not obviously so) diagram. We therefore constrain an encoder-respecting form into a ZX canonical form (ZXCF) via four additional rules. For notational convenience, we will interchangeably refer to applying an operation on a node's free edge as applying an operation on the node.

\begin{definition}
\label{def:zxcf}

A ZX diagram in encoder-respecting form said to be in \textit{ZX canonical form (ZXCF)} if it satisfies the following rules.

\begin{enumerate}[(1)]
    \item \textit{Edge Rule}: The ZXCF must have exactly one $Z$-node per free edge and every internal edge must have a Hadamard on it. 
    
    \item \textit{Hadamard Rule}: No output node $v$ may both have a Hadamard gate on its output edge and have edges connecting $v$ to lower-numbered nodes or input nodes.

    \item \textit{RREF Rule}: $M_{\mathcal{D}}$ must be a full rank matrix in reduced row-echelon form (RREF).

    \item \textit{Clifford Rule}: Let $\CP \subseteq \widetilde{\CO}$ be the nodes associated with the pivot columns of the RREF matrix $M_{\mathcal{D}}$, so that $|\CP| = k$. The only local Cliffords allowed on nodes are the set $\set{I, S, Z, SZ, H, HZ}$, but there can only be identity operations on input and pivot nodes. There can be no input-input edges or pivot-pivot edges. 
\end{enumerate}
\end{definition}

In subsequent chapters, we will also refer to diagrams in ZXCF as defined in Definition~\ref{def:zxcf} as diagrams in KLS form, diagrams in KLS canonical form, or KLS diagrams.
This is done to disambiguate this particular canonical form from other canonical forms constructed through the ZX-calculus.

The notion of a \textit{pivot node} from the Clifford rule is incredibly important to the results that follow.

\begin{definition}
\label{def:pivot-node}

    For any code or diagram in ZXCF, or more generally obeying the RREF rule, we define the \textit{pivot nodes} of the diagram as follows.
    The adjacencies between the input and output nodes determine the $k\times n$ matrix $M_{\mathcal{D}}$, where the inputs correspond to rows of $M_{\mathcal{D}}$ and outputs correspond to columns.
    Since $M_{\mathcal{D}}$ is in RREF, $k$ of its columns are the \textit{pivot columns}, meaning they each contain exactly one non-zero entry and that entry is the first non-zero entry in its row.
    Each row is thus associated with a pivot column.
    This gives a bijection between inputs and a special subset of $k$ outputs.
    We call these output nodes \textit{pivot nodes} and make frequent use of their properties that each is connected to a distinct unique input node.
\end{definition}

\begin{figure}
    \centering
    \includegraphics[scale=0.5]{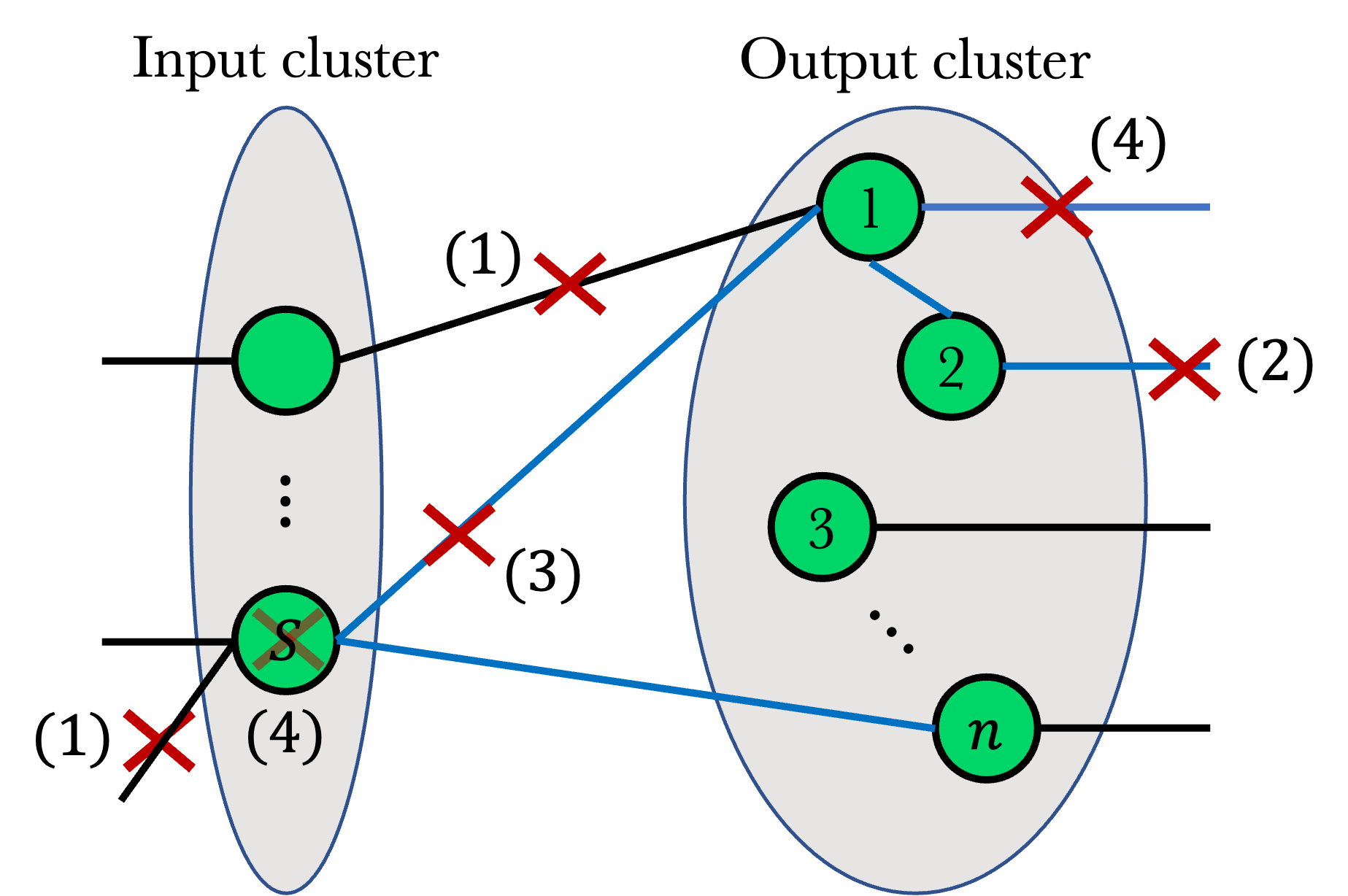}
    \caption[Examples of ZXCF rule violations]{Example of an encoder-respecting form and some ways it might violate the 4 rules. The violations are as follows.
    Rule (1): an internal edge does not have a Hadamard gate on it and a node does not have exactly one free edge. Rule (2): a free edge with a Hadamard gate is connected to a lower-numbered node or output node.
    Rule (3): the partial adjacency matrix $M_{\mathcal{D}}$ is not in RREF.
    Rule (4): there is a non-identity local operation on a pivot node and on an input node.}
    \label{fog:example_violations}
\end{figure}

We have elected for emphasis to include the constraints on only using $Z$-nodes and not using input-input edges constraint in the Edge and Clifford rules, respectively, even though these are redundant with the constraints of the encoder-respecting form.

To provide some visual intuition on these 4 rules, Figure~\ref{fog:example_violations} depicts a generic example of some possible violations to the rules.
We define the four rules such that the ZXCF is as homogeneous as possible in its node type ($Z$ or $X$) and generically prefers not including edges where possible. It can be shown that each of these rules corresponds to certain equivalence transformations~\cite{van2020zx,backens2014zx}. For example, row operations on $M_{\mathcal{D}}$ correspond to unitary controlled-$X$ operations on the input qubits.
To show the uniqueness of the ZXCF, we devise a map from an encoder to a ZXCF and prove that such a map is efficient and bijective.

\begin{theorem}
\label{thm:canonicity}
    Any Clifford encoder has a unique equivalent ZXCF satisfying the Edge, Hadamard, RREF, and Clifford rules. There exists an algorithm, running in worst-case $O(n^3)$ time, that on an input of a description of a Clifford encoder outputs the corresponding ZXCF. This map is bijective.
\end{theorem}

We briefly sketch the proof idea of Theorem~\ref{thm:canonicity} but defer the details to Appendix~\ref{app:compiler}. The proof involves a series of transformations via ZX equivalence rules.

\begin{proof}[Proof Sketch]\renewcommand{\qedsymbol}{}

In the first step, we apply an isomorphism between circuits and states, and proceed to apply the result of Theorem~\ref{thm:canonicalform-state} to preliminarily represent the circuit as the ZX diagram of a corresponding state. We then transform the diagram into one representing a circuit, and carefully choose a series of equivalence transformations that iteratively force the diagram to satisfy a rule above without at the same time violating a different rule that has previously been satisfied. The resulting product is precisely a ZXCF.
\end{proof}

In fact, since the steps are equivalence transformations, distinct encoders map to distinct ZXCFs.
The only remaining aspect of the proof of Theorem~\ref{thm:canonicity} is whether this one-to-one map is bijective. We prove this claim in Section~\ref{sec:compiler}.

\subsection{Tableau to encoder}

There are numerous algorithms for mapping a stabilizer tableau into a Clifford encoding circuit~\cite{nielsen2002quantum}. If we quotient out operations which leave the code space invariant, then any such map is also bijective. We describe one such algorithm, which can be summarized as follows.
For a $(n-k) \times n$ stabilizer tableau, we begin by drawing $n$ output wires. At each step, we first simplify the tableau by applying a Clifford operation and then measure out one of the qubits to remove one row and column from the tableau. Repeating inductively yields a Clifford circuit that takes $k$ qubits to $n$ qubits.

To build an encoding circuit from a tableau, we apply the following sequence of operations.
We start with a circuit containing $n$ wires.
We then apply the following procedure $n-k$ times, as many times as there are generators in the stabilizer tableau.

We consider a generator in our tableau, $g$.
Let $v$ be the index of a non-identity Pauli in $g$
We find any Clifford operation $C$ such that $gC=CZ_v$.
Such an operation can be built up as follows.

We start with a candidate $C'=I$ and repeatedly evaluate $P={C'} g{C'}^\dagger$ until $P$ becomes $Z_1$.
If $P_v$ is equal to $Y$, we left-multiply $C'$ by $S_v$ to turn the $Y$ into an $X$, transforming $C'\to S_vC'$.
If $P_v$ is now equal to $X$, we multiply $C'$ by $H_v$ to turn the $X$ into a $Z$.
We do the same for all other non-identity terms in $P$.
If the sign of $P$ is negative, we multiply $C'$ by $X$ to correct the sign.
For each $Z$ remaining in $P$ at position $u\neq v$, we multiply $C$ by $CX_{u,v}$, which turns the $Z_u$ into the identity in $P$.
What remains now is $P=Z_v$.
Now our candidate $C'$ is a valid choice for $C$.
We apply the operation $C$ to our wires.
Lastly, we apply the post-selected measurement $\bra0$ on qubit $v$, capping it off.
Now, there are 1 fewer wires in our diagram and we can ignore the column of the stabilizer tableau corresponding to $v$.
We can also ignore the generator $g$.
Effectively, our tableau's height and width both shrank by 1 and we reindex our rows and columns to match our remaining wires.

After applying the above subroutine $n-k$ times, our tableau has 0 rows in it and we have $k$ wires remaining.
At this point, we simply read our circuit in the reverse direction.
The $k$ remaining wires become $k$ input wires and the $n$ original input wires become outputs.
Any post-selected measurement $\bra0$ now becomes $\ket0$, an initialization of a work qubit.
Any operation $C$ applied to our wires becomes $C^{\dagger}$ when applied backwards.
Since the only post-selected states allowed by the circuit we built were those stabilized our tableau, all states produced by this encoder must be as well.

Going backwards from an encoder to a stabilizer tableau is much simpler. Every time an ancilla qubit $v$ is initialized in the $\ket0$ state, this adds a stabilizer generator of the form $Z_v$.
Every time a Clifford operation $U$ is applied in the encoder, all existing generators are conjugated by $U$, with generator $g$ turning into $UgU^{\dagger}$.

\subsection{Proof of canonicity}
\label{sec:compiler}

We next prove that our proposed ZXCF is indeed canonical. In other words, two equivalent stabilizer tableaus, generators of the same subspace of the $n$-qubit Hilbert space, will map to the same ZXCF. 
We proceed by counting the number of stabilizer tableaus and ZXCFs. For notational convenience, we will let $m$ be the number of rows of the stabilizer tableau, so $m=n-k$.

\begin{claim}
The number of stabilizer tableaus with $m$ stabilizers on $n$ qubits is
\begin{equation}
\label{eq:zx-count}
\frac{\prod\limits_{i=1}^m\pars{2\cdot4^n/2^{i-1}-2\cdot2^{i-1}}}{\prod\limits_{i=1}^m\pars{2^m-2^{i-1}}}=
\prod_{i=1}^m\frac{2^{2n-i+2}-2^i}{2^m-2^{i-1}}.
\end{equation}
\end{claim}
\begin{proof}
For each row $i$, there are $2 \cdot 4^n$ possible Pauli strings (including the sign) for that row's stabilizer, but the requirement that they commute with previous rows divides the count by $2^{i-1}$. Of these $2 \cdot 4^n/2^{i-1}$ valid strings, $2^{i-1}$ strings are linear combinations of previous strings.
Another $2^{i-1}$ strings are exactly the negation of a linear combination of previous strings.

We have overcounted, however, since there are many ways to find a set of stabilizer generators for a particular tableau. In particular, when choosing a set of generators for a Clifford encoder, we have $2^m-1$ choices for the first generator (subtracting the identity). There are then $2^m-2$ ways of choosing the next element, as we cannot pick anything in the span of the elements chosen so far.
This product counts the multiplicity by which we have overcounted the number of ways to represent each encoder.
Repeating inductively gives the denominator above.
\end{proof}

Next, the number of ZXCF diagrams with $k$ inputs and $n$ outputs can be expressed by the following fourfold recursive function $f$ evaluated at $p = o = 0$, where we still have $m=n-k$.
\begin{equation}
\label{eq:zxcf-count}
f(n,m,p,o)=\begin{cases}
1 & \text{if } n=m=0 \\
A_{n,m,p,o}+B_{n,m,p,o}&\text{else}
\end{cases}
\end{equation}
where we have
\begin{align}
A_{n,m,p,o}=&
\begin{cases}
0 & \text{if } n=m \\
2^of(n-1,m,p+1,o)&\text{else}
\end{cases}\\
B_{n,m,p,o}=&
\begin{cases}
0 & \text{if } m=0 \\
(2^{2p+o+2}+2)f(n-1,m-1,p,o+1)&\text{else}
\end{cases}.
\end{align}

This function is computed with a base case of an empty tableau when $n=m=0$. Derivation of this function is left to Appendix~\ref{app:recursion}. We solve for an explicit form to obtain
\begin{equation}
    f(n,m,p,o)  = 2^{o (n - m)}  \prod_{i=1}^m \frac{(2^{n+1} - 2^i) (2^{n - i + 1 + 2 p + o} + 1)}{2^m - 2^{i-1}}.
\end{equation}
It can be verified by induction that $f$ satisfies the conditions of Equation~(\ref{eq:zxcf-count}) and that $f(n,m,0,0)$ gives the same expression as in Equation~(\ref{eq:zx-count}), proving that the ZXCF is indeed canonical.

\subsection{ZXCF to graph}
In the final step, we reduce a ZXCF directly to a graph. We remark that any local Clifford operation must, by the Clifford rule, be placed only on non-pivot output nodes. However, local Cliffords do not change the weights of the set of correctable errors. Thus, local Cliffords do not change the code distance or any other code parameters. If we also quotient out this equivalence, we may remove all local Cliffords entirely. Since all internal edges have Hadamards, we can remove those as well and simply remember that they exist for inversion. Similarly, we can remove all free edges. The remaining diagram is identical to a graph, specified entirely by a collection of nodes and edges. As a consequence, we have completed the map sequence from Equation~(\ref{eq:map_sequence}), except for the graph to ZXCF map.
We have also thus proved the following theorem.
\begin{theorem}
    For any stabilizer code, there exists at least one equivalent code that has a representation as a semi-bipartite graph with no local Clifford operations.
\end{theorem}

A well-studied family of stabilizer codes are the CSS codes, which can be formed by using two classical codes obeying certain orthogonality properties.
CSS codes are named after Calderbank, Shor, and Steane~\cite{calderbank1996good,steane1996multiple}.

\begin{definition}
\label{def:css-code}

A CSS code is a stabilizer code which can be expressed as a set of stabilizer generators $S$ obeying the following property.
There exists a way to partition $S$ into two sets, $S_X$ and $S_Z$ such that the only non-identity Pauli matrices in each set of generators are $X$ and $Z$, respectively.
The $S_X$ measurements detect $Z$ errors and vice versa.
CSS codes generally make decoding a code easier as the decoding of $X$ and $Z$ errors can be done completely independently.
\end{definition}

From experimental testing we can make the following observation.

\begin{theorem}
\label{thm:CSS_bipartite}
    Bipartite graph correspond to CSS codes.
    Specifically, the graph of any CSS code will be bipartite and any bipartite graph can be turned into a CSS code.
\end{theorem}

We prove this theorem fully in the next chapter, which focuses more heavily on CSS codes.

\begin{proof}[Proof Sketch]\renewcommand{\qedsymbol}{}
We sketch the proof of one of the directions of the theorem.
If a graph is bipartite, its vertices can be split into two sets, $V_1$ and $V_2$.
The graph is invariant under the operation of applying an $X$ operation to a node $v$ and $Z$ operations to all of $v$'s neighbours.
If we apply Hadamard gates to $V_1$, producing a locally equivalent code, then for all vertices $v\in V_1$ or $v\in V_2$, the invariants will only contain $Z$ or $X$ operations, respectively.
\end{proof}

\section{Application to quantum codes}
\label{sec:applications}

It is illuminating to observe the graph representation of well-known quantum codes. Not only does it inform of geometric structure embedded in the code, but it also provides a guide as to how such codes may be generalized.
We provide three examples for study.
Consider first the nine-qubit code, due to \textcite{shor1995scheme}, which uses 9 physical qubits to encode 1 logical qubit. The Shor code may be represented by the stabilizer tableau \begin{align}
\begin{bmatrix}
    Z & Z & I & I & I & I & I & I & I \\
    Z & I & Z & I & I & I & I & I & I \\
    I & I & I & Z & Z & I & I & I & I \\
    I & I & I & Z & I & Z & I & I & I \\
    I & I & I & I & I & I & Z & Z & I \\
    I & I & I & I & I & I & Z & I & Z \\
    X & X & X & X & X & X & I & I & I \\
    X & X & X & I & I & I & X & X & X 
\end{bmatrix} .
\end{align}
In Figure~\ref{fig:nine-qubit-code}, we give the graph representation in (a) and the ZXCF in (b). In the ZXCF, the input node is denoted by \textbf{I} and free edges are dashed for emphasis.
In the graph representation, the input node is blue and the outputs are black.
Let us first examine the ZXCF. There are three identical sectors of the outputs, two with Hadamarded outputs and one without. This resembles our expectations from examination of the unnormalized qubit representation of the Shor code, $(\ket{000} \pm \ket{111})^{\otimes 3}$. The graph, on the other hand, is a star-shaped tree. It is easy to imagine how one might extend the Shor code into an infinite family of graphs, namely by applying several iterations of recursively giving every outer (leaf) node several child nodes.

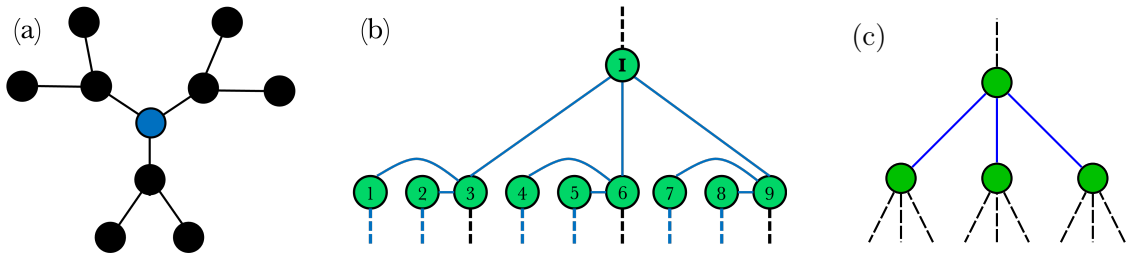
\begin{figure}[ht]
    \centering
    \includegraphics[scale=0.4]{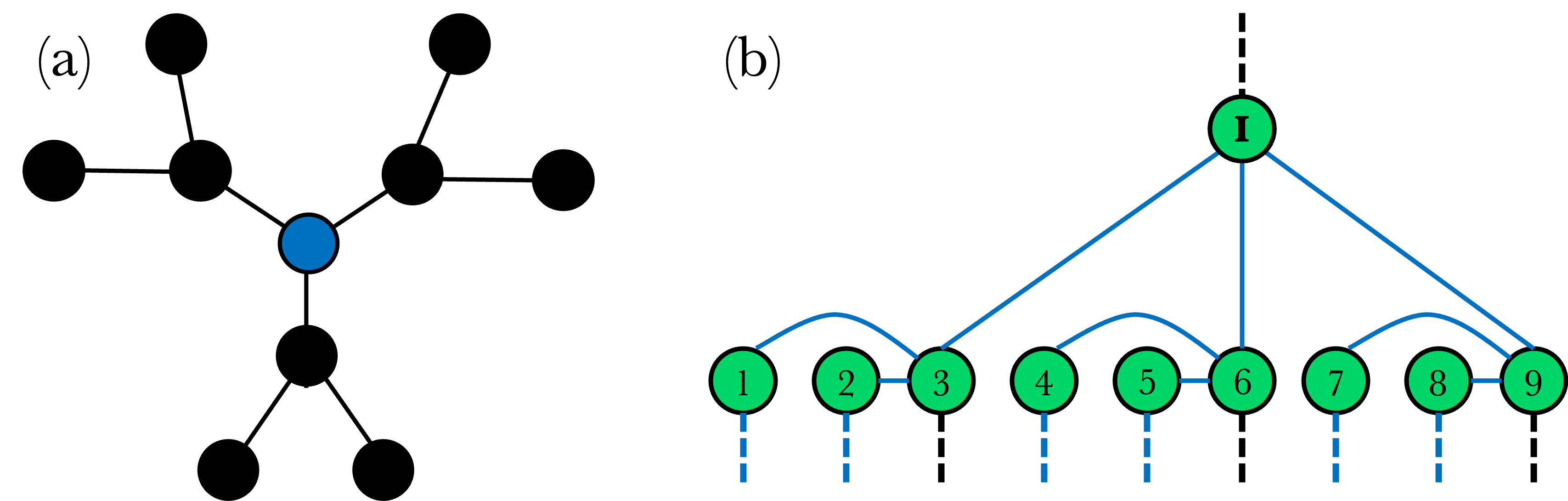}\hspace{4ex}\raisebox{0.3em}{
    \begin{tikzpicture}[scale=1.7]
	\begin{pgfonlayer}{nodelayer}
        \node [style=none] at (-1, 1.1) {\small (c)};
		\node [style=z-node-demo,thick] (0) at (-0.75, 0) {};
		\node [style=z-node-demo,thick] (1) at (0, 0) {};
		\node [style=z-node-demo,thick] (2) at (0.75, 0) {};
		\node [style=z-node-demo,thick] (3) at (0, 0.75) {};
		\node [style=none] (4) at (0, 1.25) {};
		\node [style=none] (5) at (-0.75, -0.5) {};
		\node [style=none] (6) at (-0.5, -0.5) {};
		\node [style=none] (7) at (-0.25, -0.5) {};
		\node [style=none] (8) at (0, -0.5) {};
		\node [style=none] (9) at (0.25, -0.5) {};
		\node [style=none] (10) at (0.5, -0.5) {};
		\node [style=none] (11) at (0.75, -0.5) {};
		\node [style=none] (12) at (1, -0.5) {};
		\node [style=none] (13) at (-1, -0.5) {};
	\end{pgfonlayer}
	\begin{pgfonlayer}{edgelayer}
		\draw [draw=blue,thick] (0) to (3);
		\draw [draw=blue,thick] (3) to (1);
		\draw [draw=blue,thick] (2) to (3);
		\draw [thick,dashed,dash pattern={on 5pt off 1pt}] (3) to (4.center);
		\draw [thick,dashed,dash pattern={on 5pt off 1pt}] (0) to (13.center);
		\draw [thick,dashed,dash pattern={on 5pt off 1pt}] (0) to (5.center);
		\draw [thick,dashed,dash pattern={on 5pt off 1pt}] (6.center) to (0);
		\draw [thick,dashed,dash pattern={on 5pt off 1pt}] (1) to (7.center);
		\draw [thick,dashed,dash pattern={on 5pt off 1pt}] (8.center) to (1);
		\draw [thick,dashed,dash pattern={on 5pt off 1pt}] (1) to (9.center);
		\draw [thick,dashed,dash pattern={on 5pt off 1pt}] (10.center) to (2);
		\draw [thick,dashed,dash pattern={on 5pt off 1pt}] (2) to (11.center);
		\draw [thick,dashed,dash pattern={on 5pt off 1pt}] (12.center) to (2);
	\end{pgfonlayer}
\end{tikzpicture}}

    \caption[9-qubit code]{(a) Graph representation and (b) ZXCF of the 9-qubit code. The input in the graph is shown in blue. Note that the graph encoder uses a slightly different basis as since its output edges do not contain Hadamard gates. (c) A simpler representation of the 9-qubit code, formed by applying the identity removal and Hadamard cancellation rules from Definition~\ref{def:zx-basic-rewrite-rules}.}
    \label{fig:nine-qubit-code}
\end{figure}

Next, we compile Steane's seven-qubit code, a central construction in many fault tolerant quantum computation schemes~\cite{steane1996multiple}. The code is represented by the following tableau: \begin{align}
    \begin{bmatrix}
        I & I & I & X & X & X & X & \\ 
        I & X & X & I & I & X & X & \\ 
        X & I & X & I & X & I & X \\
        I & I & I & Z & Z & Z & Z \\
        I & Z & Z & I & I & Z & Z \\
        Z & I & Z & I & Z & I & Z
    \end{bmatrix} .
\end{align}

\begin{figure}[ht]
    \centering
    \includegraphics[scale=0.475]{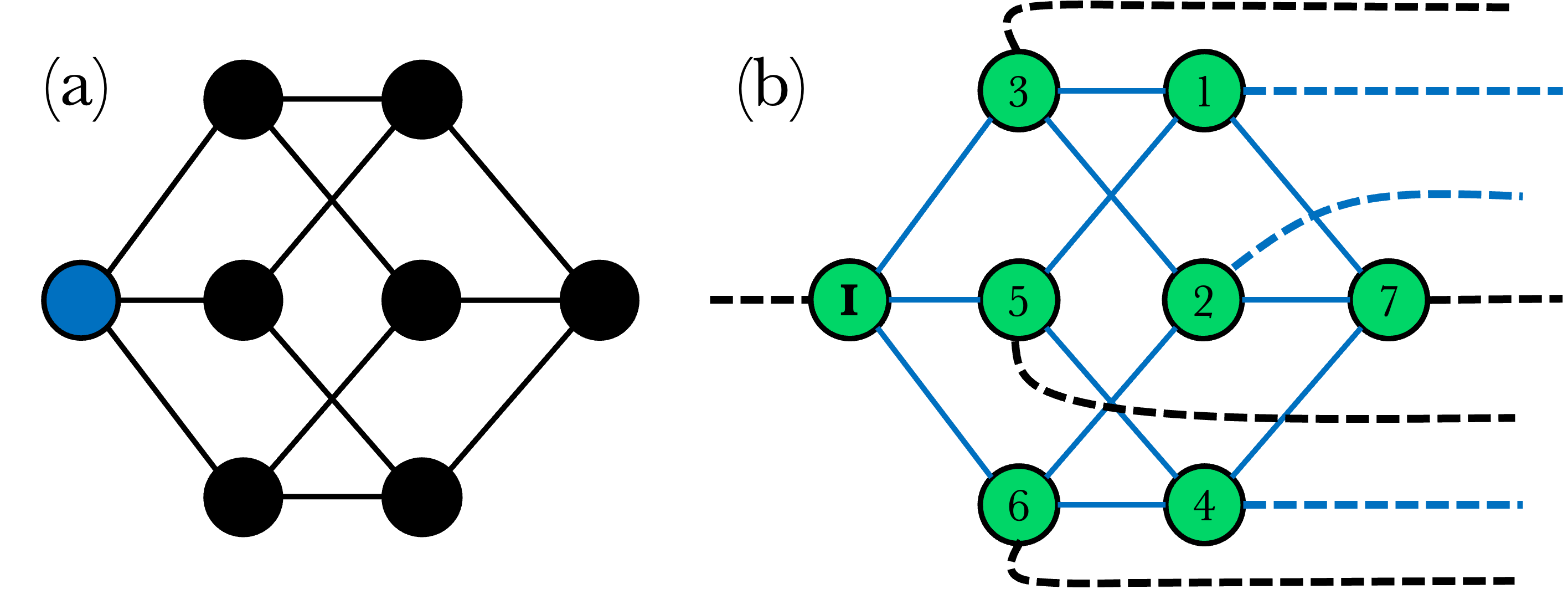}
    \caption[7-qubit code]{(a) Graph representation and (b) ZXCF of the 7-qubit code. The input in the graph is shown in blue. Note that the graph encoder uses a slightly different basis as since its output edges do not contain Hadamard gates.}
    \label{fig:7-qubit-code}
\end{figure}

In our ZXCF, this code takes the form given in Figure~\ref{fig:7-qubit-code}. In particular, the nodes of the diagram are simply the corners of a cube. A similar picture has been given in a different context by \textcite{duncan2013verifying}.
Several generalizations of the Steane code come to mind.
One involves using graphs for polyhedra other than the cube and another involves constructing hypercubes of arbitrary dimension. We explore precisely such generalizations in Section~\ref{sec:constructions}.

We observe that the $X$ and $Z$ operations in the stabilizers of the 7-qubit code are positioned at exactly the $1$-indices in the binary representation of the qubit's index.
We can start to see some elegant symmetries in Figure~\ref{fig:7-qubit-code} with the positions of the nodes and their expressions in binary.

As a final example we consider the five-qubit code, which is the smallest code one can achieve for correction of an arbitrary single-qubit error and has been studied experimentally~\cite{gottesman2009introduction,knill2001benchmarking}. Its tableau representation is given by
\begin{align}
    \begin{bmatrix}
        X & Z & Z & X & I \\
        I & X & Z & Z & X \\
        X & I & X & Z & Z \\
        Z & X & I & X & Z 
    \end{bmatrix} .
\end{align}

\begin{figure}[ht]
    \centering
    \includegraphics[scale=0.475]{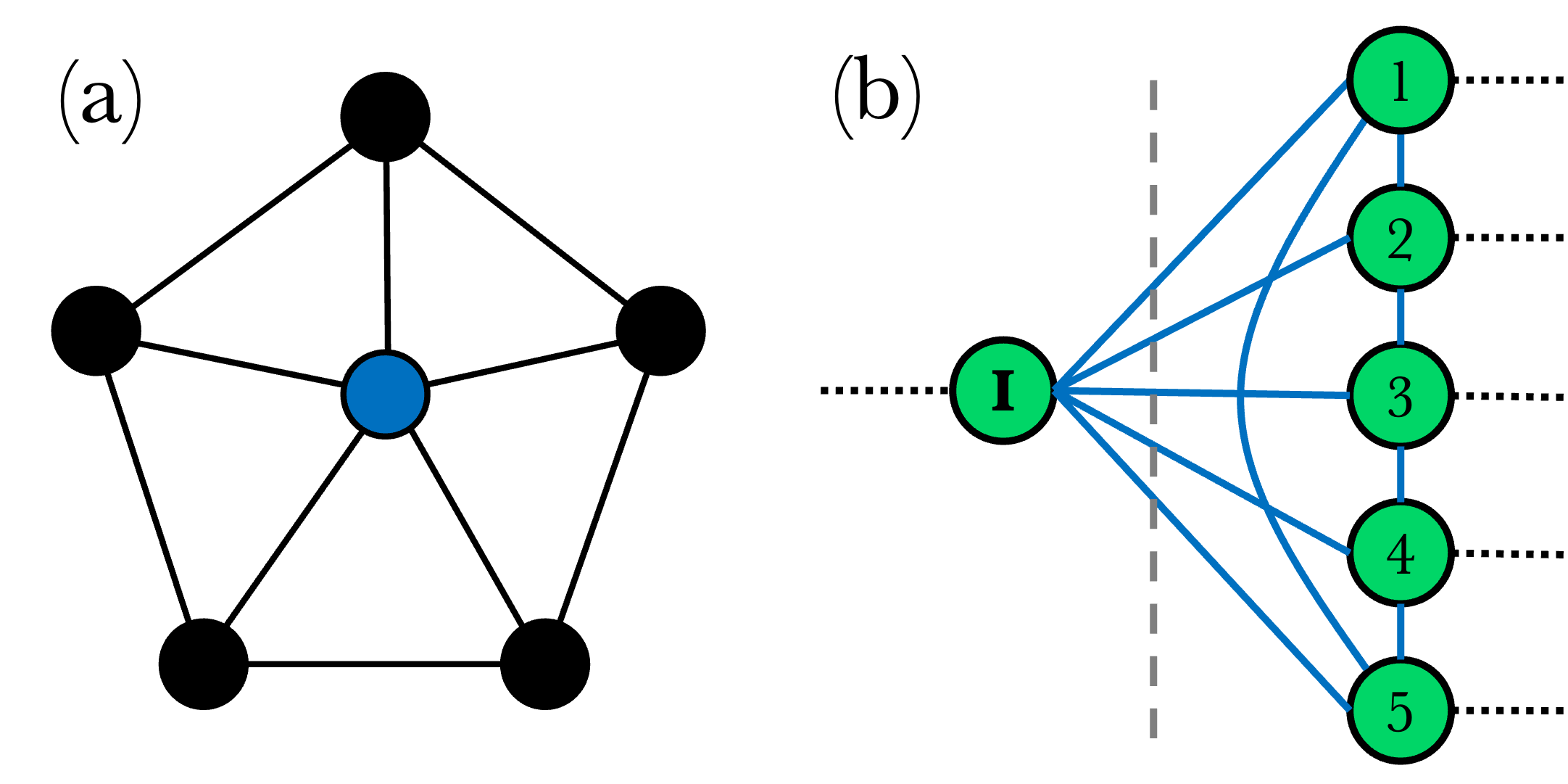}
    \caption[5-qubit code]{(a) Graph representation and (b) ZXCF of the 5-qubit code. The input in the graph is shown in blue.}
    \label{fig:5-qubit-code}
\end{figure}

Our construction produces the ZXCF shown in Figure~\ref{fig:5-qubit-code}.
The 5-qubit code takes the shape of either a pentagon with an input in the middle or a cone, depending on how one sets up the graph.
Consequently, a simple generalization of such a code would be to use a $n$-gon, or equivalently add more nodes to the circle in the cone.

Such elegant and simple structure in the 5, 7, and 9-qubit codes suggest that other graphs with similar elegant structure may yield interesting quantum codes, either of constant size or in family. We explore this idea in Section~\ref{sec:constructions}.

\section{The inverse map: from graphs to stabilizer codes}
\label{sec:inversion}

The bijective relation between stabilizer codes and graphs, particularly those in the ZX canonical form, give justification that graphs are at least as powerful as stabilizer tableaus for the expression of codes. In this section, we lay the foundation for the use of graphs as \textit{constructive tools} for code design. Such a mapping from graphs to stabilizer tableaus can be viewed as an inverse map to the algorithm given in Section~\ref{sec:compiler}, with some caveats relating to the fact that not all graphs can be equivalently represented by ZXCFs and therefore not all graphs can be mapped into codes.

Let $G$ be a graph with $n+k$ nodes, partitioned into $k$ input nodes $\CI$, a choice of $k$ pivot nodes $\CP$, and $n-k$ output nodes $\CO$. To ensure that our operations are well-defined, there must exist equivalence operations that map $G$ into a ZXCF.
The only nontrivial restriction this requirement poses is on the choice of pivots, whose conditions are given in Definition~\ref{def:pivot-node}. The RREF rule ensures that every distinct pivot is connected to exactly one distinct input; in other words, $\CP$ and $\CI$ are perfectly matched. 

\begin{figure}[ht]
    \centering
    \includegraphics[width=0.4\linewidth]{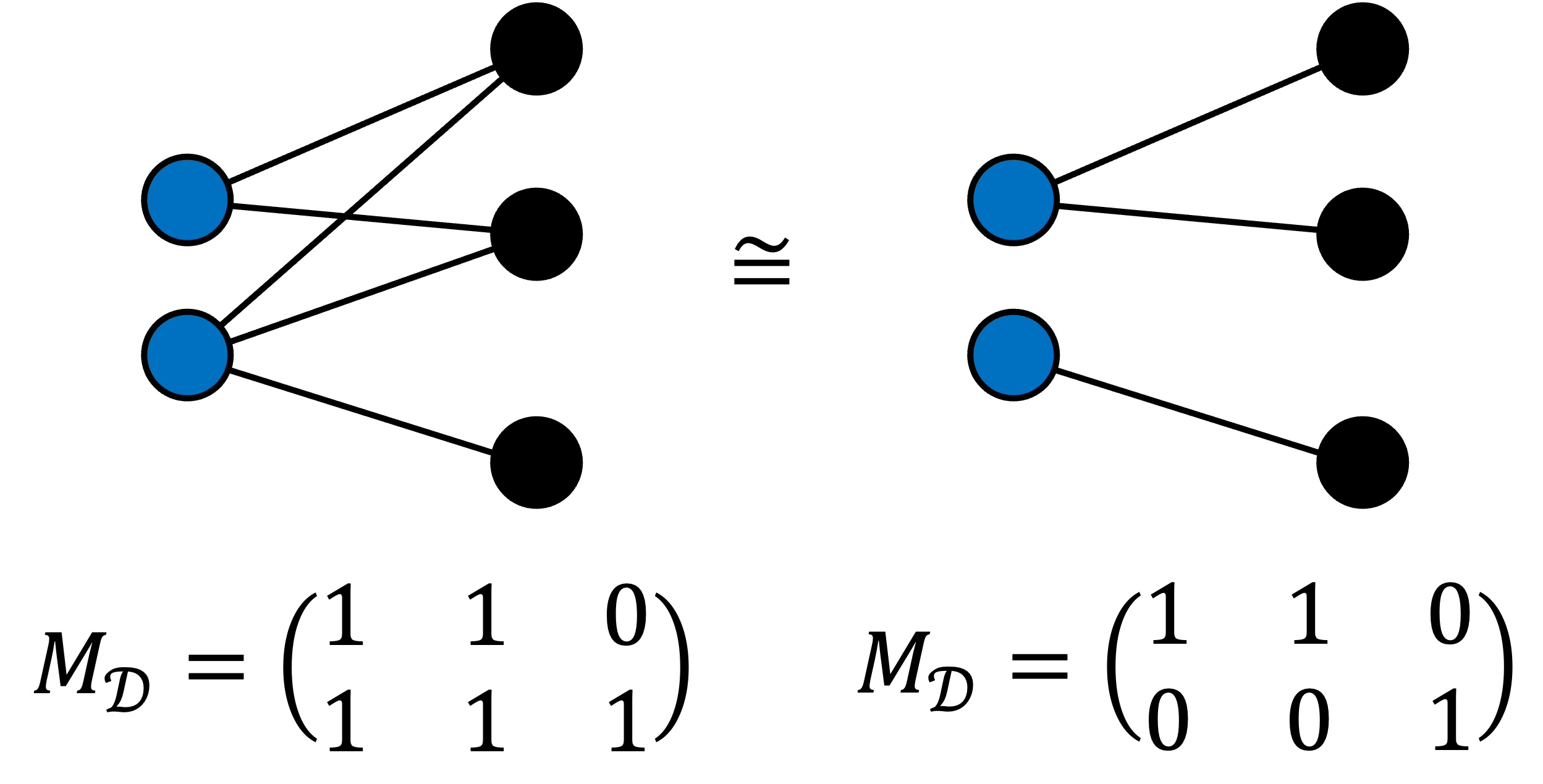}
    \caption[Pivot assignment constraint example]{A simple example for which the third output node is necessarily a pivot. The inputs are blue and outputs are black. An example graph code is given on the left, and an equivalent graph is given on the right.}
    \label{fig:badpivot}
\end{figure}

Since row operations on the partial adjacency matrix are equivalences, it suffices to choose pivots such that there exists a row-reduced partial adjacency matrix $M_{\mathcal{D}}$ for which the chosen $k$ pivot nodes correspond to the $k$ pivot columns of $M_{\mathcal{D}}$. A simple example of this subtlety is shown in Figure~\ref{fig:badpivot}. Here, it is impossible to choose the first two outputs to be the pivots because the third output is forced to be a pivot when we row-reduce. Note, however, that this subtle restriction poses no issue if every distinct pivot we choose is already connected to exactly one distinct input (as in the case in Section~\ref{subsec:randcodes}) or if we can observe directly that the pivots we have chosen do correspond to the pivot columns after row reduction. Henceforth, we will denote $G$ as a graph for which a set of valid pivots $\CP$ have been specified alongside a set of input nodes $\CI$.

Assuming the choice of $\CP$ is valid, $G$ will be equivalent to a ZXCF by a series of equivalence transformations. First, we place a free edge on every node and place a Hadamard on all internal edges, satisfying the edge rule. The Hadamard rule is vacuously satisfied, as we will not place any Hadamard on output nodes. We can row-reduce the partial adjacency matrix to satisfy the RREF rule. Next, we strip input-input edges, which are unitary controlled-$Z$ operations on the inputs and thus represent equivalent encoders. Lastly, as discussed in Appendix~\ref{app:compiler}, there is an equivalence relation that strips pivot-pivot edges. This procedure results in a ZXCF with no local Clifford on nodes at all. While there exist encoders that can only be expressed with local Clifford operations that this map cannot produce, the map can always produce a locally equivalent encoder, and is otherwise completely general.

In practice, we never transform our graphs into a ZXCF. The only transformation we shall make is endow every node with a free edge. This endowment enables a simpler description of operator actions on the graph. However, in our diagrams we omit the free edges to reduce clutter. We can also strip input-input edges and pivot-pivot edges if desired for convenience. Aside from that, knowing that the graph is equivalent to a ZXCF, we can write down the stabilizer generators of its corresponding ZXCF directly. To do so, we first establish notation. Let $\mathop{\D}\limits_{i=1}^n A_i = A_1 \,\D\, \cdots \,\D\, A_n$ for sets $A_1, \hdots, A_n$, where $\D$ is the symmetric difference. The symmetric difference is both commutative and associative, so this iterative procedure is well-defined.

Consider some vertex $v$ in our graph. We define $i(v) = \set{w \in \CI \,|\, w \sim v}$ the inputs connected to $v$ and $p(v) =\set{w \in \CP \,|\, w \sim v}$ the pivots connected to $v$, which each consist of only one vertex if $v$ is a pivot or an input, respectively. For later use we similarly define $o(v)=\set{w \in \CO \,|\, w \sim v}$, the set of non-pivot outputs connected to $v$. We also define $N_o(v) = \set{w \in \CO \cup \CP \,|\, w \sim v}=o(v)\cup p(v)$ the non-input neighbours of $v \in \CO$, and more generally $N_o(A) = \mathop{\D}\limits_{v \in A} N_o(v)$ and similarly for $i(A)$, $p(A)$, and $o(A)$. The symmetric difference is useful for correctly counting the parities of Pauli operations applied to certain nodes, i.e. if a Pauli is applied twice in the set, it is effectively not applied at all.

Let $X_w$ be a Pauli on (the free edge of) node $w$, and $X_A = \bigotimes\limits_{v \in A} X_v$.
Pauli $Z$ operators are defined analogously. Note that we may only directly apply operators to (the free edge of) \textit{output} nodes because these are physical operators and not logical ones.
With this notation, we define the following map $\CS : G \mapsto S$ from graph to stabilizer tableau, the output of which we call the \textit{canonical stabilizers} of the code represented by $G$. \begin{align} \label{eq:inverse}
    \CS \coloneq \set{X_v Z_{N_o(v)} X_{p(i(v))} Z_{N_o(p(i(v)))}}_{v \in \CO} .
\end{align}
Denote each stabilizer $S_v$ for $v \in \CO$. (Note that we are not using local Clifford operations, so $S$ refers to a Pauli stabilizer and not a phase gate.)
These stabilizers are readily seen to be a valid generating set. There are exactly $n-k$ generators, and all are independent because only $S_v$ contains the term $X_v$. Every pair of stabilizers $S_u$ and $S_v$ also have even $X$-$Z$ overlap, and thus commute, because the graphs are undirected and if $u\in N_o(v)$, then $v\in N_o(u)$ and vice versa. In principle, one could disregard the entirety of the previous section, \textit{define} Equation~(\ref{eq:inverse}) as a map from graphs to stabilizers, and explore their properties. However, it is not immediately clear from this approach how expressive such a class of stabilizers would be and the motivation behind such an assertion would be obscure.

We next show that such a map generates precisely the stabilizer group of the encoder that the corresponding ZXCF of $G$ represents. Since graphs do not have local Clifford operations associated to them, we can treat $\CS$ as the inverse map from a graph to a ZXCF, up to local equivalence. To do so, we appeal to the two most fundamental equivalence rules of the ZX-calculus, the merging/unmerging rule and the $\pi$-copy rule from Definition~\ref{def:zx-basic-rewrite-rules}. We will use both of these rules many times in the course of our calculations.
More details on the ZX-calculus and its rules can be found in Appendix~\ref{app:zx-calculus}.

\begin{lemma}[ZX-calculus fundamental equivalences]
\label{lemma:rules}
Let $v$ be a $Z$-node in a ZX diagram. Operators may be placed on edges and associated with the vertex on either side of the edge.
    \begin{enumerate}[(1)]
    \item If an operator placed on an internal edge $(u,v)$ is \textit{associated} with $u$, it can be equivalently associated with $v$ instead. However, if the internal edge is Hadamarded, the operator is flipped $X \leftrightarrow Z$ during this change of association. An operator placed on a free edge of $v$ is vacuously associated with $v$.
    \item If an $X$ operator is placed on an edge associated with $v$, it can be equivalently commuted through $v$ into $X$'s on \textit{every other edge} (including the free edge) connected to $v$, with each new $X$ still associated with $v$.
    \item If a $Z$ operator is placed on an edge associated to $v$, it may be equivalently commuted through as a $Z$ on \textit{any one} other edge associated to $v$, remaining associated to $v$.
\end{enumerate}
\end{lemma}

The rules provided in Lemma~\ref{lemma:rules} give a prescription for generating stabilizers. Fix $v \in \CO$. First, place a $X$ on the free edge of $v$. By rule (2), this $X$ equivalently transforms into $X$ on all other edges connected to $v$, associating with $v$. For edges connecting $v$ to a non-input node $u$, we re-associate them by rule (1) to $u$, and since all internal edges are Hadamarded, these operators become $Z$. They commute through by rule (3) to the free edge of $u$ and remain there as $Z$. We perform a similar transformation except that we do not commute the $Z$'s to the free edges of inputs because these correspond to operations on inputs (i.e. logical operations) rather than operations on outputs (i.e. physical operations). Instead, we commute them through canonically to the pivot associated with that input, turning the $Z$ back into the $X$ as we re-associate the operator with the pivot. Finally, we commute the $X$ through the pivot. Since the $X$ was on an internal edge between the input and pivot, a $X$ appears on the free edge of the pivot. The remaining $X$'s become $Z$'s associated with the output neighbours of the pivot and commute through to their free edges by rule (3). To reverse all these operations and achieve the identity, thereby creating a stabilizer, we apply all these operations in reverse, precisely $X_v Z_{N_o(v)} X_{p(i(v))} Z_{N_o(p(i(v)))}$. Therefore, we have proven the following theorem.

\begin{theorem}[Inversion]
    Given a graph $G$ with a feasible choice of pivots $\CP$, the stabilizer tableau associated with the ZXCF equivalent to $G$ is given by Equation~(\ref{eq:inverse}).
\end{theorem}

Note that the sets on which the $X$ and $Z$ operations are applied are not necessarily disjoint, so pairs of $Y$ terms can appear if $v$ is adjacent to both an input and its pivot.
The $Y$ operations will appear on $v$ and on the respective pivot.
As stated, the expression in Equation~(\ref{eq:inverse}) is correct, but the order of the terms matters as they do not necessarily commute.

From here, we can also easily write down a canonical set of \textit{logical} Pauli operations for a given graph, by using the same two transformation rules from above. For each input node $v \in \CI$ corresponding to the $v^{\text{th}}$ logical qubit, construct $\overline{X}_v$ by placing $X_v$, which is by definition a logical $X$, on the $v^{\text{th}}$ logical qubit. To obtain the physical operator, we commute $X_v$ to the (free edges of the) output nodes, where $X_v$ commutes through to become $Z_{N_o(v)}$. Hence, $\overline{X}_v = Z_{N_o(v)}$. Similarly, $\overline{Z}_v = X_p Z_{N_o(p)}$, where $p$ is the pivot associated with $v$. We place an $X$ on $p$ to canonically define logical $Z$, but any neighbour of $v$ would suffice. In this way, we have used the equivalence rules on every possible node: on $\CO$, they create stabilizers; on $\CI$, they create logical $X$'s; and on $\CP$, they create logical $Z$'s.
We have proven the following.
\begin{theorem}[Logical operations]
\label{thm:logical-ops}
    Given a graph $G$, the logical operations associated with the ZXCF equivalent to $G$ are given by the following. If $v$ is an input vertex with a neighbour $u$, $\overline{X}_v = Z_{N_o(v)}$ and $\overline{Z}_v = X_u Z_{N_o(u)}$.
\end{theorem}

As a consequence of our canonical logical operator construction, we observe that the degree of $G$ controls the upper bound of the distance as well as the stabilizer weights of $\CS(G)$.
\begin{corollary} \label{corollary:degree}
The distance of the code with stabilizers $\CS(G)$ is at most the minimum degree of any vertex in $\CI \cup N_o(\CI)$.
\end{corollary}
\begin{proof}
    We use the logical operators from Theorem~\ref{thm:logical-ops}.
    For each input vertex $v\in\CI$, $X_v=Z_{N_o(v)}$ is a logical operation.
    The weight of this operation is equal to the degree of $v$.
    Similarly, for each vertex $u\in N_o(\CI)$ which neighbours some input $v$, we have the operation $Z_v=X_u Z_{N_o(u)}$
    This size of this set is equal to $\left|N_o(u)\right|+1$, which is at most the degree of $u$, since $u$ neighbours $v$.
    The construction of these logical operations proves the desired upper bound on the distance.
    This also shows that the statement of the corollary could be strengthened slightly by replacing the degree of input neighbours $u$ by $\left|N_o(u)\right|+1$.
\end{proof}

\begin{corollary} \label{corollary:stab_weight}
    Given a graph $G$, define \begin{align}
        \d^*_{\CO} = \max_{v \in \CO} \deg(v) ,\quad \d^*_{\CO \CI} = \max_{v \in \CO} |i(v)| ,\quad \d^*_{\CP \CO} = \max_{v \in \CP} |N_o(v)|  .
    \end{align}
    Then $\max\limits_{S \in \CS(G)} |S| \leq 1 + \d^*_{\CO} + \d^*_{\CO \CI} \d^*_{\CP \CO}$, where $|S|$ is the weight of a stabilizer $S$. Moreover, if every path between every pair of nodes $v, w \in \CI$ is of length at least 3, so that every node $u \in \CO$ is connected to at most one node in $\CI$, then $\max\limits_{S \in \CS(G)} |S| \leq \d^*_{\CO} + \d^*_{\CP \CO}$.
\end{corollary}
This result is a direct consequence of Equation~(\ref{eq:inverse}).

In the next section, we will show that another important property, encoding depth, is also bounded above by the degree. We make two remarks regarding the utility of such degree bounds. First, for a general stabilizer code it is unclear as to how one may give a nontrivial upper bound on properties like distance and circuit depth. The degree allows us to algorithmically do so by compiling a stabilizer tableau into a graph and then examining the degree. Second, it is possible, at least in certain cases, to upper bound the degree itself by some function of the code parameters. We show bounds of this type in Section~\ref{subsec:decoder} and Section~\ref{subsec:randcodes}. When such a bound is possible, it provides nontrivial bounds on code properties in terms of functions of other parameters.

\subsection{Encoding circuit with degree-bounded depth}
\label{subsec:encoding_circuit}

A particularly appealing quality of the graph formalism is that it enables us to efficiently and canonically construct encoding circuits which are bounded by a small linear function of the \textit{degree }of the graph. This construction can be taken advantage of in two ways. If we already have a stabilizer code and would like to construct a nice encoding circuit, we can compile the code's tableau into a graph and then apply the procedure below. Note that Theorem~\ref{thm:encoding} assumes a graph representation, so if the ZXCF has local Cliffords, they are not included in the encoding circuit. They can be appended to the encoding circuit produced by the theorem, at the cost of increasing the depth by 1, since all local Cliffords are single-qubit operations. Alternatively, if we have a graph which represents a stabilizer code, we can directly write down this encoding circuit for the code and have a guarantee of its depth.

\begin{theorem}
\label{thm:encoding}
    Suppose that one may initialize qubits to the $\ket{+}$ state.
    Given a graph $G$ with maximum degree $\d^*$, there exists an efficient algorithm which outputs an encoding circuit $\CC$ for the code with stabilizers $\CS(G)$, such that $\operatorname{depth}(\CC) \leq 2 \d^* + 3$.
\end{theorem}

\begin{proof}
We construct explicitly the transformation from the graph to its encoding circuit. An informal description is as follows. A graph with $|\CI| = k$ and $|\CO \cup \CP| = n$ requires a circuit taking $k$ input and $n-k$ ancillary qubits to $n$ output qubits. We choose to send the $k$ inputs to the $k$ pivots on the output side.
Edges in $G$ will generally become controlled-$Z$ gates in the encoding circuit, with the exception of edges from $\CI$ to $\CP$, which become a layer of Hadamards.
More precisely, we construct the circuit $\CC$ as follows.
\begin{enumerate}[(1)]
    \item \textit{Initialization.} Create $n$ wires, one for each node in $\CO \cup \CP$. Label the output side of the wires by the index of the node. For each wire associated to a pivot $v$, label the input side of the wire by the input corresponding to $v$. The remaining $n-k$ wires unlabelled on the input side are the ancillary qubits. Initialize these unlabelled wires to $\ket{+}$.
    \item \textit{Input-output edges.} For each input $u \in \CI$ and $v \in \CO$, if $v \in o(u)$ then add a $CZ_{u, v}$ to the circuit. The order does not matter since all $CZ$ gates commute.
    \item \textit{Input-pivot edges.} For each input $u \in \CI$ apply $H_u$. These will correspond to the edges between each input and its corresponding pivot.
    \item \textit{Output-output edges.} For each pair of non-input vertices $u,v \in \CO \cup \CP$ where the edge $(u,v)$ is in $G$, apply $CZ_{u,v}$.
\end{enumerate}
There are two claims to show: (a) this construction correctly yields an encoding circuit for the code described by $\CS(G)$, and (b) $\operatorname{depth}(\CC) \leq 2\d^* + 3$. The first claim is a direct consequence of the rules of the ZX-calculus and the form of the ZXCF.
By turning the encoding circuit into a ZX diagram and merging $Z$ nodes together, we recover the ZXCF of the code.
Examples of this transformation are shown in Section~\ref{subsec:dodecahedron} and Appendix~\ref{app:convert-zxcf-to-circuit}.

Specifically, the initialization of qubits in the state $\ket+$ in step (1) is represented by a $Z$ node with a single free edge in ZX-calculus.
In turn, $CZ$ gates in steps (2) and (4) are applied by a pair of $Z$ nodes connected by a Hadamarded edge.
The Hadamard gate in step (3) is applied as a Hadamard gate.
Now, we simplify the ZX diagram.
Every $Z$ node along each wire initialized as $\ket+$ is merged into a single node.
All $Z$ nodes along each input wire are merged into two nodes, an input and a pivot, with one cluster on each side of the central Hadamard gate.
The Hadamard gate on the wire between each input and its pivot now becomes exactly the Hadamarded edge we desired.

Note that the gates are ordered in this particular way because of the causal structure of the ZX diagram. First, input edges must propagate information to the output edges, so the associated input wires must have their $CZ$ gates to establish their interaction with the ancilla wires. Next, the circuit completes the interactions from input to output via the input-pivot Hadamards. Finally, the output edges interact, so the output-output $CZ$ gates are placed on the ancillary wires. This completes the proof of the first claim, correctness.

For the second claim, we observe that since all $CZ$ gates commute, the depth is given by the minimum number of groupings of $CZ$ gates which act on disjoint sets of qubits. Such a problem reduces to finding the optimal \textit{edge colouring} of $G$, i.e. the minimum number of colours needed to colour every edge, such that no two edges which share a node have the same colour. Although finding the optimal edge colouring is known to be \textbf{NP}-complete, Vizing's edge colouring theorem guarantees that the minimal colouring is at most $\d^* + 1$, and moreover there exists an efficient algorithm which finds such a colouring~\cite{vizing,misra1992constructive}.
Thus, step (1) does not require any depth, steps (2) and (4) each has depth at most $\d^* + 1$, and step (3) has depth 1, so the total depth is at most $2\d^* + 3$ as claimed.
\end{proof}

Although encoding is practice is often performed by a method to prepare the logical zero state, e.g. repeatedly measuring the stabilizers until they are all $+1$, Theorem~\ref{thm:encoding} shows that for codes that have sufficiently small degree, it may be possible to have a much stronger guarantee, namely to efficiently encode any state directly via an encoding circuit.

This construction is also described in Appendix~\ref{app:convert-zxcf-to-circuit} with some minor variations as well as step-by-step illustrations for an example graph code.
Additionally, Section~\ref{subsec:dodecahedron} gives a example of such an encoding circuit for one of our constructed codes and discusses how the depth may actually be substantially smaller in practice by solving small cases of minimal edge colouring. Moreover, in Section~\ref{subsec:decoder}, we will show that the degree, and thus the encoding circuit depth, can be bounded above by code parameters for certain families of codes. This result will thus directly upper bound the degree of the encoding circuit by code parameters.

\subsection{Logical non-Pauli gates on graph codes}
\label{subsec:gates}

Earlier, we showed that the logical Pauli operations can be extracted immediately from the graph representation. Specifically, $\overline{X}_v$ and $\overline{Z}_v$ for $v \in \CI$ are implemented respectively by $Z_{N_o(v)}$ and $X_{p(v)} Z_{N_o(p(v))}$. We derived these logical Paulis by placing a Pauli on the free edge of a particular input node, and then applying ZX equivalence rules to push the Pauli through the node into the output.

For a completely general logical unitary operation $\overline{U}$ and graph code $G$, the only algorithm to physically implement $\overline{U}$ is to unencode, apply $U$, and then re-encode. Such a technique is typically undesirable, since the encoding circuit is necessarily high-depth to achieve a good distance, which implies that the logical operation's implementation strongly propagates errors and therefore resists fault tolerance. Nonetheless, if the encoding circuit has a depth which is sufficiently small, it is conceivable that in some cases this generic procedure may be of use. Moreover, we already have an generically optimized encoding circuit from Section~\ref{subsec:encoding_circuit}, which in turn implies the following general result.

\begin{theorem}[Generic logical operation] \label{thm:generic_logical_gate}
    Let $G$ be a graph code with maximum degree $\d^*$, and let $U$ be a unitary which can be implemented in depth $d_U$. Then there is an efficient algorithm which constructs a circuit $\CC_U$ that logically implements $U$ on $G$, such that $\operatorname{depth}(\CC_U) \leq 4 \d^* + 6 + d_U$.
\end{theorem}

Conceivably, if a depth of $\d^*$ is not immediately prohibitive, an additional factor of $4$ significantly reduces the likelihood of practical utility. However, we will show that many important gates, including diagonal Clifford, non-diagonal Clifford, and diagonal non-Clifford gates, can be implemented generically with a reduction on the constant factor.

\begin{theorem}
    For any graph code $G$ with maximum degree $\d^*$, every diagonal gate $U$ can be implemented logically on $G$ with depth at most $2\d^* + 5$.
\end{theorem}
\begin{proof}
    The circuit is still simply the unencode-apply-reencode operation, but with cancellations due to the diagonal structure. Let $E$ be the encoding circuit from Theorem~\ref{thm:encoding}. The circuit is of depth at most $2\d^* + 3$ because it is of the form $CZ_V H_V CZ_V$, i.e. a layer of $CZ$'s acting on all the nodes $V$ in the graph, a layer of $H$'s, then another layer of $CZ$'s. These respectively contribute at most $\d^* + 1$, $1$, and $\d^* + 1$ to the encoding circuit depth. The logical operation is given by $CZ_V H_V CZ_V U CZ_V H_V CZ_V = CZ_V H_V U H_V CZ_V$ since $U$ commutes with all $CZ$ operations. This circuit has a reduced depth of $(\d^* + 1) + (1) + (1) + (1) + (\d^* + 1) = 2\d^* + 5$.
\end{proof}
As a corollary, the controlled-$Z$ ($CZ$), phase ($S$), and $\frac{\pi}8$ ($T$) gates can all be generically applied with a reduction by a factor of $2$ from Theorem~\ref{thm:generic_logical_gate}. 

For a certain non-diagonal Clifford gate, the $\sqrt{X}$ gate, we are able to even further reduce the constant to just 1. This improvement relies on a specific ZX equivalence relation involving a purely-based graph transformation, a local complementation about a vertex (LCV) as defined in Definition~\ref{def:graph-local-complementation-vertex}.
In Definition~\ref{def:graph-local-complementation-edge}, we showed how three local complementations about a pair of adjacent vertices create a local complementation about an edge (LCE).
We have already seen the effect of an edge local complementation on a graph state in Equation~(\ref{eq:edge-local-complementation}).
Local complementations are defined for any graph and has been presented under various guises in the literature~\cite{backens2014zx,van2020zx}.

We have already seen the LCV used in the proof of Lemma~\ref{lemma:orbit-under-local-comp}, in the expressions in Table~\ref{table:psipq-expressions}, and in Appendix~\ref{app:compiler} to enforce the Clifford rule of the ZXCF.

Below, we fix a vertex set $V$ and consider different graphs which can be made with the nodes in $V$.
Given a graph $G = (V, E)$, let $N(u) = \set{v \in V \,:\, (u, v) \in E}$ be the neighbours of $u \in V$ in $G$.
As always, we can assume that $G$ does not permit self-edges; that is, $\forall u \in V,\, (u, u) \notin E$. 

Another way to express local complementations about a vertex is as follows.
Let $V$ be fixed and $G=(V,E)$ be a graph. A \textit{local complementation about a vertex} $v\in V$ is a graph transformation $\text{LCV} \,:\, \mathcal{G} \times V \to \mathcal{G}$ given by $(V, E) \mapsto (V, E \,\D\, K(N(v)))$, where $K(S)$ is the edge set of a complete graph formed by nodes in $S \subseteq V$. As before, we denote the resulting graph $\text{LCV}(G,v)$ as $L_v(G)$.

A LCV about a vertex $v$ can thus be described as follows.
Consider the edges of the complete graph built from $N(v)$, and then toggle all these edges in $G$. That is, if $e \in E \cap K(N(v))$, then remove $e$ from $E$, and if $e \in K(N(v)) \setminus E$, add $e$ to $E$.

The LCV builds a bridge between unitary operations on graph codes and inherent graph transformations. Suppose we have a graph $G = (\CI \cup \CO \cup \CP, E)$ and we wish to apply a logical unitary $\overline{U}$; that is, $\overline{U}_v$ for $v \in \CI$.
Generically, if we attempt to use the ZX equivalence rules to push $\overline{U}_v$ to the output, the rules will add more local operators and change edges, resulting in $\overline{U}$ applied to the output $\CO$ of a \textit{different graph}, which corresponds to a different code.
To make this distinction clearer, let $U_v(G)$ be the operator $U$ applied to $v \in V$ on $G$. If $v \in \CI$, then $U_v(G)$ is a logical operator; otherwise, $U_v(G)$ is a physical operator. For brevity, if $U = U^{(1)} \dots U^{(m)}$, we denote $U_{v}(G) = U^{(1)} \dots U^{(m)}_v(G)$. We also denote $U_{v_1}(G) \otimes \cdots \otimes U_{v_m}(G) = U_{\set{v_1, \dots, v_m}}(G)$.

The specific connection between the LCV and the $\sqrt{X} = HSH$ gate is given in Theorem~\ref{thm:vertex-local-complementation} and Equation~(\ref{eq:simplification-hsh}), which we restate here.
Proofs of this relation in the ZX-calculus are given by \textcite{backens2014zx} and \textcite{van2020zx}. 
\begin{claim}
\label{clm:logical}
Let $G = (V, E)$ be a graph. Then $\sqrt{X}_{v}(G) = HSH_{v}(G) =S_{N(v)}(L_v(G))$.
\end{claim}

Therefore, in order to apply a logical $\sqrt{X}$ on some input vertex $v \in \CI$, we apply the LCV about $v$, and then apply a $S$ gate on all nodes in $N(v)$. Since $N(v) \subseteq \CO \cup \CP$, all the $S$ gates are physical.

\begin{lemma} \label{lemma:lc-bound}
    Given a graph $G$ representing a code, the graph transformation $L_v(G)$ for input vertex $v \in \CI$ can be implemented with a quantum circuit of depth at most $\deg(v)$. There is an efficient algorithm outputting this circuit. 
\end{lemma}
\begin{proof}
    The graph $L_v(G)$ is constructed from $G$ by applying $CZ_{u,w}$ for all $(u, w) \in K(N(v))$. 
    Since $|N(v)| = \deg(v)$, each node in $K(N(v))$ has degree $\deg(v) - 1$. By Vizing's theorem, there is an efficient algorithm to arrange the $CZ$ gates to have depth at most $(\deg(v) - 1) + 1 = \deg(v)$.
\end{proof}

\begin{theorem}
    Let $G$ be a graph. There is an efficient algorithm which finds a circuit that implements the logical $\sqrt{X} = H S H$ gate on an input node $u$ with depth at most $\deg(u) + 1$.
\end{theorem}
The proof is a direct consequence of Claim~\ref{clm:logical} and Lemma~\ref{lemma:lc-bound}.

In general, it seems difficult completely reduce the implementation depth of an arbitrary Clifford circuit because the Hadamard gate does not appear amenable to further optimization. To have a complete set of Clifford operations which can be implemented in depth strictly better than that of Theorem~\ref{thm:generic_logical_gate}, one solution is to start with a self-dual CSS code, which admits a transversal Hadamard. In this case, the local complementation simplifications that push a $H$ gate through to the output will cancel out in just the precise way as to produce $H$'s on all output nodes while preserve the graph structure. $\set{H, \sqrt{X}, CZ}$ is a generating set of the Clifford group, and thus the above results imply that any self-dual CSS code implements logical Cliffords at a depth that is reduced by at least a factor of 2 from the generic procedure of Theorem~\ref{thm:generic_logical_gate}. 

\subsection{Unification of stabilizer coding algorithms}
\label{subsec:game}

Thus far, given a graph, we have described its corresponding stabilizers, given a set of logical operators, and provided a simple upper bound on the distance. To complete our fully graph formalism of stabilizer codes, we will show that important code algorithms, distance calculation or approximation, stabilizer weight reduction, and decoding, can all be expressed in a unified graph manner, namely by a single class of one-player games on the graph.

We begin by mentioning an existing game knows as ``Lights Out''. In Lights Out, a player aims to turn all lights off on a grid with some lights initially on.
Each cell in the grid contains both a light and a button or a switch.
The player uses moves consisting of pressing a button or toggling a switch at a given cell, thereby \textit{toggling} all of the adjacent cells' lights as well as the light in the cell.
Note that this might cause some lights that were previously off to turn off.
The game ends when all lights are off.

\begin{definition}[Quantum lights out]
An instance of a quantum lights out (QLO) game on a graph $G$ is described as follows. The graph $G$ consists of both input nodes $\CI$, pivot nodes $\CP$ and non-pivot output nodes $\CO$, as usual.
However, the notion of a free edge is not needed for QLO and will be discarded.
Every node is endowed with both a light and a switch, both of which are binary (i.e. on of off) and which are both initially off.
Flipping a switch at node $v$ both (a) permanently destroys the light at $v$ if $v \notin \CI$ and (b) toggles all intact (i.e. non-destroyed) lights on the neighbours of $v$. Let Alice be the player. The graph begins with some initial configuration of lights. The game consists of two rounds: \begin{enumerate}[(1)]
    \item Depending on the instance of the game, Alice may be allowed to flip switches of input nodes $\CI$.
    \item Alice chooses a set of $s$ nodes $v_1, \dots, v_s \in \CO \cup \CP$, subject to some some instance-dependent constraints. For each node $v_i$, in order, Alice chooses whether to toggle and destroy the switch at $v_i$ or only to destroy the light on $v_i$. The game ends when the lights reach an instance-specific desired final configuration, usually with all non-input lights off or destroyed.
\end{enumerate}
The specific instance of the QLO game specifies the initial light configuration, whether round (1) exists, the constraints in round (2), and the desired final configuration, e.g. all lights are turned off. Regardless of instance, Alice must make at least one move in round (2). We say that Alice wins in $s$ moves if she makes $s$ choices that lead to the desired final configuration of the game.
\end{definition}

By correctly choosing the instance of the game, we can represent many important stabilizer coding algorithms as strategies for QLO games. We begin by showing that finding the distance of a code is QLO. Note that the distance is invariant under unitary operations on inputs, and hence every equivalent Clifford encoder has the same distance.
\begin{theorem}[Distance is QLO]
    Given a graph $G$, denote the distance of the code represented by $d(G)$, which is calculated by the following QLO instance. Initially, all lights are off. In round (1), Alice may flip any number $s$ of switches in $\CI$, the input nodes of $G$. In round (2), Alice has no constraints and may flip any non-input switch or destroy any non-input light. The final configuration consists of all non-input lights off or destroyed.
    Additionally, we require that either at least one input light is on or that $s \geq 1$. Under this instance,
    \begin{equation}
        d(G) = \min_{\CA} m_{\text{D}}(\CA, G) ,
    \end{equation}
    where $\CA$ is Alice's strategy and $m_{\text{D}}(\CA, G)$ is the number of moves she made in round (2). As a consequence, every distance approximation algorithm on $G$ which is correct up to a multiplicative error of $1+\e$ yields a strategy in the distance QLO game instance that is optimal up to a factor of $1+\e$ in the number of steps.
\end{theorem}
\begin{proof}
    Each of the switches flipped on the $v^{\text{th}}$ input corresponds to $\overline{X}_v$, i.e. logical $X$ on logical qubit $v$.
    If the game allowed for flipping a single light at a node, this would correspond to applying a $Z$ on the node. There is no point to applying a $Z$ operator multiple times to the same output node, as they will cancel (and we can order them so that no phase occurs). Hence, instead we simplify by destroying the light altogether as a reminder that re-toggling this light does not increase an operator's weight.
    Next, toggling lights on all neighbours of a node $v$ corresponds to applying $X_v$.
    There is no additional Pauli weight if we apply both $X_v$ and $Z_v$, so for convenience we also destroy the light on $v$ when we flip the switch.
    Thus, destroying lights and flipping switches correspond to physical $Z$ and $X$ operators.
    Lights turned on at inputs after all lights on outputs have been extinguished are logical $Z$ operators.
    Input lights that would have been on had they not been destroyed are logical $Y$ operators.
    Consequently, the entire distance QLO game consists of applying some logical $X$ operators and translating them into physical operators and logical $Z$ operators. In other words, the physical operator represented by the $n$ moves apply the logical product of the initial logical $X$'s and the final logical $Z$'s, which is a nontrivial logical operator if at least one of the two is not vacuous, i.e. if at least one input light is left on or $s \geq 1$. Since each move in round (2) increases the physical operator's weight by exactly 1, the minimum number of moves is precisely the distance.
\end{proof}
Since finding the distance of a stabilizer code is \textbf{NP}-complete~\cite{kapshikar2023hardness}, we may immediately conclude that optimally playing QLO on arbitrary graphs is also \textbf{NP}-complete. 

\begin{theorem}[Weight reduction is QLO] 
    Let $\set{S_i}$ be the canonical stabilizer generators of $G$. Without loss of generality, let $S_1$ be the first stabilizer. Then the minimum-weight stabilizer $S_1'$ which can replace $S_1$ in the tableau, i.e. that remains independent of $S_2, \dots, S_{n-k}$ is given by the following QLO game instance. Initially, all lights are off. There is no round (1). In round (2), Alice may freely flip switches on any non-input node except $v_1$, the vertex in $\CO$ corresponding to $S_1$. Specifically, she must flip the switch at $v_1$ \textit{exactly once}. The final configuration consists of all lights destroyed or off, including those at inputs. Under this game instance,
    \begin{equation}
        |S_1'| = \min_{\CA} m_{\text{WR}}(\CA, G) ,
    \end{equation}
    where $\CA$ is Alice's strategy and $m_{\text{WR}}(\CA, G)$ is the number of moves she made in round (2).
    Thus the minimum weight is equal to the minimum number of moves to end the game.
\end{theorem}
We formulate the weight reduction theorem in terms of the canonical stabilizers in the graph representation. We observe that an analogous result applies for a different set of stabilizers, so long as they are independent.

\begin{proof}
    Initially, we apply $X_{v_1}$ when Alice toggles the switch at $v_1$ in round (2).
    The canonical stabilizer represents the trivial strategy of flipping pivot switches to remove lights from inputs, and then destroying all remaining lights. However, any sequence of moves that turns all lights off, i.e. effectively applies the identity on the graph, is also a stabilizer. Since we prohibit Alice from either not flipping the switch at $v_1$ at all or flipping it more than once, the resulting stabilizer is both nontrivial (i.e. not the identity) and independent from all other $S_i$, since $S_1'$ is still the only stabilizer containing $X_{v_1}$. Hence, the minimum number of moves represents the new weight as claimed.
\end{proof}

By examining the proof, we also observe that the act of writing down the canonical stabilizers and the canonical logical operators are also instances of substantially simpler QLO instances.

\begin{theorem}[Decoding is QLO]
\label{thm:decodingQLO}
    Suppose that after the application of a noisy channel to a codeword, we measure the canonical (or more generally, any known set of) stabilizer generators of the graph of our code, $G$.
    We thereby obtain syndrome bits $s_1, \dots, s_{n-k} \in \set{\pm 1}$. Then, the recovery map that should be applied is given by a QLO game instance.
    Initially, the light at each vertex $v_i\in\CO$ corresponding to stabilizer $S_i$ is turned on if $s_i = -1$.
    All remaining lights start off.
    In round (1), Alice may flip any number of switches in $\CI$.
    In round (2), Alice has no constraints aside from only being able to act on nodes in $\CO\cup\CP$, as usual.
    The final configuration consists of all non-input lights off or destroyed.
    Under this instance, \begin{align}
        r(S,G) = \min_{\CA} m_{\text{DC}}(\CA,S,G) ,
    \end{align}
    where $S$ is the set of syndromes, $r$ is the number of moves in the optimal recovery map, $\CA$ is Alice's strategy and $m_{\text{DC}}(\CA,S,G)$ is the number of moves she made in round (2). We note that input lights being on does not affect anything in round (2) or in the minimization condition.
    This means that, where convenient, input lights can be completely ignored.
\end{theorem}
\begin{proof}
    We show that at each step of the decoding procedure, the current syndromes $s_i$ (corresponding to $v_i$) are given by $(-1)^{\eta_i}$, where $\eta_i$ is the parity of the number of lights that are currently on in certain sets $\mathcal{T}_i \coloneq \set{v_i} \cup p(i(v_i))$. We proceed by casework on each possible type of error, and the stabilizers that they affect.
\begin{enumerate}[(1)]
    \item $E = Z_v$ for $v \in \CO$. $E$ anticommutes with stabilizers which contain $X_v$. The unique canonical generator which contains $X_v$ is $S_v$. Thus, applying a $Z$ to a node in $\CO$ toggles its light.
    \item $E = Z_v$ for $v \in \CP$. $E$ anticommutes with stabilizers which contain $X_v$. These are the stabilizers associated with the \textit{non-pivot output neighbours} of the \textit{input} associated with $v$, i.e. $S_u$ for all $u \in o(i(v))$. Applying a $Z$ to a node in $v \in \CP$ negates the syndromes of \textit{all} 
    stabilizers of the vertices $o(i(v))$, and therefore should toggle all lights in $o(i(v))$.
    We can either view the action of the $Z$ as toggling the lights $o(i(v))$ or, equivalently, as just toggling the light at $v$, since this will also negate all of the desired syndromes.
    The reason for this is that $v\in \mathcal{T}_i$ for all $v_i$ in $o(i(v))$.
    Toggling the lights $o(i(v))\cup\{v\}$ can be done in round (1) by toggling the switch at $i(v)$, so this equivalence is justified.
    This means that toggling the light at $v$ is equivalent to toggling the lights of $o(i(v))$.
    Thus, applying a $Z$ to a node in $\CP$ toggles its light.
    \item $E = X_v$ on any $v \in \CO \cup \CP$. $E$ anticommutes with stabilizers which contain $Z_v$. These are $S_u$ for $u \in o(v)$ as well as $w \in o(i(p(v)))$ due to the presence of $Z$'s in $S_u$ on $o(p(i(u)))$ (note that the order of $i$ and $p$ is reversed when we consider stabilizers affected by a node). Applying an $X$ to a node $v \in \CO \cup \CP$ toggles all lights in $o(v) \,\D\, o(i(p(v)))$, or equivalently, $o(v) \cup p(v)$.
    As shown previously, $X(v)$ is equivalent to $Z$ on all neighbours of $v$, some of which may be pivots. 
    Applying an $X$ to $v$ flips the switch at $v$, and thereby toggles the lights on $N_o(v)$, its non-input neighbours. This is equivalent to toggling the lights $o(v)$ and then toggling $o(i(p(v)))$.
    Thus, applying an $X$ to a node in $\CO\cup\CP$ toggles all neighbouring lights.
\end{enumerate}
A decoding procedure involves turning off lights by acting on nodes in $\CO \cup \CP$ (corresponding to physical recovery operators), using the fewest number of moves possible. For $v \in \CO$, let $S_v$ be the canonical stabilizer associated with $v$.
The moves in round (1) do not change the syndromes, since for any stabilizer $S_i$, toggling any input switch will necessarily toggle either zero or two lights in $\mathcal{T}_i=\set{v_i} \cup p(i(v_i))$, depending on whether or not the input neighbours $v_i$.
Additionally, we note that breaking a light at a node after applying the switch at that node represents ambivalence between applying an $X$ or a $Y$, but if Alice is interested in finding the proper recovery map instead of just the map's size, she should leave all lights unbroken.
Then, when she is done toggling switches, she can choose to turn off any lights at vertices whose switches she toggled without having this count as a move in round (2).
We conclude that decoding algorithms are QLO strategies.
\end{proof}
We note that construction of the recovery map is very closely related to the calculation of the distance.

The unified expression of all three algorithms, combined with the many complex techniques to perform optimizations on graphs that have been developed in past decades, suggest that the study of approximating optimal QLO strategies may be a promising pathway for devising coding algorithms.

\section{Applications of the graph formalism}
\label{sec:constructions}

We proceed to provide constructive evidence in three areas for the utility of a graph representation. First, in Section~\ref{subsec:smallcodes}, we argue that for near-term experimental purposes, the graph formalism provides a simple and flexible prescription for the design of codes with a desired rate $R$ and distance $d$. Intuitively, this advantage occurs because we can appeal to the geometrical and topological structure of graphs, as well as take inspiration from well-studied graphs, to get close to a distance $d$. We then describe improved results for random codes and construct an efficient graph decoder for a family of stabilizer codes constructed from graphs.

\subsection{Flexible code design}
\label{subsec:smallcodes}

Generally, stabilizer tableaus have proven more useful as a description rather than a constructive mechanism. Most well-known stabilizer codes constructed are actually CSS codes~\cite{panteleev2021degenerate,panteleev2022asymptotically,bravyi2024high,kitaev2003fault}, and non-CSS stabilizer codes are usually found by some alternate representation. The lack of use of stabilizer tableaus is not unusual, as it is simply not obvious as to what collections of Pauli strings construct codes with high distance, easy decoding, gates, geometric properties etc. The universality of the graph representation of stabilizer codes, combined with a rich history of graphs in error correction~\cite{tanner1981recursive,sipser1996expander,bravyi2024high,kissinger2022phase}, suggests that a promising technique to construct certain codes lies in finding sufficiently well-behaved graphs.

Due to its universal expressive power, the graph representation in this chapter has parameters that scale with distance. As we showed previously, the distance is bounded above by degree parameters of the graph, and the canonical stabilizers have weights that also scale with the degree. Therefore, the graph representation is best suited for finding codes with a desired approximate numerical values of code parameters, as opposed to asymptotically good LDPC codes. In other words, graphs provide a flexible technique to generate potentially practical small-scale codes.

As an example of this flexibility, if one wanted to create a code with a particular distance, one can set the degree of the graph in accordance with Corollary~\ref{corollary:degree}.
On graphs that do not have large overlaps in sets of neighbours, the distance will be determined by the degree of the input vertices as well as the output-degrees of input neighbours.
Achieving a certain rate translates to having a sufficiently high density of inputs, which in turn can influence the required degree of the input neighbours.
Some experimental devices or algorithms may benefit from codes that have at most a particular number of physical qubits or encode a minimum number of logical qubits, both of which are easy to set in the graph formalism.
If ensuring that the physical qubits behave similarly in a highly uniform manner, a graph could be chosen to be highly regular and symmetric, such as the 5-, 7-, and 9-qubit codes from Section~\ref{sec:applications}.
Further still, a graph can be chosen to be planar and local in two dimensions, or perhaps supporting only a fixed number of long-range connections for ease of practical implementation.
On the other hand, a graph embedded into a higher-dimensional space, such as a high-dimensional grid graph for example, could be easier to represent or simulate classically. We will give examples of codes with each of these properties below.

The distance calculations of codes in this section were performed by direct computation. The code is freely available for use~\cite{khesinlu_graphcodes}.

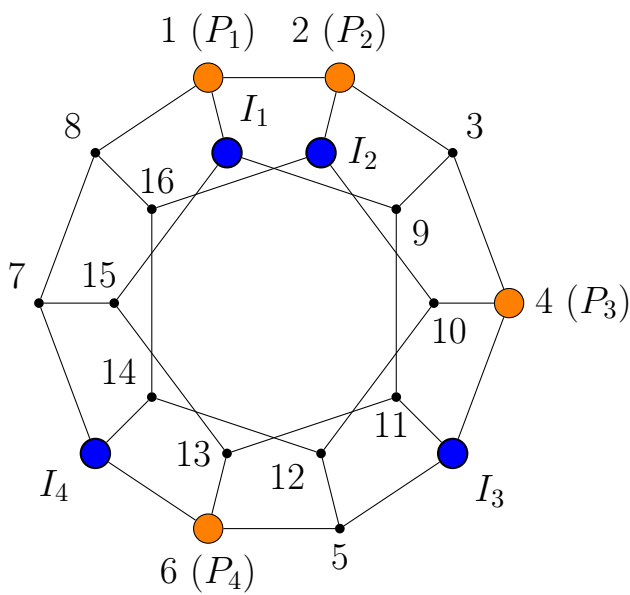
\begin{figure}[ht]
\centering
\begin{tikzpicture}
    \begin{pgfonlayer}{nodelayer}
        \node [style=input,label={[label distance=0]85:{\large $I_1$}}] (1) at (0.00, 2.00) {};
        \node [style=input,label={[label distance=0]0:{\large $I_2$}}] (2) at (1.25, 2.00) {};
        \node [style=output,label=90:{\large \phantom.16}] (3) at (-1.00, 1.25) {};
        \node [style=output,label=350:{\large 9}] (4) at (2.25, 1.25) {};
        \node [style=output,label={270:{\large \phantom{...}10}}] (5) at (2.75, -0.00) {};
        \node [style=pivot,label=0:{\large 4 ($P_3$)}] (6) at (3.75, -0.00) {};
        \node [style=output,label=60:{\large 3}] (8) at (3.00, 2.00) {};
        \node [style=pivot,label={\large 2 ($P_2$)}] (9) at (1.50, 3.00) {};
        \node [style=pivot,label={\large 1 ($P_1$)}] (10) at (-0.25, 3.00) {};
        \node [style=output,label=120:{\large 8}] (11) at (-1.75, 2.00) {};
        \node [style=output,label=91:{\large 15\phantom{...}}] (12) at (-1.50, -0.00) {};
        \node [style=output,label=120:{\large 7}] (13) at (-2.50, -0.00) {};
        \node [style=output,label=150:{\large 14}] (14) at (-1.00, -1.25) {};
        \node [style=output,label=180:{\large 13}] (15) at (0.00, -2.00) {};
        \node [style=output,label=190:{\large 12}] (16) at (1.25, -2.00) {};
        \node [style=output,label=270:{\large 11\phantom.}] (17) at (2.25, -1.25) {};
        \node [style=pivot,label=270:{\large 6 ($P_4$)}] (20) at (-0.25, -3.00) {};
        \node [style=output,label=270:{\large 5}] (21) at (1.50, -3.00) {};
        \node [style=input,label=315:{\large $I_3$}] (22) at (3.00, -2.00) {};
        \node [style=input,label=200:{\large $I_4$}] (23) at (-1.75, -2.00) {};
    \end{pgfonlayer}
    \begin{pgfonlayer}{edgelayer}
        \draw (1) to (12);
        \draw (1) to (4);
        \draw (1) to (10);
        \draw (2) to (3);
        \draw (2) to (5);
        \draw (2) to (9);
        \draw (3) to (14);
        \draw (3) to (11);
        \draw (4) to (17);
        \draw (4) to (8);
        \draw (5) to (16);
        \draw (5) to (6);
        \draw (6) to (8);
        \draw (6) to (22);
        \draw (8) to (9);
        \draw (9) to (10);
        \draw (10) to (11);
        \draw (11) to (13);
        \draw (12) to (15);
        \draw (12) to (13);
        \draw (13) to (23);
        \draw (14) to (16);
        \draw (14) to (23);
        \draw (15) to (17);
        \draw (15) to (20);
        \draw (16) to (21);
        \draw (17) to (22);
        \draw (20) to (23);
        \draw (20) to (21);
        \draw (21) to (22);
    \end{pgfonlayer}
\end{tikzpicture}
\caption[Dodecahedral code]{The $\llbracket 16, 4, 3\rrbracket$ dodecahedral code. Blue nodes are inputs, black nodes are outputs, and orange nodes are pivots. This code is the optimal packing of input nodes into the dodecahedral graph so that no inputs have a path between them of length less than 3. As the dodecahedron has odd cycles, it is not a bipartite graph and therefore this code is a non-CSS stabilizer code. The rate of the code is higher than that of the 5-qubit code. The labels indicate an ordering of the input and output nodes as well as the pivots (lowest numbered nodes next to each input). There are no extra input-pivot edges but we allow pivot-pivot edges where this does not change our constructions. An ordering can be chosen on the nodes to avoid all pivot-pivot edges, such as by swapping nodes 1 and 15.}
\label{fig:dodeca}
\end{figure}

\subsection{The dodecahedral code}
\label{subsec:dodecahedron}

We now showcase a few examples of novel small codes that we constructed from the graph formalism.
As a first example, we examine the 7-qubit Steane code in Figure~\ref{fig:7-qubit-code}, which takes the shape of a cube. Such a geometric structure motivates the exploration of other highly regular geometric solids as graph codes, such as the platonic solids.
While the graphs of a tetrahedron and octahedron perform rather poorly as codes due to their particular symmetries, the graph of a dodecahedron represents a relatively good code.
Figure~\ref{fig:dodeca} exhibits the dodecahedral code, which has parameters $\llbracket 16, 4, 3\rrbracket$. Since the degree of every node is 3, the dodecahedral code saturates its distance-degree bound.
The rate has also been optimized while not compromising the distance.
To maintain a distance equal to the degree, no two inputs may share a neighbour (as otherwise there exists a logical operator of lower weight), so the maximum number of inputs that can be selected is 4. We remark that this graph is not a CSS code, being non-bipartite, but has rate $1/4$. As such, it has the same distance and a higher rate than perhaps the most well-known non-CSS stabilizer code, the $\llbracket 5, 1, 3\rrbracket$ 5-qubit code from Section~\ref{sec:applications}.

\begin{figure}[ht]
    \centering
    \resizebox{0.5\linewidth}{14em}{
    \begin{tikzpicture}[scale=0.3]
    \begin{pgfonlayer}{nodelayer}
		\node [style=none,] (0) at (0, 0) {\scriptsize $I_1$\hspace{3ex}\phantom.};
		\node [style=none] (1) at (0, -1) {\scriptsize $I_2$\hspace{3ex}\phantom.};
		\node [style=none] (2) at (0, -3) {\scriptsize $I_3$\hspace{3ex}\phantom.};
		\node [style=none] (3) at (0, -5) {\scriptsize $I_4$\hspace{3ex}\phantom.};
		\node [style=z-node] (4) at (0, -2) {};
		\node [style=z-node] (5) at (0, -4) {};
		\node [style=z-node] (6) at (0, -6) {};
		\node [style=z-node] (7) at (0, -7) {};
		\node [style=z-node] (8) at (0, -8) {};
		\node [style=z-node] (9) at (0, -9) {};
		\node [style=z-node] (10) at (0, -10) {};
		\node [style=z-node] (11) at (0, -11) {};
		\node [style=z-node] (12) at (0, -12) {};
		\node [style=z-node] (13) at (0, -13) {};
		\node [style=z-node] (14) at (0, -14) {};
		\node [style=z-node] (15) at (0, -15) {};
		\node [style=z-node] (20) at (1, 0) {};
		\node [style=z-node] (21) at (1, -8) {};
		\node [style=z-node] (22) at (2, -1) {};
		\node [style=z-node] (23) at (2, -9) {};
		\node [style=z-node] (25) at (3, -10) {};
		\node [style=z-node] (26) at (3, -3) {};
		\node [style=z-node] (27) at (4, -5) {};
		\node [style=z-node] (28) at (4, -13) {};
		\node [style=z-node] (29) at (5, 0) {};
		\node [style=z-node] (30) at (5, -14) {};
		\node [style=z-node] (31) at (6, -1) {};
		\node [style=z-node] (32) at (6, -15) {};
		\node [style=z-node] (33) at (7, -3) {};
		\node [style=z-node] (34) at (7, -4) {};
		\node [style=z-node] (35) at (7, -5) {};
		\node [style=z-node] (36) at (7, -6) {};
		\node [style=z-node] (37) at (7, -6) {};
		\node [style=h-box] (38) at (8, 0) {};
		\node [style=h-box] (39) at (8, -1) {};
		\node [style=h-box] (40) at (8, -3) {};
		\node [style=h-box] (41) at (8, -5) {};
		\node [style=z-node] (42) at (9, 0) {};
		\node [style=z-node] (43) at (9, -1) {};
		\node [style=z-node] (44) at (9, -2) {};
		\node [style=z-node] (45) at (9, -3) {};
		\node [style=z-node] (46) at (9, -4) {};
		\node [style=z-node] (47) at (9, -5) {};
		\node [style=z-node] (48) at (9, -6) {};
		\node [style=z-node] (49) at (9, -7) {};
		\node [style=z-node] (50) at (9, -8) {};
		\node [style=z-node] (51) at (9, -10) {};
		\node [style=z-node] (52) at (10, -9) {};
		\node [style=z-node] (53) at (10, -11) {};
		\node [style=z-node] (54) at (9, -12) {};
		\node [style=z-node] (55) at (9, -14) {};
		\node [style=z-node] (56) at (10, -13) {};
		\node [style=z-node] (57) at (10, -15) {};
		\node [style=z-node] (58) at (11, 0) {};
		\node [style=z-node] (59) at (11, -7) {};
		\node [style=z-node] (60) at (12, -1) {};
		\node [style=z-node] (61) at (12, -2) {};
		\node [style=z-node] (62) at (12, -3) {};
		\node [style=z-node] (63) at (12, -9) {};
		\node [style=z-node] (64) at (14, -4) {};
		\node [style=z-node] (65) at (14, -11) {};
		\node [style=z-node] (66) at (15, -5) {};
		\node [style=z-node] (67) at (15, -12) {};
		\node [style=z-node] (68) at (16, -6) {};
		\node [style=z-node] (69) at (16, -14) {};
		\node [style=z-node] (70) at (13, -2) {};
		\node [style=z-node] (71) at (13, -8) {};
		\node [style=z-node] (72) at (17, -7) {};
		\node [style=z-node] (73) at (17, -15) {};
		\node [style=z-node] (74) at (12, -13) {};
		\node [style=z-node] (75) at (12, -11) {};
		\node [style=z-node] (76) at (11, -12) {};
		\node [style=z-node] (77) at (11, -10) {};
		\node [style=none] (78) at (18, -15) {\scriptsize\phantom.\hspace{3ex}16};
		\node [style=none] (79) at (18, -14) {\scriptsize\phantom.\hspace{3ex}15};
		\node [style=none] (80) at (18, -13) {\scriptsize\phantom.\hspace{3ex}14};
		\node [style=none] (81) at (18, -12) {\scriptsize\phantom.\hspace{3ex}13};
		\node [style=none] (82) at (18, -11) {\scriptsize\phantom.\hspace{3ex}12};
		\node [style=none] (83) at (18, -10) {\scriptsize\phantom.\hspace{3ex}11};
		\node [style=none] (84) at (18, -9) {\scriptsize\phantom.\hspace{3ex}10};
		\node [style=none] (85) at (18, -8) {\scriptsize\phantom.\hspace{3ex}9};
		\node [style=none] (86) at (18, -7) {\scriptsize\phantom.\hspace{3ex}8};
		\node [style=none] (87) at (18, -6) {\scriptsize\phantom.\hspace{3ex}7};
		\node [style=none] (88) at (18, -5) {\scriptsize\phantom.\hspace{7ex}6 ($P_4$)};
		\node [style=none] (89) at (18, -4) {\scriptsize\phantom.\hspace{3ex}5};
		\node [style=none] (90) at (18, -3) {\scriptsize\phantom.\hspace{7ex}4 ($P_3$)};
		\node [style=none] (91) at (18, -2) {\scriptsize\phantom.\hspace{3ex}3};
		\node [style=none] (92) at (18, -1) {\scriptsize\phantom.\hspace{7ex}2 ($P_2$)};
		\node [style=none] (93) at (18, 0) {\scriptsize\phantom.\hspace{7ex}1 ($P_1$)};
	\end{pgfonlayer}
	\begin{pgfonlayer}{edgelayer}
		\draw [style=h] (37) to (35);
		\draw [style=h] (34) to (33);
		\draw [style=h] (31) to (32);
		\draw [style=h] (30) to (29);
		\draw [style=h] (27) to (28);
		\draw [style=h] (25) to (26);
		\draw [style=h] (22) to (23);
		\draw [style=h] (21) to (20);
		\draw (0.center) to (20);
		\draw (20) to (29);
		\draw (31) to (22);
		\draw (22) to (1.center);
		\draw (2.center) to (26);
		\draw (26) to (33);
		\draw (34) to (5);
		\draw (27) to (35);
		\draw (37) to (6);
		\draw (8) to (21);
		\draw (23) to (9);
		\draw (10) to (25);
		\draw (13) to (28);
		\draw (14) to (30);
		\draw (32) to (15);
		\draw (27) to (3.center);
		\draw (35) to (41);
		\draw (40) to (33);
		\draw (39) to (31);
		\draw (29) to (38);
		\draw [style=h] (56) to (57);
		\draw [style=h] (55) to (54);
		\draw [style=h] (53) to (52);
		\draw [style=h] (51) to (50);
		\draw [style=h] (48) to (49);
		\draw [style=h] (46) to (47);
		\draw [style=h] (45) to (44);
		\draw [style=h] (43) to (42);
		\draw (32) to (57);
		\draw (55) to (30);
		\draw (28) to (56);
		\draw (54) to (12);
		\draw (11) to (53);
		\draw (51) to (25);
		\draw (23) to (52);
		\draw (50) to (21);
		\draw (7) to (49);
		\draw (48) to (37);
		\draw (41) to (47);
		\draw (46) to (34);
		\draw (40) to (45);
		\draw (44) to (4);
		\draw (39) to (43);
		\draw (42) to (38);
		\draw [style=h] (69) to (68);
		\draw [style=h] (66) to (67);
		\draw [style=h] (65) to (64);
		\draw [style=h] (63) to (62);
		\draw [style=h] (61) to (60);
		\draw [style=h] (59) to (58);
		\draw (42) to (58);
		\draw (60) to (43);
		\draw (44) to (61);
		\draw (62) to (45);
		\draw (46) to (64);
		\draw (66) to (47);
		\draw (48) to (68);
		\draw (59) to (49);
		\draw (52) to (63);
		\draw (69) to (55);
		\draw [style=h] (73) to (72);
		\draw [style=h] (71) to (70);
		\draw [style=h] (75) to (74);
		\draw [style=h] (76) to (77);
		\draw (93.center) to (58);
		\draw (60) to (92.center);
		\draw (91.center) to (70);
		\draw (70) to (61);
		\draw (62) to (90.center);
		\draw (89.center) to (64);
		\draw (66) to (88.center);
		\draw (87.center) to (68);
		\draw (72) to (86.center);
		\draw (72) to (59);
		\draw (71) to (85.center);
		\draw (71) to (50);
		\draw (63) to (84.center);
		\draw (83.center) to (77);
		\draw (77) to (51);
		\draw (79.center) to (69);
		\draw (78.center) to (73);
		\draw (73) to (57);
		\draw (54) to (76);
		\draw (76) to (67);
		\draw (67) to (81.center);
		\draw (74) to (56);
		\draw (53) to (75);
		\draw (75) to (65);
		\draw (65) to (82.center);
		\draw (80.center) to (74);
	\end{pgfonlayer}
\end{tikzpicture}
}\hfill
\raisebox{-0.9em}{
\resizebox{0.48\linewidth}{15.1em}{
\begin{tikzpicture}
\node[scale=10]{
\input{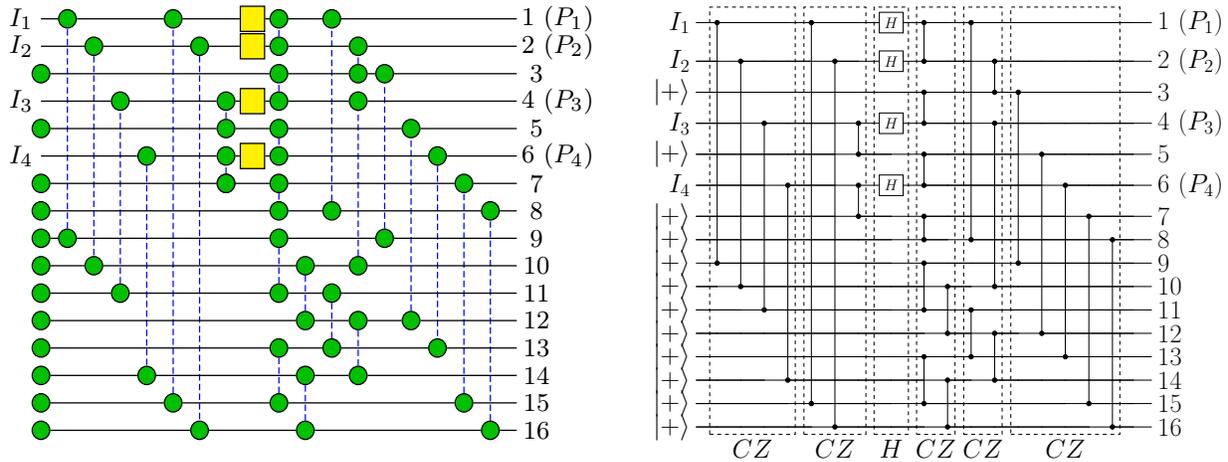}
};
\end{tikzpicture}
}}
\caption[Encoder examples]{Two representations of the encoding circuit of the dodecahedral code, labelled in the same manner as the graph in Figure~\ref{fig:dodeca}.
The construction of this circuit is described in Section~\ref{subsec:encoding_circuit} and a similar construction is detailed in Appendix~\ref{app:convert-zxcf-to-circuit}.
On the left, the representation is of a ZX diagram. Yellow boxes indicate $H$ gates and blue dashed edges are Hadamarded.
On the right, the representation is of a standard quantum circuit.
Each dashed block represents a set of transversal (depth-1) operations, so that the circuit depth is the number of blocks, in this case 6. The representation here is partially optimized, by inspection, to improve upon the upper bound given in Theorem~\ref{thm:encoding}. It is possible to move some $CZ$ gates connecting output nodes from one block to the other side of the layer of Hadamard gates to produce a depth-5 (but more visually complicated) circuit, which is provably optimal.}
\label{fig:dodeca-encoder}
\end{figure}

To complete our discussion of the dodecahedral code, we apply Theorem~\ref{thm:encoding} to give the code a low-depth encoding circuit. Figure~\ref{fig:dodeca-encoder} illustrates two representations of the encoding circuit. On the left, we draw a ZX diagram which appears as an intermediate step in the algorithmic construction in Theorem~\ref{thm:encoding}. On the right, we translate the ZX diagram into standard quantum circuit notation.
The translation is an immediate consequence of the ZX-calculus rules discussed in Appendix~\ref{app:zx-calculus}: a $Z$ node with a single output is a $\ket+$ state and a Hadamarded edge between two $Z$ nodes implements a $CZ$ gate. Corresponding to Theorem~\ref{thm:encoding}, the $CZ$ gates before the layer of Hadamards form the input-output edges; the Hadamards effectively form input-pivot edges; and the $CZ$ gates after the layer of Hadamards form the edges from non-inputs to non-inputs.
We observe also that the application of 
the ZX-calculus simplification rules to the diagram in Figure~\ref{fig:dodeca-encoder}, such as merging $Z$ nodes together, recovers the graph representation of the dodecahedral code itself. In this sense the encoding circuit appears almost immediately from the graph representation.

While our results about the worst-case encoder depth suggest that it may have a depth of up to 9, we have applied some heuristic analysis to reduce the depth to 6.
Further compilation of this circuit would allow us to reduce the depth further.
This is because output-output edges commute with every gate in the circuit as they are not on the input-pivot wires containing the only non-diagonal gates.
In particular, the 5 gates in the last block can be moved to the second, first, fourth, first, and first blocks, respectively, resulting in a circuit of depth 5.
This is the minimum possible depth as that is the number of gates acting on some qubits, such as the first input.

\subsection{Covering spaces and the 5-covered icosahedral code}

We next consider a code based on the graph of an icosahedron.
Since the vertices have an icosahedron have degree 5, we can hope that a code based on this graph would have distance 5. Unfortunately, we find that the code underperforms the degree bound because a $Y$ gate on two antipodal vertices is equivalent to applying a $Z$ to each vertex.
Hence, a logical $Y$ on any input is a logical operation of weight 3, with one $Y$ on the vertex opposite the input and two more on any two antipodal vertices.
This observation also implies that if we wish to add more than one input to the icosahedron, our distance will be reduced further, as a logical operation would just consist of $Y$'s on two vertices opposite the inputs.

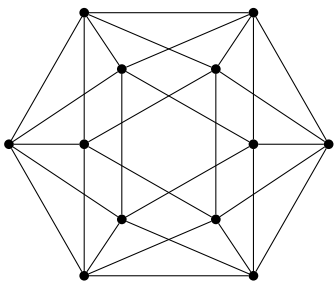
\begin{figure}[ht]
    \centering
    \begin{tikzpicture}
    \begin{pgfonlayer}{nodelayer}
        \node [style=output] (0) at (0.00, 2.00) {};
        \node [style=output] (1) at (0.00, -0.00) {};
        \node [style=output] (2) at (1.75, 1.00) {};
        \node [style=output] (3) at (1.25, 2.00) {};
        \node [style=output] (4) at (1.25, -0.00) {};
        \node [style=output] (5) at (-0.50, 1.00) {};
        \node [style=output] (6) at (-0.50, 2.75) {};
        \node [style=output] (7) at (1.75, -0.75) {};
        \node [style=output] (8) at (2.75, 1.00) {};
        \node [style=output] (9) at (1.75, 2.75) {};
        \node [style=output] (10) at (-1.50, 1.00) {};
        \node [style=output] (11) at (-0.50, -0.75) {};
    \end{pgfonlayer}
    \begin{pgfonlayer}{edgelayer}
        \draw (0) to (1);
        \draw (0) to (2);
        \draw (0) to (6);
        \draw (0) to (9);
        \draw (0) to (10);
        \draw (1) to (2);
        \draw (1) to (11);
        \draw (1) to (7);
        \draw (1) to (10);
        \draw (2) to (8);
        \draw (2) to (7);
        \draw (2) to (9);
        \draw (3) to (5);
        \draw (3) to (4);
        \draw (3) to (9);
        \draw (3) to (6);
        \draw (3) to (8);
        \draw (4) to (5);
        \draw (4) to (7);
        \draw (4) to (11);
        \draw (4) to (8);
        \draw (5) to (10);
        \draw (5) to (11);
        \draw (5) to (6);
        \draw (6) to (9);
        \draw (6) to (10);
        \draw (7) to (8);
        \draw (7) to (11);
        \draw (8) to (9);
        \draw (10) to (11);
    \end{pgfonlayer}
\end{tikzpicture}
\caption[Icosahedral graph]{The graph of an icosahedron.  Although it has degree 5, marking a single node as an input only creates a code of distance 3. We can puncture the icosahedron through the two central faces and take the double cover of the graph with an added index of each vertex matching the winding number of a path around the puncturing axis. If we select one of the 6 outer nodes, nodes that are not vertices of the punctured faces, and mark both copies of the selected nodes as inputs, we create a $\llbracket 22,2,5\rrbracket$ code. 
We can improve the rate by covering the icosahedron 5 times and selecting nodes along the perimeter, each rotated $300\degree$ (5 steps) around the puncturing axis from the last. Since we mark one node every $300^\circ$, we create 6 input nodes, since $6\cdot300\degree=5\cdot360\degree$. The 5-covered icosahedral code has parameters $\llbracket 54,6,5\rrbracket$.}
\label{fig:icosa}
\end{figure}

To get around this issue, we lift the graph into a finite covering space.
Specifically, we can imagine taking a double cover of the icosahedron by puncturing the icosahedron in the center of two opposite faces and attaching a winding number to each vertex, which is computed modulo 2.
In doing so, we find that the double-covered graph with a single input has distance 5.
Such a graph can even support two inputs, located at a pair of vertices in the same position in the icosahedron but merely at different indices.
This produces a $\llbracket 22,2,5\rrbracket$ code.

In fact, we can improve the rate slightly by taking an $n$-cover of the icosahedron. Specifically, when $n=5$, we can fit 6 input vertices, where the inputs are not vertices of the punctured faces, but rather on the ring between them, and where each input is wound $300\degree$ or 5 steps around the punctured axis from the previous input.
In Figure~\ref{fig:icosa}, if we puncture the two central faces, the inputs in the 5 copies of the icosahedron will be on copies of the 6 vertices along the perimeter, each rotated $300\degree$ from the last.
This construction creates a $\llbracket 54, 6, 5\rrbracket$ code.

The ideas used to create the five-covered icosahedron can be extended more generally to extend otherwise finite graphs.
In fact, arbitrarily complicated covers can be proposed, leading all the way up to the universal cover of the graph.
If regularity and a finite set of vertices are desired, some periodic boundary conditions can be established to keep the graph finite.
In general, such covering space graphs may provide a promising path towards code design recipes in this framework. In the case of Cayley graphs, similar ideas have already been used to construct quantum LDPC codes from Cayley graphs of free groups with generators comprised of columns of a classical parity check matrix~\cite{couvreur2013construction}.

\subsection{Planar and higher-dimensional lattice-type codes}
\label{sec:hypercube-code}

Lattice-type codes enable the construction of families of codes with increasing degree, and thus potentially increasing distance, while also preserving some notions of regularity that are often desirable in code construction.
If we wish to remain in two dimensions and also have a Euclidean lattice, we are restricted in the sense that the triangular lattice, which has degree 6, is the highest-degree such lattice.
As long as a triangular lattice is taken to be of a certain minimum size, it does yield distance-6 codes.
Beyond degree 6, we can either raise the dimension or embed graphs in non-Euclidean spaces. For example, an embedding of triangular surfaces into hyperbolic spaces yields the Poincar\'e's disk map. However, since the Poincar\'e disc packs infinitely many points into a finite space, we would have to construct finite approximations and then study the effect of artificially establishing a finite boundary. Analysis of such effects might lead to discovering interesting code families. In particular, such classes of codes in non-Euclidean spaces may prove useful and might be connected to HaPPY codes~\cite{pastawski2015holographic}, which also use graphs in hyperbolic space.

The second approach, increasing the dimension, is more amenable to analysis. A simple, general recipe to construct a code from a lattice is to cut off the lattice in each dimension at a finite number of steps, endow the finite version with periodic boundary conditions in each direction, and embed some inputs and pivots within the vertices such that the choices of pivots are valid and inputs are not connected.

One immediate advantage of this approach is a robustness to deformation that does not hold for CSS codes. In particular, 
we noted in Theorem~\ref{thm:CSS_bipartite} that a graph code is CSS if and only if it is bipartite. For lattices with periodic boundary conditions, the bipartite condition is equivalent to the length of every single dimension of the torus being even. Such a requirement is very non-local and does not change the central architecture of the code.
Even for CSS torii, the addition of one vertex far from any inputs such that it forms an odd-length cycle will make the code no longer CSS but will not meaningfully affect any of the code's parameters.

When constructing families of codes in the graph formalism, it is important to consider the geometry of the graph as well as the placement of the inputs.
To make use of the graph formalism for computing stabilizers and logical operations, the identification of pivot nodes is paramount.
While the ZXCF does not allow for pivot-pivot edges, such edges do not affect the expressions of the graph formalism.
Thus, for easy pivot selection, it suffices for all of the inputs to be separated by paths of length at least 3 from each other in the graph. Although such a restriction is not necessary for code construction, it does allow for any input's neighbour to be chosen as its pivot as that node would necessarily not border any other inputs.
The downside to such graphs is that their rate is bounded as $R \leq \frac{1}{d}$, where $d$ is the largest degree of the graph.
In other words, such graphs have a strong rate-distance trade-off. Nonetheless, these graphs may be interesting at small scales, as we have already seen. 

In higher dimensions, perhaps the simplest lattice we can use to build a code is the boolean hypercube in an arbitrary $m$ dimensions. We represent the hypercube by labelling nodes as length-$m$ bitstrings and connecting two nodes if their Hamming distance is exactly 1, i.e. they differ by one bit flip. We observe immediately that the hypercube code is CSS for all $m$, since the hypercube is bipartite. 
To select inputs, we observe that if we wish to place inputs such that they are separated by paths of length at least 3, we can leverage classical coding theory to provide this guarantee automatically.
Specifically, we consider the \textit{classical} $[2^r - 1, 2^r - r - 1, 3]$ Hamming codes~\cite{hamming1950error}. Letting $m = 2^r - 1$, we can equivalently express the parameters as $[m, m - \log(m+1), 3]$, where the code is perfect when $m$ is one less than a power of 2. Note that all logarithms in this chapter are taken to be base 2.
Since this code has $m - \log(m+1)$ logical bits, there are a total of $\frac{2^m}{m + 1}$ codewords represented as length-$m$ bitstrings.
Since the Hamming code's distance is 3, all such codewords are guaranteed to have a separation in the hypercube by paths of length at least 3.
Therefore, by choosing the input nodes in the $m$-dimensional hypercube to be those whose bitstring is a codeword in the $[m, m - \log(m+1), 3]$ Hamming code, we obtain a code with $\frac{2^m}{m+1}$ logical qubits, $2^m -\frac{2^m}{m+1} = \frac{m 2^m}{m+1}$ physical qubits, and a distance $D_m$ bounded above by the degree $m$.
That is, potentially a $\llbracket \frac{m 2^m}{m+1}, \frac{2^m}{m+1}, D_m\rrbracket$ family of codes.
Due to the length-3 separation of inputs, any choice of pivots uniquely connected to a corresponding input will do; later, we will show that it is both possible and particularly convenient to choose the pivots such that they are also separated by paths of length at least 3. We will also later prove in Section~\ref{subsec:decoder} that $D_m \gtrsim m/2$, up to an additive error of 2, and give evidence that $D_m = m$. For example, for $m = 7$, the code is $\llbracket 112, 16\rrbracket$ with distance between 4 and 7. Asymptotically, the code scales as $\llbracket n, \Th(\frac{n}{\log n}), \Th(\log n)\rrbracket$.

To improve the rate further without the strong distance-rate trade-off barrier $R \leq \frac{1}{d}$, we will need to use the most general form of graph codes, namely those with inputs separated by paths of length 2. We conclude this section with some remarks on how such codes may be constructed.
In a grid lattice on a $m$-dimensional torus, we wish to pack inputs more closely while being mindful of the fact that our distance will be bounded above by the number of output nodes next to each output node.
One such construction is as follows.
First, consider a foliation of the $m$-torus into sequentially numbered ``layers'' of $(m-1)$-dimensional toric hypersurfaces that together stack into a $m$-dimensional torus.
In each layer, we will consider sets of vertices grouped by their \textit{index}, the sum of their integer coordinates taken mod 3.
We can assume that the size of each torus dimension is a multiple of 3 to make sure these sets are well-defined.
Note that the index does not consider the $m$th coordinate, as that is determined by the number of the layer that a node ends up in.
\begin{enumerate}[(1)]
    \item
In layer 1, we set all nodes of index 0 to be inputs, with pivots in layer 2.
    \item
In layer 2, we set all nodes of index 0 or 1 to be pivots.
    \item
In layer 3, we set all nodes of index 1 to be inputs, with pivots in layer 2.
    \item
We repeat this 3-layer structure as many times as desired.
\end{enumerate}
Such a pattern is essentially a 3-layered, high-dimensional analogue of a checkerboard pattern.
Again, we see that our distance upper bound is $d$, as vertices such as those of index 1 in layer 1 have $d$ inputs among their $2d$ neighbours. More sophisticated techniques are necessary to lower-bound the distance.
However, the rate of this code is a constant $\frac27$, which is achieved at any scale of this code, without the need for extreme numbers of qubits for the asymptotic behaviour of the code parameters to be expressed.
Additionally, we note that if an odd number of 3-layer structures are used or if any of the dimensions of each layer is odd, the code graph will have a cycle of odd length, meaning that the resulting code will not be a CSS code.

These examples showcase a wide variety of families and constructions, both small and large, that can be achieved with the graph formalism. While some are sufficiently small to admit brute-force computational distance calculations, others are more difficult to analyze and will require additional distance-bounding tools to understand. We take a first step in this direction in the next section.

\subsection{Decoding algorithms and analysis}
\label{subsec:decoder}

We continue our discussion by constructing a simple greedy decoding algorithm on graphs. We give a sufficient-condition characterization of graphs for which the greedy decoder successfully recovers errors up to half of the theoretical maximum. That is, graphs which satisfy this certain property have a provable distance \textit{lower bound} based on its degree. Finally, we will show that the  $\llbracket \frac{m 2^m}{m+1}, \frac{2^m}{m+1}\rrbracket$ boolean hypercube code, defined in Section~\ref{sec:hypercube-code}, satisfies this characteristic and thus is efficiently decodable. 

Denote for simplicity $oip(v) \coloneq o(i(p(v)))$, $ip(v) = i(p(v))$, and $oi(v) = o(i(v))$. In Theorem~\ref{thm:decodingQLO}, we showed that decoding amounts to strategies in QLO. There, we presented a more sophisticated approach in which the lights on pivots encoded information about stabilizers that were connected to the pivots' corresponding inputs. In our greedy decoder, it will be simpler for analysis purposes to do away with this notion and decode more directly. That is, $X$ errors on a node $v$ toggle lights in $o(v) \,\D\, oip(v)$, $Z$ errors on a pivot $v$ toggle lights in $oi(v)$, and $Z$ errors on a non-pivot output $v$ toggle the light at $v$. ($Y$ errors are treated as both an $X$ and a $Z$ error.)
As a consequence, there are no lights at all on pivots (since there are no stabilizers there). Moreover, we apply recoveries by directly performing the same operations, e.g. if we wish to do a $X$-recovery on a node $v$, we toggle lights in $o(v) \,\D\, oip(v)$ and destroy the light at $v$.
We note that this paradigm is equivalent to the QLO game instance from Theorem~\ref{thm:decodingQLO}. In particular, suppose that after every switch flip on $v$ Alice also flips the switches on $ip(v)$, which immediately toggles the lights on the neighbours, i.e. turning off the lights in $p(v)$ and turning on the lights in $oip(v)$. This rule maps the QLO game in the theorem to the setup we have given here. But the flipping of input switches can be moved post-hoc to round (1) in the game, not changing the number of moves made in round (2), and thus the two paradigms are entirely equivalent. 

\begin{algorithm}[ht]
\caption{Greedy graph decoder}
\label{alg:greedy}
\DontPrintSemicolon
  \SetKwFunction{FMain}{GreedyDecoder}
  \SetKwProg{Fn}{Function}{:}{}
  \Fn{\FMain{$G$}}{
        $\text{explored} \leftarrow [],\, (v_0, n_0) \leftarrow (\text{nil}, 0)$\;
        \While(\texttt{// X error recovery loop}){True}{
            \For{$v \in \CO \cup \CP$}{
                $\text{gap} \leftarrow 2 \text{\texttt{NumLights}}(o(v)) - |o(v)|$\;
                \tcp{If more than half of lights around $v$ are on, `gap' is positive}
                \If{$\text{gap} > n_0$}{
                    $v_0 = v,\, n_0 = \text{gap}$\;
                }
            }
            \If{$v_0 \in \textrm{explored}$ \textbf{or} $\text{gap} = 0$}{
                \textbf{break}\;
            }
            \texttt{Recover}$_X(v_0)$ \tcp{toggle lights at $o(v_0) \,\D\, oip(v_0)$, destroy light at $v_0$}
            \texttt{Append} $v_0$ \textbf{to} explored\;
        }
        $\text{explored} \leftarrow [],\, (v_0, n_0) \leftarrow (\text{nil}, 0)$ \tcp{Reset tracking variables}
        \While(\texttt{// Z error \textit{on pivots} recovery loop}){True}{
            \For{$v \in \CP$}{
                $\text{gap} \leftarrow 2 \text{\texttt{NumLights}}(oi(v)) - |oi(v)|$\;
                \If{$\text{gap} > n_0$}{
                    $v_0 = v,\, n_0 = \text{gap}$\;
                }
            }
            \If{$v_0 \in \textrm{explored}$ \textbf{or} $\text{gap} = 0$}{
                \textbf{break}\;
            }
            \texttt{Recover}$_Z(v_0)$ \tcp{toggle lights at $oi(v_0)$}
            \texttt{Append} $v_0$ \textbf{to} explored\;
        }
        \For(\texttt{// Z error \textit{on outputs} recovery loop}){$v \in \CO$}{
            \If{\texttt{LightIsOn}$(v)$}{
                \texttt{Recover}$_Z(v)$ \tcp{Destroy light at $v$}
            }
        }
        \KwRet
  }
\end{algorithm}
The greedy decoder is a strategy in the decoding instance of QLO. In particular, the operations permitted are \texttt{Recover}$_X$ and \texttt{Recover}$_Z$, which flip switches or destroy lights depending on the node $v$ chosen. We showed earlier in Theorem~\ref{thm:decodingQLO} how a QLO strategy maps to a physical recovery operation.
The construction of the algorithm is based on the following intuition. If $v \in \CO$ suffers a $X$ error, the stabilizers at $o(v) \,\D\, oip(v)$ will measure $-1$, so in the QLO game the lights in $o(v) \,\D\, oip(v)$ will turn on. Perhaps there will be some $Z$ errors in $o(v) \subseteq \CO$, each of which correspond to turning on a single light in $o(v)$. Therefore, assuming it is difficult to turn on more than half of the lights in $o(v)$, a good guess is that if more than half of the lights are on in $o(v)$ then $v$ has suffered an $X$ error, at which point we can apply $X$ on $v$ to reverse the error, via \texttt{Recover}$_X$.
Note that we do not consider the status of the lights in $oip(v)$ when deciding whether to toggle $v$, as all lights in $oip(v)$ can be toggled by applying an $X$ to any neighbour of $p(v)$.
The recovery operation may itself turn on more lights in other places, so to avoid issues related to parallelism, we work in sequence and apply $X$ recoveries one at a time, in order of nodes which have the most neighbours with lights on. A similar story holds for $Z$ errors on $v \in \CP$. They are detected by stabilizers in $oi(v)$, so we apply the analogous greedy recovery operation there. Finally, for remaining lights on in $v \in \CO$, we directly destroy the lights with $Z$ operations. 

We note that this decoder echoes the greedy strategy using by \textcite{sipser1996expander}.
There, variables were similarly toggled to satisfy as many unsatisfied constraints as possible.
A key distinguishing feature of our work is the ability of nodes to act both as constraints and as variables.
When a graph is bipartite, the nodes can be partitioned into variables and constraints, but this only allows one to work with CSS codes, as discussed in Theorem~\ref{thm:CSS_bipartite}.
Studying more general expander graphs through the graph formalism seems like a promising avenue for future work.

The key property behind the intuition about graphs that enable greedy-type decoding is the idea that it must be difficult to simulate an $X$ error on some $v$ by using a small amount of moves to turn on most of the lights in $o(v)$. We now formalize this property to give a simple performance guarantee of the greedy decoder.

In this definition, we are specifically interested in how many lights an operation at node $v$ can toggle in the set $o(u)$, the set we examine to determine whether to toggle the switch at $u$.

\begin{definition}
\label{def:sensitive}
    Let $G = (V, \set{o, i, p})$ be a graph with input, output, and valid pivot nodes. We say that $G$ is \textit{$B$-sensitive} if \begin{enumerate}[(a)]
        \item for all $u \in \CO \cup \CP,\, v \in \CO \cup \CP$, $u \neq v$, $|o(u) \cap (o(v) \,\D\, oip(v))| \leq B$ ($X$ error on any output or pivot $v$), 
        \item for all $u \in \CO \cup \CP ,\, v \in \CO$, $u \neq v$, $|o(u) \cap (\set{v} \,\D\, o(v) \,\D\, oip(v))| \leq B$ ($Y$ error on non-pivot output $v$), 
        \item for all $u \in \CO \cup \CP ,\, v \in \CP$, $u \neq v$, $|o(u) \cap (oi(v) \,\D\, o(v) \,\D\, oip(v))| \leq B$ ($Y$ error on pivot $v$),
        \item and for all $u \in \CI,\, v \in \CP$, if $u \neq i(v)$, i.e. $v$ is not the pivot corresponding to input $u$, then $|o(u) \cap oi(v)| \leq B$ ($Z$ error on pivot $v$) .
    \end{enumerate}
The case $|o(u) \cap \set{v}| \leq B$ ($Z$ error on non-pivot output $v$) is not included because it is always satisfied, since we always take $B \geq 1$. 
\end{definition}
These four conditions together form a sufficient condition to place a relatively tight guarantee on the greedy decoder's performance.

\begin{theorem}[Greedy guarantee] \label{thm:greedy_guarantee}
     Let $G$ be a $B$-sensitive graph. Denote $\d_* \coloneq \min\limits_{v \in V} \deg(v)$ the minimum degree of $G$. Then the greedy graph decoder corrects at least $\lfloor \frac{\d_*}{2B} \rfloor$ errors. In other words, the less sensitive a graph is, the more effective the greedy decoder. For constant $B$, the decoder thus corrects $\Th(\d_*)$ errors.
\end{theorem}

\begin{proof}
    The decoder will fail if it makes either of the following two errors: (1) it applies \texttt{Recover}$_P(v)$ when $v$ did \textit{not} suffer a Pauli $P$ error ($Y$ errors are treated separately as both an $X$ and a $Z$ error) or (2) it fails to apply a \texttt{Recover}$_P(v)$ when $v$ \textit{did} suffer a Pauli $P$ error. 
    By the construction of the greedy decoder, this phenomenon occurs only if enough errors happened to flip at least $\ceil{\frac{d}{2}}$ lights in $o(v)$ or in $oi(v)$ if $v$ is a pivot. Indeed, in the former case, enough errors occurred to have turned on the majority of lights in $o(v)$ or $oi(v)$ without the corresponding error having occurred on $v$, and in the latter case it turned off the majority of lights even when the corresponding error occurred. Conditions (a)-(c) of Definition~\ref{def:sensitive} guarantee that it takes strictly more than $\floor{\frac{1}{B} \floor{\frac{\d_*}{2}}} = \floor{\frac{\d_*}{2B}}$ errors to trick the decoder in either way during the $X$-recovery part of the algorithm, and condition (d) guarantees the same lower bound for the $Z$-on-pivots-recovery part of the algorithm. In condition (d), note that we equivalently used $o(u)$ for $u \in \CI$ rather than $oi(u)$ for $u \in \CP$ (as in the $Z$-on-pivots-recovery loop in Algorithm~\ref{alg:greedy}), and we have disregarded the $u = i(v)$ case which corresponds to an actual $Z$ error on $p(u)$ which would be correctly decoded. Thus, in all cases at least $\floor{\frac{\d_*}{2B}}$ errors are required to cause the decoder to act incorrectly, which completes the proof.
\end{proof}

Recall that key code properties, the weight of canonical stabilizers, the distance, and the canonical encoding circuit depth, are all bounded above by the degree.
Since the number of errors corrected $t$ and the distance are related by $2t+1 \leq \text{dist}(C) \leq 2t + 2$, our result also implies a two-way bound on the distance of $B$-sensitive graphs.
\begin{corollary}[Distance bound] \label{corollary:dist_lower_bound}
    Let $G$ be a $B$-sensitive graph and minimum degree $\d_*$, and let $C(G)$ be the code that $G$ represents. Then $\floor{\frac{\d_*}{B}} \leq 2\floor{\frac{\d_*}{2B}} + 1 \leq \text{dist}(C(G)) \leq \d_*$. 
\end{corollary}
Similarly, we can give a bound on the encoding circuit based directly on the distance.
\begin{corollary}[Encoding circuit bound] \label{corollary:encoding_upper_bound}
    Let $G$ be a $B$-sensitive graph and minimum degree $\d_*$. Let $C(G)$ be the code that $G$ represents, and let $\CC$ be its canonical encoding circuit as defined by the algorithm from Theorem~\ref{thm:encoding}. Then $\operatorname{depth}(\CC) \leq \d_* = 2 B \frac{\d_*}{2B} + B - B \leq B (2 \floor{\frac{\d_*}{2B}} + 1) - B \leq B (\operatorname{dist}(C(G)) - 1)$.
\end{corollary}
The proof is an immediate consequence of needing a depth of at least $\d_*$ from Theorem~\ref{thm:encoding} and of Corollary~\ref{corollary:dist_lower_bound}.

Although they are not necessary for the purposes of establishing Theorem~\ref{thm:greedy_guarantee}, several end-of-line optimizations to Algorithm~\ref{alg:greedy} can improve its practical performance. \begin{enumerate}[(1)]
    \item \textit{Neglecting destroyed lights.} Redefine \texttt{gap} to be the difference between the number of lights on in $o(v)$ and the number of lights that are \textit{off but not destroyed} in $o(v)$. The only difference between this definition and the original is that lights which are destroyed are not counted. This follows because once a light is destroyed, it does not contribute to the QLO game anymore. Equivalently, a node with a destroyed light has already had some Pauli on it, which increased the weight of the total Pauli string by 1. Therefore, adding another Pauli to the same node does not increase the weight.
    \item \textit{Gap edge case.} In the $X$ loop, if \texttt{gap = 1} for the chosen node $v$, then only apply the recovery if the majority of un-destroyed lights in $oip(v)$ are on. This is because it costs one move in QLO to apply the recovery, and if the majority of lights are off in $oip(v)$ then one might actually turn more lights on than off by recovering, which could slightly increase the total number of moves.
    
    \item \textit{Parity redundancy.} In the $X$ recovery loop, consider only $o(v) \setminus oip(v)$ instead of $o(v)$, since lights in $o(v) \cap oip(v)$ cannot be turned on by $X$ on $v$ because $X_v$ toggles each such light twice. Nodes in $o(v) \cap oip(v)$ can be affected by all other nodes which are neighbours of $p(v)$.
\end{enumerate}

For a given graph code, a smaller sensitivity implies an increased effectiveness of the greedy decoder. We now give a family of codes for which the sensitivity is nearly optimal.
Previously, we defined the hypercube code in Section~\ref{sec:hypercube-code} with parameters $\llbracket \frac{m 2^m}{m+1}, \frac{2^m}{m+1}\rrbracket$. We label the nodes of the hypercube by their length-$m$ binary string representation. Denote by $e_i$ the $i$th standard basis vector in $\Z_2^m$, i.e. $(e_i)_i = 1$ and $(e_i)_{j \neq i} = 0$. Moreover, we showed by injection of the classical Hamming code that the input nodes of the hypercube code are separated by paths of length at least $3$. Before we prove that the hypercube has sensitivity $2$, we give a useful lemma.

\begin{lemma}[Hypercube separation] \label{lemma:hypercube_separation}
    It is possible to choose pivots on the hypercube code such that every path between pivots has length at least 3.
\end{lemma}
\begin{proof}
    The minimum path length between inputs is already $3$. By homogeneity of the hypercube, it is enough to choose the pivots in a uniformly consistent way for every input. One example which works is to set $p(v) = \set{v + e_1}$ for all $v \in \CI$. This assignment will never choose $p(v) \in \CI$ since the inputs are separated by paths of length at least 3. Now suppose there exists a path between two pivots $p_1$ and $p_2$ with length less than $3$. In other words, with at most 2 bit flips, we can go from $p_1$ to $p_2$, so $p_2 = p_1 + e_a + e_b$ for some $a, b \in [m]$. Since $p_1 = i_1 + e_1$ and $p_2 = i_2 + e_1$, where $i_j$ is the input corresponding to $p_j$, $i_2 = i_1 + e_a + e_b$, which contradicts the fact that inputs are separated by paths of length at least 3.
\end{proof}

We note that there are several other transformations that can be applied to the hypercube to choose a set of pivots.
One such other example is a $\frac\pi2$ axis-aligned rotation.

\begin{figure}[ht]
    \centering
    \includegraphics[width=0.75\linewidth]{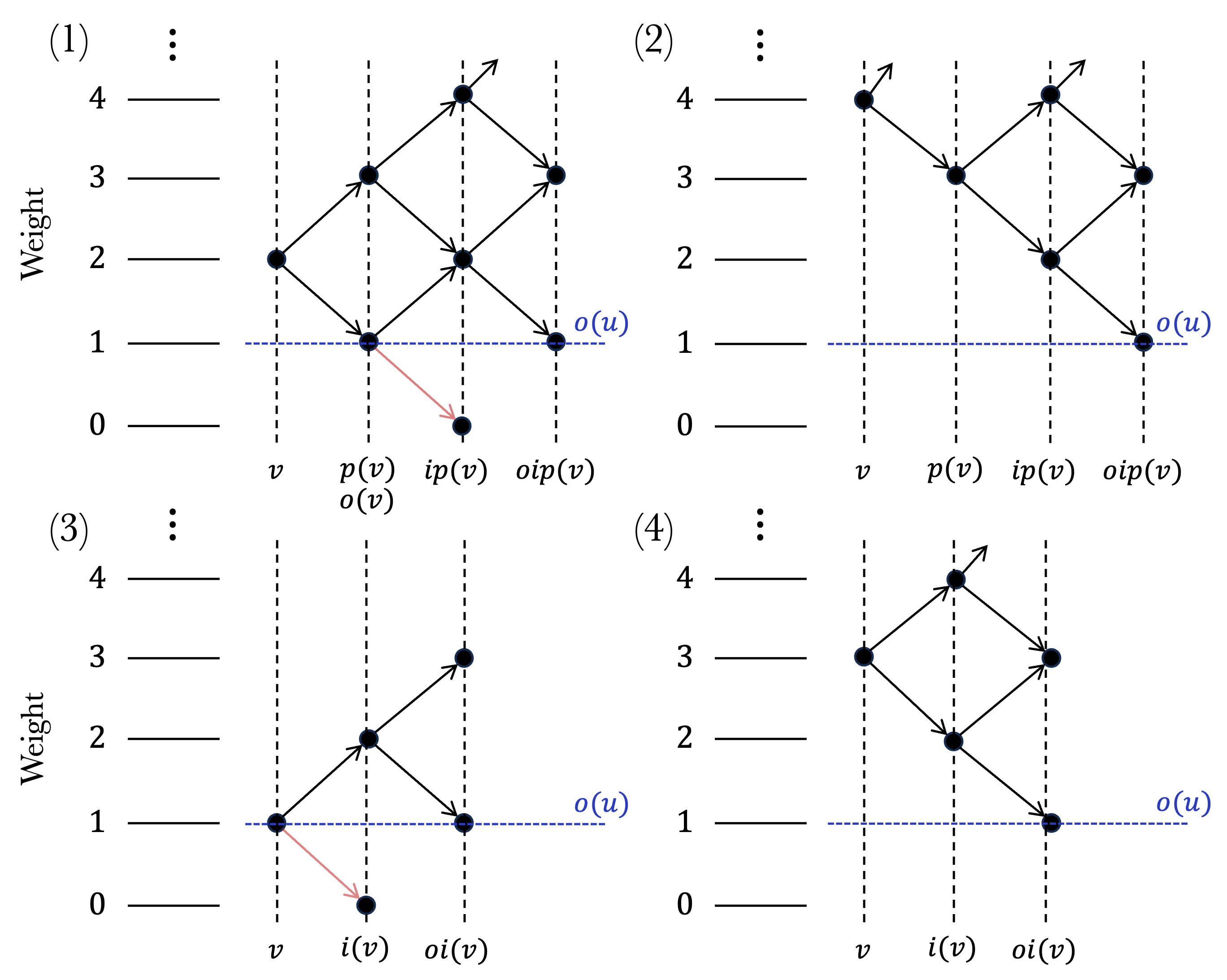}
    \caption[Toggling nodes in hypercube code]{Possible weights of nodes in a hypercube code after neighbouring operations (1) $o(v)$ and $oip(v)$ with starting weight $|v| = 2$, (2) $oip(v)$ with starting weight $|v| = 4$, (3) $oi(v)$ with starting weight $|v| = 1$, and (4) $oi(v)$ with starting weight $|v| = 3$. $u = 0^m$ in all cases, so $|w| = 1$ for all $w \in o(u)$. The red arrows correspond to paths that a priori seem possible but do not actually exist.}
    \label{fig:leveldiagram}
\end{figure}

\begin{theorem}[Hypercube sensitivity]
    Every hypercube code is $2$-sensitive.
\end{theorem}
\begin{proof}
    By homogeneity, let $u = 0^m$ without loss of generality, so that $o(u)$ has only weight-1 bitstrings. We expand into each of the four cases in Definition~\ref{def:sensitive}. The general idea is that when we take neighbour-type operations $p(\cdot), o(\cdot), i(\cdot)$ on a bitstring with weight $M$, the resultant bitstring's weight can only change by 1 due to the connectivity of the hypercube. 
    \begin{enumerate}[(a)]
        \item Let $u, v \in \CO \cup \CP$. If $o(u) \cap (o(v) \,\D\, oip(v)) \neq \varnothing$, then $|v| \in \set{2, 4}$. These cases are shown in Figure~\ref{fig:leveldiagram}(1)-(2). \begin{enumerate}[i.]
            \item If $|v| = 2$, let $v = e_i + e_j$ for $i \neq j$. So $o(u) \cap o(v) \subseteq \set{e_i, e_j}$. It is possible, as shown in Figure~\ref{fig:leveldiagram}(1) for $oip(v)$ to contain weight-1 bitstrings as well. We work forwards. By Lemma~\ref{lemma:hypercube_separation}, $p(v)$ has at most one element. Since inputs correspond bijectively to pivots, $ip(v)$ thus also has at most one element. Suppose $\exists w \in ip(v)$. Then $oip(v)$ intersects $o(u)$ nontrivially only if $|w| \in \set{0, 2}$. However, $u$ is the unique element whose weight is $0$, and $u \notin \CI$. Thus, since $w \in \CI$, $|w| \neq 0$ so $|w| = 2$. Let $w = e_k + e_l$ for some $k \neq l$. Since $w$ and $v$ are separated by a path of length 2, $|w - v| = 2$.
            Moreover, $|v| = 2$. The only way to transform $e_i + e_j \to e_k + e_l$ is if at least one of $\set{e_i, e_j}$ equals one of $\set{e_k, e_l}$. Without loss of generality, say $e_i = e_k$. Then $o(u) \cap o(v) \subseteq \set{e_i, e_j}$ and $o(u) \cap oip(v) \subseteq \set{e_i, e_l}$, so $o(u) \cap (o(v) \,\D\, oip(v)) \subseteq \set{e_j, e_l}$, and hence $|o(u) \cap (o(v) \,\D\, oip(v)| \leq 2$.

            \item If $|v| = 4$, then $o(v)$ has weight at least $3$, so the only possible intersection must come from $oip(v)$, as shown in Figure~\ref{fig:leveldiagram}(2). We argue similarly as before; namely that $ip(v)$ has at most one element, and if such an element $w$ exists it must have weight 2 for $oip(v) \cap o(u) \neq \es$. We write $w = e_i + e_j$ for $i \neq j$. Hence, $oip(v) = o(w)$ has at most two weight-1 elements, $e_i$ and $e_j$. So $|o(u) \cap (o(v) \,\D\, oip(v))| \leq 2$.
        \end{enumerate}

    \item In case (a) we showed that any nontrivial overlap with $o(v) \,\D\, oip(v)$ requires $|v| \equiv 0 \pmod{2}$. But $v \in o(u)$ only if $|v| = 1$. In the former case, i.e. $|v| \equiv 0 \pmod{2}$, $v \notin o(u)$, so we reduce to bounding $|o(u) \cap (o(v) \,\D\, oip(v))|$ which we have done already in case (a). In the latter case, $o(u) \cap (o(v) \,\D\, oip(v)) = \es$ so $|o(u) \cap (\set{v} \,\D\, o(v) \,\D\, oip(v))| = |o(u) \cap \set{v}| \leq 1$.

    \item Now $v \in \CP$, and $u \in \CO \cup \CP$ as before. If $o(u) \cap oi(v) \neq \es$, then $|v| \in \set{1, 3}$ as shown in Figure~\ref{fig:leveldiagram}(3)-(4). As seen earlier $|v| \equiv 0 \pmod{2}$ in order for $o(u) \cap (o(v) \,\D\, oip(v)) \neq \es$. For $|v|$ even we have shown before the $B = 2$ bound. For $|v|$ odd it suffices to consider the overlap $o(u) \cap oi(v)$. Since $v \in \CP$, $i(v)$ is the unique element $w$ associated to the pivot $v$. 
    \begin{enumerate}[i.]
        \item If $|v| = 1$, then $|w| \in \set{0, 2}$ as shown in Figure~\ref{fig:leveldiagram}(3). Since $w \in \CI$ and $u \notin \CI$, $w \neq u$ so $|w| \neq 0$. Write $w = e_i + e_j$ for some $i \neq j$. Thus, $o(u) \cap oi(v) \subseteq \set{e_i, e_j}$, so $|o(u) \cap oi(v)| \leq 2$.
        \item If $|v| = 3$, then unless $|w| = 1$, $oi(v) \cap o(u) = \es$ as shown in Figure~\ref{fig:leveldiagram}(4). Then $w = e_i + e_j$ for $i \neq j$, so $o(u) \cap oi(v) \subseteq \set{e_i, e_j}$ and thus $|o(u) \cap oi(v)| \leq 2$.
    \end{enumerate}

    \item Now, $u \in \CI,\, v \in \CP$, and $u \neq i(v)$. This case has essentially already been proven by part (c). The only difference is that the input node $w = i(v)$ cannot have weight $0$ because $u \neq i(v)$.
\end{enumerate}
\end{proof}

As a consequence of the optimal sensitivity of the hypercube, we can not only correct its errors well by a greedy approach, but also obtain a good distance lower bound that is at most a factor of 2 from optimal.
\begin{corollary}
    The $d$-dimensional hypercube cube has distance between $2\ceil{\frac{d}{4}} + 1 \approx d/2$ and $d$.
\end{corollary}

We believe that a more careful analysis requiring detailed casework can show that our algorithm is correct up to $\frac{d}2$ errors, and thus distance $d$.
However, this example already shows that we can use this formalism to construct graphs that give simple and intuitive algorithms for decoding.

\begin{conjecture}
The $d$-dimensional hypercube code has distance $d$ and a greedy decoding algorithm can decode up to $\lfloor\frac{d-1}2\rfloor$ errors.
\end{conjecture}

\subsection{Random codes with reduced stabilizer weights}
\label{subsec:randcodes}

One direct technique for which a graph formalism enables more controlled analysis is random codes. Under the stabilizer tableau representation, the quantum Gilbert-Varshamov-bound gives an asymptotic (large $n$) relation between code parameters $n, k, d$ that, if satisfied, guarantee the existence of a $\llbracket  n, k, d\rrbracket$ code. \begin{theorem}[Quantum Gilbert-Varshamov] \label{thm:quantum_gv}
    For parameters $d \leq \frac{n}{2}$ and $k$ which depend implicitly on $n$, as $n \to \infty$ there exists a $\llbracket  n, k, d\rrbracket$ code if $n H(d/n) + d \log 3 < n - k$, where $H(p) \coloneq -p \log p - (1-p) \log(1-p)$ is the binary entropy function of $p \in [0, 1]$.
\end{theorem}
We give a self-contained proof of Theorem~\ref{thm:quantum_gv} in Appendix~\ref{app:QGV}.
The proof proceeds probabilistically, namely by generating a uniformly random stabilizer tableau and arguing that such a tableau has a nonzero probability of having distance $d$ if the bound is satisfied. However, because the expected weight (i.e. number of non-identity elements) of a random Pauli string $P_{1} \otimes \cdots \otimes P_{n}$, where $P_i \in \set{I, X, Y, Z}$, is $O(n)$, the random codes generated in the Gilbert-Varshamov fashion will have linear stabilizer weight. A natural question is whether the average weight can be decreased by a stochastic algorithm that is more fine-grained than uniform randomness. Such a task has proven difficult in the stabilizer tableau representation due to the challenge of generating Paulis that both commute and are not uniformly distributed. On the other hand, we can easily generate random graphs of a certain structure by defining first the possible edges that may be included, and then randomly choosing edges to include in the graph. This insight enables us to extend the result of the quantum Gilbert-Varshamov bound by way of the following theorem. 

\begin{theorem} \label{thm:random_graphs}
    Let $n$ be the number of physical qubits and define $R(n) < 1$ and $d(n) \in [\w(1), \frac{n}{\log n}]$. There exists an efficiently sampleable family of distributions $\CS_n$, parameterized by $n$, over stabilizer codes that as $n \to \infty$ satisfy three properties.
    \begin{enumerate}[(a)]
        \item The stabilizer weights are \begin{align}
            |S_i| \leq \min(n, (4+o(1)) R d^2 \log^2 n) ,
        \end{align}
        where $S_i$ is a stabilizer from a tableau sampled from $\CS_n$, 
        \item the rate is $R(n)$, and
        \item with probability at least $1 - n^{-d}$, a random code from $\CS_n$ will have distance $d(n)$.
    \end{enumerate}
\end{theorem}

We remark that the idea of finding a smaller class of objects which achieve desirable properties is conceptually similar to an idea explored recently by \textcite{cleve2015near}, who showed that there is a strict subgroup of the Clifford group that also uniformly mixes non-identity Paulis.

Theorem~\ref{thm:random_graphs} is positive in that it provides a nontrivial bound on the stabilizer weight of random codes for many choices of $R(n)$ and $d(n)$, and it gives a strong probabilistic concentration of $1 - n^{-d}$. For any $d \in [\w(1), O(\frac{n}{\log n})]$, this concentration implies that almost all codes in $\CS_n$ have the desired properties. On the other hand, Theorem~\ref{thm:random_graphs} complements, as opposed to subsuming, the Gilbert-Varshamov bound because it provides a nontrivial weight reduction at substantially sublinear distances. It also forces the stabilizer weight, at least for the generating set we construct, to be at least as large as the distance, which is an artifact of our proof technique.
However, the result nonetheless remains interesting because of the guarantees it is able to provide and the flexibility that it offers. Indeed, practically useful codes are seldom asymptotically good. Since the distance and the rate are both determining factors for the stabilizer weight, we obtain a three-way distance-rate-weight trade-off of random codes that can be leveraged to study sublinear-distance codes in which we care much less about one of the three properties than the others. As with the other claims in this chapter, this result is thus complementary of LDPC constructions that are asymptotically good but have intractably large constants.
Instead, we interpret this result as evidence that for moderately large codes, e.g. $n \sim 50$ at a scale where we expect little-$o$ asymptotics to take effect, there is a great deal of flexibility in the graph construction to build codes of certain ranges of $R$ and $d$ which also may have a relatively small stabilizer weight.

\begin{figure}[ht]
    \centering
    \includegraphics[width=0.67\linewidth]{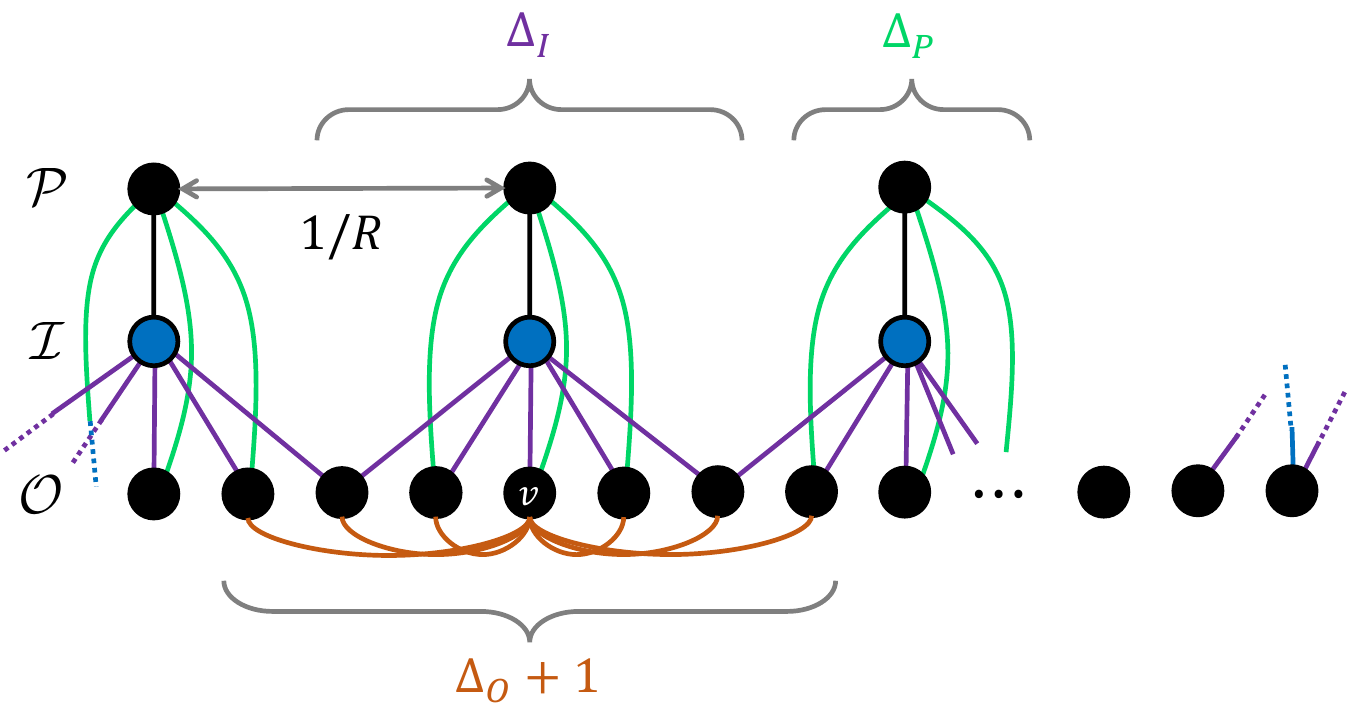}
    \caption[Family of random graphs]{General localized random graph construction. Nodes are separated into $k$ input nodes $i \in \CI$, (which are spaced apart by $1/R$), $k$ pivot nodes $p \in \CP$, and $n-k$ non-pivot output nodes $o \in \CO$. Each pivot $p$ (input $i$) can connect to output nodes within a diameter $\D_{P}$ of $p$ ($\D_I$ of $i$); here, $\D_P = 3$ ($\D_I = 5$) and the possible lines are green (purple). Similarly, output nodes $o$ may connect to outputs in a diameter $\D_O$ of $o$. For clarity, only the possible connections for a single output $v$ is shown, and here $\D_O = 6$.}
    \label{fig:localgraph}
\end{figure}

We proceed to the proof of Theorem~\ref{thm:random_graphs}.
\begin{proof}
    We construct elements of $\CS_n$ graphically as shown in  Figure~\ref{fig:localgraph} and we denote this graph $G$. In particular, we begin with an encoder-respecting form, separating the inputs $\CI$, pivots $\CP$, and (non-pivot) outputs $\CO$. Next, we assign each set of nodes a \textit{locality diameter}, respectively $\D_I, \D_P, \D_O+1$ (the $+1$ since the output node is included in its diameter but does not have a self-edge). The boundary conditions are periodic. Each node may only connect to output nodes within their locality diameter. We sample from $\CS_n$ by iterating through each connectable edge and randomly including or excluding it from the graph with probability $1/2$; the resulting stabilizer tableau is the one defined by the collection $S \coloneq \set{S_v}_{v \in \CO}$ where $S_v \coloneq X_v Z_{N_o(v)} X_{p(i(v))} Z_{N_o(p(i(v)))}$ as described in Section~\ref{sec:inversion}. For notation, we say that $S_v$ contains a Pauli $X_i$ ($Z_i$) if the $i^{\text{th}}$ element of $S_v$ is $X_i$ or $Y_i$ ($Z_i$ or $Y_i$). Our analysis of this stochastic algorithm adapts the techniques in the analysis of Theorem~\ref{thm:quantum_gv} to the graph picture. We fix a non-identity Pauli $P$ and consider four cases.
    \begin{enumerate}[(1)]
        \item $P = Z^{\otimes T}$ for $T \subseteq [n]$ and $T$ does not contain any pivots. Pick $i\in T$ and consider $v_i \in G$. Then $S_{v_i}$ must anticommute with $P$ because $S_{v_i}$ contains $X_{v_i}X_{p(i(v_i))}$ which only overlaps $P$ in one position. Hence, $P$ does not commute with $S$.

        \item $P = Z^{\otimes T}$ for $T \subseteq [n]$ and $T$ contains at least one pivot $p \in \CP$. Let $i \in \CI$ be the unique input connected to $p$. Each of the $\D_I$ output nodes $v$ within the diameter $\D_I$ of $i$ have probability $1/2$ of connecting to $i$. $S_{v}$ contains $X_p$ if and only if $i$ and $v$ are connected. Since the (anti-)commutativity of $P$ and $S_v$ are determined only by the parity of the number of qubits on which $X$'s in $S_v$ overlap with $Z$'s in $P$, the probability that $S_v$ and $P$ commute is $1/2$ regardless of the other Paulis in $S_v$ and $P$. Furthermore, consider any other $S_w$ where $w$ is also in the diameter of $i$. Since $w$ has an edge connecting to $i$ distinct from that of $v$, the probability of $S_w$ and $P$ commuting is independently $1/2$. This argument is general to all output nodes in the diameter of $i$, and remaining stabilizer act on disjoint subspaces. Hence the probability that $P$ commutes with $S$ is $2^{-\D_I}$. 

         \item $P$ contains at least one $X_p$ for $p \in \CP$. Analogously, $\Pr[P \text{ commutes with } S] = 2^{-\D_P}$ by considering the $\D_P$ nodes in $\CO$ which each independently have probability $1/2$ of having a $Z_p$.

        \item $P$ contains at least one $X_v$ for $v \in \CO$. $\Pr[P \text{ commutes with } S] \leq 2^{-\D_O}$ by an analogous argument. We neglect some stabilizers which could anticommute with $P$ by acting through neighbouring pivots, i.e. the $Z_{N_o(p(i(v)))}$ component, since that can only decrease the probability and complicates the expression.
    \end{enumerate}
    Set $\D_O = \D_P = \D_I \coloneqq \D$.
    We proceed to bound the probability of the event $A(S)$ that there exists non-identity $P$ with weight $\leq d$ which commutes with $S$. As shown in Appendix~\ref{app:QGV}, there are $2^{n H(d/n) + d \log 3 + O(\log n)}$ such Paulis, where $H(p) = -p\log p - (1-p) \log(1-p)$ is the binary entropy. By the union bound, \begin{align}
        \Pr[A(S)] \leq 2^{n H(d/n) + d \log 3 + O(\log n) - \D} \leq \e ,
    \end{align}
    where $\e$ bounds the fraction of codes which do not have distance at least $d$. We therefore solve the inequality \begin{align}
        \D \geq \log \frac{1}{\e} + n H\left(\frac{d}{n}\right) + d \log 3 + O(\log n) .
    \end{align}
    Since $\frac{d}{n} \leq \frac{1}{\log n}$, we can asymptotically expand the binary entropy as $H(d/n) \to (\frac{1}{\ln 2} - \log \frac{d}{n}) \frac{d}{n}$, giving \begin{align}
        \D \geq \log \frac{1}{\e} + d \log n + o(d \log n) .
    \end{align}
    Here we have used our assumption that $d = \w(1)$, so that $O(\log n) = o(d \log n)$.
    Let $\e = n^{-d}$, so that if $\D = 2 d \log n + o(d \log n)=(2+o(1))d\log n$ at least $1 - n^{-d}$ of codes asymptotically have distance $d$. Lastly, we calculate the stabilizer weights. Each stabilizer is of the form $S_v = X_v  Z_{N_o(v)}X_{p(i(v))} Z_{N_o(p(i(v)))}$. These terms have at most $1$, $\D_O$, $\D_I R$, and $\D_P\D_I R$ Paulis, respectively. The worst case occurs when the sets have no overlap and all edges are connected. Hence, \begin{equation}
        |S_v| \leq 1 + \D + \D R + \D^2 R = (4+o(1)) R d^2 \log^2 n.
    \end{equation}
    Additionally, $|S_v| \leq n$, which gives the claimed result. 
\end{proof}
An important special case of Theorem~\ref{thm:random_graphs} is for graphs in which all inputs are at least 3 nodes away from each other, in which case $R = 1/\D$ and thus $|S_v| \leq (2 + o(1)) d\log n$.

\section{Conclusion and outlook}
\label{sec:conclusion}

We introduced a graph structure which universally represents all stabilizer codes. Stabilizer tableaus can be efficiently compiled into such graphs, and vice versa. Our primary motivation for such a representation was to gain access to natural graph properties and notions, such as their degree, geometry, and connectivity, and then to leverage them to improve code construction and analysis. As first steps in this direction, we chose several geometric shapes discretized into graphs and analyzed them as code representations. In doing so, we found a number of constant-size codes with desirable rates and reasonably large numerical distances, as well as a family of hypercube codes that have near-linear rate, logarithmic distance, and good encoding/decoding properties. In similar spirit, we also constructed a class of codes for which, given a desired rate $R$ and distance $d$, a random code drawn from the class has high probability of having rate $R$, distance $d$, and, so long as $R$ and $d$ were not both too large, nontrivially bounded stabilizer weights. This analysis extends the result of the quantum Gilbert-Varshamov bound, which has a distance-rate trade-off but no bound on stabilizer weight due to its coarse-grained analysis.
 \chapter{Equivalent Quantum Codes}
\label{chapter:code-equivalence}

\section{Introduction}

A number of approaches have been created to represent the components of quantum error-correcting codes. The stabilizer formalism is a method that expresses quantum error-correcting codes in terms of stabilizers, operators that, when applied to certain stabilizer states, preserve the state~\cite{gottesman1997stabilizer}. This approach borrows ideas from group theory to represent the whole class of stabilizers with a finite number of generators. To make the idea of quantum error-correcting codes visual, recent advances have made progress on the topic of representing quantum states through graphs~\cite{vandennest2004graphical}.
We as we saw in the last few chapters, graphs can be incredibly powerful in their expressions of quantum states, circuits, codes, and more.

\begin{figure}[ht]
    \begin{center}
    \includegraphics[width=0.8\textwidth]{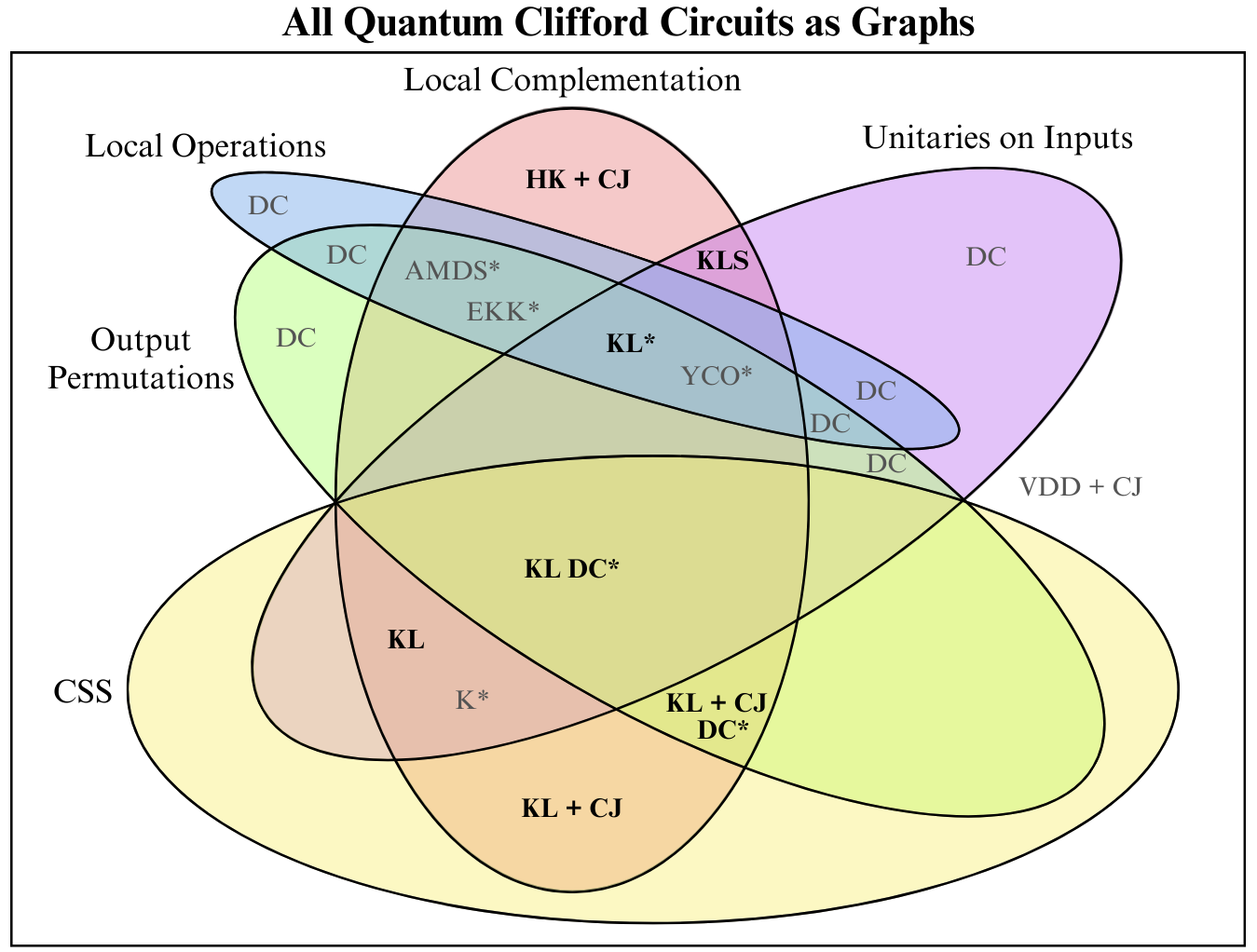}
    \end{center}
    \caption[Works on graph representations under various equivalences]{A summary of the work done on the equivalence classes of graphs of quantum Clifford encoders. Four of the categories, output permutations, local operations, local complementation, and unitaries on inputs, are different equivalences that preserve the information received. The CSS category is not an equivalence relation and shows results specific to CSS codes. Note that the CSS and local operations categories do not intersect because equivalence by local operations can transform a CSS code into a non-CSS code. The regions of the diagram are labelled with works on studying graphs under the corresponding equivalences. Asterisks denote results that give non-unique forms, and DC denotes results that are direct consequences of previous works. Abbreviations used are explained in the main text.
}
    \label{fig:venn-diagram}
\end{figure}

Following the work on graph states, progress has been made on representing Clifford codes using the ZX-calculus~\cite{coecke2008interacting, coecke2011interacting,backens2014zx}, a graphical language for expressing and manipulating certain families of tensor networks. The properties of the ZX-calculus allow it to do anything that can be done in the stabilizer formalism, an expression of codes and states as a list of generators of commuting operators.
Specifically, these properties include its universality, meaning it can express every quantum operation, its soundness, meaning that we can independently derive equivalence rules for ZX-calculus diagrams, and its completeness, meaning ZX-calculus diagrams can derive equivalences of tableaus~\cite{coecke2008interacting,backens2014zx}. This graphical language has had various applications in quantum information~\cite{peham2022equivalence, cowtan2022quantum, van2021constructing, east2022aklt} and quantum computation problems~\cite{de2020zx, kissinger2020reducing}.

Figure~\ref{fig:venn-diagram} gives a summary of the work done in graphically presenting quantum Clifford encoders. The set of all quantum Clifford encoders as graphs is split into different categories based on the type of code represented (CSS or general Clifford codes) and four different equivalence operations (output permutations, local operations, local complementation, and unitaries on inputs). Note that equivalence under local operations means equivalence under any local Clifford operations on the outputs. All four of these categories represent different equivalences, which preserve the information the receiver gets while changing the encoder's ZX diagram in some way. Some of them change the encoder (output permutations, local operations, and unitaries on inputs), and a different set of them change the code (output permutations and local operations). 

Expressing quantum Clifford encoders as graphs and finding canonical forms for equivalent graphs has had recent advances in the past few decades.
A depiction of this is shown in Figure~\ref{fig:venn-diagram}. The work of \textcite{vandennest2004graphical}~(VDD) provides a conversion between any Clifford state and a graph with local Clifford gates, and the Choi-Jamio\l{}kowski isomorphism~\cite{choi1975completely,jamiolkowski1972linear}~(CJ) extends this to a conversion between any Clifford circuit and a graph with local Clifford gates. Starting from the graphs of the Clifford encoders, we find the canonical forms of the direct consequences~(DC's) located only within the bubbles for output permutations, local operations, and unitaries on inputs. For output permutations, we remove the numbering on the outputs. For local operations, we remove all local Clifford operations on the outputs. For unitaries on inputs, we remove input-input edges, local Clifford operations on the inputs, and the numbering on the inputs. Any combination of these three become DC in Figure~\ref{fig:venn-diagram}.

The HK form from Chapter~\ref{chapter:qstatesgraphs} provides a canonical form for quantum Clifford states. In the context of quantum encoders, this is equivalent to having no inputs and only outputs.
Additionally, the KLS form from Chapter~\ref{chapter:qcodes-graphs} built on the HK form, providing a canonical form for Clifford encoders.
Chapter~\ref{chapter:qcodes-graphs} shows the process of transforming stabilizer tableaus into the ZX-calculus, then performing operations that preserve equivalence to transform the graph into its canonical form.
A similar construction to the one used in Chapter~\ref{chapter:qcodes-graphs} has been studied by \textcite{yu2007graphical}~(YCO).
\textcite{adcock2020mapping}~(AMDS) and \textcite{englbrecht2022transformations}~(EKK) considered equivalence of graph states under local complementations and the effects of relabelling the nodes.

\textcite{kissinger2022phase}~(K) found a way to represent the CSS codes using internal measurement nodes. Section~\ref{sec: kl forms} establishes the KL canonical form of CSS codes and states. In Figure~\ref{fig:venn-diagram}, direct consequences (DC's) of Section~\ref{sec: kl forms} follow through a removal of the numbering on output nodes.
These fall within the CSS region.
In Section~\ref{sec 4}, Section~\ref{sec: tabulations from code}, and Section~\ref{sec:bipartite forms}, we consider equivalence classes under all four equivalences for general Clifford encoders.

Informally, our main results in this chapter are specializing the HK and KLS canonical forms to CSS codes.
This simpler set of rules for our canonical form allows for the simpler generation of CSS codes.
Simply by constructing a bipartite graph and designating nodes in one of the two vertex groups to be inputs, one can design new CSS codes.
In this chapter, we also investigate and classify families of small stabilizer codes under the afore-mentioned equivalence relations.

In this chapter, Section~\ref{section:background} contains key definitions and background on the ZX-calculus and Clifford encoders. Section~\ref{sec: kl forms} presents our main result, the KL canonical form for CSS codes, giving a unique, phase-free form for CSS codes that minimizes the number of nodes and clearly shows the $Z$ stabilizers and logical $Z$ operators. These results are also extended to CSS states, which will be defined later. Section~\ref{section:surfacetoric} contains our results on toric codes and specific surface codes, and it shows different forms of these codes, including the canonical form for the toric code based on Section~\ref{sec: kl forms}. These two sections build on the KLS forms for Clifford codes and recent work~\cite{huang2023graphical, li2023graphical, kissinger2022phase} that introduced the normal form of CSS codes, which can efficiently determine the stabilizers from the ZX normal form. Our representations of CSS codes will also have this property. Furthermore, our representations reduce the number of nodes used in the ZX diagrams so that each node corresponds to either an input or output.

Section~\ref{sec: prime codes} introduces prime code diagrams, which are code diagrams that are composed of one connected component. Furthermore, we prove the Fundamental Theorem of Clifford Codes, showing a unique prime decomposition of Clifford codes.

Section~\ref{sec 4} provides another definition of equivalence, permitting outputs to be permuted as a valid operation among equivalent graphs.
This definition of equivalence is also considered since changing the order of the outputs does not change any code parameters. Simplifications on the set of non-equivalent Clifford encoders are given to narrow down the search for the canonical form. In Section~\ref{sec: tabulations from code} and Section~\ref{sec:bipartite forms}, we show our results on identifying equivalence classes and finding representative forms. We analyze the equivalence class sizes and the presence of bipartite forms among these classes. Section~\ref{sec:bipartite forms} expands on the equivalence classes containing bipartite forms and considers some classes that do not have bipartite forms.

\section{Background}
\label{section:background}

We begin by defining key terms and background on error-correcting codes and the ZX-calculus.

First, we recall the definition of Pauli matrices, as given in Definition~\ref{def:pauli}.

The \textit{Pauli matrices} are
\begin{equation}
    I \coloneq \begin{pmatrix}1&0\\0&1\end{pmatrix}, \quad X \coloneq \begin{pmatrix}
    0&1\\1&0
    \end{pmatrix}, \quad
    Y \coloneq \begin{pmatrix}
       0&-i\\i&0
    \end{pmatrix},\quad Z\coloneq \begin{pmatrix} 1&0\\0&-1\end{pmatrix}.
\end{equation}

    The Pauli matrices represent quantum gates that can act on qubits and alter their state. All four gates are Hermitian, and the three gates $X,Y,Z$ are pairwise anti-commutative. The \textit{Pauli operators on $n$ qubits} are $n$-fold tensor products of Pauli matrices, multiplied by a factor of the form $i^k$ where $k \in \{0,1,2,3\}$ and $i = \sqrt{-1}$.

The Pauli operators are all equal to their conjugate transposes, so all Pauli operators are Hermitian and unitary. The Pauli operators form a group, called the \textit{Pauli group}.

Pauli operators can act on states in multi-qubit systems. For example, in a three qubit system, the tensor product $Z\otimes Z \otimes I$ will make $Z$ act on the first qubit, $Z$ act on the second qubit, and $I$ act on the third qubit. The notation for the tensor product can be simplified to $Z_1Z_2$, with the subscripts showing which qubits the operators are acting on. We may also write these three gates as $ZZI$, omitting the tensor product symbols. Other Pauli operators on multiple qubits can be written analogously. 

Other quantum gates that are commonly used in quantum error-correcting codes are the Hadamard ($H$), controlled-NOT (CNOT), phase ($S$), and $\pi/8$ ($T$) gates, as defined in Definition~\ref{def:clifford}.
\begin{equation}
H \coloneq \frac{1}{\sqrt{2}} \begin{pmatrix}
1 & 1 \\ 1& -1 \end{pmatrix}, \quad \text{CNOT} \coloneq \left(\begin{smallmatrix}
1&0&0&0\\ 0&1&0&0\\ 0&0&0&1\\0&0&1&0
\end{smallmatrix}\right), \quad S \coloneq \begin{pmatrix}
1 & 0 \\ 0&i\end{pmatrix}, \quad T \coloneq \begin{pmatrix}
1&0\\0&e^{i \pi  /4}\end{pmatrix}.
\end{equation}

These operations have the property of universality, the ability to approximate any operator to arbitrary accuracy~\cite{nielsen2002quantum}.

We define \textit{encoders} of quantum error-correcting codes as families of quantum processes that apply a transformation on some number of input qubits, mapping it to a given range.
Specifically, we will only be concerned with the case of \textit{full-rank} or \textit{non-degenerate} encoders where the encoding operation is injective, meaning that the dimension of the range is at least as large as the number of inputs.
We write that an encoder takes $k$ inputs, or logical qubits, and gives $n$ outputs, or physical qubits.
Note that an encoder can mean any quantum map in the family of such processes defined by the image.
Thus, two different circuits that differ only by a unitary operation on the inputs represent the same encoder.

Encoding quantum information into a larger number of qubits provides redundancy, making it possible to correct certain errors. The \textit{stabilizers} of an encoder are a set of commuting operations that can be composed with the output of an encoder without any change to the overall process. In the stabilizer formalism, the stabilizers determine the entire quantum code~\cite{gottesman1996class}.
A family of encoders are \textit{Clifford codes}, quantum error-correcting codes such that each stabilizer is a Pauli operator on $n$ qubits. All stabilizers of a Clifford code on $k$ inputs form a group isomorphic to $\mathbb{Z}_2^{n-k}$ and can be defined by a set of linearly independent \textit{generators}. For $k$-to-$n$ codes, or $\llbracket n,k\rrbracket$ codes, we have $n-k$ generators. The code maps the input qubits onto elements of the \textit{codespace}, the range of the code, which is the intersection of the +1 eigenspaces of the code's stabilizers.

The ZX-calculus is a graphical language used for expressing quantum states, circuits, and codes through tensor networks made of two types of nodes, $Z$ and $X$ nodes.
This chapter makes heavy use of ZX-calculus diagrams and rewrite rules.
An introduction to the ZX-calculus with associated definitions and rewrite rules can be found in Appendix~\ref{app:zx-calculus}.

We recall the definition of ZX diagrams of codes expressed in their \textit{encoder-respecting form} from Definition~\ref{def:respect}.
We restate this definition here with a slight modification: we allow $X$ nodes in the diagram.
If these $X$ nodes were to be replaced by $Z$ nodes with Hadamard gates on every incident edge, the result would always be a diagram satisfying the definition of Definition~\ref{def:respect}, so relaxing this constraint does not meaningfully change the set of encoder-respecting forms.

The \textit{encoder-respecting form} is a ZX representation of a Clifford code that contains $Z$ (green) and $X$ (red) nodes, with each node corresponding to an input or output. Each node has a corresponding free edge, an edge not connected to any other nodes, with input nodes having input edges and output nodes having output edges.
The outputs may have local operations on their free edges. Each of the $k$ input nodes may only have connections with the output nodes, while each of the $n$ output nodes may have connections to each other. Additionally, the output edges are numbered from 1 to $n$.

\begin{figure}[ht]
\begin{center}
\includegraphics[scale=0.45]{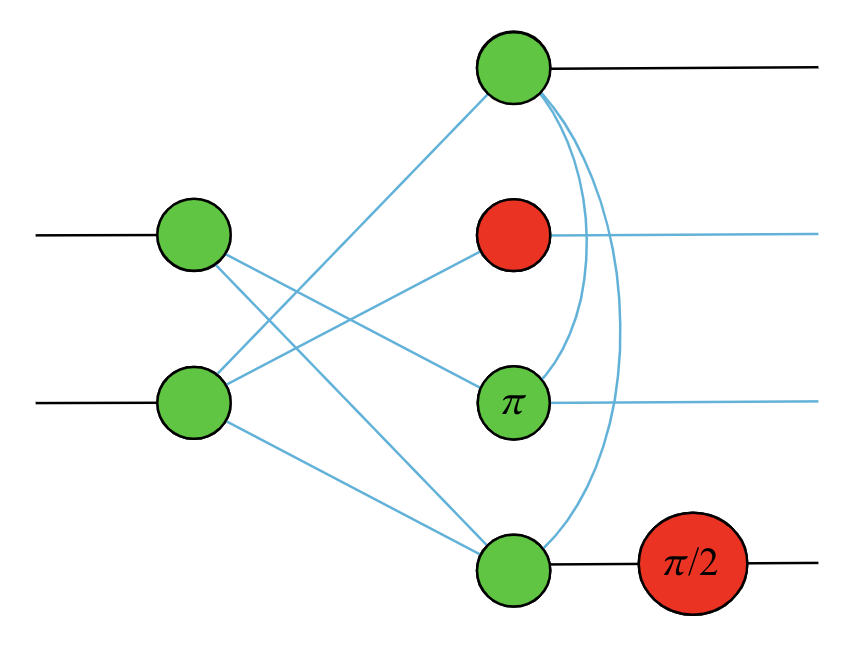}
\end{center}
\caption[Encoder in the ZX-calculus]{Example of an encoder in the ZX-calculus. The incoming edges from the left side are input edges (sending information in) and the outgoing edges on the right side are output edges (sending encoded information out). Note the local operations applied on the output qubits, with blue edges representing edges with Hadamards. This encoder is not in encoder-respecting form.}
\label{fig:encoder-example}
\end{figure}

An example of a Clifford code expressed in encoder-respecting form is shown in Figure~\ref{fig:encoder-example}. There are local operations on the output edges, as shown by the blue free edges and $\sqrt{X}$ gate ($X$ node with phase $\pi/2$). The $X$ node is not considered an output node since it is isolated as a local operation on its neighbouring $Z$ node. There are internal edges in the graph between the input and output nodes and from the output nodes to other outputs, but not from inputs to inputs. Two qubits of are encoded into four qubits in this example.

We recall from Definition~\ref{def:zxcf} and Theorem~\ref{thm:canonicity} that encoders with the same sets of stabilizers have the same ZXCF, also known as their KLS canonical form. This form consists of four rules that can be efficiently checked for a given ZX diagram.
We restate the main result of Theorem~\ref{thm:canonicity} here.

For any stabilizer code there is a unique equivalent ZX diagram satisfying the following rules.
We say that a ZX diagram satisfying these rules is a KLS diagram and is in KLS form.
A KLS diagram must be in encoding-respecting form and all of its nodes must be $Z$ nodes. Such a diagram must also satisfy the following rules.
\begin{enumerate}[(1)]
    \item \textit{Edge rule:} All internal edges have Hadamards, and there is exactly one $Z$ node per free edge.
    \item \textit{Hadamard rule:} Output nodes with Hadamard gates on their free edges cannot share an edge with a lower-numbered output node or with an input node.
    \item \textit{RREF rule:} The adjacency matrix representing the edges between input nodes and output nodes is in reduced row-echelon form (RREF).
    \item \textit{Clifford rule:} In the RREF matrix, the pivot columns of the input to output adjacency matrix correspond to pivot output nodes. There are no local Clifford operations on the pivot or input nodes, or their free edges. There are also no input-input edges or pivot-pivot edges.    
\end{enumerate}

Additionally, we know from Theorem~\ref{thm:canonicity} that a given ZX diagram can be efficiently transformed to its KLS canonical form using a series of ZX rewrite rules.
Furthermore, as discussed in Section~\ref{subsec:encoding_circuit}, KLS canonical forms may be efficiently transformed into quantum circuits.
We list a slightly modified version of the transformation procedure, which only differs from the encoder in Section~\ref{subsec:encoding_circuit} in where it places the layer of Hadamard gates.
The exact ZX-calculus transformation rules that transform a graph code into an encoding circuit described below are shown in Appendix~\ref{app:convert-zxcf-to-circuit}. 

Consider an $\llbracket n,k\rrbracket$ encoder given in its KLS form. Then, it can be efficiently transformed into an equivalent quantum circuit using the following steps.

\begin{enumerate}[(1)]
\item Start with $k$ open wires representing the inputs of the circuit.
\item Add a $\ket{0}$ state for each of the $n-k$ non-pivot output nodes.
\item Apply an $H$ gate to all $n$ wires.
\item Apply a $CX$ gate between the wires corresponding to the edges between inputs and non-pivot outputs. The input node is the target qubit, and the output node is the controlled qubit.
\item Apply a $CZ$ gate between the wires corresponding to the edges between only outputs.
\item Apply the local operations attached to the outputs.
\end{enumerate}

This procedure works by building up the encoder's quantum circuit representation in layers, starting from the input qubits, which correspond to the ZX diagram input nodes, adding auxiliary qubits that encode the information from the input qubits, and connecting the qubits using the appropriate gates.

Lastly, we recall several key definitions regarding graphs.

The \textit{neighbourhood} $N(v)$ of a vertex $v$ in a graph $G=(V,E)$ is the set of all vertices in $V$ adjacent to $v$, not including $v$ itself. An operation commonly used to apply equivalence transformations to ZX diagrams is local complementation, as defined in Definition~\ref{def:graph-local-complementation-vertex}.

Let $G$ be a graph $G=(V,E)$, where $V$ and $E$ are the sets of vertices and edges in $G$, respectively. Consider a vertex $v\in V$. A \textit{local complementation} about vertex $v$ is a transformation of the graph $G$ where all edges connecting two vertices in $N(v)$ are toggled. That is, if the edge existed before the local complementation, it is removed; if it did not exist before, it is added.
We denote the effect of this operation as $L_v(G)$.

\section{Canonical form for CSS codes and states}
\label{sec: kl forms}

Calderbank-Shor-Steane (CSS) codes are a commonly studied class of quantum error-correcting codes constructed starting from two classical codes~\cite{zarei2017strong, steane1999enlargement, sarvepalli2009sharing, harris2018calderbank}.
The generators of a CSS code's stabilizers can be chosen such that each of the generators is either a Pauli operator with only $I$ and $Z$ gates or a Pauli operator with only $I$ and $X$ gates.

We introduce the notion of a CSS state.

\begin{definition}
A \textit{CSS state} is a CSS code with 0 inputs.
\end{definition}

We now present the KL canonical form of CSS codes and CSS states in the ZX-calculus, and prove that it is canonical in this section. The KL forms for CSS codes are special cases of KLS forms, which are for Clifford codes.
However, since CSS codes have nice properties and structure, it is interesting to study the rules and properties governing this subclass of canonical forms, KL forms.

All KL forms are in encoder-respecting form as in Definition~\ref{def:respect}, meaning every input and output node has its own free edge, and the output nodes are numbered from $1$ to $n$.
However, as stated in Section~\ref{section:background}, we allow encoder-respecting forms to have both $Z$ and $X$ nodes.

\begin{figure}
    \begin{center}
    \includegraphics[scale=0.35]{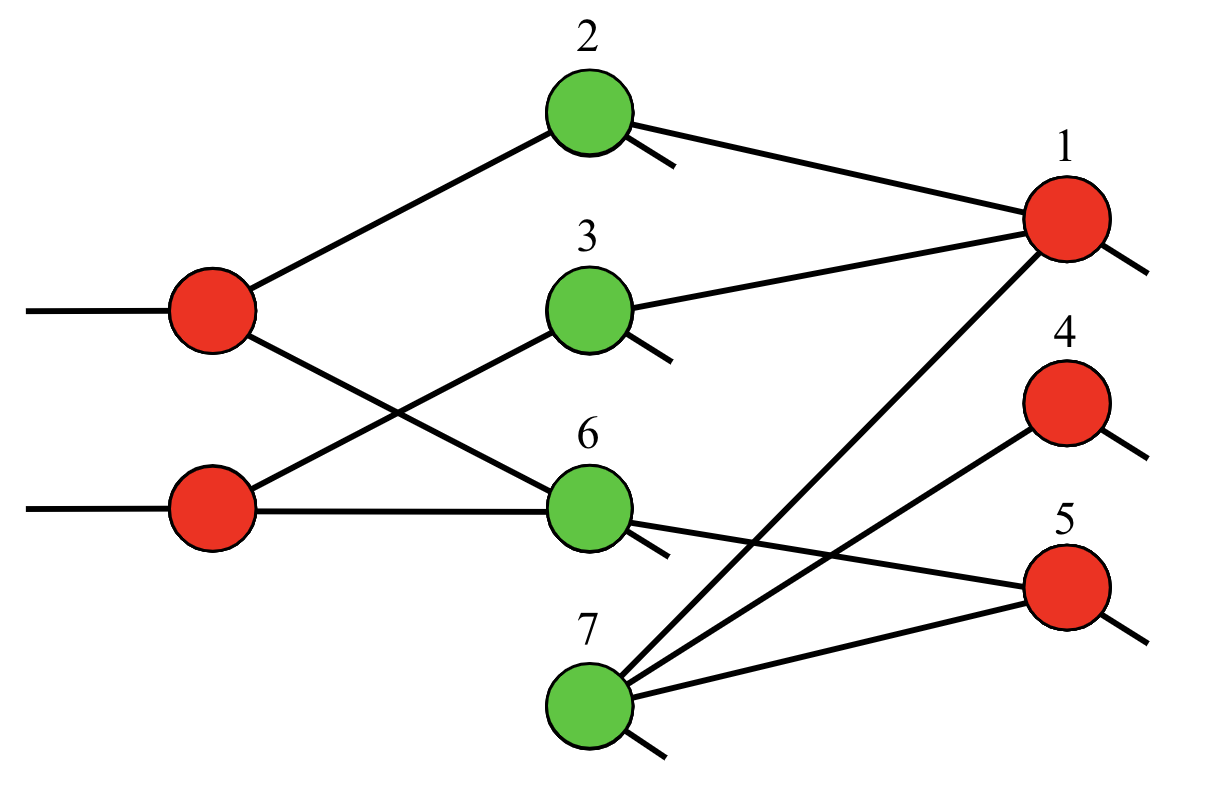}
    \end{center}

    \caption[KL form code example]{
        Example of a CSS code in KL canonical form. The short diagonal edges are the free output edges of the 7 output nodes.}
            \label{fig: KL form example}

        \end{figure}
        
\begin{theorem}[KL canonical form for CSS codes]
\label{CSS code canonical form}
For any CSS code, there is a unique equivalent ZX diagram satisfying the following rules.
\begin{itemize}
\item \textit{Bipartite rule:} The nodes can be split into two groups. One group consists of the input nodes (which are all $X$ nodes) as well as the output $X$ nodes.
The other group consists of the output $Z$ nodes. The only interior edges allowed are between two nodes from two different groups.
\item \textit{Phase-free rule:} All input and output nodes have phase 0 and there are no local operations on any free edges.
\item \textit{RREF rule:} The adjacency matrix between the input $X$ nodes and output $Z$ nodes is in reduced row echelon form (RREF). The adjacency matrix between the output $X$ nodes and \textit{all output nodes}, where the output $X$ nodes are marked as connected to themselves, is also in RREF.

\end{itemize}

\end{theorem}

\begin{remark}
    Note that the first part of the RREF rule is equivalent to the one set in the KLS form, in which the adjacency matrix between the inputs and outputs is in RREF. Additionally, note that the second part of the RREF rule is equivalent to stipulating that no output $X$ node be connected to a lower-numbered $Z$ node. This implies that the second part of the RREF rule is equivalent to the Hadamard rule set in the KLS form, in which outputs without Hadamards cannot connect to lower-numbered output nodes or input nodes.
\end{remark}

An example of the KL form is shown in Figure~\ref{fig: KL form example}. The Bipartite rule is made clear by the division of the nodes into columns, with the first and last columns of nodes forming the first group ($X$ nodes) and the middle column forming the second group ($Z$ nodes). The edges satisfy the constraint that the output $X$ nodes only connect to higher-numbered output $Z$ nodes. Lastly, the input $X$ nodes to output $Z$ nodes adjacency matrix is in RREF.

In CSS states, the part in the RREF rule pertaining to input nodes is unnecessary because of the lack of input nodes while the Bipartite and Phase-free rules still apply. As a corollary, we also find the canonical form for CSS states.

\begin{corollary}[Canonical form for CSS states]
\label{canonical of CSS state}
    For any CSS state, there is a unique equivalent ZX diagram satisfying the following rules.
    \begin{itemize}
        \item \textit{Bipartite rule:} The nodes can be split into two groups, one with $X$ nodes and the other with $Z$ nodes. The only interior edges allowed are between nodes of different groups.
        \item \textit{Hadamard rule:} Each $X$ node can only connect to higher-numbered $Z$ nodes.
        \item \textit{Phase-free rule:} All output nodes have phase 0, and there are no local operations on any free edges.
    \end{itemize}
\end{corollary}

Now, we build up to the proof of Theorem~\ref{CSS code canonical form}. To do this, we proceed analogously to Section~\ref{sec:compiler}.
Specifically, we show that each KL form corresponds to a distinct CSS code by counting the number of CSS codes and comparing it to the number of KL diagrams, and we later show that any KL form can be converted into the stabilizer representation of the CSS code. These two steps establish the bijection between the KL forms and CSS codes. We start by finding the number of CSS codes.

\begin{lemma}
\label{lemma: number of CSS codes}
    The number of CSS codes with $n$ physical qubits, $p$ $Z$ stabilizers, and $q$ $X$ stabilizers is
\begin{equation}\frac{\prod\limits_{i=1}^{p} (2^n - 2^{i -1})}{\prod\limits_{i =1}^p(2^p - 2^{i  -1})}\cdot \frac{\prod\limits_{i  =1}^q (2^{n-p} - 2^{i-1})}{\prod\limits_{i=1}^q (2^q - 2^{i-1})}.\label{eq: CSS states}\end{equation}
\end{lemma}

\begin{proof}
    Consider choosing the $Z$ stabilizers of the CSS code's stabilizer tableau. Since each stabilizer must be linearly independent from each other, there are $\prod\limits_{i=1}^p(2^n - 2^{i-1})$ ways to choose $p$ independent stabilizers. However, different sets of $p$ generators could represent the same set of stabilizers. For a given set of stabilizers, there are $2^p - 2^{i-1}$ ways to choose the $i^{\text{th}}$ generator, so we must divide to find
    \begin{equation}
    \frac{\prod\limits_{i=1}^{p} (2^n - 2^{i -1})}{\prod\limits_{i =1}^p(2^p - 2^{i  -1})}
    \end{equation}
    as the number of ways to choose the set of $Z$ stabilizers.

    Next, each $X$ stabilizer must commute with all of the $Z$ stabilizers. Since the number of $X$ stabilizers that commute or anti-commute with any single $Z$ stabilizer is equal, there are $2^n /2^p = 2^{n-p}$ $X$ stabilizers to choose from. Using a similar analysis as above, there are
    \begin{equation}
    \frac{\prod\limits_{i  =1}^q (2^{n-p} - 2^{i-1})}{\prod\limits_{i=1}^q (2^q - 2^{i-1})}
    \end{equation}
    ways to choose the set of $X$ stabilizers. 

    Multiplying these two counts gives the number of CSS codes with the given parameters, as desired.
\end{proof}

Now, we find the number of KL diagrams with similar parameters. We will show later that the parameters chosen below make the KL forms correspond exactly to the CSS codes with $n$ outputs, $p$ $Z$ stabilizers, and $q$ $X$ stabilizers.

Consider a KL form of a CSS code. Let there be $p$ output $X$ nodes and $n - p  -q = k$ input nodes, so that there are $n - p$ output $Z$ nodes, of which $k$ are pivot nodes in the RREF adjacency matrix between the input nodes and output $Z$ nodes.

First, consider the output nodes. Note that there are $n$ total nodes among the $X$ and $Z$ groups. To count the number of ways to connect edges between the output $X$ and $Z$ nodes in the KL form, note that the Hadamard rule restricts the connections to those from lower-numbered output $X$ nodes to higher-numbered output $Z$ nodes.

\begin{lemma}
    \label{lemma: number of ways to connect green output to red output}
    In the KL form, the number of ways to connect the $n-p$ output $Z$ nodes and $p$ output $X$ nodes is
    \begin{equation}
    \frac{\prod\limits_{i=1}^{p} (2^n - 2^{i -1})}{\prod\limits_{i =1}^p(2^p - 2^{i  -1})}.
    \end{equation}
\end{lemma}

\begin{proof}
In the proof of this lemma, we only consider the output nodes of the encoder.

Consider building up the ZX diagram as follows. Each time an output node is added, some number of possible connections could be made between the newly added output node and the nodes that have already been placed. The nodes are added in order, so each node has a higher index than all the nodes that came before it.

Let $p'$ be the number of remaining $X$ nodes that need to be added in the ZX diagram. Let $q'$ be the number of remaining $Z$ nodes that need to be added in the ZX diagram. The function $f(p,p',q')$ counts the number of ways to add nodes and edges starting from some arbitrary state that has $p'$ remaining $X$ nodes and $q'$ remaining $Z$ nodes. The total number of $X$ nodes after placing all nodes will be $p$. We find a recursive relation for $f(p,p',q')$.

When $p' = 0 $ and $q' \ne 0$, all of the remaining $Z$ nodes can be added in and connected arbitrarily to the $p$ pivots in $2^{pq'}$ ways, giving
$f(p,0,q') = 2^{pq'}$.

When $p' \ne 0 $ and $q' = 0$, all of the remaining $X$ nodes can be added in, but they cannot connect to anything since they necessarily are higher-numbered than all of the $Z$ nodes. Therefore,
$f(p,p',0) = 1$.

When $p' \ne 0$ and $q' \ne 0$, either a $X$ node is added or a $Z$ node is added (note there are $2^{p-p'}$ ways to connect edges between the new $Z$ node and $p-p'$ current $X$ nodes), giving the recursive equation
\begin{equation}
f(p,p',q') = f(p,p'-1, q') + 2^{p-p'}f(p,p',q'-1).
\end{equation}

We can check that the following function $f(p,p',q')$ satisfies this recursive equation and the base cases considered above:

\begin{equation}
    f(p,p',q') =  \frac{2^{(p-p')q'}\prod\limits_{i = 1}^{p'}(2^{q' + i} - 1)}{\prod\limits_{i = 1}^{p'}(2^i-1)}.
\end{equation}

The total number of ways to connect the $n-p$ $Z$ and $p$ $X$ nodes must be $f(p, p, n-p)$. Evaluating the above expression with these parameters gives a count of 
\begin{equation}
\frac{\prod\limits_{i = 1}^{p}(2^{n-p + i} - 1)}{\prod\limits_{i = 1}^{p}(2^i-1)} = \frac{\prod\limits_{i = 1}^{p}(2^{n} - 2^{i-1})}{\prod\limits_{i = 1}^{p}(2^p-2^{i-1})},
\end{equation}
as desired.

\end{proof}

Now, we count the number of ways to form connections between the input nodes and output $Z$ nodes.

\begin{lemma}
    \label{lemma: input to green output connections}
    In the KL form, the number of ways to connect the $k$ input nodes and $n-p$ output $Z$ nodes is
    \begin{equation}
    \frac{\prod\limits_{i=1}^q (2^{n-p} - 2^{i-1})}{\prod\limits_{i=1}^q (2^q - 2^{i-1})}.
    \end{equation}
\end{lemma}

\begin{proof}
    We have to count the number of ways to connect the input nodes and output $Z$ nodes to form an RREF adjacency matrix between them. Imagine each input node and its corresponding pivot node (in the RREF matrix) as a single \textit{super-node}. Super-nodes may not connect with each other, since pivots may not connect with other inputs, by the RREF rule. Also, the pivot node is the lowest-numbered node among the $Z$ outputs that the input node connects to, so the super-node can only connect with non-pivot nodes that are higher-numbered.

    Therefore, we have $k$ super-nodes and $q = n-p-k$ single output $Z$ nodes, and connections are restricted to those between super-nodes and higher-numbered single output $Z$ nodes. This is equivalent to $f(k, k, q)$ from the proof of Lemma~\ref{lemma: number of ways to connect green output to red output}, so the number of ways to connect the input nodes and output $Z$ nodes must be
    \begin{equation}
\frac{\prod\limits_{i = 1}^{k}(2^{q + i} - 1)}{\prod\limits_{i = 1}^{k}(2^i-1)} = \frac{\prod\limits_{i  =1}^q (2^{n-p} - 2^{i-1})}{\prod\limits_{i=1}^q (2^q - 2^{i-1})},
    \end{equation}
giving the desired result.

\end{proof}

From Lemma~\ref{lemma: number of ways to connect green output to red output} and Lemma~\ref{lemma: input to green output connections}, we find the total number of KL forms by multiplying the two expressions found, giving us a total count of
\begin{equation}
\frac{\prod\limits_{i=1}^{p} (2^n - 2^{i -1})}{\prod\limits_{i =1}^p(2^p - 2^{i  -1})}\cdot \frac{\prod\limits_{i  =1}^q (2^{n-p} - 2^{i-1})}{\prod\limits_{i=1}^q (2^q - 2^{i-1})}.
\end{equation}

Since this matches with Lemma~\ref{lemma: number of CSS codes}, the number of KL forms is indeed equal to the number of CSS codes under certain parameters, so each KL form can correspond to a distinct CSS code. Now, we need to show that restricting the CSS codes to $n$ physical qubits, $p$ $Z$ stabilizers, and $q$ $X$ stabilizers is equivalent to restricting the KL forms to $n$ output nodes, $n-p$ output $Z$ nodes, and $p$ output $X$ nodes.

In Equation~(\ref{eq:inverse}), we found the canonical stabilizer generators of codes by examining their KLS form.
In this section, we derive analogous results for CSS codes in KL form.
Specifically, starting from the KL form of a CSS code, we determine the $Z$ stabilizers, $X$ stabilizers, and logical $Z$ operators of the code.

\begin{lemma}
\label{lemma: red nodes have Z checks}
Each of the output $X$ nodes corresponds to its own linearly independent $Z$ check.
\end{lemma}

\begin{proof}
Consider an output $X$ node $v$. We identify a $Z$ stabilizer that includes the operator $Z_v$ by sliding a $Z$ gate through node $v$. By the $\pi$-copy rule from Definition~\ref{def:zx-basic-rewrite-rules},
this $Z$ gate splits into $Z$ gates on all the other incident edges of node $v$. Each of these $Z$ gates travels down an incident edge and combines with the output $Z$ node at the other end of the edge. Ultimately, this results in a phase of $\pi$ on all the neighbours of node $v$.

Since this is equivalent to placing $Z$'s on all the neighbours of node $v$, we have made a $Z$ stabilizer, which is the product of $Z_v$ and the $Z$ gates on the neighbouring output $Z$ nodes.

The $Z$ checks of the output $X$ nodes are linearly independent from each other because each $Z$ check contains a $Z$ operator on a distinct output $X$ node.
\end{proof}

\begin{lemma}
\label{lemma: X checks for non-pivots}
    Each of the non-pivot output $Z$ nodes corresponds to its own linearly independent $X$ check.
\end{lemma}

\begin{proof}
    Consider a non-pivot output $Z$ node $v$. We identify an $X$ stabilizer by sliding an $X$ gate through this node. By the $\pi$-copy rule, this $X$ gate splits into $X$ gates on all the other incident edges of node $v$. This results in a phase of $\pi$ on all the neighbours of node $v$. For each of the neighbours that are input nodes, an $X$ node of phase $\pi$ can be unmerged from the input $X$ node and slid towards the input's corresponding pivot node. The unmerging rule is given in Definition~\ref{def:zx-basic-rewrite-rules}. By the $\pi$-copy rule, all other incident edges, including the free output edge, of the pivot node get a $X$ $\pi$ node. Then, the output $X$ neighbours of the pivot node receive additional phases of $\pi$.

    Since this is equivalent to placing $X$'s on all the pivots of input nodes connected to $v$ and placing $X$'s on all the output $X$ nodes that end up with phase $\pi$, we have made an $X$ stabilizer, which is the product of the $X$ gates just described and $X_v$.

    The $X$ checks of the non-pivot output $Z$ nodes are linearly independent from each other because each $X$ check contains an $X$ operator on a distinct non-pivot output $Z$ node.
    
\end{proof}

The \textit{logical operator} of a code maps an element of the codespace onto another element of the codespace. We denote the logical $Z$ operators as $\overline{Z}$ and the logical $X$ operators as $\overline{X}$. The $Z$ and $X$ stabilizers can be used to determine all the logical operators of a CSS code. Analogously, the $Z$ stabilizers and logical $Z$ operators of a CSS code can be used to determine all the $X$ stabilizers.

\begin{lemma}
\label{lemma: inputs have logical ops}
The adjacencies of the inputs determine the logical $Z$ operators.
\end{lemma}

\begin{proof}
Consider an input node $v$. Similar to Lemma~\ref{lemma: red nodes have Z checks}, we identify a logical $Z$ operator by sliding a $Z$ gate through node $v$.

Since all of the input nodes are $X$ nodes, the $Z$ gate will split into $Z$ gates onto all the incident edges of node $v$. Ultimately, all the neighbours of node $v$ will have a phase of $\pi$ due to $\overline{Z}_v$. 

Since this is equivalent to placing $Z$'s on all these neighbouring nodes of node $v$, we have made a logical $Z$ operator, which is the product of the $Z$ gates on the neighbouring output $Z$ nodes.

Note that, in the adjacency matrix between the input nodes and output $Z$ nodes, each input node has a corresponding pivot. Therefore, the logical $Z$ operators determined by the input nodes will be linearly independent from each other because each has a $Z$ gate on a distinct element (i.e. the pivot node) in the set of output $Z$ nodes.
\end{proof}

We are now ready to prove our main result.

\begin{proof}[Proof of Theorem~\ref{CSS code canonical form}]

For a stabilizer code with $n$ physical qubits, the number of stabilizers and logical operators adds up to $n$. Considering the KL form of a CSS code with $p$ red output nodes and $k$ input nodes, we find, from Lemma~\ref{lemma: red nodes have Z checks}, Lemma~\ref{lemma: X checks for non-pivots}, and Lemma~\ref{lemma: inputs have logical ops}, that there are $p$ $Z$ stabilizers, $q$ $X$ stabilizers, and $k = n-p-q$ logical $Z$ operators. Since these three numbers add up to $n$, all stabilizers and logical $Z$ operators are accounted for by these three lemmas. This means that the KL forms with $p$ output $X$ nodes, $n-p$ output $Z$ nodes, and $k$ input nodes correspond exactly to the CSS codes with $n$ physical qubits, $p$ $Z$ stabilizers, and $q$ $X$ stabilizers.

From Lemma~\ref{lemma: red nodes have Z checks} and Lemma~\ref{lemma: X checks for non-pivots}, we also see there is a clear way to convert from a KL form into a representation of the CSS code entirely in terms of its $Z$ stabilizers and $X$ stabilizers.

Then, since KL forms can be converted into a stabilizer representation and the number of KL forms is equal to the number of CSS codes of analogous parameters, it follows that there is a bijection between CSS codes with $n$ physical qubits, $p$ $Z$ stabilizers, and $q$ $X$ stabilizers and KL forms with $n$ output nodes, of which $p$ are red output nodes and $n-p$ are green output nodes, finishing the proof of Theorem~\ref{CSS code canonical form}.
\end{proof}

This construction allows us to prove several propositions.
First, we complete correspondence between the KL form and the stabilizer tableau.
Lemma~\ref{lemma: red nodes have Z checks}, Lemma~\ref{lemma: X checks for non-pivots}, and Lemma~\ref{lemma: inputs have logical ops} explain how to construct both
types of stabilizers and $Z$ logical operators.
Although this is enough to determine the $X$ logical
operators, we can also use the following.

\begin{proposition}\label{prop: logical x ops}
   The adjacencies of the pivots determine the logical $X$ operators.
\end{proposition}
\begin{proof}
Consider an input node $v$. We identify a logical $X$ operator by sliding a $X$ gate through node $v$.
Since all inputs nodes are $X$ nodes and the $X$ gate is an $X$ node with a phase of $\pi$, we merge and unmerge it with the input node to move it along the edge connecting the input to its pivot, an output $Z$ node.
Then, we use the $\pi$-copy rule to turn the $X$ gate
into $X$ gates on each other edge connected to the pivot.
Note that this means that an $X$ gate will be placed on the pivot's free edge as well as each other internal edge connected to the pivot.
Since the pivot node cannot be connected to any inputs other than $v$ and the graph is bipartite, this means that these $X$ gates can be moved towards the output $X$ nodes and then merged and unmerged along each of those outputs' free edges.

Note that, in the adjacency matrix between the
input nodes and output $Z$ nodes, each input node
has a corresponding pivot. Therefore, the logical $X$
operators determined by the input nodes will be linearly independent from each other because each has
an $X$ gate on a distinct element, the pivot node.
\end{proof}

\begin{remark}
    The stabilizers of a CSS code need to commute with each other and with each logical operation, while the logical $Z$ and $X$ operations on a single input should anti-commute.
    From Lemma~\ref{lemma: inputs have logical ops} and Lemma~\ref{prop: logical x ops} we see that logical $Z$ and $X$ operations can only overlap on pivot nodes, and since each only has a local operation on the pivot node corresponding to the input of the logical operator, the only anti-commuting operations will be logical $Z$ and $X$ operations on the same input.

    For stabilizer commutation, we see from Lemma~\ref{lemma: red nodes have Z checks} and Lemma~\ref{lemma: X checks for non-pivots} that a $Z$ and $X$ stabilizer can only intersect at an $X$ output node or at its $Z$ output neighbours, as those are the elements of a $Z$ stabilizer.
    Each of the $Z$ output nodes present in the $Z$ stabilizer will also place a term on the neighbouring $X$ output node.
    This will make the parity of total overlapping nodes even, making sure the stabilizers commute.

    The same argument as above shows why $Z$ stabilizers commute with logical $X$ operations. Since $X$ stabilizers only have gates on $Z$ output nodes at the corresponding non-pivot node and on several pivot nodes, a logical $Z$ operation will either not intersect it or intersect the stabilizer exactly twice, at a pivot and a non-pivot connected to the input of the logical $Z$.
    Either way, this proves that the stabilizers commute with the logical operations and each other.
\end{remark}

We may consider CSS codes in terms of only its $Z$ stabilizers and logical $Z$ operators, which is equivalent to the usual representation of the codes in terms of their $Z$ stabilizers and $X$ stabilizers. From Lemma~\ref{lemma: red nodes have Z checks} and Lemma~\ref{lemma: inputs have logical ops}, the $Z$ stabilizers and logical $Z$ operators can be found directly from the KL form by examining the connections to output $Z$ nodes and input nodes, respectively. 

The above construction allows us to easily transform a KL diagram into a stabilizer tableau.
The reverse is also easy to accomplish by applying the following procedure.
If we start with a stabilizer code, we first compute its $Z$ stabilizers and logical $Z$ operations.
We then row-reduce the list of $Z$ stabilizers, identifying the pivot columns.
We create $n$ output nodes, with an $X$ node for each pivot of the row-reduced $Z$ stabilizers and a $Z$ node for each non-pivot.
Note that these are not the same as the pivot and non-pivot nodes in the context of the adjacencies from the input nodes.
Each output node is connected to a free edge and the connections between the $Z$ and $X$ output nodes are made in accordance with the operations in the row-reduced list of $Z$ stabilizers.
The list of logical $Z$ operations is then transformed by the stabilizer operations until none of the logical operations have any terms on nodes which are pivots of the stabilizers.
The remaining logical operations will only have terms on the $Z$ output nodes, at which point they can be row-reduced and connected to input $X$ nodes, using the row-reduced matrix as an adjacency matrix.
This completes the construction of a KL form from a stabilizer tableau.

In some cases, it is possible to further reduce the number of nodes in a diagram in its canonical KLS form. Some of these simplifications are also relevant to the KL form.

\begin{definition}
    The \textit{reduced KLS form} can be found by completing the following steps:
    \begin{enumerate}[(1)]
        \item We first remove any output nodes of phase zero and degree 2 using the identity removal rule.
        \item If an output node now has more than one free output edge, is only connected to an input node, and the input node is only connected to this output node, the Hadamards on the output node's free edges can be manipulated so that the topmost free edge of the node has no Hadamard.
        \item Any input nodes of phase zero and degree 2 may be removed.
    \end{enumerate}
\end{definition}

The second step to the reduced KLS form can be done by using the Hadamard-sliding rule from Definition~\ref{def:zx-derived-rewrite-rules} on the output node and a connected input node. Unitary operations that appear on the input node are removed.
The output node's free edges all get an additional Hadamard gate due to the rewrite rule, causing the toggling of Hadamards on all the free edges. Out of the two possible Hadamard configurations, one configuration results in the topmost free edge having no Hadamard gate.

These reduced forms allow for further, albeit limited, reduction of the number of nodes in the KLS form of a Clifford code.

\section{Toric and surface codes}
\label{section:surfacetoric}

As an example of the CSS codes explored in the previous section, we now consider toric codes, a specific class of surface codes with periodic boundary conditions, as introduced in~\cite{kitaev2003fault} and certain surface codes, in the form studied by~\textcite{kissinger2022phase}. We first construct a symmetrical form for the toric code starting from Kissinger's \textit{ZX normal form}, a representation of CSS codes in the ZX-calculus that clearly shows stabilizers using internal nodes representing quantum measurements and shows logical operators using connections from input nodes to output nodes.
The ZX normal form generally has fewer edges and long-range connections at the cost of having more nodes.
A consequence of Theorem~\ref{thm:canonicity} and its proof in Appendix~\ref{app:compiler} is that the KL canonical form can be efficiently derived starting from any equivalent diagram, including the ZX normal form.
We also provide the KL forms for some toric and surface codes.

\begin{figure}
\begin{center}
\includegraphics[scale=0.33]{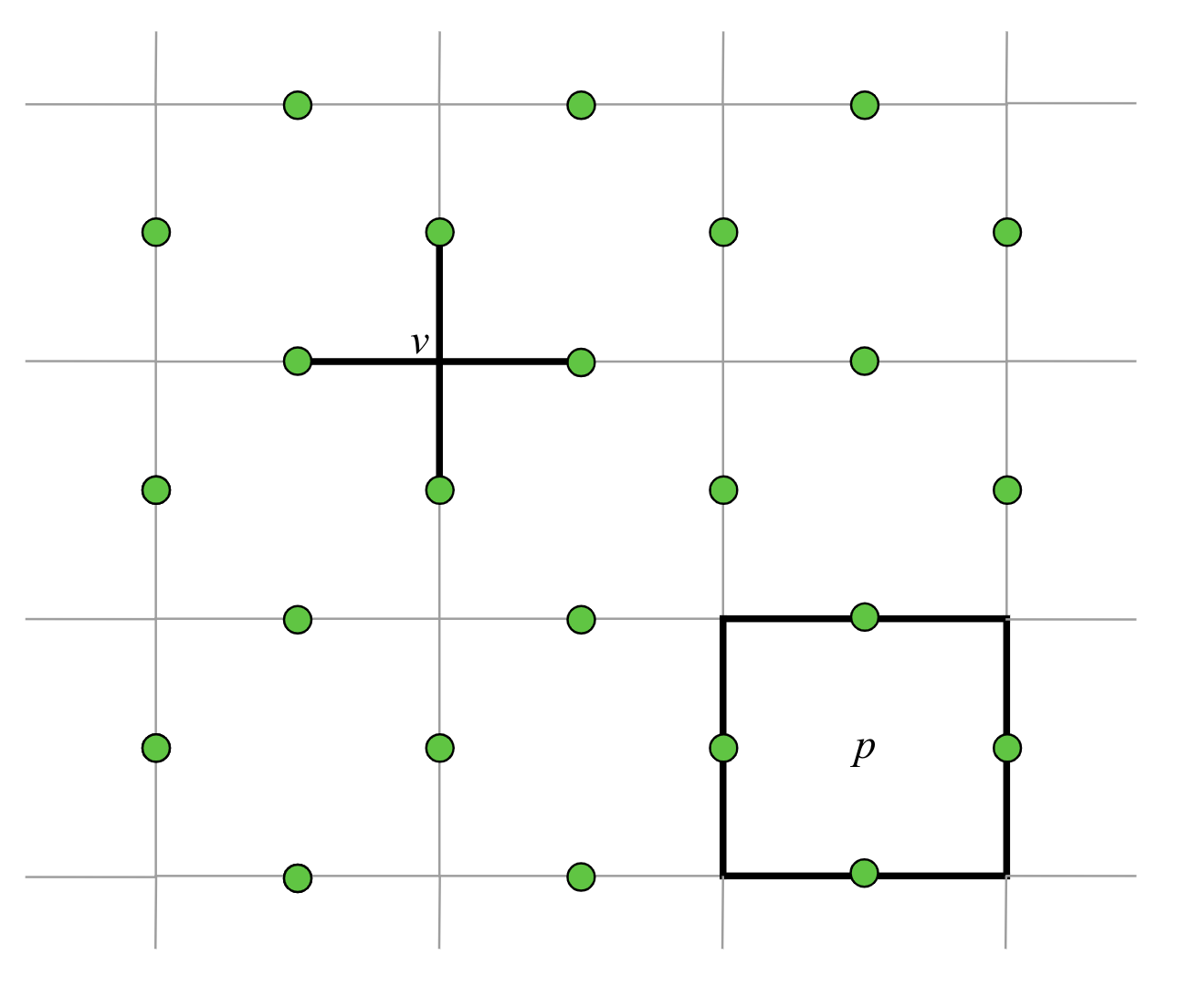}
\end{center}
    \caption[Toric code stabilizers]{A section of the torus. The stabilizers associated with the intersections of the grid such as $v$ have $X$ gates on the nodes immediately surrounding the intersection. The stabilizers associated with the plaquettes such as $p$ have $Z$ gates on the nodes immediately surrounding the plaquette. All nodes are shown in green, but this will change once we add in edges.}
    \label{generic toric code diagram}
\end{figure}

\begin{figure}
    \centering
    \includegraphics[scale=0.6]{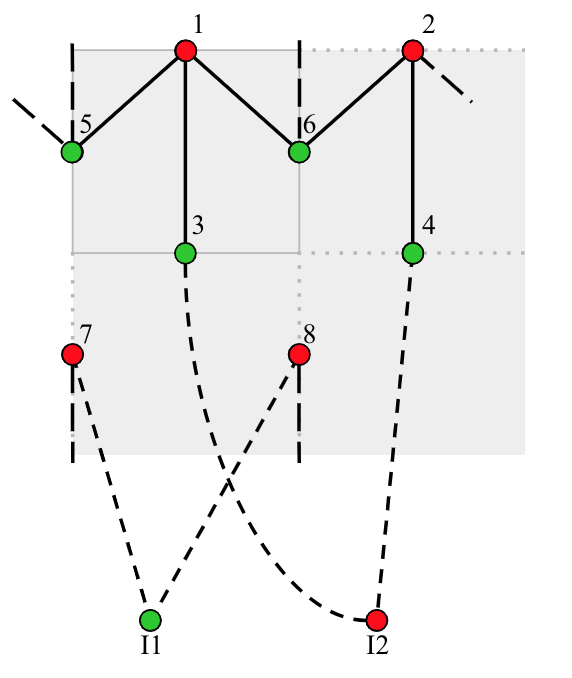}
    \caption[ZX diagram of 2 $\times$ 2 toric code]{A ZX-calculus diagram of the 2 $\times$ 2 toric code. Edges with long dashes represent edges that wrap around the torus. For example, the vertical dashed edge coming from node 5 meets node 7 and the dashed edge from node 2 meets node 5. Edges with short dashes are input-output edges. The free input and output edges are not shown.
    Note that this diagram is not in KL form, since there is a green input node.
    It could be transformed into KL form using ZX rewrite rules.}
    \label{final 2 by 2}
\end{figure}

\begin{figure}
    \centering
    \includegraphics[scale=0.55]{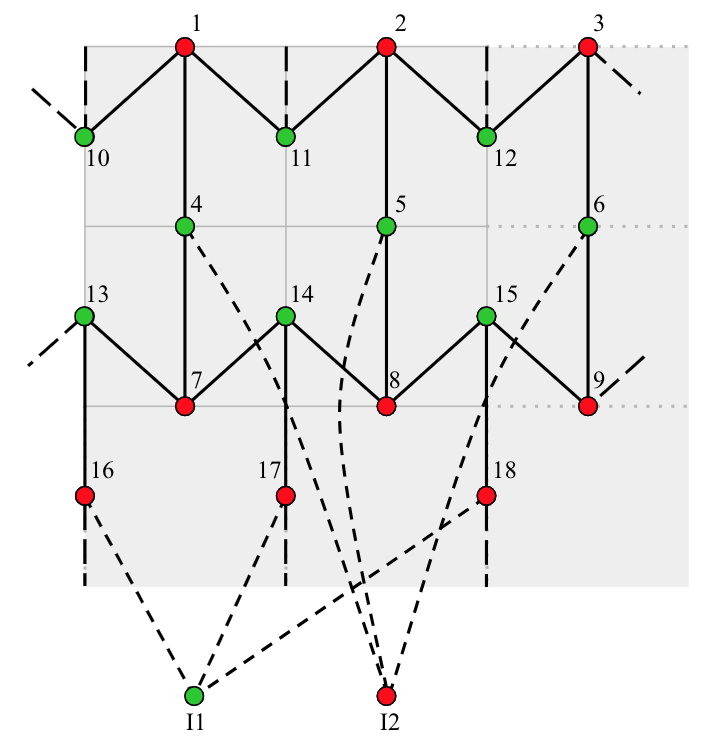}
    \caption[ZX diagram of 3 $\times$ 3 toric code]{A ZX-calculus diagram of the 3 $\times$ 3 toric code. Edges with long dashes represent edges that wrap around the torus. For example, the vertical dashed edge coming from node 10 meets node 16 and the dashed edge from node 3 meets node 10. Edges with short dashes are input-output edges. The free input and output edges are not shown.
    Note that this diagram is not in KL form, since there is a green input node.
    It could be transformed into KL form using ZX rewrite rules.}
    \label{final 3 by 3}
\end{figure}

\begin{figure}
\begin{center}
\includegraphics[scale=0.33]{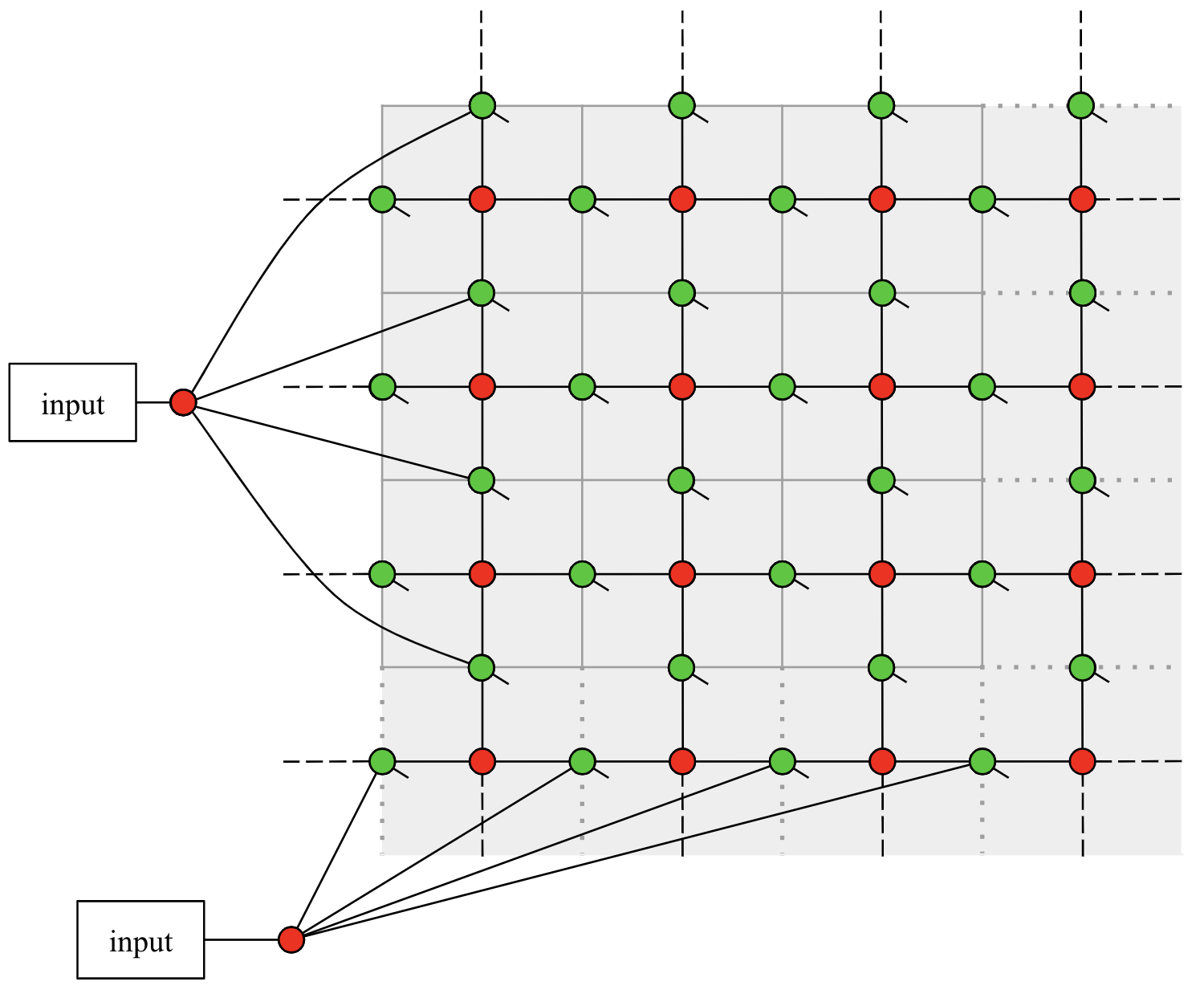}
\hfill
\includegraphics[scale=0.36]{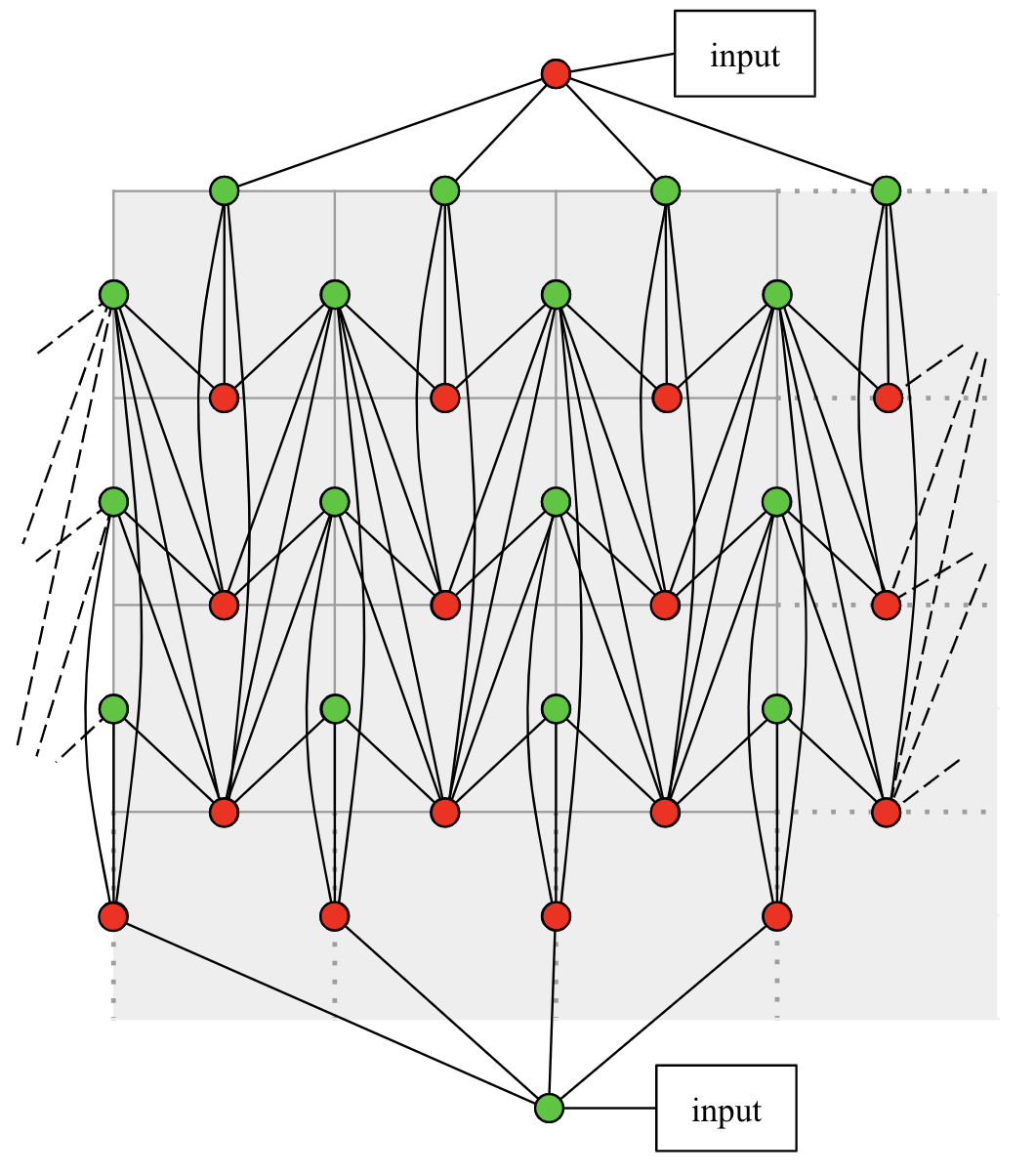}
\end{center}
\caption[ZX normal form and diagram of 4 $\times$ 4 toric code]{The left diagram is the ZX normal form of the 4 $\times$ 4 toric code. The output nodes, which are all the $Z$ nodes, are shown with their free edges.
The $X$ nodes in the toric grid are internal nodes.
This representation of the code has fewer edges and local connections, at the cost of having internal nodes.
The right diagram is a symmetrical form of the 4 $\times$ 4 toric code. There is a symmetry between the $X$ and $Z$ nodes.
Though this is not in KL form, it is easily converted into the KL form.
This presentation has a minimal set of nodes, one for each input and output free edge, although the edges are not local.}
\label{4-by-4}
\end{figure}

We begin with a definition of toric codes.

\begin{definition}
\label{def:toric-code}
    A \textit{toric code} is a quantum error-correcting code that can be represented on a torus $\mathcal{T}$. For an $m \times n$ toric code, $\mathcal{T}$ is wrapped by $m$ circles along one axis and $n$ circles along another.
    These circles form an $m\times n$ grid. 

    A node or equivalently, a qubit, is placed at the midpoint of each of the $2mn$ edges on $\mathcal{T}$. The stabilizers are defined as follows.
    
    The four nodes surrounding each of the $mn$ four-sided faces form a $Z$ check (stabilizer with only $Z$'s and $I$'s) consisting of $Z$'s on these four nodes and $I$'s on all other nodes.

    The four nodes surrounding each of the $mn$ intersections form an $X$ check (stabilizer with only $X$'s and $I$'s) consisting of $X$'s on these four nodes and $I$'s on all other nodes.

    An illustration of these stabilizers are shown in Figure~\ref{generic toric code diagram}.
\label{def:thing}
\end{definition}

We now determine the structure of the symmetrical forms of the 2 $\times$ 2 and 3 $\times$ 3 toric code's ZX diagram. We do this by using the stabilizers to deduce the output edges with Hadamards and internal edges between output nodes.

We present the ZX-calculus form of the 2 $\times$ 2 toric code in Figure~\ref{final 2 by 2}. This has an arrow-like structure among the output-output edges, as seen by the group of nodes 1, 3, 5, and 6, as well as nodes 2, 4, 5, and 6. 

Also, the resulting 3 $\times$ 3 toric code is shown in Figure~\ref{final 3 by 3}, with its full derivation given in Appendix~\ref{app:3-by-3-toric}. Note that by wrapping this pattern around a torus, we see that it is horizontally periodic. The edges among vertices 1, 4, 10, and 11 form an upward-arrow-like structure. Similarly, nodes 2, 5, 11, and 12 form this structure and, on a torus, nodes 3, 6, 10, and 12 do so as well. By the simplicity of this diagram, it is relatively easy to read off the stabilizers by placing gates on the free output edges and seeing how the $\pi$-copy rule affects the diagram.

\begin{figure}
    \centering
    \includegraphics[scale=0.3]{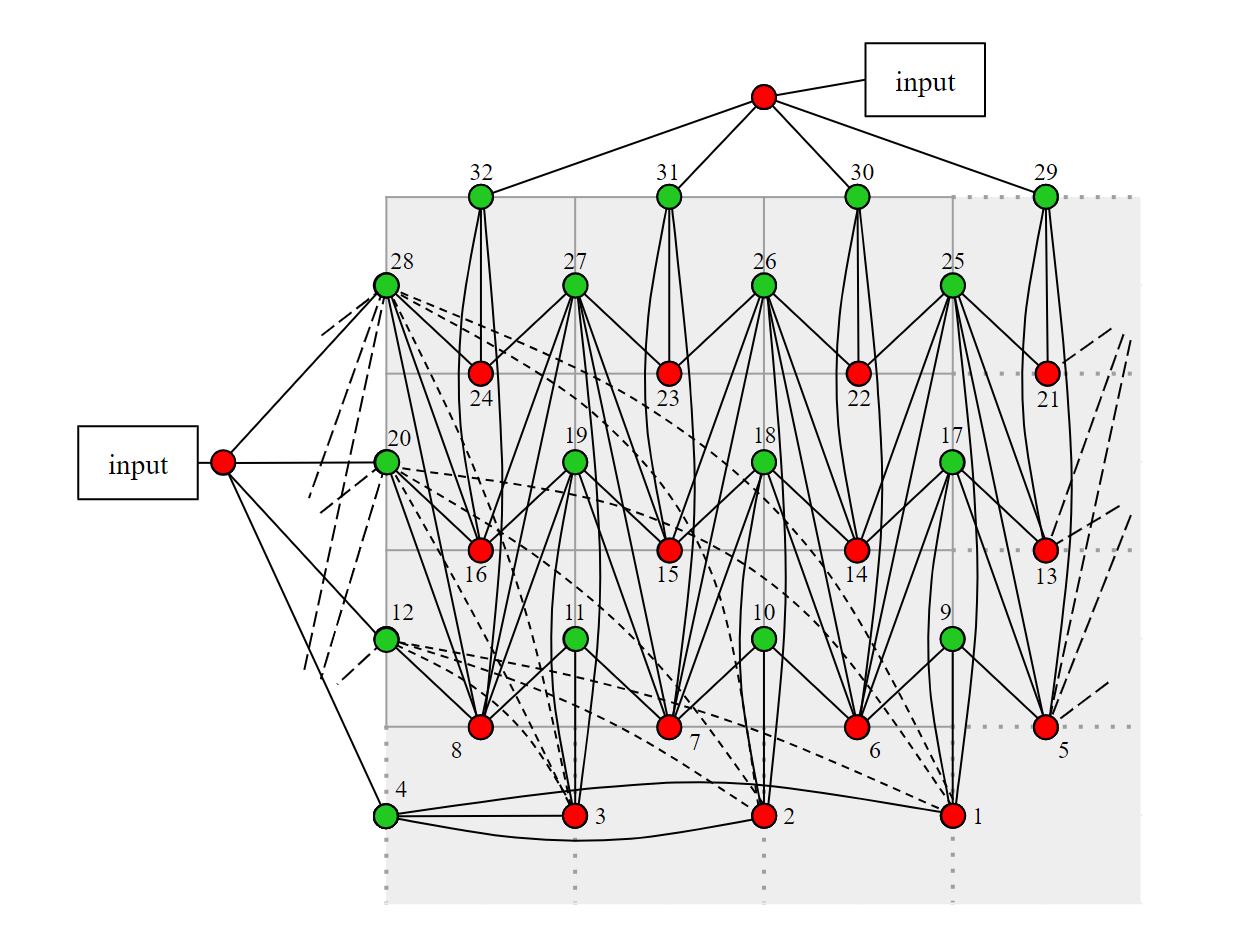}

    \caption[KL form of 4 $\times$ 4 toric code]{The 4 $\times$ 4 toric code in KL form. Long-dashed edges on the left and right edges of the grid wrap around the torus. The curved edges between nodes 1, 2, 3 and nodes 12, 20, 28 are dashed for clarity. The free output edges are not shown.}
    \label{fig: 4-by-4 toric code KL form}
\end{figure}

For larger toric codes and the surface codes, we use the ZX-calculus software \textit{Quantomatic}~\cite{quantomatic} to simplify the known ZX normal form~\cite{kissinger2022phase} of a code into its canonical form. In our algorithm, the main focus is on performing the bialgebra rule on internal nodes, so that after running the first part of the algorithm all internal nodes will be removed from the diagram, leaving only input and output nodes. Then, the Hadamard sliding rule from Definition~\ref{def:zx-derived-rewrite-rules} provides the operation that can repeatedly move Hadamards until the encoder diagram is symmetric.

The procedure we follow may be written as the following algorithm.
\begin{enumerate}[(1)]
\item Use basic simplifications, by merging nodes, applying the state-copy rule, applying the Hopf rule, removing scalars, removing loops, or combining two Hadamards into the identity from Definition~\ref{def:zx-basic-rewrite-rules} and Definition~\ref{def:zx-derived-rewrite-rules}.
\item Apply one iteration of the bialgebra rule from Definition~\ref{def:zx-basic-rewrite-rules} to remove an internal node. If there are no more internal nodes, proceed to step (3). Otherwise, return to step (1).
\item Repeatedly apply the Hadamard-sliding rule from Definition~\ref{def:zx-derived-rewrite-rules} until reaching the desired colours of nodes. It is not the case than any colour configuration can be reached. For example, in the toric code, it turns out that it is impossible for the colours to alternate every row, as one row on either end has the opposite colour from the rows it is connected to.
\end{enumerate}

The reason why we can perform step (2) is due to the fact that after applying the bialgebra rule from Definition~\ref{def:zx-basic-rewrite-rules} to an internal node, the additional nodes that appear can be merged with other nodes in the graph.
If neither of the two nodes the bialgebra rule is applied to is internal, this will not decrease the number of nodes.
Thus, while at least one internal node remains in the graph, the number of nodes can be reduced this way.

We now use the above algorithm to derive the general ZX diagram for the $m\times n$ toric code.
We begin by presenting our symmetrical form of the 4 $\times$ 4 toric code, which was derived from the ZX normal form. These are both shown in Figure~\ref{4-by-4}.
As can be seen in the diagrams, the number of output nodes is reduced by a factor of 2, and the diagram in Figure~\ref{4-by-4}(b) retains a high degree of symmetry. When moving horizontally, it can be seen that there is are periodic patterns of nodes, with one column having 3 $Z$ nodes and 1 $X$ node and the next having 3 $X$ nodes and 1 $Z$ node. Also, the edges between the columns of nodes are local in one direction, as their length does not scale with the horizontal dimension of the toric code.

Using the algorithm on larger dimension $m\times n$ toric codes shows that they have the same general structure as that of Figure~\ref{4-by-4}(b). To construct it geometrically, first place an input $X$ node at the top of the diagram and an input $Z$ node at the bottom of the diagram. Then, on the unfolded toric grid, the first and second rows (out of $2n$ rows) of $m$ nodes each are designated as $Z$ nodes. Then, the rows alternate as rows of $X$ and $Z$ nodes until the very bottom two rows of the output nodes, which are rows of $X$ nodes. This is reflected in the 4 $\times$ 4 example.

We now need to identify the diagram's edges.
We start by drawing an edge from each of the top row's $Z$ nodes to each of the $X$ nodes within its column.
Furthermore, the second row's $Z$ nodes have edges connecting them to each of the $X$ nodes in the neighbouring columns, as well as one edge connecting it to the bottommost $X$ node in the same column. The next row of $Z$ nodes (the fourth row of $m$ nodes from the top) have edges to all the $X$ nodes in the neighbouring columns that are in rows strictly below it, as well as one edge to the bottommost $X$ node in the same column. This pattern follows for the other rows of $Z$ nodes.

Since there is a symmetry between the $Z$ and $X$ nodes, we see the same arrow-shaped patterns extending upwards from rows of $X$ nodes.

We convert this symmetrical form of the toric code into its KL form by changing the input $Z$ node into an input $X$ node and keeping the stabilizers the same. This gives Figure~\ref{fig: 4-by-4 toric code KL form} as the KL form for the toric code, with the numbering reflecting the Hadamard rule.

We also extend our results to the rotated surface codes, shown in Equation~(12) from \textcite{kissinger2022phase}. This is a surface code defined such that each face shaded in green represents a $Z$ check while each face shaded in red represents an $X$ check. A $Z$ check consists of $Z$ gates on each of the nodes on the perimeter of the green face, and $X$ checks are similarly defined.

\begin{figure}
    \centering
    \includegraphics[scale=0.4]{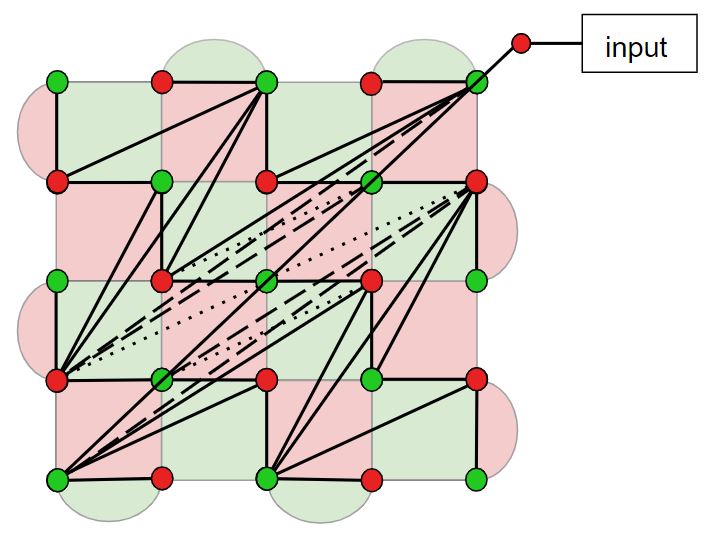}
    \caption[KL form of 5 $\times$ 5 surface code]{A 5 $\times$ 5 surface code in its ZX canonical form. Note its resemblance with the toric code diagram when viewed at a $45^\circ$ angle. The red input node connects to all 5 $Z$ nodes in the main bottom-left to top-right diagonal. Some edges are dotted or dashed for clarity. The free output edges are not shown.}
    \label{surface 5-by-5}
\end{figure}

In Figure~\ref{surface 5-by-5}, we show the result of simplifying the 5 $\times$ 5 surface code. If it is rotated by $45^\circ$ counterclockwise, it closely resembles the patterns of edges and colours seen in the general toric code. The neighbouring diagonals (from bottom-left to top-right) of nodes of different colours connect in arrow-shaped patterns, just as in the toric code.

In general, the $(2k+1)\times(2k+1)$ surface code can be made to have a similar structure as shown in Figure~\ref{surface 5-by-5}.

\section{Prime codes}
\label{sec: prime codes}

When expressed in the ZX-calculus, Clifford codes can have multiple connected components. The tensor product of operations of the connected components is the operation of the entire ZX diagram. For example, any connected component with no free edges is a scalar which scales the tensor of the rest of the diagram. In this section, we consider codes with respect to their connected components.

We note the following statement about ZX diagrams with multiple connected components. 

\begin{proposition}
If a ZX diagram has multiple connected components that share no connections between each other, then these components are not entangled. Equivalently, if the components are entangled, then they share some connections.
\end{proposition}

We introduce the notion of \textit{prime} codes, defined as follows.

\begin{definition}
    A \textit{prime code diagram} is a ZX diagram in KLS form that cannot be expressed as a disconnected graph after a sequence of rewrite rules. 
\end{definition}

We show the following result about prime codes, which we the name the Fundamental Theorem of Clifford Codes (FTCC), analogously to the Fundamental Theorem of Arithmetic.

\begin{theorem}[Fundamental Theorem of Clifford Codes]

    Consider a Clifford code with an encoder-respecting form satisfying the constraint that the input-output adjacency matrix is full rank. Then, there exists a unique decomposition of the code into a product of prime codes (up to a permutation of input nodes).
\end{theorem}

We prove the FTCC later in this section. To begin, we first prove the following lemma considering the properties of connected components of KLS forms.

\begin{lemma}
\label{lemma: KLS components are all prime}
    The connected components of a KLS form are all prime.
\end{lemma}

\begin{proof}
    Consider a connected component $\mathcal{D}$ in the KLS form of the code $\mathcal{C}$. Since the adjacency matrix between all inputs and outputs of $\mathcal{C}$ is full-rank and in RREF, the adjacency matrix $N$ between the inputs and outputs of $\mathcal{D}$ has pivot nodes for each of its input nodes. Also, $N$ is full rank, so $N$ must be in RREF. Furthermore, note that any row operations on $N$ preserve the input-pivot connections.

    Since the component $\mathcal{D}$ is in RREF, satisfying the RREF rule, and since it is a subgraph of a KLS diagram, it satisfies the Edge, Hadamard, and Clifford rules of Section~\ref{subsec:zx-construction}.
    Thus, $\mathcal{D}$ must also be in KLS form.

    The ZX-calculus is complete for stabilizer quantum mechanics~\cite{backens2014zx} so any two equivalent diagrams can be made equal through a sequence of basic rewrite rules from Definition~\ref{def:zx-basic-rewrite-rules}. The basic rewrite rules of the ZX-calculus that can affect whether two nodes are in the same connected component
    are the state copy and Hopf rules from Definition~\ref{def:zx-basic-rewrite-rules}. The state copy rule requires an internal node (a node without free edges) connected to exactly one node of the opposite colour. None of the other rules are able to produce such an internal node, so the state copy rule cannot be applied. For example, the bialgebra rule can never cause an internal node to have only one edge. The Hopf rule requires two nodes sharing two edges. Since none of the other basic rewrite rules cause two nodes to share two edges, the Hopf rule cannot be applied either. 

    We now show that no sequence of row operations on the adjacency matrix $N$ can turn the initial connected graph of $\mathcal{D}$ into disconnected components. For the sake of contradiction, suppose $\mathcal{D}$ is split into multiple connected components due to row operations. Then, row operations can turn each connected component's input-output adjacency matrix into RREF, so the adjacency matrix for $\mathcal{D}$ is disconnected and (after appropriate input permutations) in RREF. Since $N$ was initially connected and in RREF, and the RREF is unique, we reach a contradiction. Thus, no sequence of row operations on $N$ can turn $\mathcal{D}$ into disconnected components.

    Therefore, all the nodes in $\mathcal{D}$ are always in the same connected component. Since $\mathcal{D}$ is in KLS form and it cannot become disconnected, $\mathcal{D}$ is prime.

    Since $\mathcal{D}$ was chosen arbitrarily, this implies that the connected components of a KLS form are all prime.
\end{proof}

The encoders considered below will always be in encoder-respecting form. We now go through the process of converting an arbitrary encoder into its KLS form to show that the end result is the same had we converted two disconnected components into KLS, possibly having to rearrange the input nodes.
As defined in Appendix~\ref{app:compiler}, the \textit{ZX-HK form} is an intermediate form of a code that satisfies the Edge and Hadamard rules for KLS forms, but not necessarily the RREF or Clifford rules.

\begin{lemma}
\label{lemma: k(A) otimes K(B) sim K(C)}
    If $\mathcal{X}$ is an arbitrary Clifford encoder in the ZX-calculus, let $\text{KLS}(\mathcal{X})$ denote its KLS form in the ZX-calculus. If an encoder $\mathcal{C}$ can be decomposed into components $\mathcal{A}$ and $\mathcal{B}$, where there are no edges between the components, then $\text{KLS}(\mathcal{C})$ can be decomposed into components $\text{KLS}(\mathcal{A})$ and $\text{KLS}(\mathcal{B})$, again with no edges between the components.
\end{lemma}

\begin{proof}
    As we saw in Appendix~\ref{app:compiler}, the conversion from an arbitrary encoder-respecting form to the KLS form involves only irrelevant unitaries on inputs, local complementation, and row operations on the input-output adjacency matrix.

    Transforming a code into its ZX-HK form relies only on local complementations in Equation~(\ref{eq:simplification-sh}), Equation~(\ref{eq:simplification-hs}), and Equation~(\ref{eq:edge-local-complementation-hsliding}).
    Since local complementation only toggles edges of the neighbours of a node, it cannot affect whether two nodes are in the same connected component, since any vertices it might disconnect will always neighbour the node about which the local complementation is applied. Thus, if the ZX-HK form of the encoders are $\mathcal{A'},\mathcal{B'},$ and $\mathcal{C'}$, respectively, we find that $\mathcal{C'}$ can be decomposed into components $\mathcal{A'}$ and $\mathcal{B'}$, with no edges between the components, since local complementation happens entirely within the components.

    The next step is to transform the ZX-HK form's adjacency matrix between the input and output nodes into RREF. When transforming $\mathcal{C'}$ into RREF, note that we can first transform the separate components $\mathcal{A'}$ and $\mathcal{B'}$ into RREF. Each input node of $\mathcal{C'}$ has a corresponding pivot node. 
    Then there exists a permutation of the inputs that will result in $\mathcal{C'}$ being in RREF if we order the inputs to match the order of the pivots. Permuting can be done by row operations. Since the RREF of the adjacency matrix is unique, this is the desired RREF for $\mathcal{C'}$. Note that the resulting diagram $\mathcal{C''}$ can still be decomposed into components $\mathcal{A'}$ and $\mathcal{B'}$, with no edges between the components. This is because the order of the inputs within each component is unchanged.

    The next step in converting into KLS form involves local complementing at the inputs of pivot nodes.
    Then, if there are edges between pivot nodes in the resulting form, we perform local edge complementations on the corresponding input nodes, followed by a permutation of the same input nodes. The former does not affect whether two nodes are in the same connected component. The latter also does not alter this property because local complementation does change graph connectivity and the input-output adjacency matrix is preserved in RREF after the permutation. This completes the transformation into the KLS form.

    Thus, we have shown that $\text{KLS}(\mathcal{C})$ can be decomposed into components $\text{KLS}(\mathcal{A})$ and $\text{KLS}(\mathcal{B})$, which share no edges.
\end{proof}

Note that the output nodes of $\mathcal{A}$ and $\mathcal{B}$ remain the same after transforming them into KLS form. Thus, ordering the inputs of $\text{KLS}(\mathcal{A})$ and $\text{KLS}(\mathcal{B})$ so that they match with the ordering of the pivot output nodes, the resulting diagram is the same as $\text{KLS}(\mathcal{C})$. 

Next we show why equivalent Clifford encoders share the same set of prime components.

\begin{theorem} 
Given two Clifford encoders with the same codespace, each entirely made up of prime components, they have equivalent ZX diagrams up to a permutation of the inputs.

\label{thm: two equivalent encoders entirely made of primes}
\end{theorem}

\begin{proof}

First, we note that two encoders with the same codespace must have the same KLS form. Therefore, the statement is equivalent to showing that, given a Clifford encoder $\mathcal{C}$ entirely made up of prime components, its ZX diagram is equivalent to the KLS form up to a permutation of inputs.

We now induct on the number of prime components in $\mathcal{C}$. If $\mathcal{C}$ has exactly one prime component, then its entire diagram must be in KLS form. Since the KLS form for $\mathcal{C}$ is unique, this means $\mathcal{C}$ and its KLS form are identical, and thus equivalent up to a permutation of inputs.

Now, suppose all encoders with $i \leq n$ prime components are equivalent to their KLS form up to a permutation of inputs. We want to show that this implies any encoder $\mathcal{C}$ with $n+1$ prime components is equivalent to its KLS form up to a permutation of inputs.

Denote the connected component of $\mathcal{C}$ that has output node 1 as $\mathcal{A}$. Then, the remaining part of $\mathcal{C}$ that shares no edges with $\mathcal{A}$ can be denoted as $\mathcal{B}$. Note that $\mathcal{A}$ is already in KLS form, since $\mathcal{A}$ must be prime. By the inductive hypothesis, we have that $\mathcal{B}$ is equivalent to $\text{KLS}(\mathcal{B})$ up to input permutations, so, after an appropriate sequence of rewrite rules, $\mathcal{C}$ can be decomposed into $\mathcal{A}$ and $\text{KLS}(\mathcal{B})$. 

 By Lemma~\ref{lemma: k(A) otimes K(B) sim K(C)}, we have that $\text{KLS}(\mathcal{C})$ can be decomposed into $\text{KLS}(\mathcal{A})$, which is just $\mathcal{A}$, and $\text{KLS}(\mathcal{B})$. This proves that $\mathcal{C}$ is equivalent to $\text{KLS}(\mathcal{C})$ up to input permutations.

\end{proof}

Now, we prove the Fundamental Theorem of Clifford Codes.

\begin{proof}[Proof of Theorem 5.3 (FTCC)]
    We begin with a Clifford code with an encoder-respecting form satisfying the constraint that the input-output adjacency matrix is full rank. Then, after converting it into its KLS form, Lemma~\ref{lemma: KLS components are all prime} gives that all its connected components are prime. Therefore, we have constructed a decomposition of the code into primes.

    Now, we show why this construction is unique. For the sake of contradiction, suppose there is another different decomposition of the code into primes. Then, by Theorem~\ref{thm: two equivalent encoders entirely made of primes}, these two ZX diagrams are equivalent up to a permutation of inputs. Note that we can rearrange the inputs at will using row operations without changing the structure of any of the primes, so this second decomposition must be the same as the one constructed above.

    Thus, the decomposition exists and it is unique.
\end{proof}

\section{Another definition of equivalence}
\label{sec 4}

Previous works have examined the equivalence classes of graphs under local complementation~\cite{adcock2020mapping, bahramgiri2007enumerating, bouchet1993recognizing}. In Clifford codes, the presence of designated input and output vertices makes the definition of equivalence more exotic.

In the following sections, we consider only the ZX diagrams for Clifford codes that have no local operations on the free output edges.
Such codes are exactly represented as graphs as discussed in Chapter~\ref{chapter:qcodes-graphs}.

\begin{figure}
\begin{center}
\begin{subfigure}[c]{0.3\textwidth}
\centering{
Transforming the\\adjacency matrix:
\begin{equation*}\begin{pmatrix}
    0 & 0 & 1 & 0 & 1 & 0 & 1\\
    0 & 0 & 0 & 1 & 0 & 0 & 0\\
    1 & 0 & 0 & 0 & 0 & 0 & 0\\
    0 & 1 & 0 & 0 & 0 & 0 & 1\\
    1 & 0 & 0 & 0 & 0 & 0 & 0\\
    0 & 0 & 0 & 0 & 0 & 0 & 0\\
    1 & 0 & 0 & 1 & 0 & 0 & 0\\    
\end{pmatrix}\end{equation*}
}
\end{subfigure}\hspace{-5ex}
\begin{subfigure}[c]{0.5\textwidth}
  \centering{\includegraphics[scale=0.30]{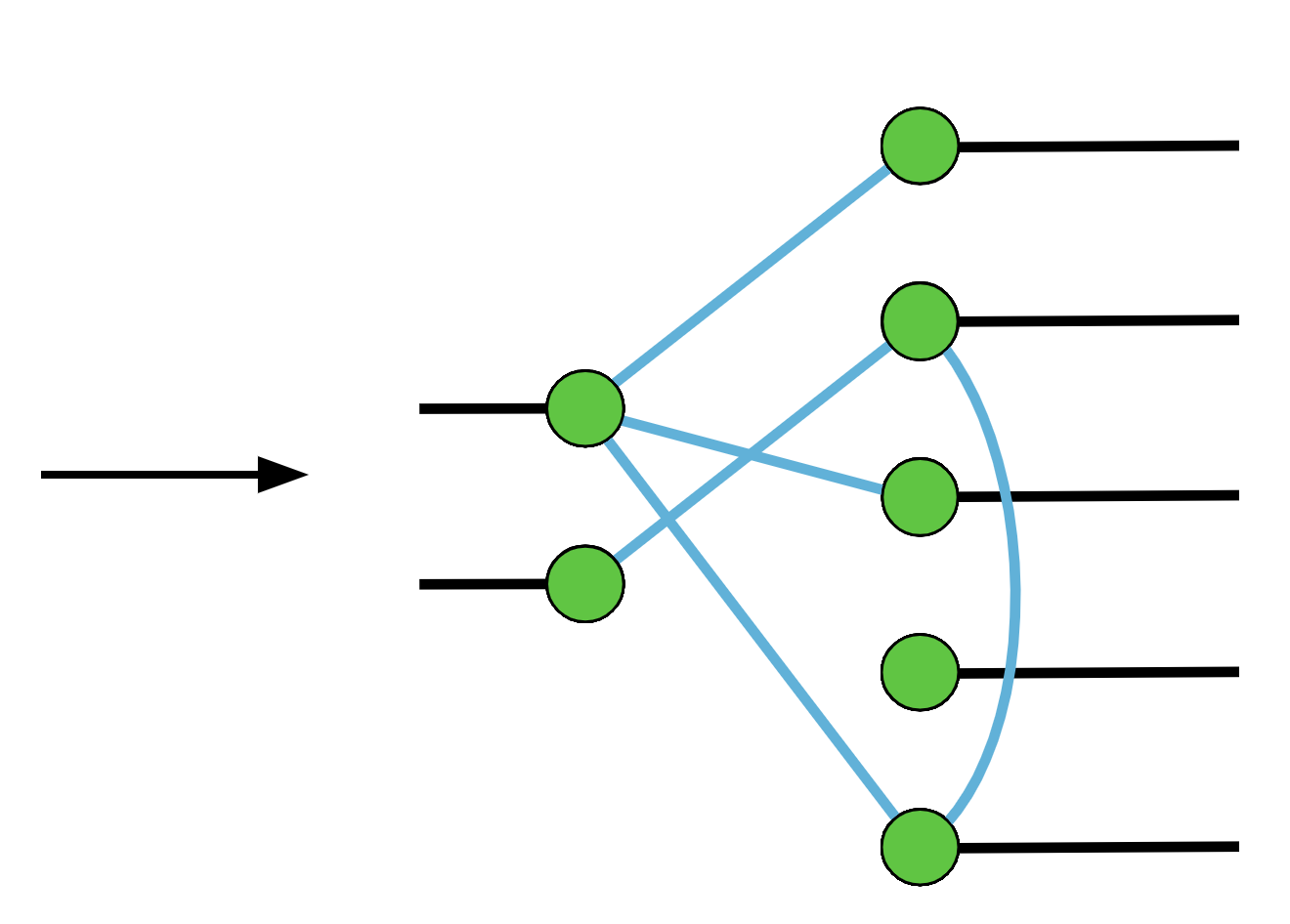}}
\end{subfigure}
\end{center}
\caption[Graph code from adjacency matrix]{In the adjacency matrix, the first row corresponds to the first input, the second row corresponds to the second input, and so on. After the inputs, the following row corresponds to the first output, and the other outputs follow. The same applies for the columns as well as the rows.
We see that the matrix is symmetric, as an adjacency matrix should be. In the ZX diagram on the right, we see the ZX encoder for the graph code with the given adjacency matrix.}
\label{transforming example}
\end{figure}

\begin{definition}
\label{def:locally-equivalent}
    Two ZX diagrams are \textit{locally equivalent} if and only if one can be converted to the other through a sequence of local complementations and local operations on the free output edges.
\end{definition}

We now provide another definition of equivalence that turns the focus to the topological structure of the code, letting us ignore specific orderings of output nodes and the local operations attached on the free output edges.

\begin{definition}
\label{def:equivalent-clifford-codes}
    We say that two Clifford codes $\mathcal{C}_1$ and $\mathcal{C}_2$ are \textit{equivalent} if and only if their ZX diagrams are locally equivalent up to some permutation of the output nodes or applications of any unitary operators on the inputs.
\end{definition}

We now list six operations which preserve encoder graph equivalence.
We conjecture that this is a list of equivalent operations generates all others.

\begin{conjecture}
\label{operations}
The ZX diagrams for two Clifford codes $\mathcal{C}_1$ and $\mathcal{C}_2$ are \textit{equivalent} if and only if the diagram for one of the codes can be reached from the other after a sequence of operations consisting only of the following:
\begin{enumerate}[(1)]
\item Applying a local complementation about any vertex of the graph.
\item Permuting the output vertices.
\item Permuting the input vertices.
\item Performing row operations on the adjacency matrix of input to output edges.
\item Removing an input-input edge.
\item Applying local operations on the output edges.
\end{enumerate}
\end{conjecture}

All of the operations in Conjecture~\ref{operations} are reversible, so, if code $\mathcal{C}_1$ can be made equivalent to $\mathcal{C}_2$, the reverse is also true.

Operation (1), local complementation as in Definition~\ref{def:graph-local-complementation-vertex}, is included to account for equivalence of encoder graphs based on their entanglement~\cite{adcock2020mapping}.
Additionally, local complementation can be performed on a graph by applying the correct sequence of local operations and applying Equation~\ref{eq:vertex-local-complementation}.

Two encoder diagrams should also be equivalent if the information they produce can be ordered differently to become the same. In this way, operations (2) and (3) reflect this, since connections among the vertices of the graph remain the same and these operations only change the order in which the information is inputted or outputted.

Operation (4) consists of adding rows of the adjacency matrix between input and output nodes in modulo 2.
As discussed in Appendix~\ref{app:compiler}, this operation is equivalent to applying a controlled-$X$ gate between two inputs, thus preserving the equivalence of the encoder.
Similarly, operation (5) applies a controlled-$Z$ operation to the inputs.
Both of these are allowed by Definition~\ref{def:equivalent-clifford-codes}. Lastly, operation (6) preserves equivalence since all local operations can be removed by multiplying by their corresponding conjugate, which is allowed by Definition~\ref{def:locally-equivalent}.

We have thus shown one of the directions of Conjecture~\ref{operations}.

Note that this definition of equivalence does not allow two encoders to be in the same equivalence class if they only differ by an extra output which is not connected to anything else. That is, if the two encoders perform the same encoding, but one appends an additional fixed output state, this definition of equivalence marks them as different.

An example of these extra outputs is shown in Figure~\ref{transforming example}.
The fourth output vertex is not connected to any input or output. It does not provide any more encoding of information from the inputs than if it was not present. For this reason, removing extra states is more useful for practical purposes.

There are some simplifications that can be made on the set of encoders we consider by using the above operations. This is so that we consider only encoder graphs that could possibly be non-equivalent.

We can now try to use these operations to find equivalences between prime codes.
These operations allow us to impose conditions on selecting representative elements of equivalent codes.
Operations (3) and (4) of Conjecture~\ref{operations} allow the input-to-output portion of the encoder diagram to be expressed in RREF. All encoder diagrams considered from here on will be expressed already satisfying the RREF rule from Section~\ref{subsec:zx-construction}.

Continuing from the RREF of the encoder graph, operation (2) from Conjecture~\ref{operations} can be used to move the pivot nodes to have lower-numbered indices than all the other outputs. In this way, the first output node can be made into the pivot node corresponding to the first input node, the second output node can be made into the pivot node corresponding to the second input node, and so on. Thus, these $k$ pivot nodes are fixed as the first $k$ output nodes.
For brevity, the other $n-k$ non-pivot output nodes are called \textit{free} output nodes.

In Conjecture~\ref{operations}, we could freely manipulate local Clifford gates at nodes. Therefore, for our purposes, we remove all local Clifford gates present at the nodes of the ZX diagram for the encoder, just like in the graph codes studied in Chapter~\ref{chapter:qcodes-graphs}.

A further simplification is that graphs with pivot-pivot edges are omitted, since they can always be transformed into a graph without pivot-pivot edges using a sequence of local complementations.

We group these simplifications into the following result.

\begin{claim}
\label{claim: simplifications}
    To find the distinct encoders in an equivalence class of the set of Clifford codes, it is sufficient to find the distinct encoders in the equivalence class satisfying the following constraints:
    \begin{itemize}
        \item The input-output adjacency matrix is in RREF.
        \item The pivot nodes are the lowest-numbered output nodes, and they are ordered to match the order of the input nodes.
        \item Local Clifford gates on the free output edges are all removed, leaving only output $Z$ nodes with phase 0.
        \item Diagrams with pivot-pivot edges are not included.
    \end{itemize}

\end{claim}

An example of the simplifications on the diagrams as well as how the adjacency matrices transform into encoders is shown in Figure~\ref{transforming example}.

\begin{figure}[ht!]
    \centering
    \begin{subfigure}[t!]{0.5\textwidth}\centering\bgroup\def\arraystretch{1.5}
    \begin{tabular}{|c|c|c|c|c|}\hline
        $n = 1$ & $n = 2$ & $n = 3$ &  $n = 4$ & $n = 5$ \\
        \hline
        1 & 1 & 2 & 6 & 17 \\\hline
    \end{tabular}
    \egroup
    \centering
        \caption{Number of equivalence classes for $\llbracket n,1\rrbracket$ codes.}
    \end{subfigure}\hfill
    \begin{subfigure}[t!]{0.4\textwidth}
        \centering
        \bgroup
\def\arraystretch{1.5}
\setlength{\tabcolsep}{0.5em} 
    \begin{tabular}{|c|c|}\hline
        Rep: & \includegraphics[scale=0.09]{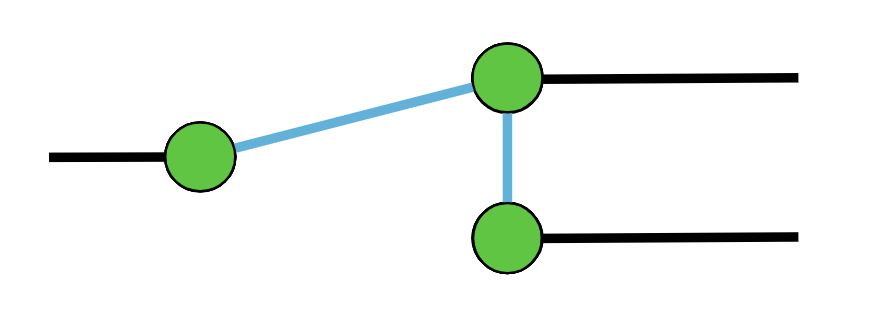} \\
        \hline
         Size: & 3\\\hline
    \end{tabular}
\egroup
        \centering\caption{$\llbracket 2,1\rrbracket$ code equivalence class representatives and sizes.}
    \end{subfigure}
\begin{center}
\bgroup
\def\arraystretch{1.5}
\setlength{\tabcolsep}{0.5em}
\begin{tabular}{|c|c|c|}\hline
Rep: & \includegraphics[scale=0.08]{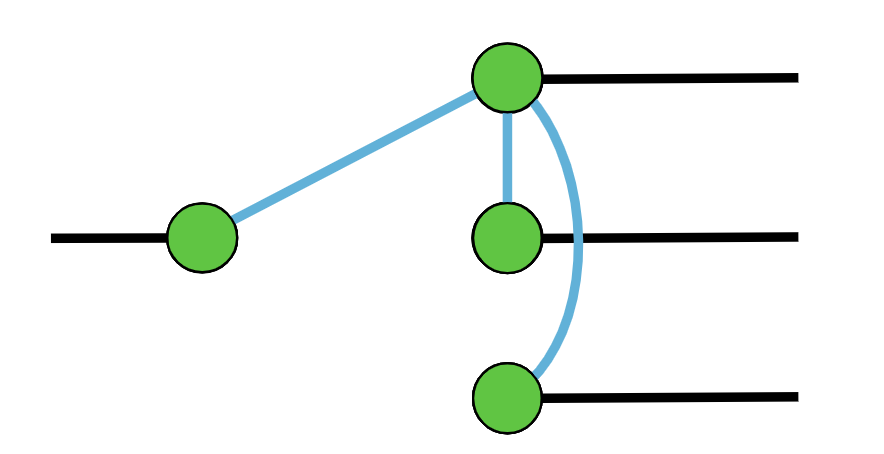}& \includegraphics[scale=0.08]{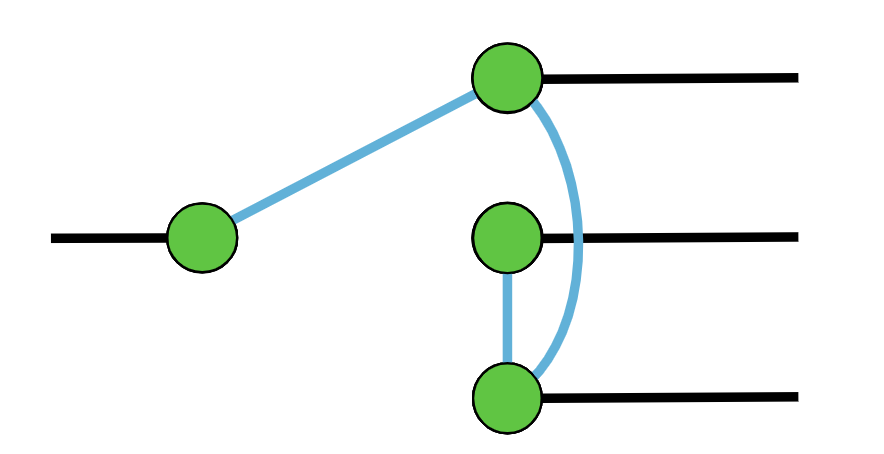}\\
\hline
Size: & 3&21 \\\hline
\end{tabular}
\egroup
\bigskip

{\small (c) $\llbracket 3,1\rrbracket$ code equivalence classes showing the size of the class underneath a representative.}

\bigskip

\bgroup
\def\arraystretch{1.5}

\setlength{\tabcolsep}{0.5em}

\begin{tabular}{|c|c|c|c|c|c|c|}
\hline
Rep: & \includegraphics[scale=0.08]{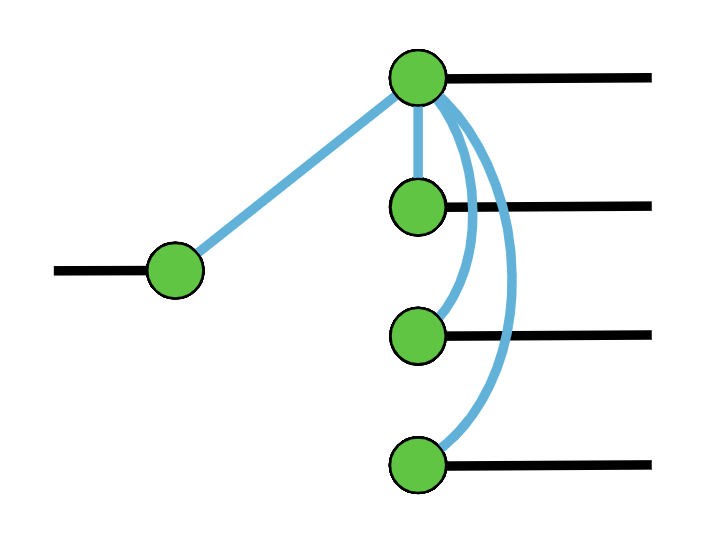}& \includegraphics[scale=0.08]{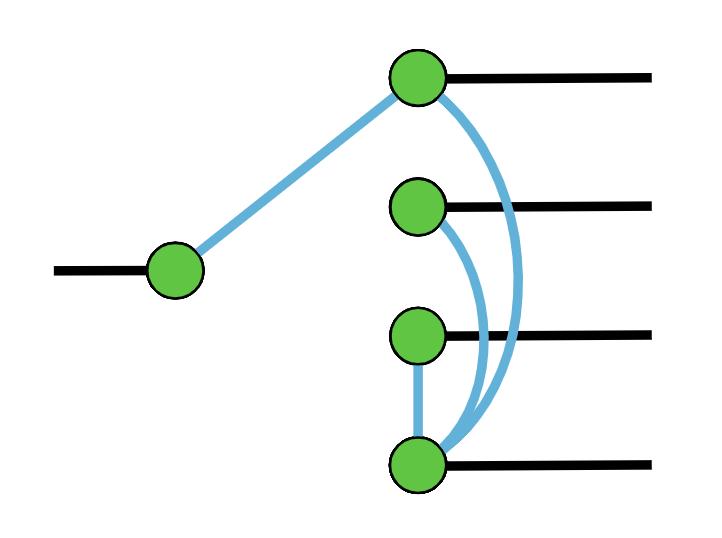} &  \includegraphics[scale=0.08]{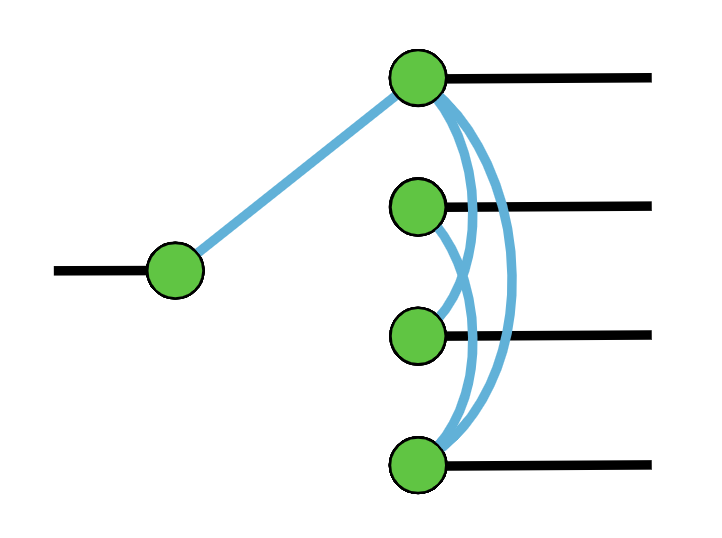}& \includegraphics[scale=0.08]{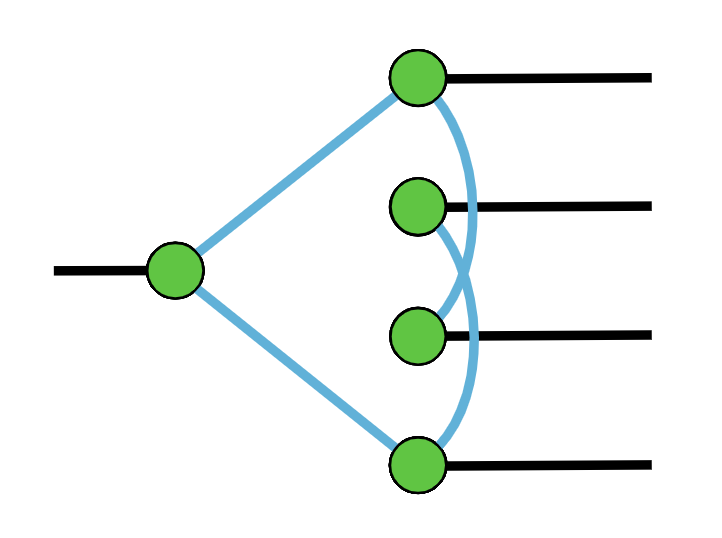} &  \includegraphics[scale=0.08]{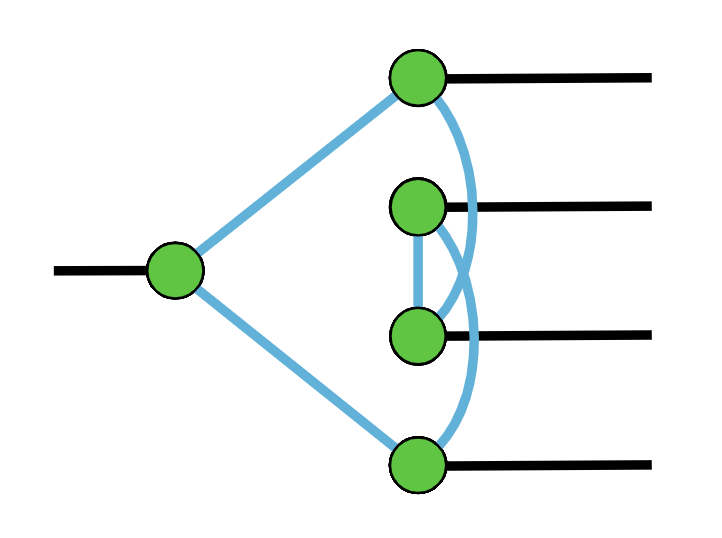}& \includegraphics[scale=0.08]{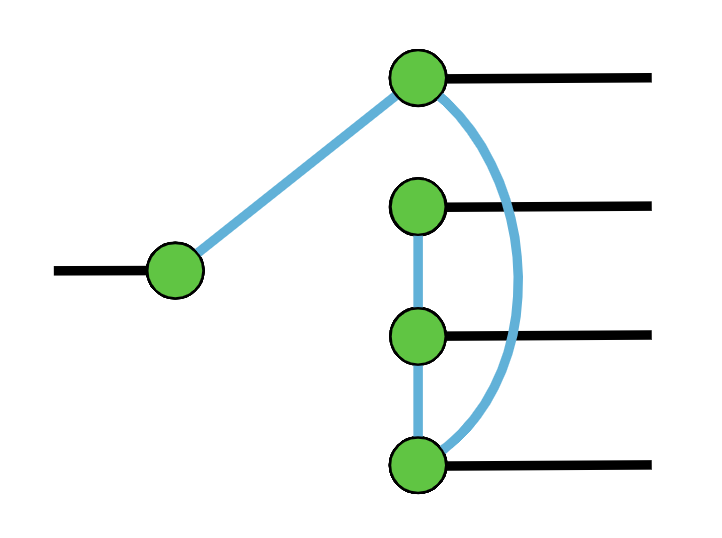}\\
\hline
Size: & 3&30 & 45 & 54 & 84 & 198 \\\hline
\end{tabular}
\egroup

\bigskip

{\small (d) $\llbracket4,1\rrbracket$ code equivalence classes showing the size of the class underneath a representative.}

\bigskip    

\bigskip

\bgroup
\def\arraystretch{1.5}

\setlength{\tabcolsep}{0.5em}

\begin{tabular}{|c|c|c|c|c|c|}
\hline
Rep: & \includegraphics[scale=0.08]{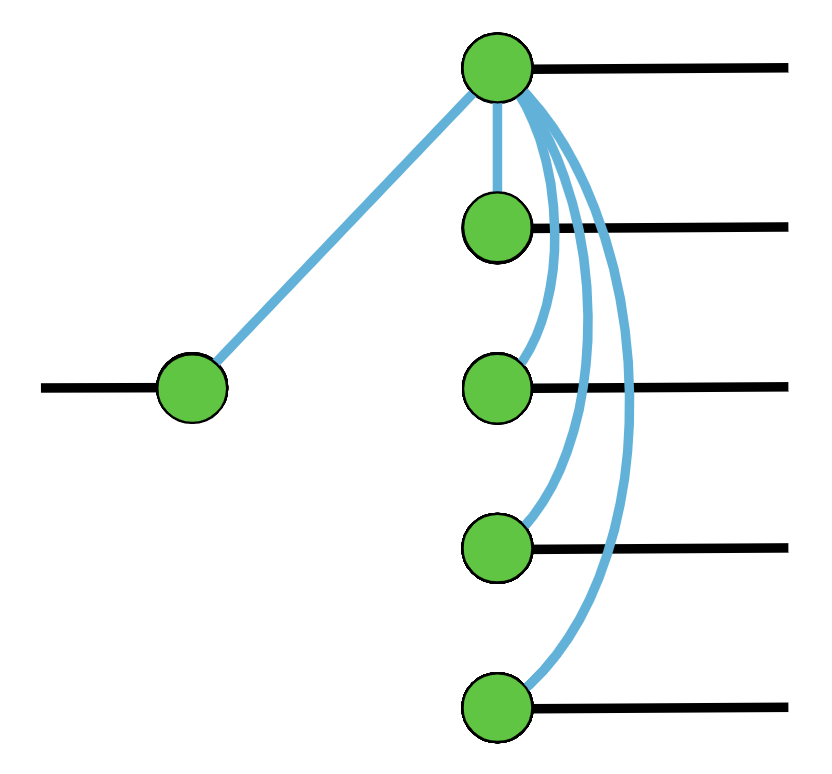}& \includegraphics[scale=0.08]{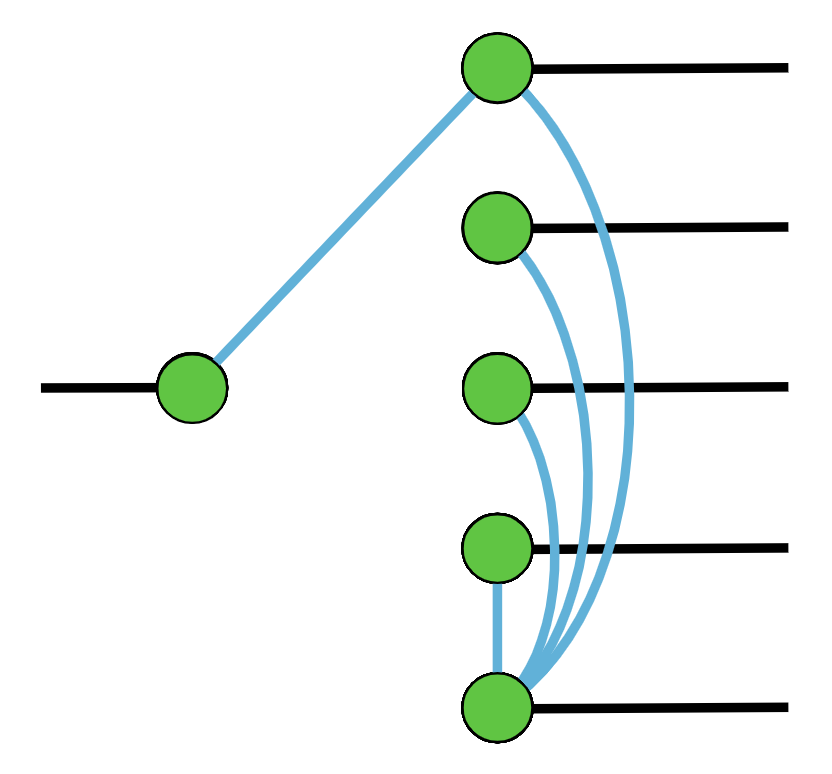} &  \includegraphics[scale=0.08]{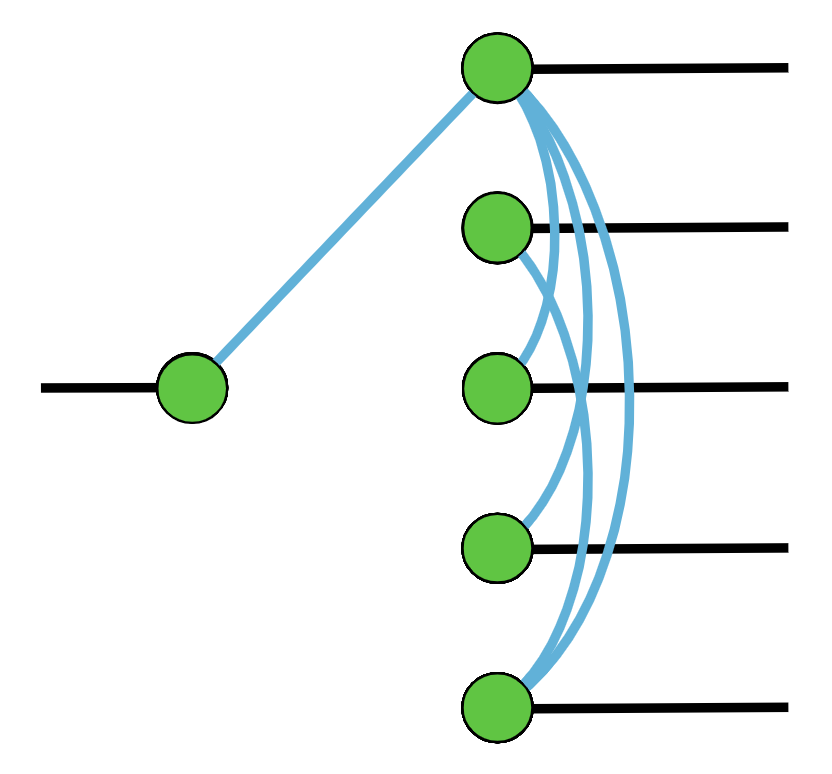}& \includegraphics[scale=0.08]{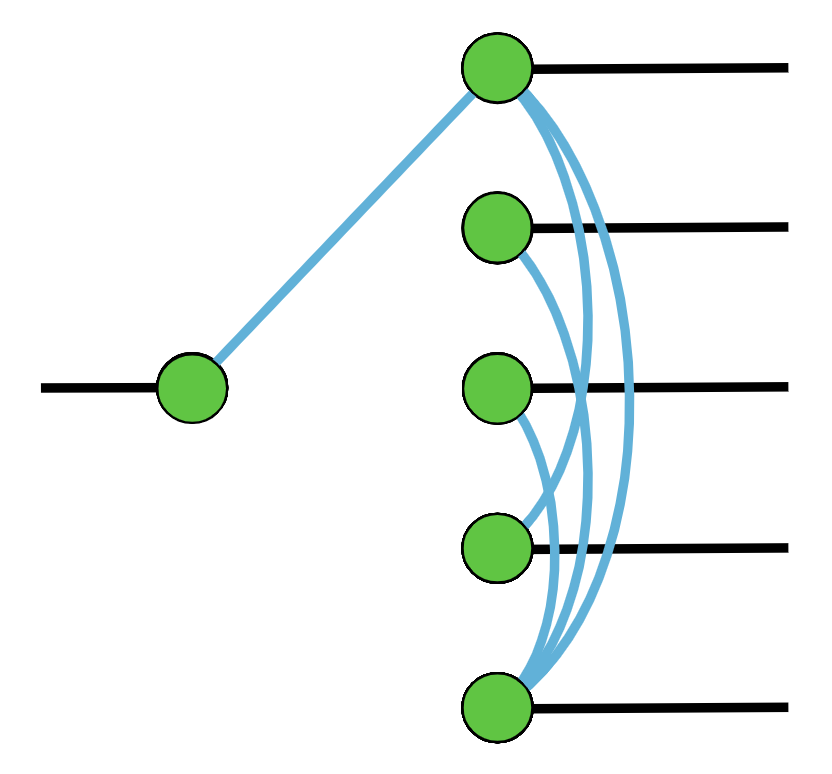} &  \includegraphics[scale=0.08]{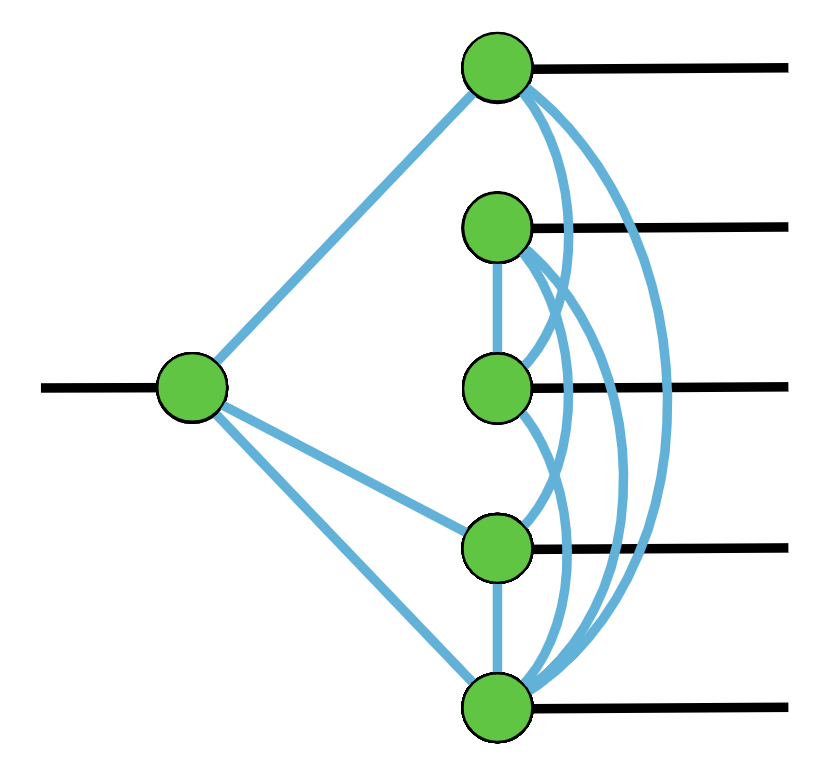}\\
\hline
Size: & 3&39 & 78 & 84 & 84 \\

\hhline{|=|=|=|=|=|=|}

\includegraphics[scale=0.08]{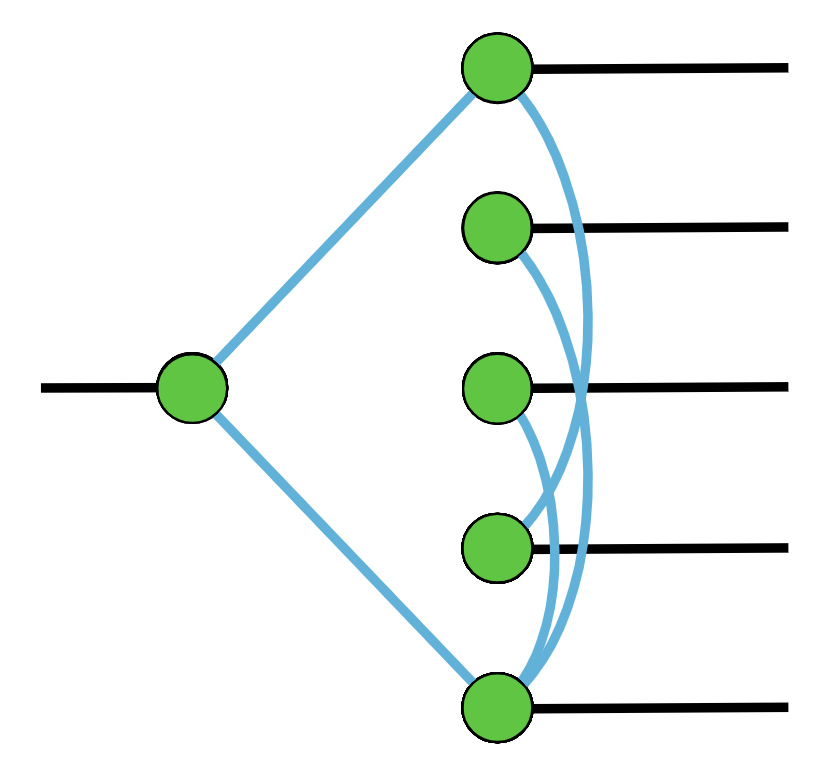}& \includegraphics[scale=0.08]{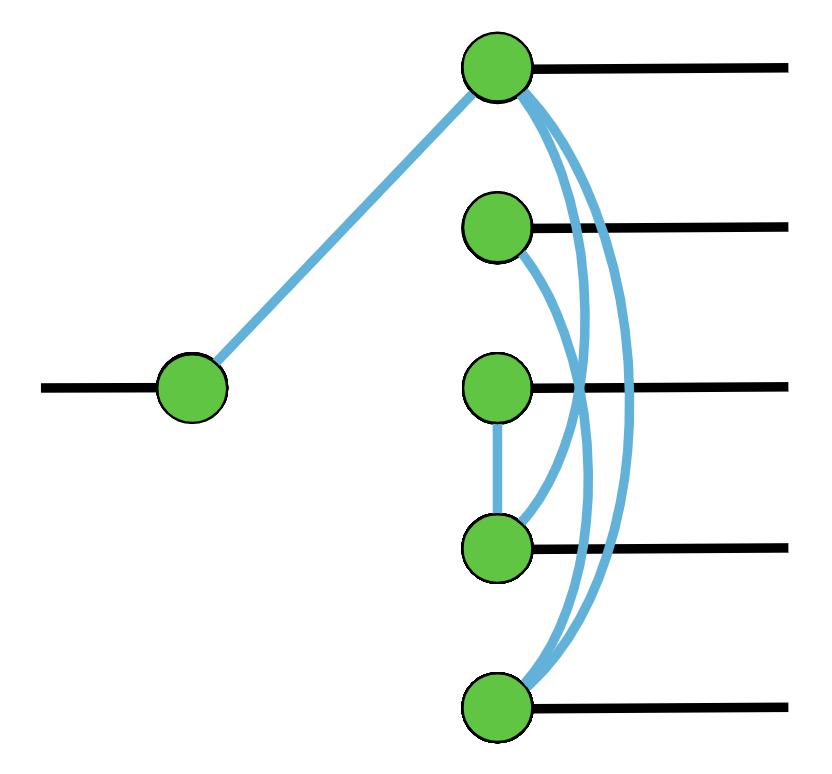} &  \includegraphics[scale=0.08]{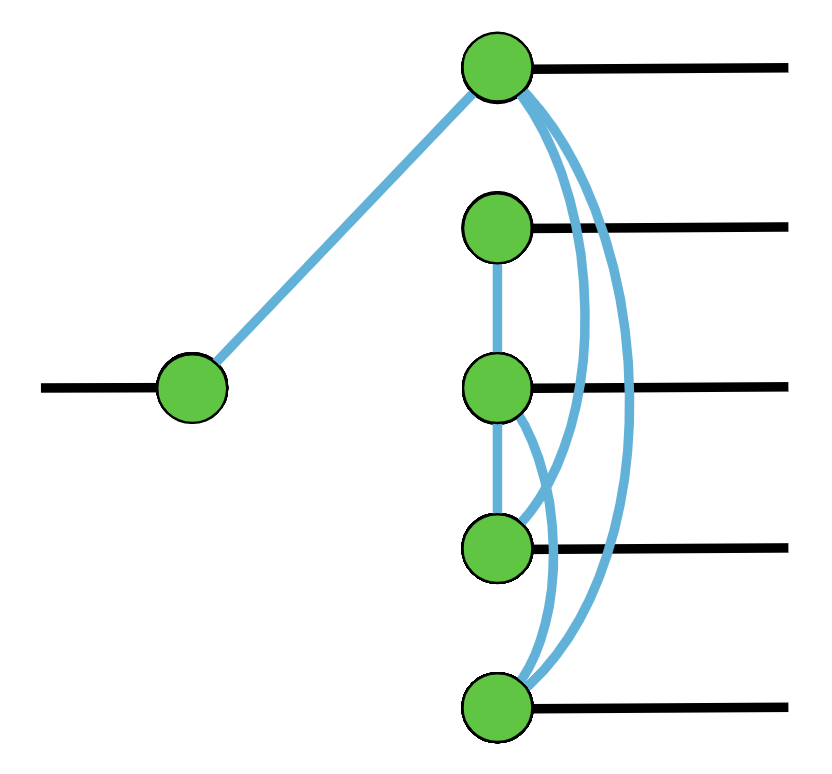}& \includegraphics[scale=0.08]{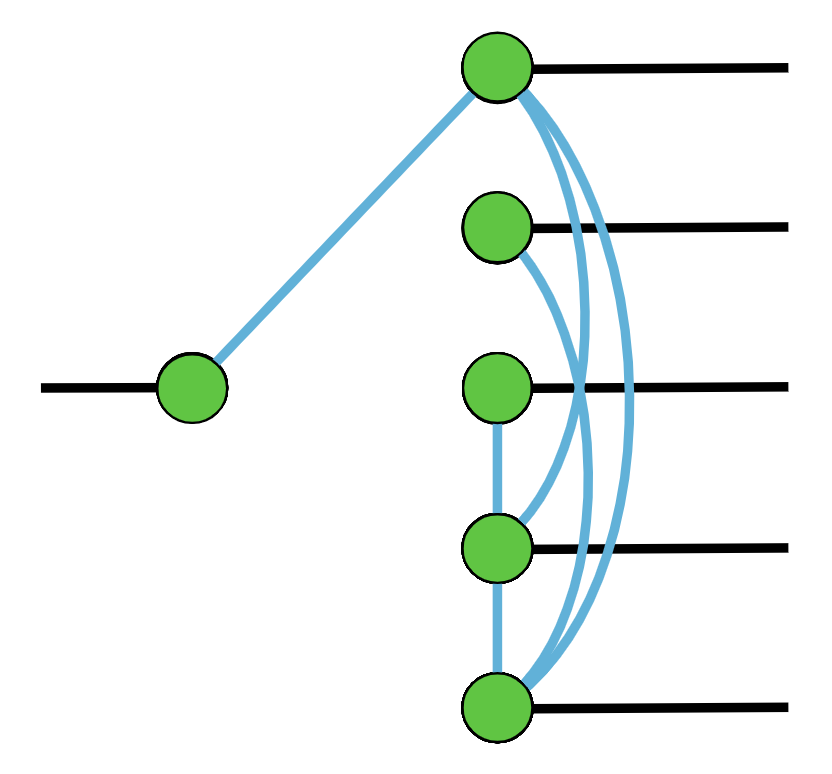} &  \includegraphics[scale=0.08]{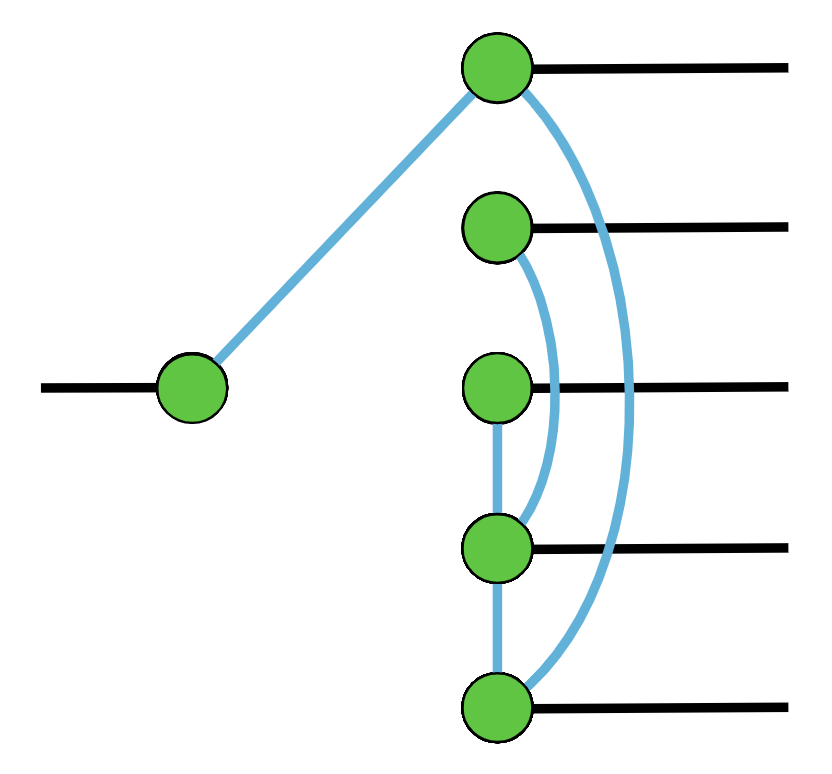}& \includegraphics[scale=0.08]{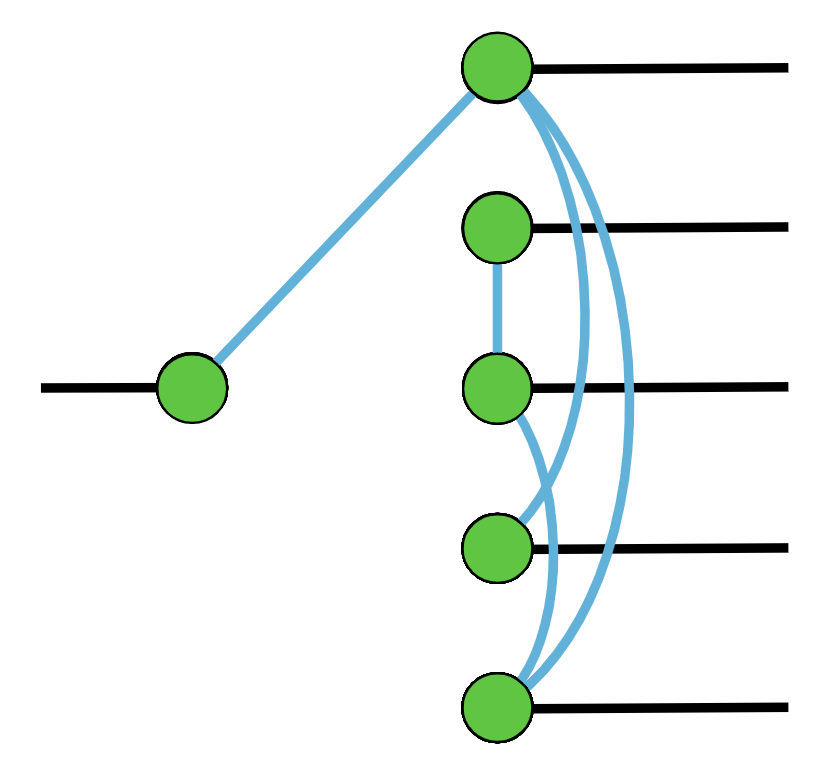}\\
\hline
204 & 297&306 & 315 & 360 & 540 \\

\hhline{|=|=|=|=|=|=|}

\includegraphics[scale=0.08]{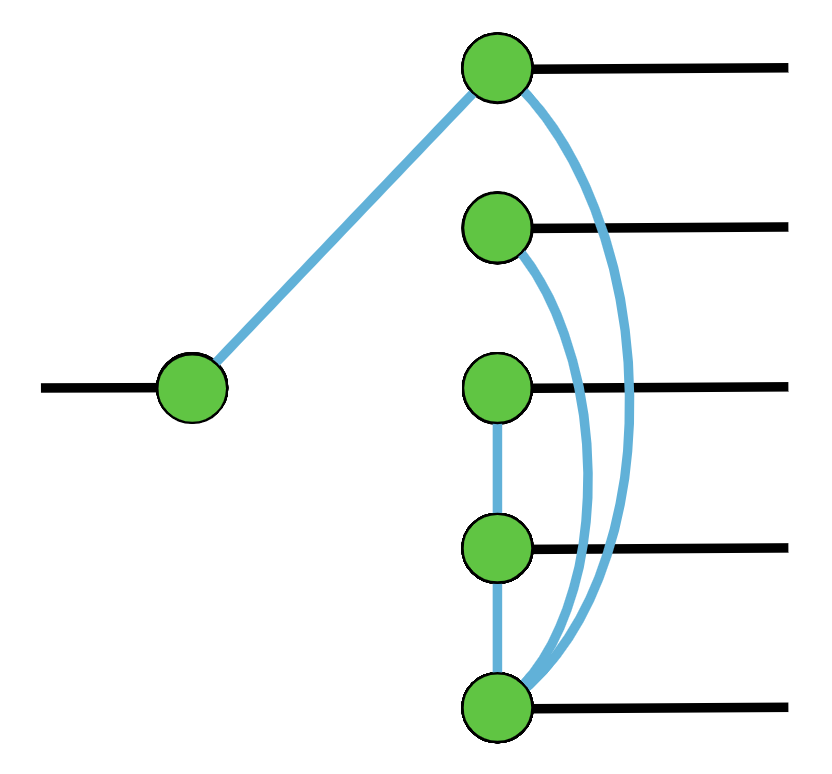}& \includegraphics[scale=0.08]{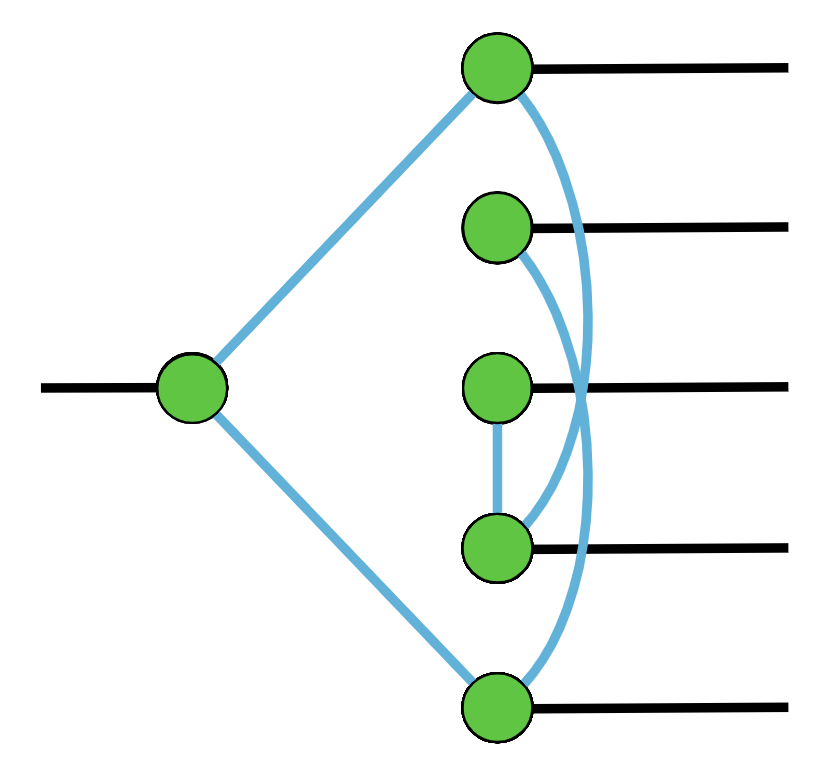} &  \includegraphics[scale=0.08]{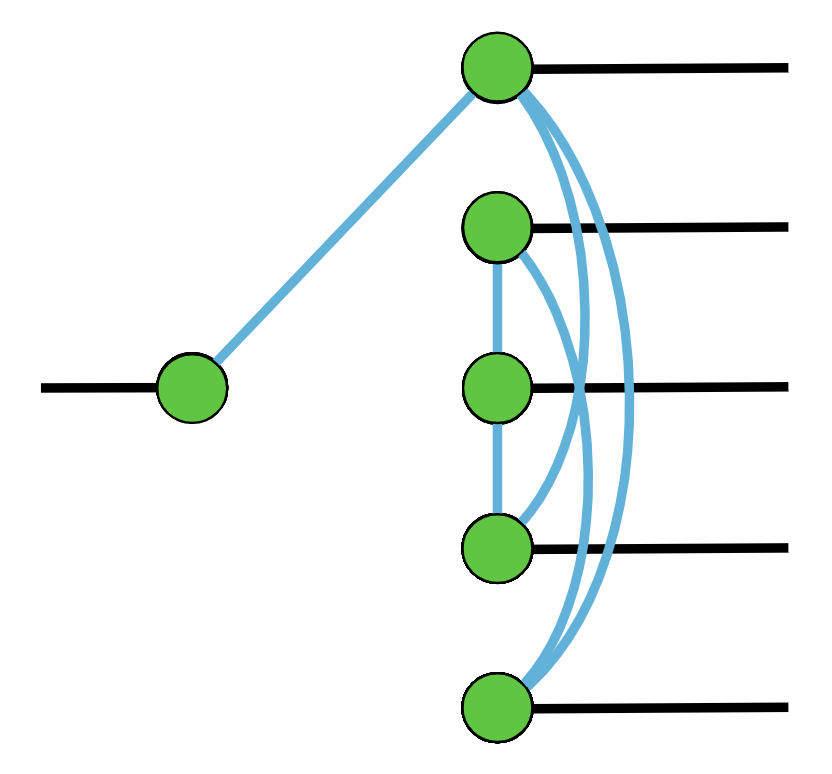}& \includegraphics[scale=0.08]{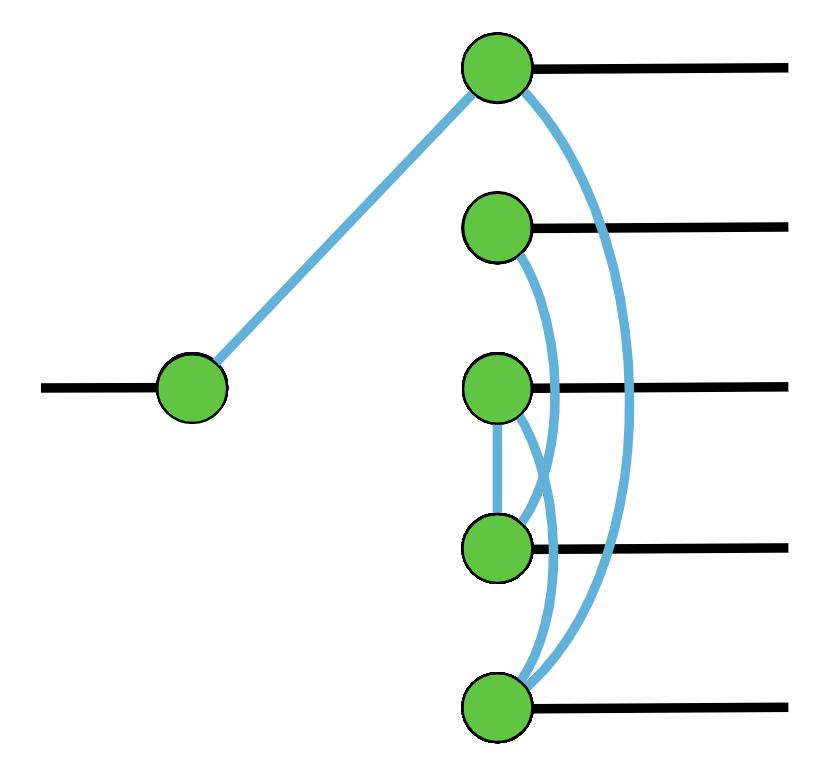} &  \includegraphics[scale=0.08]{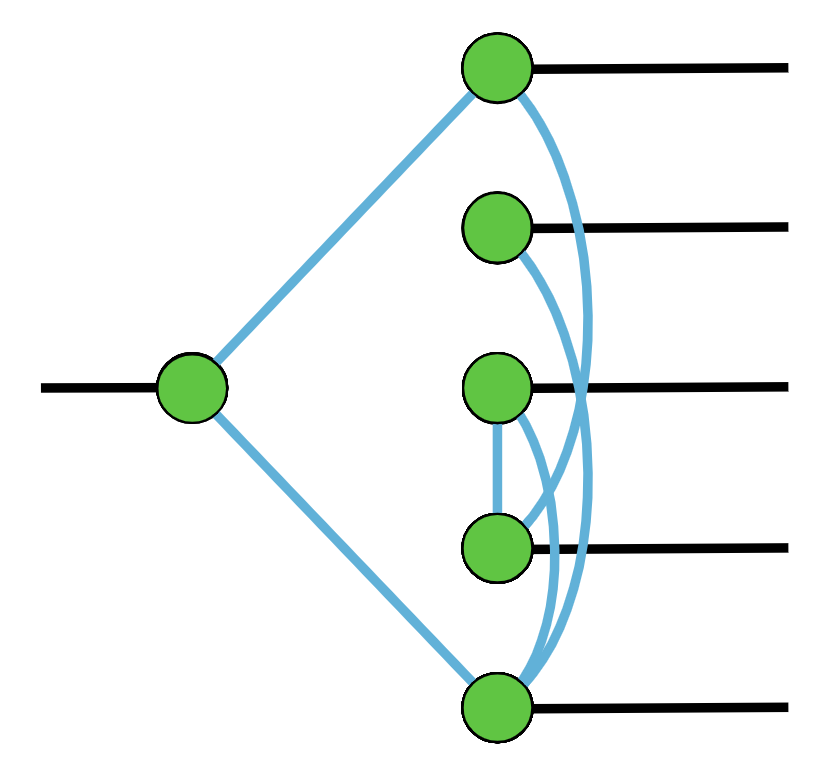}& \includegraphics[scale=0.08]{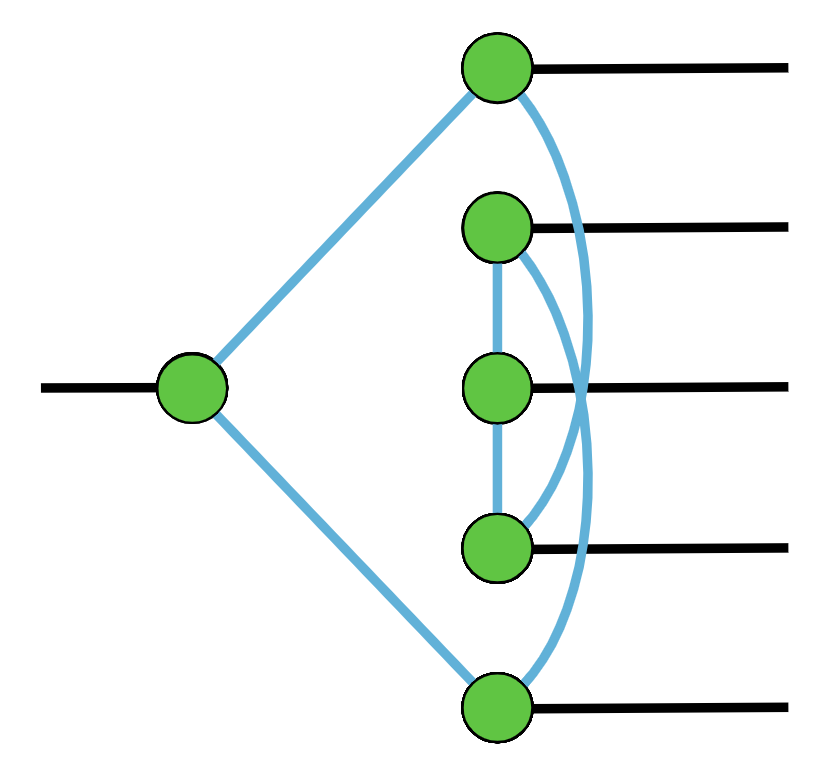}\\
\hline
558 & 1332&1404 & 2376 & 3024 & 3276 \\\hline
\end{tabular}
\egroup

\bigskip

{\small (e) $\llbracket 5,1\rrbracket$ codes equivalence classes showing the size of the class underneath a representative.}

\bigskip    

\end{center}
    
    \caption[$\llbracket n,1\rrbracket$ encoder graph equivalence classes]{(a) shows the number of equivalence classes for $\llbracket n,1\rrbracket$ encoder graphs. (b-e) show an element of the equivalence classes to denote the representative of the class and gives the size of the class. We only consider classes in which every graph is prime.}

    \label{[n,1] codes big tables}
\end{figure}

\begin{figure}[ht!]
    \centering
    \begin{subfigure}[t!]{0.4\textwidth}
        \centering
         \bgroup
\def\arraystretch{1.5}

\setlength{\tabcolsep}{0.5em} 
    \begin{tabular}{|c|c|c|c|}
        \hline $n = 2$ & $n = 3$ & $n = 4$ &  $n = 5$ \\
        \hline
        0 & 1 & 4 & 18 \\\hline
    \end{tabular}
    \egroup
        \caption{Number of equivalence classes for $\llbracket n,2\rrbracket$ codes.}
    \end{subfigure}
    \hfill
    \begin{subfigure}[t!]{0.4\textwidth}
        \centering
        \bgroup
\def\arraystretch{1.5}
\setlength{\tabcolsep}{0.5em} 
    \begin{tabular}{|c|c|}
        \hline
        Rep: &\includegraphics[scale=0.08]{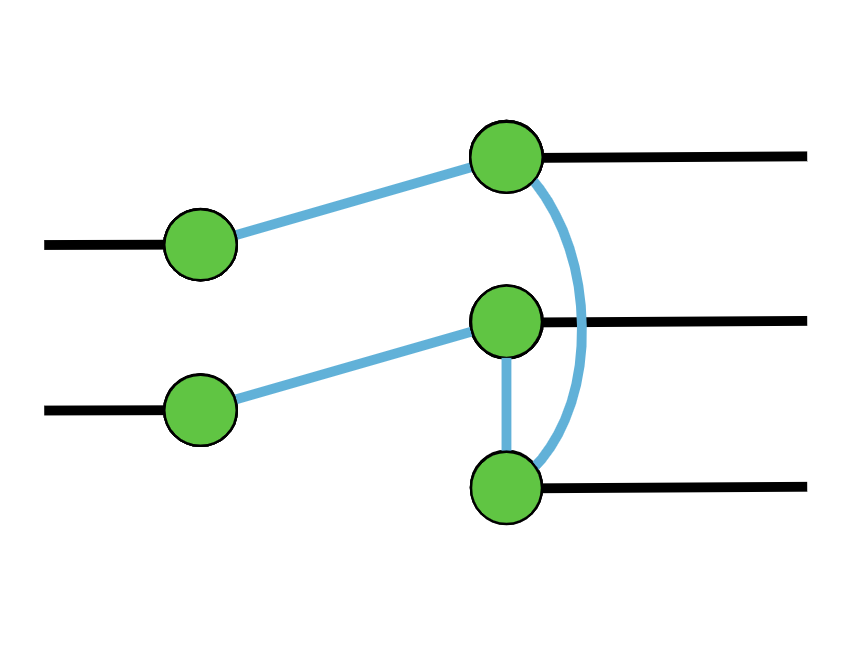} \\
        \hline
         Size: & 9\\\hline
    \end{tabular}
\egroup
        \caption{$\llbracket 3,2\rrbracket$ codes equivalence class reps. and sizes.}
    \end{subfigure}

\begin{center}

\bgroup
\def\arraystretch{1.5}

\setlength{\tabcolsep}{0.5em}

\begin{tabular}{|c|c|c|c|c|}
\hline
Rep: & \includegraphics[scale=0.08]{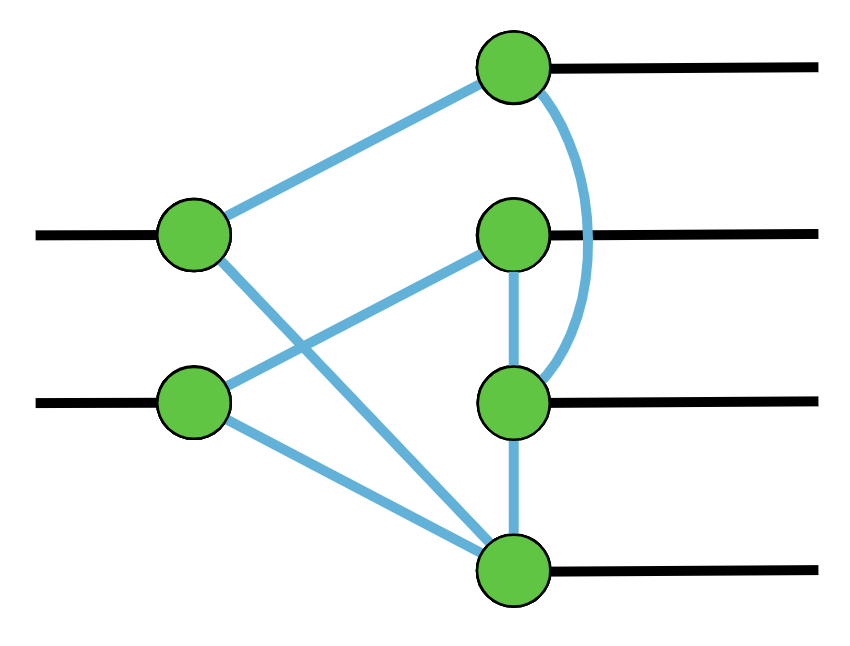}& \includegraphics[scale=0.08]{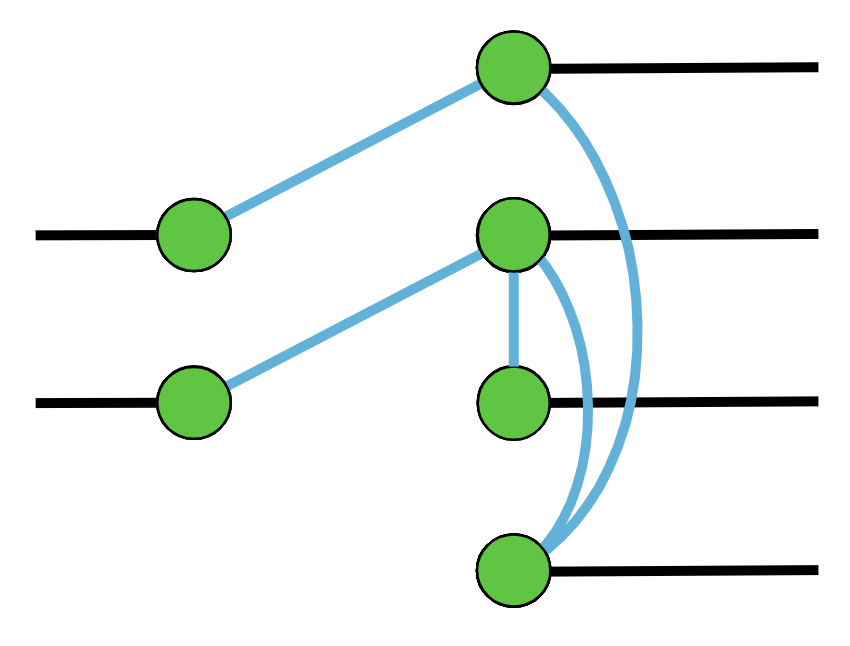}& \includegraphics[scale=0.08]{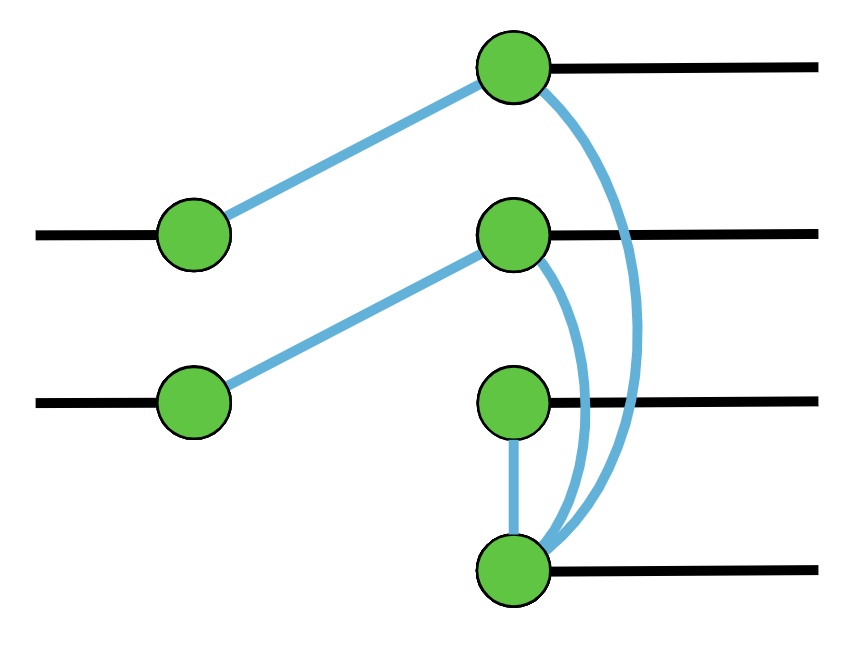}&\includegraphics[scale=0.08]{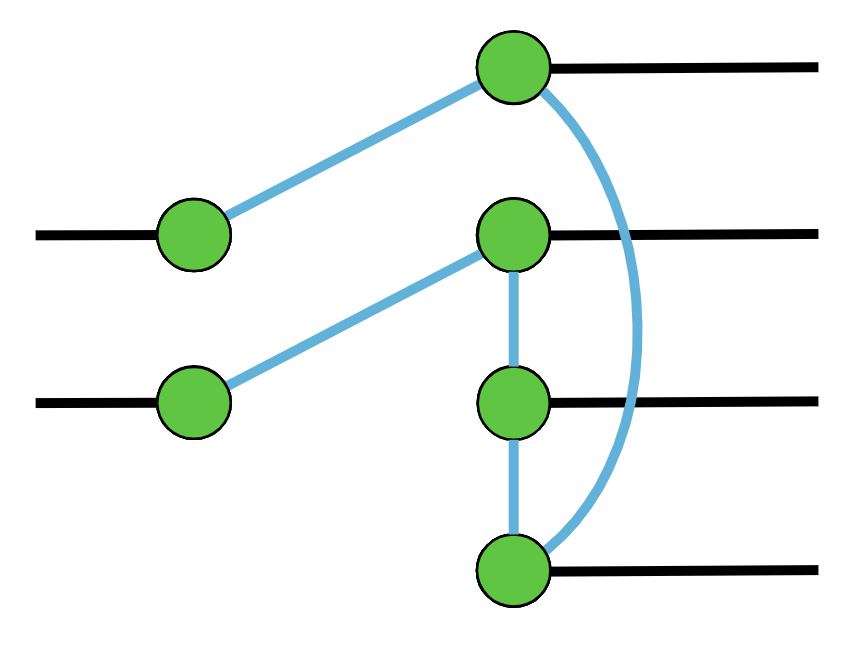}\\
\hline
Size: & 36&45&99&234 \\\hline
\end{tabular}
\egroup

\bigskip

{\small (c) $\llbracket 4,2\rrbracket$ codes equivalence classes showing the size of the class underneath a representative.}

\bigskip
\bgroup
\def\arraystretch{1.5}

\setlength{\tabcolsep}{0.5em} 
    \begin{tabular}{|c|c|c|c|c|c|}
        \hline
        \includegraphics[scale=0.08]{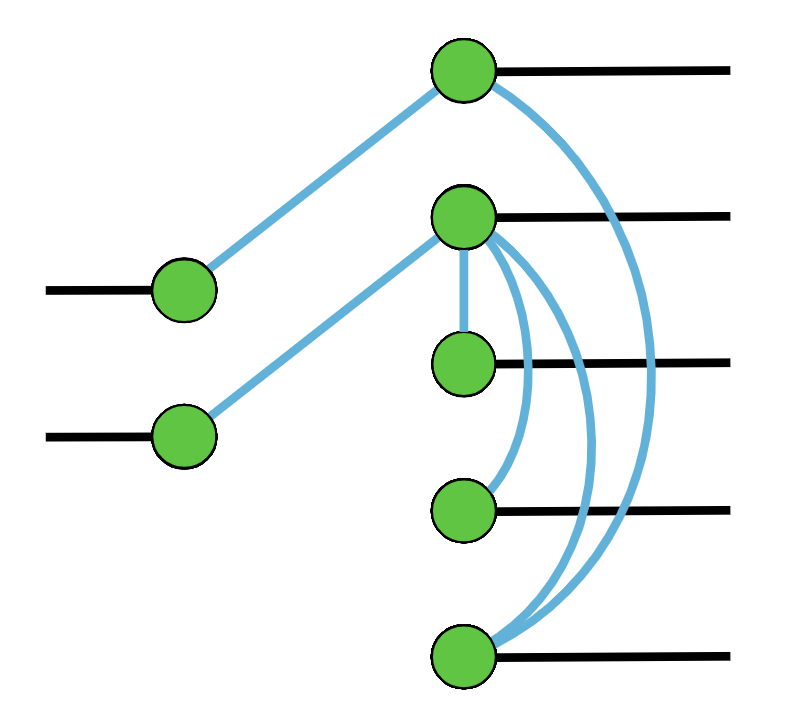} & \includegraphics[scale=0.08]{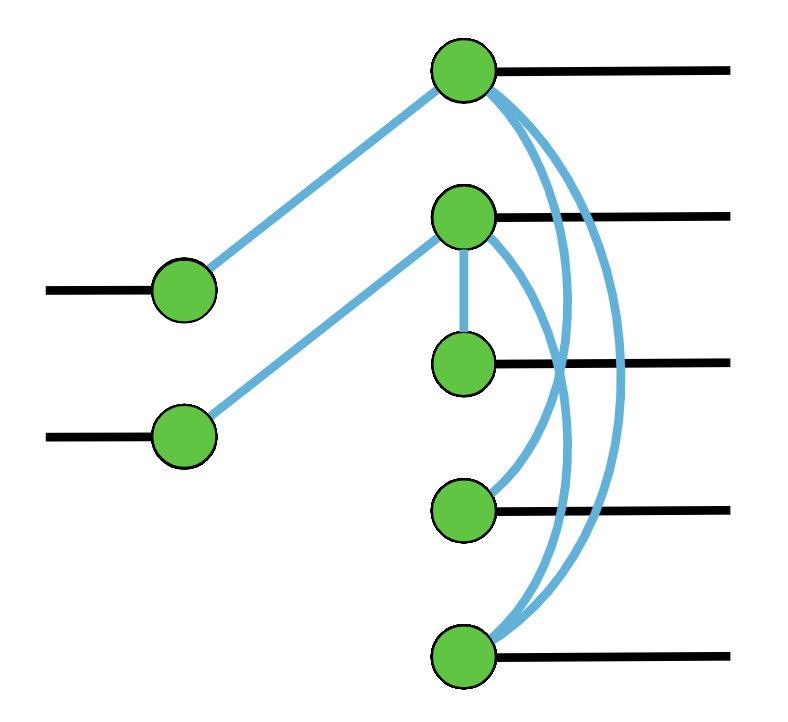} &                                     \includegraphics[scale=0.08]{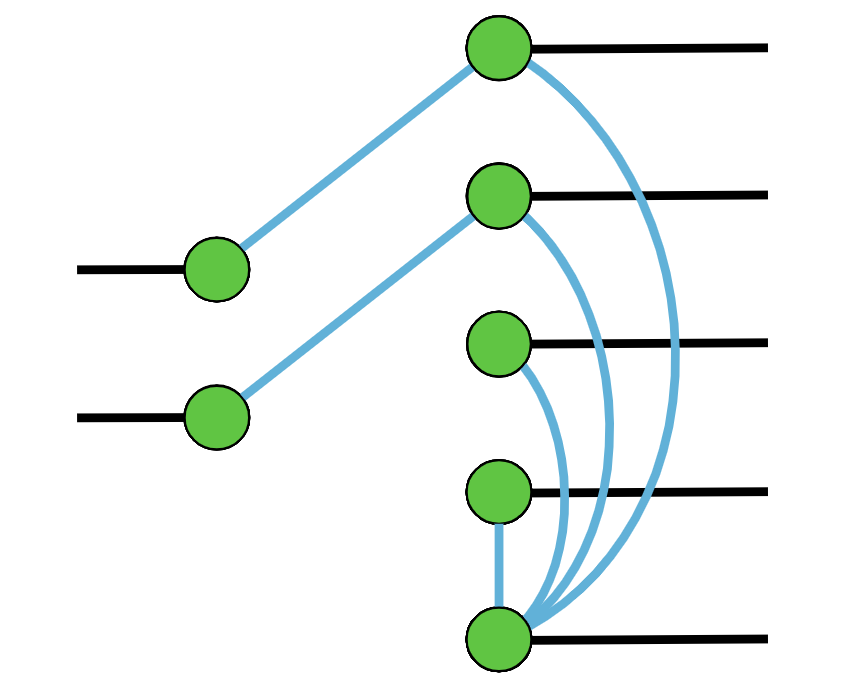} &  \includegraphics[scale=0.08]{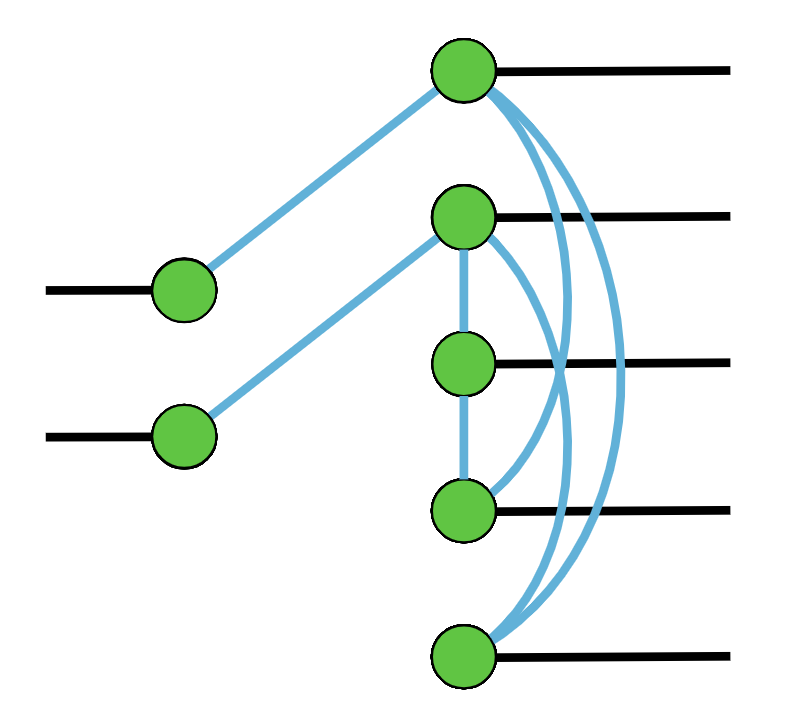} &                    \includegraphics[scale=0.08]{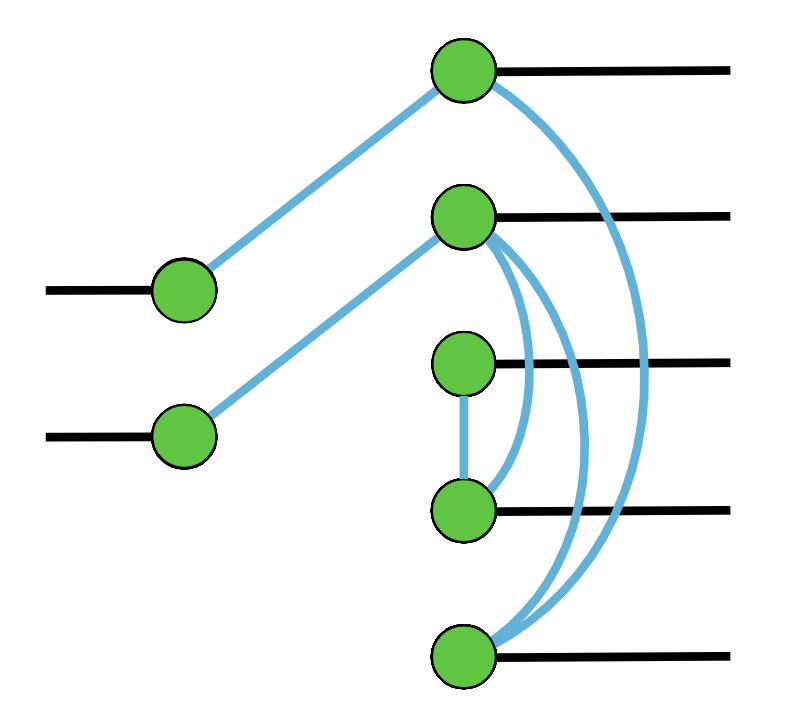} &                \includegraphics[scale=0.08]{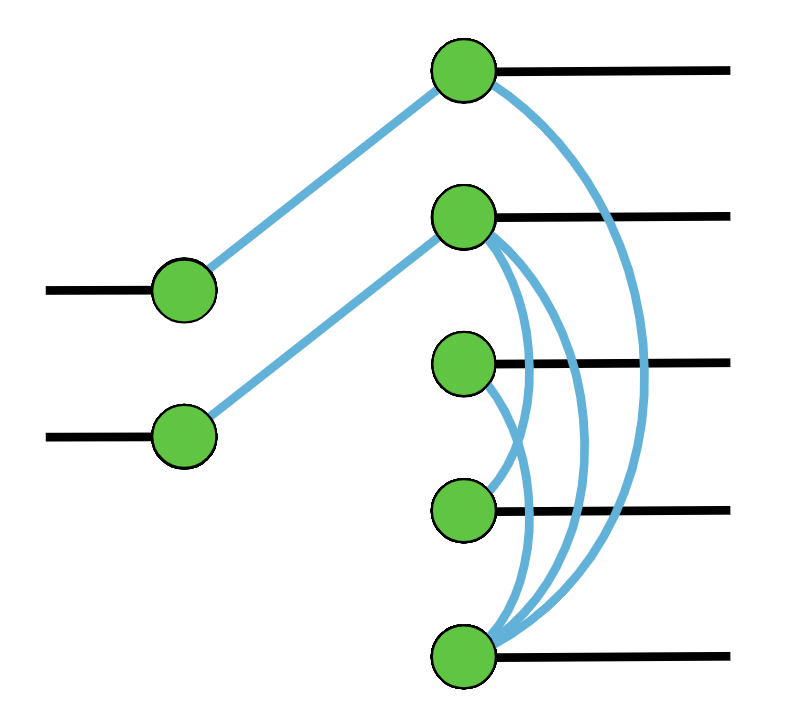} \\
         \hline
          63 & 108 & 144 & 414 & 459 & 486\\
                \hhline{|=|=|=|=|=|=|}
        \includegraphics[scale=0.08]{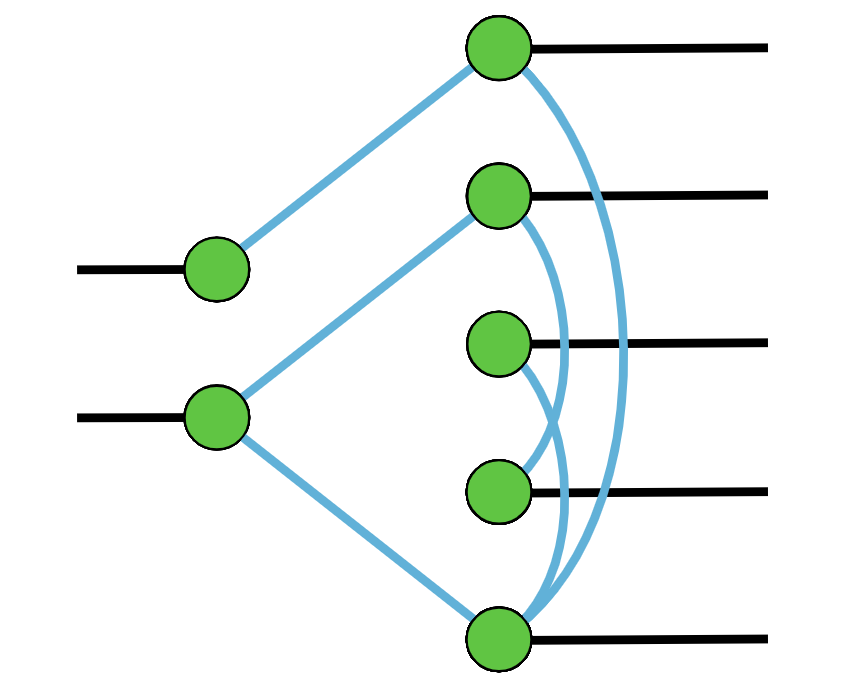}&                \includegraphics[scale=0.08]{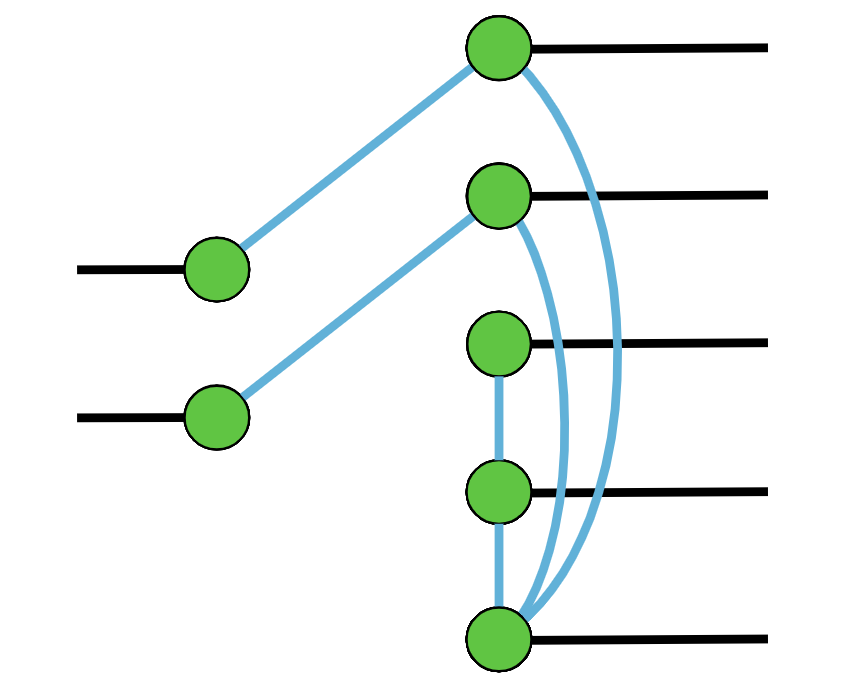}&                  \includegraphics[scale=0.08]{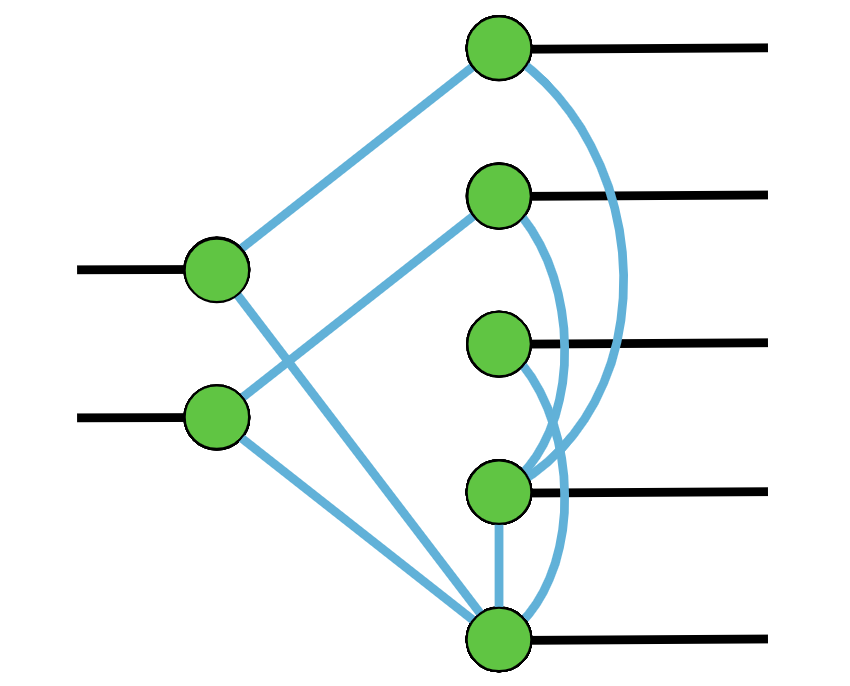} &                 \includegraphics[scale=0.08]{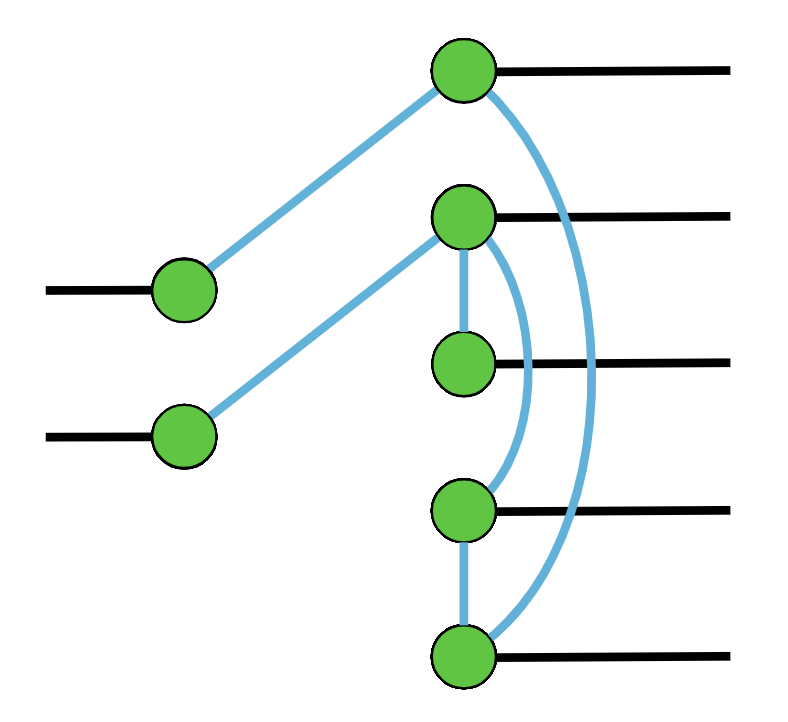} & \includegraphics[scale=0.08]{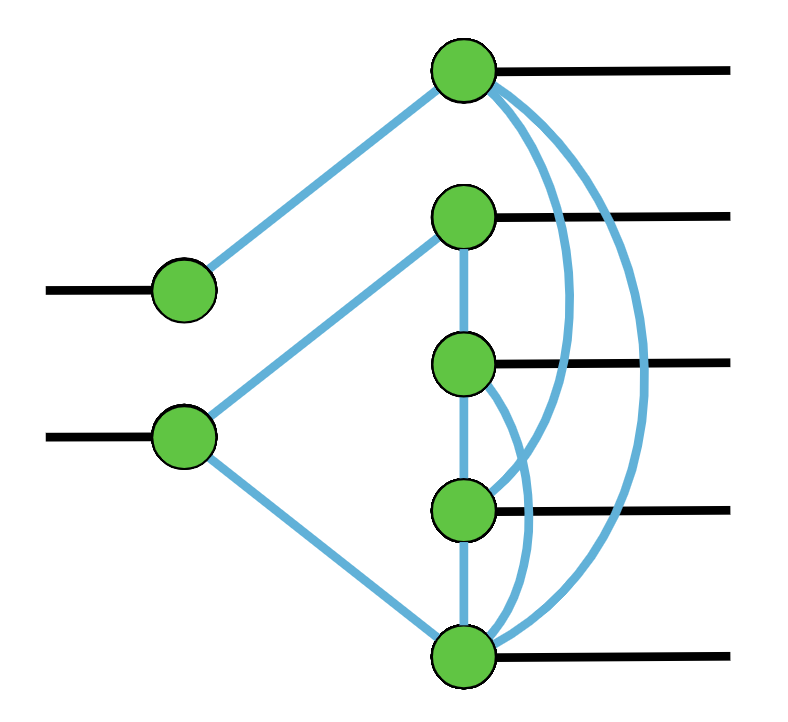} &                      \includegraphics[scale=0.08]{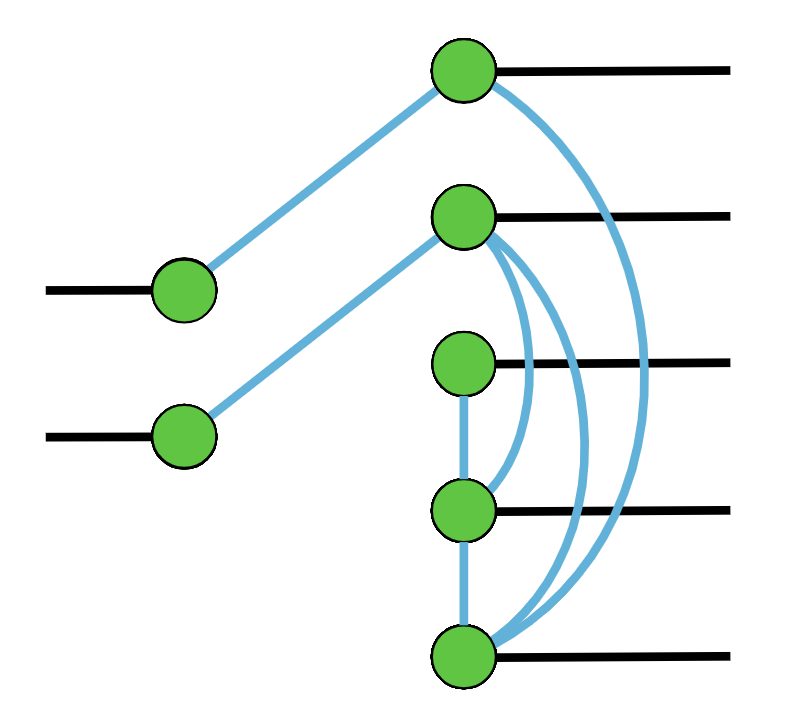}\\
        \hline
        540 & 972 & 1080 & 1080 & 1152 & 1188\\
             \hhline{|=|=|=|=|=|=|}
        \includegraphics[scale=0.08]{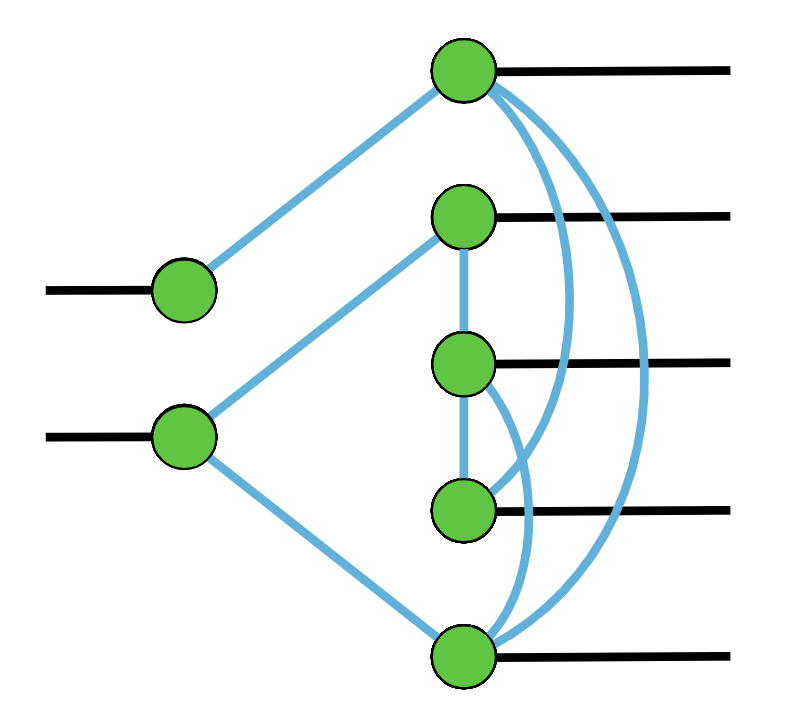} &                    \includegraphics[scale=0.08]{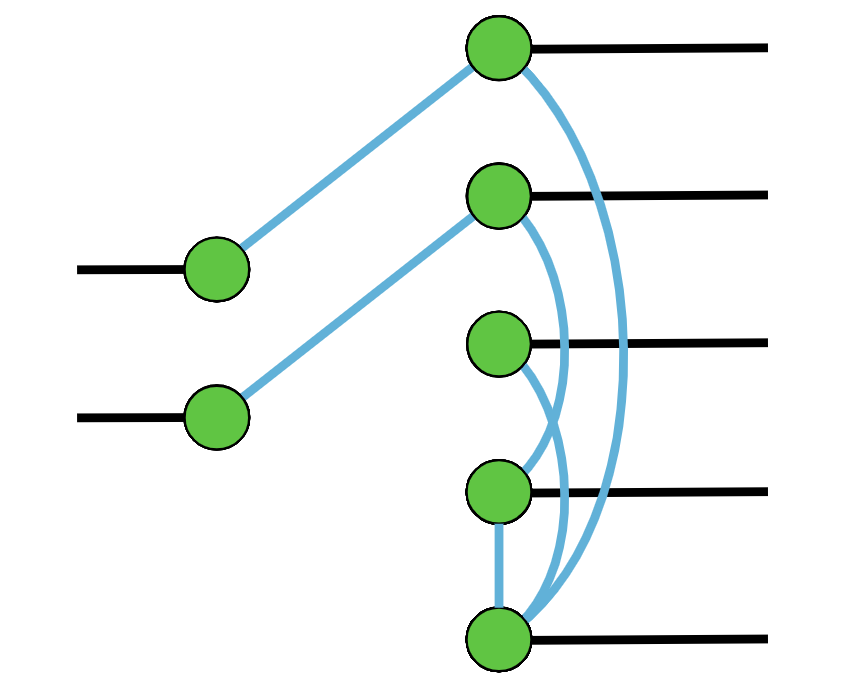}&                         \includegraphics[scale=0.08]{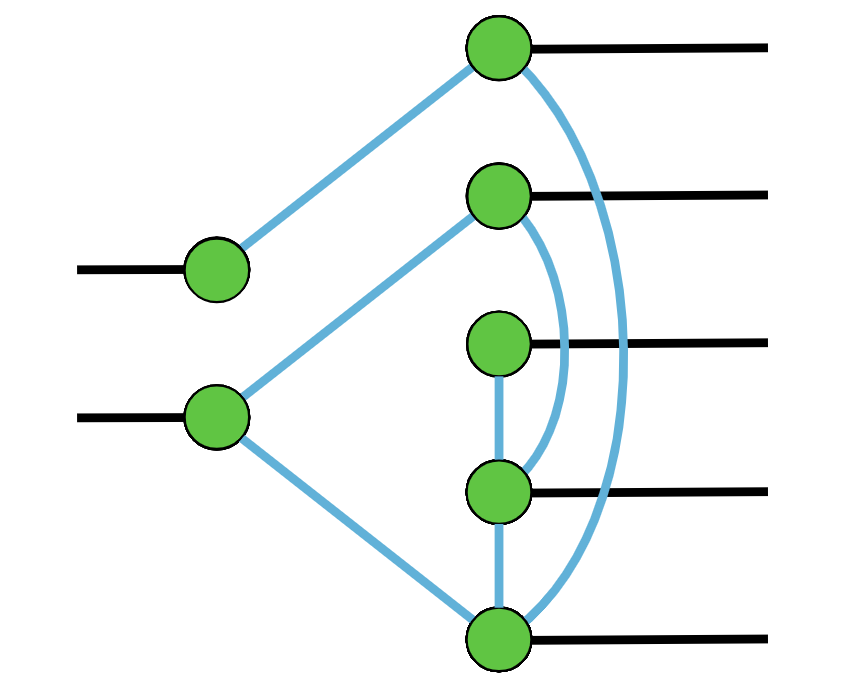} &                                           \includegraphics[scale=0.08]{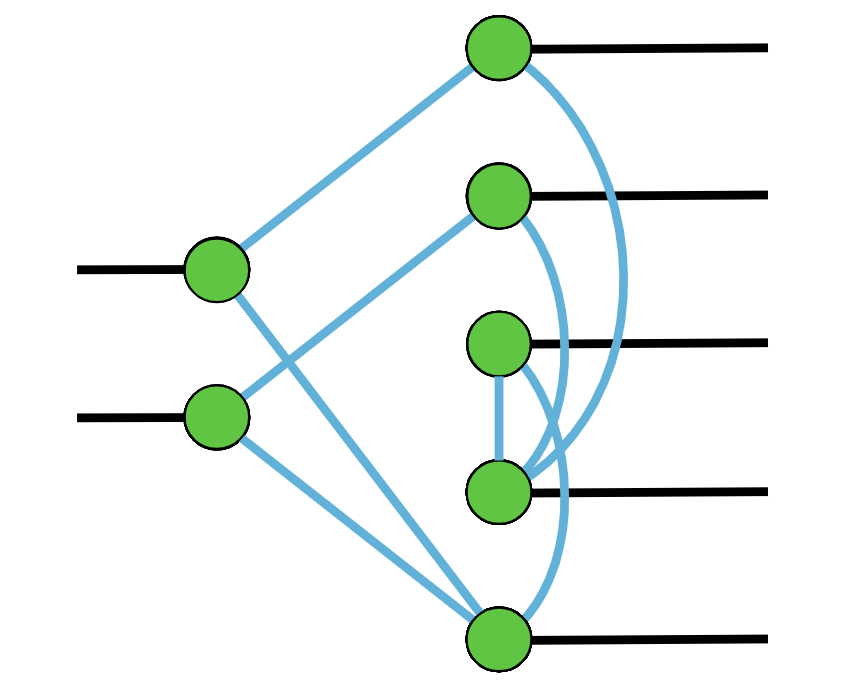} &                                 \includegraphics[scale=0.08]{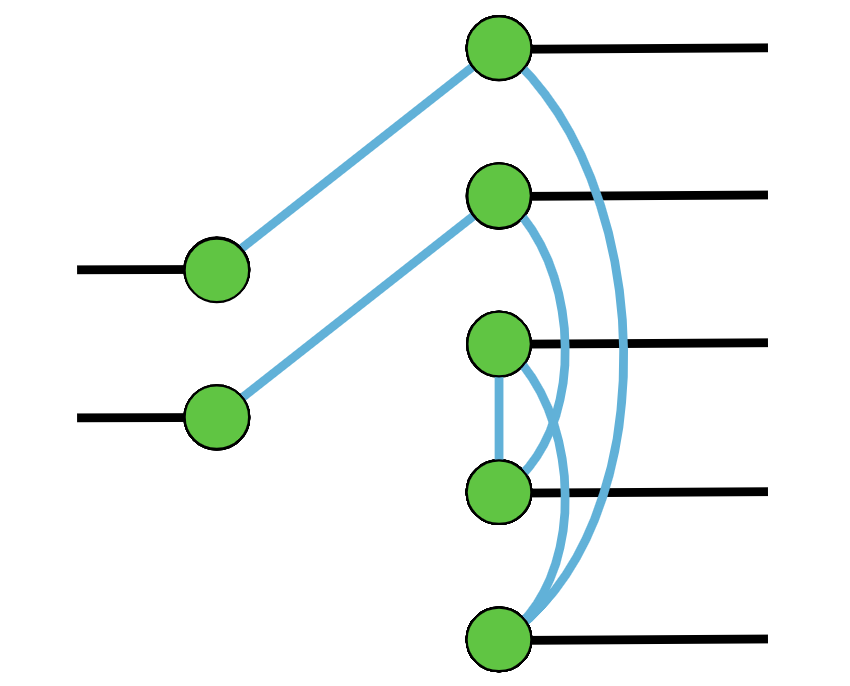} &                      \includegraphics[scale=0.08]{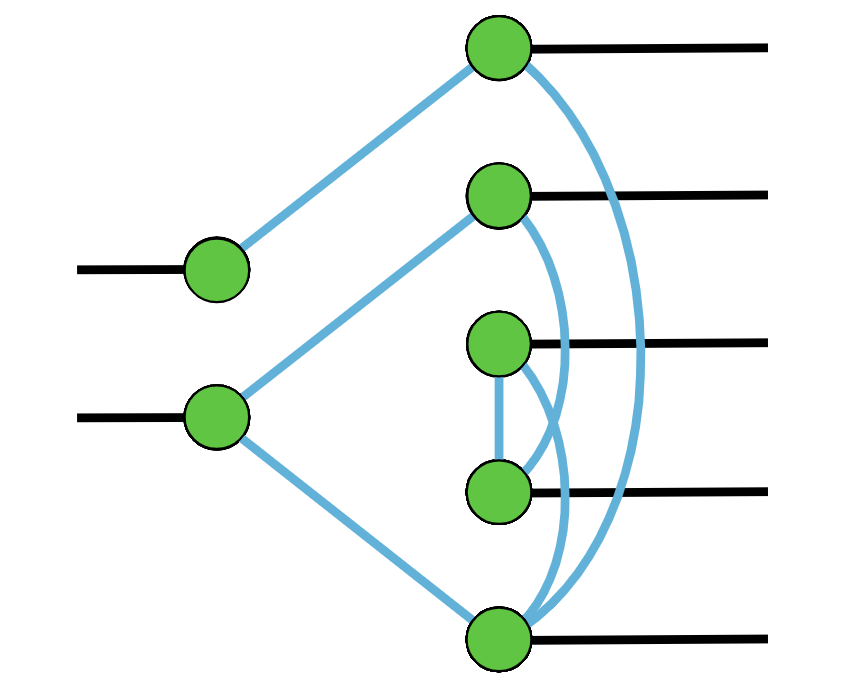} \\
        \hline
           1620 & 2268& 2484 & 4896 & 5184 & 5832\\\hline        
    \end{tabular}
    
\egroup

\bigskip

{\small (d) $\llbracket5,2\rrbracket$ codes equivalence classes showing the size of the class underneath a representative.
}
\end{center}
    \caption[$\llbracket n,2\rrbracket$ encoder graph equivalence classes]{(a) shows the number of equivalence classes for $\llbracket n,2\rrbracket$ encoder graphs. (b-d) show an element of the equivalence classes to denote the representative of the class and gives the size of the class. We only consider classes in which every graph is prime.}
    \label{coded codes}
\end{figure}

\section{Tabulations from code}
\label{sec: tabulations from code}

Keeping in mind the new definition of equivalence and the simplifications made on the set of Clifford encoder graphs being considered, we now sort the encoders into their equivalence classes. To do this, we used the {disjoint set algorithm} to split encoders into equivalence classes based on whether an operation from Conjecture~\ref{operations} caused one encoder to change into another.

Each encoder graph is converted into an integer based on the \textit{variable} edges present in the graph, which are the input-free edges, pivot-free edges, and free-free edges. Note that the input-pivot, input-input, and pivot-pivot edges are fixed, so these are not included among the variable edges.

The variable edges' values in the adjacency matrix are made into a single integer $a$ using binary representation. Note that this adjacency matrix includes all vertices, so it is a $(n+k)\times (n+k)$ matrix.

For example, in the $\llbracket 5,2\rrbracket$ codes, the $7\times 7$ adjacency matrix would look like the following:
\begin{equation*}
\begin{pmatrix}
    0 & 0 & 1 & 0 & a_{14} & a_{13} & a_{12} \\
    0 & 0 & 0 & 1 & a_{11} & a_{10} & a_{9} \\
    1 & 0 & 0 & 0 & a_{8} & a_7 & a_6 \\
    0 & 1 & 0 & 0  & a_5 & a_4 & a_3\\
    a_{14} & a_{11} & a_{8} & a_5 & 0 & a_2 & a_1\\
    a_{13} & a_{10} & a_7 & a_4 & a_2 & 0 & a_0\\
    a_{12} & a_{9} & a_6 & a_3 & a_1 & a_0 & 0\\
\end{pmatrix}
\end{equation*}

The top-left $4\times 4$ submatrix reflects the fixed input-pivot edges, as well as the lack of input-input edges and pivot-pivot edges.
The terms $a_i$ are the binary digits of $a$.
We place $a_0$ near the bottom-right corner and fill in the rows above from right to left.

After converting the ZX diagrams into integers, we use the {disjoint set algorithm}, which is useful for separating the whole set of possible encoder graphs into equivalence classes.

The code takes an integer representation, say $a$, of an encoder graph, performs one operation from Conjecture~\ref{operations} on the encoder graph, then merges the disjoint sets of $a$ and the integer representing the resulting encoder graph. All possible such operations are applied, and the resulting values are merged with $a$'s disjoint set. 

When a local complementation is performed on a free output, it is possible that an input-pivot edge is removed. Furthermore, some input-input edges could be added. To fix this, we first employ operation (5) from Conjecture~\ref{operations} to set all input-input edges to 0. Then, operations (3) and (4) are used to turn the submatrix representing the input-to-output adjacency matrix into RREF. Operation (2) is used to put the pivots back into their fixed positions, so they once again correspond to their input vertices.

The results of the code are shown in Figure~\ref{[n,1] codes big tables} and Figure~\ref{coded codes}. We have only shown the equivalence classes that have a single connected component, similar to the prime codes discussed in Section~\ref{sec: prime codes}.
The equivalence classes with multiple connected components can always be built up from connected components of smaller sizes, and finding these classes reduces to finding ways to partition the graph into groups of nodes within connected components. As shown in Figure~\ref{[n,1] codes big tables}(a) and Figure~\ref{coded codes}(a), the number of equivalence classes increases quickly as the number of output vertices increases.

For the $\llbracket n,1\rrbracket$ codes, the number of equivalence classes for $n=1$ through 4 are $(n-1)!$.
However, this pattern seems to break for larger values of $n$. Furthermore, for the $\llbracket n,2\rrbracket$ codes, the number of equivalence classes for $n=2$ through 5 are $(n-2)\cdot (n-2)!$. It is possible that these patterns for small $n$ arise from the number of ways to permute the non-pivot output nodes.

Furthermore, in Figure~\ref{[n,1] codes big tables}(b-e) and Figure~\ref{coded codes}(b-d), the equivalence classes show a variety of sizes, with many of the sizes having a factor of 3 or 9. The values of the sizes generally have many divisors, suggesting nice combinatorial patterns.

We make the following conjecture regarding the recurring factors of 3 and 9.

\begin{conjecture}
    For positive integers $n > k$, the equivalence classes of $\llbracket n,k\rrbracket$ codes with a single connected component have sizes divisible by $3^k$.
\end{conjecture}

Note that the factor of 9 is shared across the sizes of the equivalence classes of encoder graphs for the $\llbracket 3,2\rrbracket$, $\llbracket 4,2\rrbracket$, and $\llbracket 5,2\rrbracket$ codes. Similarly, there is a common factor of 3 across the sizes for the $\llbracket 2,1\rrbracket$, $\llbracket 3,1\rrbracket$, $\llbracket 4,1\rrbracket$, and $\llbracket 5,1\rrbracket$ codes. Also, based on the diagrams in Figure~\ref{[n,1] codes big tables}(b) and Figure~\ref{coded codes}(b), an analogous diagram can be drawn for $\llbracket 4,3\rrbracket$ codes, with the free output node having an edge to each pivot node.
The resulting diagram is in a class of size 27 because each input-pivot pair has three different ways to share edges with the free node. Note that this is also the only class with prime graphs among $\llbracket 4,3\rrbracket$ codes.
Thus, for $\llbracket n,n-1\rrbracket$ codes, there is a clear way to see why the analogously constructed graphs have a class size of $3^{n-1}$.
In general, we speculate that a factor of 3 arises from each input-pivot pair and the effects of different local Clifford gates among $\{I,Z,S,SZ,H,HZ\}$ when applied on the pivot free edge. If this holds, then all equivalence classes for $\llbracket n,k\rrbracket$ codes for $n > k$ would have sizes divisible by $3^k$.

\section{Equivalence classes with bipartite forms}
\label{sec:bipartite forms}

\begin{figure*}
    \includegraphics[scale=0.5]{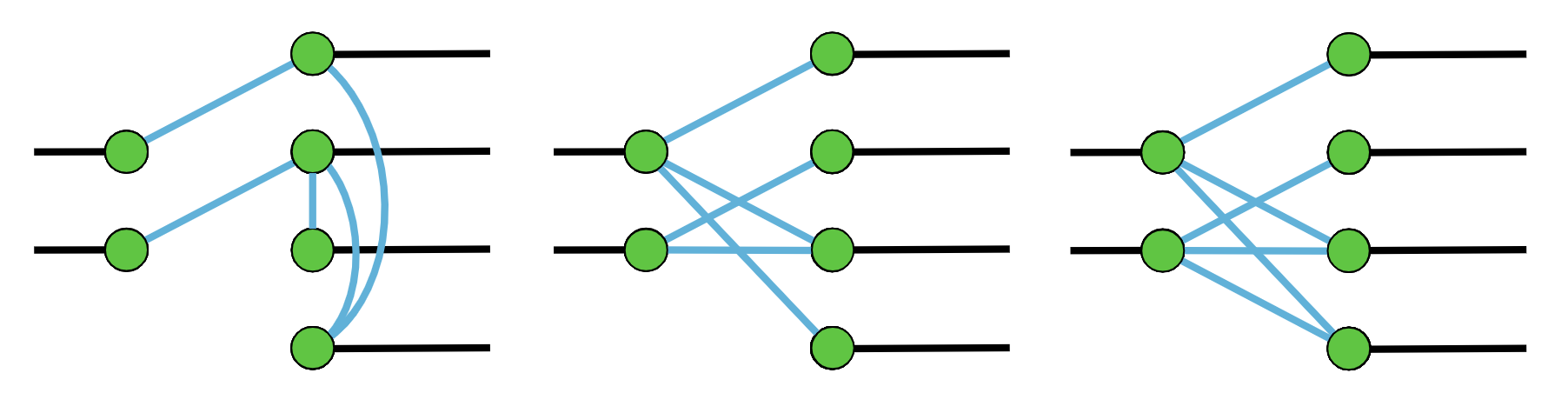}

    \caption[$\llbracket 4,2\rrbracket$ equivalence class example graphs]{Three graphs from the same $\llbracket 4,2\rrbracket$ code equivalence class. The representative graph is shown on the left. The other two graphs are both bipartite forms.}
    \label{diagram for 2-to-4}
\end{figure*}

When an equivalence class for an $\llbracket n,k\rrbracket$ code contains a bipartite graph, we can choose to narrow our search for a canonical form to bipartite forms. These graphs are simpler, with no edges among output nodes, making them good candidates as canonical forms for these equivalence classes. We now turn to setting criteria for selecting a specific bipartite form in an equivalence class that could serve as a simple representative.

First, consider the case of $\llbracket4,2\rrbracket$ codes. An example of the input-output adjacency matrix of a bipartite form, taking into account the RREF and pivot simplifications from Section~\ref{sec 4}, could look like:
\begin{equation*}
\begin{pmatrix}
    1 & 0 & \mathbf{1}&\mathbf{1}\\
    0 & 1& \mathbf{0} & \mathbf{1}\\
\end{pmatrix}.
\end{equation*}

By considering all possible combinations of 0's and 1's for the values of the bolded entries, there are 16 possible bipartite forms among all $\llbracket4,2\rrbracket$ codes. 

Per the discussion on prime codes in Section~\ref{sec: prime codes}, we will be considering codes that cannot be disconnected. In this section, codes that cannot be disconnected are called \textit{connected}.

\begin{proposition}
    \label{2-to-4 special}
    Among all connected $\llbracket4,2\rrbracket$ bipartite ZX diagrams, there is only 1 distinct diagram up to equivalence through operations 2 through 4 from Conjecture~\ref{operations}.
\end{proposition}

\begin{proof}
    There are 16 total possible input-output adjacency matrices for $\llbracket4,2\rrbracket$ bipartite codes, since there are $2^2$ ways for each input to connect to free outputs. If there are two or fewer input-to-free output edges, then either an output node is alone or the input nodes are in separate connected components.
    
    Then, there are only 5 possible input-output adjacency matrices in this case:
    \begin{align*}&\begin{pmatrix}
        1 & 0 & 1 & 1\\
        0 & 1 & 1 & 0\\
    \end{pmatrix}, \begin{pmatrix}
        1 & 0 & 1 & 1\\
        0 & 1 & 0 & 1\\
    \end{pmatrix},\\ &\begin{pmatrix}
        1 & 0 & 1 & 0\\
        0 & 1 & 1 & 1\\
    \end{pmatrix},
    \begin{pmatrix}
        1 & 0 & 0 & 1\\
        0 & 1 & 1 & 1\\
    \end{pmatrix},\\ &\begin{pmatrix}
        1 & 0 & 1 & 1\\
        0 & 1 & 1 & 1\\
    \end{pmatrix}.\end{align*}

    Consider the first matrix above. Switching the third and fourth output vertices (corresponding to the third and fourth columns of the matrix) results in 
    \begin{equation*}
    \begin{pmatrix}
        1 & 0 & 1 & 1\\
        0 & 1 & 0 & 1\\
    \end{pmatrix},
    \end{equation*}
    which is the second matrix. From here, we use operation 4 to replace the first row with the sum of the first and second rows modulo 2:
    \begin{equation*}
    \begin{pmatrix}
        1 & 1 & 1 & 0\\
        0 & 1 & 0 & 1\\
    \end{pmatrix}.
    \end{equation*}
    Permuting the outputs achieves the third and fourth matrices. Lastly, starting from the above matrix, we use operation 4 to replace the second row with the sum of the current first and second rows modulo 2 to find
    \begin{equation*}
    \begin{pmatrix}
        1 & 1 & 1 & 0\\
        1 & 0 & 1 & 1\\
    \end{pmatrix}.
    \end{equation*}
    This can be permuted to give the fifth matrix. Analogous sequences of operations can bring any of the other matrices to another, so all 5 of the graphs are equivalent, as desired.
\end{proof}

From Proposition~\ref{2-to-4 special}, the representative form for the equivalence class that contains these 5 adjacency matrices is chosen to be 
\begin{equation*}
\begin{pmatrix}
    1 & 0 & 1 & 1 \\
    0 & 1 & 1 & 0\\
\end{pmatrix}.
\end{equation*}

Note that this graph has the property of having the fewest edges. We can break ties between the four matrices with the fewest edges by selecting the graph with more edges on the first input and first output.

By writing one of these input-to-output adjacency matrices as a full $6\times 6$ matrix, we produce its ZX diagram. In Figure~\ref{diagram for 2-to-4}, the ZX diagram of the above matrix is shown in the center and the representative of the diagram's equivalence class from Figure~\ref{coded codes} is shown to the left. The rightmost ZX diagram is another bipartite form in the same equivalence class.

Now, we present a general method of simplifying a bipartite $2\times n$ input-to-output adjacency matrix. 

\begin{figure}
    \centering
    \includegraphics[scale=0.5]{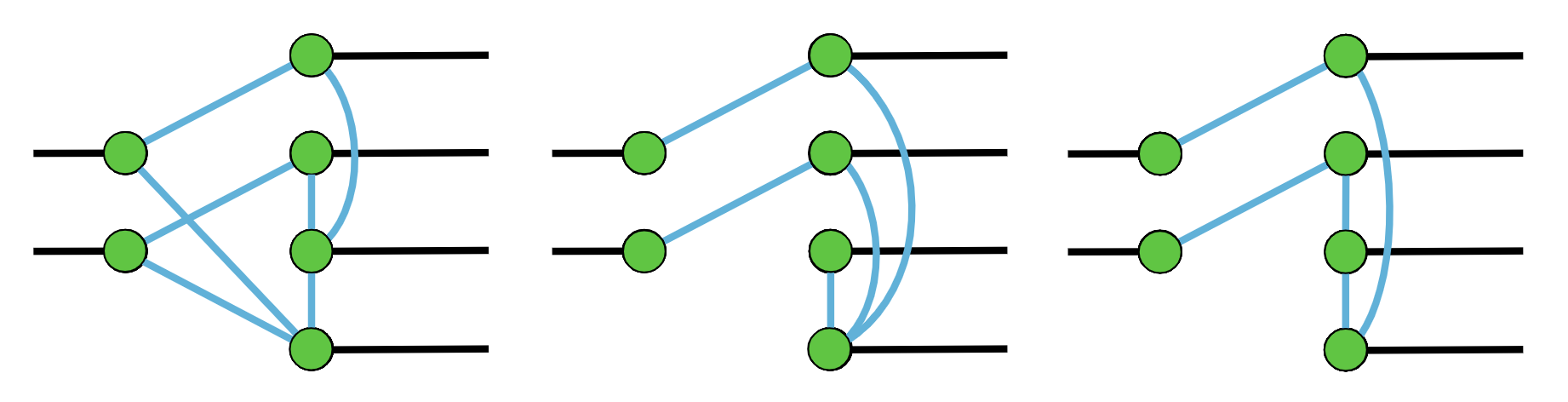}

    \caption[Non-bipartite $\llbracket4,2\rrbracket$ codes]{The representatives of the three different equivalence classes lacking bipartite forms among the classes for $\llbracket4,2\rrbracket$ ZX diagrams.}

    \label{small non-disjoint}
\end{figure}

\begin{proposition}
    \label{minimum bound 2-to-n}
    Consider a $\llbracket n,2\rrbracket$ encoder graph equivalent to some bipartite form. It is also equivalent to a bipartite form where the two inputs are both connected to at most $\big\lfloor\frac{n}{2}  - 1\big\rfloor$ of the same free outputs.
\end{proposition}

\begin{proof}In this proof, we only consider encoders that are equivalent to a bipartite form.

For the sake of contradiction, suppose all bipartite forms of this equivalence class have at least $\big\lfloor\frac{n}{2} \big\rfloor$ shared free outputs. In an input-to-output adjacency matrix, this would look like
\begin{equation*}
\begin{pmatrix}
    1 & 0 & 1 & 1 & 1 & 1 & 0 \\
    0 & 1 & 1 & 1 & 1 & 0 & 1\\
\end{pmatrix}.
\end{equation*}

The first two columns are fixed to be input-pivot edges, as usual. If there are at least $\big\lfloor\frac{n}{2} \big\rfloor$ shared free outputs, the other $n-2$ columns must have a majority of columns containing two 1's. In this example, 3 out of 5 columns contain two 1's.

However, using operation 4 from Conjecture~\ref{operations}, the top row can be replaced with the sum of the top and bottom row modulo 2. 

Note that this means all the free outputs that were shared by both inputs have their edges with the first input disconnected, so at least $\big\lfloor\frac{n}{2} \big\rfloor$ columns do not have two 1's. 

Furthermore, after the operation, the first column cannot possibly have two 1's, so one additional column does not have two 1's. The example matrix above turns into
\begin{equation*}
\begin{pmatrix}
    1 & 1 & 0 & 0 & 0 & 1 & 1 \\
    0 & 1 & 1 & 1 & 1 & 0 & 1\\
\end{pmatrix}.
\end{equation*}
We can rearrange output vertices to bring back the pivots. The matrix is thus equivalent to
\begin{equation*}
\begin{pmatrix}
1 & 0 & 1 & 1 & 1 & 0 & 0 \\
0 & 1 & 1 & 1 & 0 & 1 & 1 \\
\end{pmatrix}.
\end{equation*}
We can rearrange output vertices to bring back the pivots. The matrix is thus equivalent to

Thus, the maximum number of columns with two 1's is now $n - 1 - \big\lfloor\frac{n}{2}\big\rfloor$. However,
\begin{equation*}
    n - 1 - \Big\lfloor\frac{n}{2} \Big\rfloor < \Big\lfloor \frac{n}{2}\Big\rfloor,
\end{equation*}
so we reach a contradiction, since there are now fewer than $\big\lfloor \frac{n}{2}\big\rfloor$ shared free outputs in an equivalent bipartite form. 

Therefore, the claim holds.
\end{proof}

Proposition~\ref{minimum bound 2-to-n} demonstrates that we can choose a bipartite form that has a relatively small number of shared free outputs. In fact, if the top row is the horizontal vector $\textbf{a}$ and the bottom row is the horizontal vector $\textbf{b}$, by linear operations, there are only 3 possibilities of unordered combinations of two vectors in the rows. They are $(\textbf{a}, \textbf{b})$, $(\textbf{a+b}, \textbf{b})$, and  $(\textbf{a+b}, \textbf{a})$. Lastly, we choose which of these bipartite forms has the fewest shared free outputs and thus minimizes this number.

In an attempt to show the uniqueness of the bipartite forms in an equivalence class, we conjecture the following, which would allow for an efficient way to check whether two bipartite forms are equivalent.

\begin{conjecture}
\label{conjecture: bipartite}
    In an equivalence class that contains bipartite forms, all bipartite forms in the class can be transformed from one to another using only output permutations, input permutations, and row operations on the input-output adjacency matrix.
\end{conjecture}

One straightforward approach starts by assuming for the sake of contradiction that two bipartite graphs, $G_1$ and $G_2$, are equivalent even though they cannot be transformed from one to another using only the three operations in Conjecture~\ref{conjecture: bipartite}. A contradiction could be reached if we are able to show that the entanglement between two of the nodes is different in the diagrams. Quantifying the entanglement could be possible using the partial trace.

\begin{figure}[ht]
\centering
\includegraphics[scale=0.40]{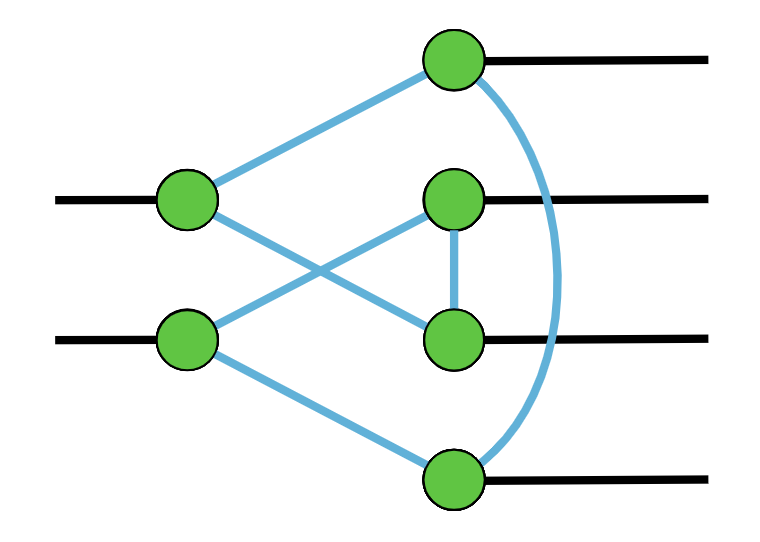}
\caption[Minimal diagram of $\llbracket4,2\rrbracket$ equivalence class]{A ZX diagram with the minimum number of output-output edges in its equivalence class. The equivalence class has no bipartite forms and does not contain any ZX diagrams with 0 output-output edges or 1 output-output edge.}
\label{new 300 graph}
\end{figure}

Besides equivalence classes with bipartite forms, there are also some classes that do not have bipartite forms. Among the equivalence classes of $\llbracket4,2\rrbracket$ codes from Figure~\ref{coded codes}, five of the classes have zero bipartite forms. Three of these classes have connected graphs, and they are shown in Figure~\ref{small non-disjoint}. Finding a clean, representative form for these classes is less intuitive, but analyzing a few other graphs in these equivalence classes could give a clue as to what to choose. For example, Figure~\ref{new 300 graph} shown above could be a better representative for the equivalence class containing the leftmost diagram in Figure~\ref{small non-disjoint}. Figure~\ref{new 300 graph} is symmetric, and it has the fewest output-output edges.

\section{Conclusion}

This chapter presented our work on producing the KL canonical forms for CSS codes, extending the work done in Chapter~\ref{chapter:qcodes-graphs}. Furthermore, we showed the resulting KL canonical forms of the toric code and certain surface codes. Additionally, we introduced the notion of prime codes and proved the Fundamental Theorem of Clifford Codes. We tabulated results found when considering codes with much looser equivalence conditions, and we analyzed possible representative forms, such as bipartite forms, for the equivalence classes found.
 \chapter{Conclusion}

In this thesis we improved the results of \textcite{garcia2017geometry} by deriving a simpler graph merging formula and by completely characterizing linearly dependent triplets of stabilizer states.
Both our characterization in terms of extended graph states and in terms of inner products reveal much structure in the additive properties of stabilizer states that can possibly be generalized.
Future work can continue characterizing linearly dependent $n$-tuples of stabilizer states for $n\ge 4$ as well as stabilizer decompositions of magic states~\cite{bravyi2019simulation,bravyi2016-stabrank, qassim2021improved, peleg2022lower}, using the graph formalism.

Our results on stabilizer states, general Clifford codes, and CSS codes find an elegant, minimal form in the ZX-calculus for these quantum operations. From Kissinger~\cite{kissinger2022phase}, it was known that phase-free ZX diagrams are CSS codes, but now we have shown an optimal and unique representation, reducing the number of nodes to have one per input and output.
In code search, it is sufficient to consider just one representative from each equivalence class, meaning a sufficiently good classification of canonical forms can massively improve our ability to discover new codes.

Additionally, we showed that graph properties are intimately tied to code properties. Key code properties such as weight, distance, encoding and logical operator circuit depths are all controlled by small linear functions of graph degrees.
Furthermore, key coding algorithms such as distance approximation, weight reduction, decoding are all strategies on various instances of quantum lights out (QLO) games.
The former observation allowed us to leverage graph algorithms to efficiently construct encoding circuits with controlled depth upper bounds as well as build logical operation circuits whose depths, in the case of diagonal or certain Clifford gates, are conceptually distinct and quantitatively better than simply unencoding, applying, and re-encoding. The latter result led us to design an efficient greedy QLO strategy for decoding, and prove that it succeeds up to a small constant factor on at least a certain class of graphs which include the hypercube code.
We believe that further study in graph-based decoding algorithms will reveal efficient decoders for many other families of graphs. Altogether, the universal graph representation enables insights and improvements that are general across all or large classes of stabilizer codes by providing avenues of construction and analysis that are substantially less obvious via conventional representations of additive quantum codes.

There are several directions in which to explore from this point. The first is general. Classically, graph representations have shown a strong correspondence between graph properties and code properties beyond simple objects like the degree. Notably, spectral expander graphs have led to good classical codes. A characterization effort connecting graph properties such as spectrum, expansion, cycle sizes, density, geometry, and topology to distance, decoding, stabilizer weight, etc. could significantly extend our insights into stabilizer codes.
It is also worthwhile to more deeply study the quantum lights out game, as insights onto their strategies have equivalent interpretations on coding algorithm performance. First steps towards the worst- and average-case analysis of coding algorithms based on the hardness of certain computational problems have appeared recently~\cite{poremba2024learning}, and their connection to QLO is an interesting question.

A second direction is directly constructive. Given particular architectures or algorithms of experimental interest, is it possible to systematically design graphs whose codes are particularly suitable for such devices or algorithms in terms of locality of connectivity, geometry, rate, and distance?

The third direction involves refining the approaches used in this thesis.
If would be useful to investigate if there exists a collection of graphs that, in exchange for representing a restricted subset of stabilizer codes, have low-weight stabilizers.
Such an approach may open new paths forward for practical code constructions and corresponding algorithms and would lie in an optimal middle ground between the approach of surface-type code constructions and the universal graph representation.
Additionally, the behaviour of graphs, and more importantly, graph codes, under logical and physical non-Clifford operations could be studied more closely.
In particular, there likely exist elegant formulas for manipulating linear combinations of graphs to apply non-Clifford gates, allowing for the study of universal fault-tolerant quantum computation.

%%% Appendices of thesis  %%%%%%%%%%%%%%%%%%%%%%%%%%%%%%%%%%%%%%%%%%%%%%%%%%%%%%%%%%%%%%%%%%%%%%%%

\appendix

\chapter{Stabilizer rank of magic states}
\label{app:magic-state-stab-rank}

Here we discuss a first attempt at finding upper bounds on the stabilizer rank of $n$-qubit magic states by using the graph formalism.
Magic states are special non-Clifford states which allow one to implement non-Clifford gates in a fault tolerant manner~\cite{bravyi2005-magicstate}.

\begin{definition}
An \textit{n-qubit magic state} $\ket{T_n}$ is the state $\ket{T}^{\otimes n}$, where $T\coloneq\frac{\ket{0}+e^{\frac{\pi i}{4}}\ket{1}}{\sqrt{2}}$.
\end{definition}

We recall from Definition~\ref{def:stabrank} that the $\textit{stabilizer rank}$ $\chi(\ket{\psi})$ of a state $\ket{\psi}$ is the smallest integer $\chi$ such that there exists a set of $\chi$ stabilizer states $S$ such that $\ket{\psi}\in \text{span}(S)$.

The stabilizer rank of a magic state is deeply tied to the runtimes of classical simulations of quantum circuits and has been explored extensively \cite{bravyi2016-stabrank,peleg2022lower,bravyi2019simulation}.
In order to tighten the upper bounds on $\chi(\ket{T_n})$, a Metropolis-Hastings numerical search algorithm that applies random transformations to stabilizer states to maximize the projection of the magic state onto their span was developed \cite{bravyi2019simulation}.
An attempt at experimenting with different variations of this method did not yield any stabilizer decompositions for larger numbers of qubits, due to the extremely large search space.
Another method, introduced in \cite{qassim2021improved}, uses \textit{cat states} and \textit{contractions}.
The method produced the improved upper bounds on $\chi(\ket{T_n})$ by finding a rank 3 decomposition of the 6-qubit cat state to be found by inspection.
Since two cat states can be combined into a magic state, this allowed for a 6-qubit decomposition of $\chi(\ket{T_6})$.
Additionally, by combining rank 3 decompositions of cat states and tracing out a qubit to join, this gave an asymptotic upper bound equivalent to finding a rank 3 decomposition of $\chi(\ket{T_4})$.

We showcase these decompositions in the graph formalism, which we used to represent the $n$-qubit magic state, $(T\ket{+})^{\otimes n}$.
It is written as a linear combination of extended graph states written in the canonical form from Definition~\ref{def:state-canonical-form}.
By comparing pairs of extended graph states, we can examine whether they can be merged together into a single state by using Theorem~\ref{thm:stab-triples}.
While we were able to express the $n$-qubit magic state as a sum of $2^{\frac{n}{2}}$ stabilizer states, it then proved impossible to simplify further simply by combining two stabilizer states into one.
Future work could try to develop more general merging criteria and formulas involving more than three extended graph states.
Then, a computer might be able to follow a heuristic with simulated annealing or any other optimization algorithm.
In order to find such a heuristic, it would be useful to study properties of sums of extended graph states that provide more insights into the structure of low-rank stabilizer decompositions. 

To find patters in known optimal stabilizer decompositions, we convert the stabilizer decompositions found in \cite{bravyi2016-stabrank,qassim2021improved} into canonical form.

We provide examples of stabilizer decompositions for two special cases, $n=3$ and $n=6$.
Let $S_{n,v}$ be the star graph on $n$ vertices with central vertex $v$, where the only edges are those of the form $(u,v)$ for all vertices $u\neq v$. Let $K_n$ be the complete graph on $n$ vertices and let $I_n$ be the empty graph on $n$ vertices.

\begin{align}
    \ket{T_3}=&\frac{i-e^{\frac{i\pi}{4}}}{2}Z_1Z_2Z_3\ket{I_3}+\frac{1+e^{\frac{i\pi}{4}}}{2}H_1H_2S_3\ket{S_{3,3}}-\frac{i+e^{\frac{i\pi}{4}}}{2}Z_1Z_2Z_3\ket{K_3}\\
    \ket{T_6}=&
    -\frac{i}{2}H_1Z_1S_2S_3S_4S_5S_6\ket{S_{6,1}}
    +\frac{e^{\frac{3i\pi}{4}}}{2}H_1\ket{S_{6,1}}\nonumber\\&
    +\frac{e^{\frac{i\pi}{4}}}{2\sqrt{2}}
    H_1Z_1
    H_2
    H_3
    H_4
    H_5
    Z_6
    \ket{S_{6,6}}
    +\frac{1}{2\sqrt{2}}H_1H_2H_3H_4H_5S_6Z_6\ket{S_{6,6}}\nonumber\\&
    -\frac{1}{2}H_1
    S_2Z_2
    S_3Z_3
    S_4Z_4
    S_5Z_5
    S_6Z_6\ket{K_6}
    -\frac{e^{\frac{i\pi}{4}}}{2}H_1Z_2Z_3Z_4Z_5Z_6\ket{K_6}
\end{align}

Even for the $3$-qubit case, there are multiple ways to decompose the magic state into $3$ stabilizer states, yet in all of the ways, there seems to be an empty graph, a complete graph, and a star graph.
In fact, the star graph can be transformed into a complete graph through local complementation.
The same is true for star graphs in the expression for $\ket{T_6}$.
This means that a search for magic state decompositions could just involve applying various local Clifford operators to complete and empty graphs.

For completeness, we give the representations of the 3-qubit and 6-qubit magic states using only empty and complete graphs.
Note that these are not in canonical form.

\begin{align}
    \ket{T_3}=&\frac{i-e^{\frac{i\pi}{4}}}{2}Z^{\otimes 3}\ket{I_3}
    +\frac{i+e^{\frac{i\pi}{4}}}{2}(HSZ)^{\otimes 3}\ket{K_3}\nonumber\\&
    -\frac{i+e^{\frac{i\pi}{4}}}{2}Z^{\otimes 3}\ket{K_3}\\
    \ket{T_6}=&
    -\frac{1}{2}S_1H_1(Z)^{\otimes 6}\ket{K_6}
    +\frac{e^{\frac{3i\pi}{4}}}{2}
    S_1H_1S_1(SZ)^{\otimes 6}\ket{K_6}\nonumber\\&
    -\frac{i}{2\sqrt{2}}S_1X_1(HS)^{\otimes 6}
    \ket{K_6}
    +\frac{e^{\frac{-i\pi}{4}}}{2\sqrt{2}}
    (HS)^{\otimes6}\ket{K_6}
    \nonumber\\&
    -\frac{1}{2}H_1S_1(SZ)^{\otimes6}\ket{K_6}
    -\frac{e^{\frac{i\pi}{4}}}{2}H_1Z_1(Z)^{\otimes6}\ket{K_6}
\end{align}
\chapter{Proofs of transformation rules}
\label{app:proofs-psipq-table}

Here we prove the expressions for $\ket{\psi_{PQ}}$ in Table~\ref{table:psipq-expressions}.
Recall that we have $\ket{\psi_{PQ}}=\frac{1}{2}\pars{(I+P_u)I_v+(I-P_u)Q_v}\ket G$ from Equation~(\ref{eq:psi-pq}).
We also have the shorthand for sets of neighbours, $N_w=N(w)$ and $M_w=N_w\cup\set{w}$.
Let $\mathds{1}_{u,v}$ be an indicator function equal to 1 if $u$ is connected to $v$ in $G$ and 0 otherwise.
In these proofs we also make use of graph stabilizers and the operator $\Delta$ for the symmetric difference.
When applying Theorem~\ref{thm:graphmerge} and Theorem~\ref{thm:graphmerging-special}, we show the sets $A$ and $B$ used.

The case of $\ket{\psi_{ZZ}}$ is just the toggling of an edge.
The cases of $\ket{\psi_{ZX}}$ and $\ket{\psi_{XX}}$ are derived in \cite{elliott2009graphical}, but we include the proofs of these cases for completeness.

\begin{lemma}[\textcite{elliott2009graphical}]
\label{xz}
    \begin{equation}
    \ket{\psi_{ZX}}=\pars{\prod_{w\in N_v}CZ_{u,w}}\ket{G}.
    \end{equation}
\end{lemma}
\begin{proof}
 Let $G'$ be the graph formed from $G$ where all the edges incident to $u$ are removed.
 
\begin{align}
\ket{\psi_{ZX}}
&=\frac{1}{2}
\pars{\underbrace{(I+Z)_u\ket{G}}_{A=M_u,\, B=\{u\}}
+X_v\underbrace{(I-Z)_u\ket{G}}_{A=M_u,\, B=\{u\}}}\\
&=\frac{1}{\sqrt{2}}H_u\ket{G'}
+X_v\frac{1}{\sqrt{2}}H_u\left(\prod_{w\in M_u}Z_w\right)\ket{G'}\\
&=\frac{1}{\sqrt{2}}
H_u\underbrace{\left(I+
(-1)^{\mathds{1}_{u,v}}
\pars{\prod_{w\in M_u}Z_w}
\pars{Z_u^{\mathds{1}_{u,v}}\prod_{w\in N_v}Z_w}
\right)\ket{G'}}_{A=\{u\},\, B=M_u\, \Delta\, N_v\, \Delta\, \set{u}^{\mathds{1}_{u,v}}}\\
&=Z_u^{1+2\mathds{1}_{u,v}}\pars{\prod_{w\in M_u\, \Delta\, N_v}CZ_{u,w}}\ket{G'}
=\pars{\prod_{w\in N_v}CZ_{u,w}}\ket{G}
\end{align}
Here, $\set{u}^{\mathds{1}_{u,v}}$ denotes that the set is empty if $u$ and $v$ are not connected in $G$ and is equal to $\set{u}$ otherwise.
\end{proof}

\begin{lemma}
\label{lemma:yz}
\begin{equation}
    \ket{\psi_{YZ}}=S_vZ_v\pars{\prod_{w\in M_u}CZ_{v,w}}\ket{G}.
\end{equation}
\end{lemma}
\begin{proof}
Let $G'$ refer to the graph resulting from toggling all the edges between vertices in the set $M_u$ in $G$.

\begin{align}
\ket{\psi_{YZ}}&=\frac{1}{2}
\pars{(I+Y)_u
+Z_v(I-Y)_u}\ket{G}\\
&=\frac{1}{2}
\pars{(I-iZX)_u\ket{G}
+Z_v(I+iZX)\ket{G}}\\
&=\frac{1}{2}
\left(\underbrace{\pars{I-i\pars{\prod_{w\in M_u}Z_u}}\ket{G}}_{A=M_u}
+Z_v\underbrace{\pars{I+i\pars{\prod_{w\in M_u}Z_u}}\ket{G}}_{A=M_u}\right)\\
&=
\frac{1-i}{2}\pars{\prod_{w\in M_u}S_w}\ket{G'}
+\frac{1+i}{2}Z_v\pars{\prod_{w\in M_u}S^3_w}\ket{G'}\\
&=\frac{1+i}{2}Z_v\pars{\prod_{w\in M_u}S^3_w}
\underbrace{\left(I-iZ_v\prod_{w\in M_u}Z_w\right)\ket{G'}}_{A=M_u\, \Delta\, \{v\}}\\
&=Z_v\pars{\prod_{w\in M_u}S^3_w}
\pars{\prod_{w\in M_u\, \Delta\,  \{v\}}S_w}
\pars{Z^{\mathds{1}_{u,v}}_v\prod_{w\in M_u}CZ_{v,w}}
\ket{G}\\
&=S_vZ_v\prod_{w\in M_u}CZ_{v,w}\ket{G}
\end{align}
\end{proof}

\begin{lemma}[\textcite{elliott2009graphical}]
\label{lemma:xx}
If the edge $(u,v)$ is in the graph $G$,
\begin{equation}
    \ket{\psi_{XX}}=H_uH_vCZ_{u,v}\pars{\prod_{p\in M_u,\, q\in M_v}CZ_{p,q}}\ket{G}.
\end{equation}
Otherwise,
\begin{equation}
    \ket{\psi_{XX}}=\pars{\prod_{p\in N_u,\, q\in N_v}CZ_{p,q}}\ket{G}.
\end{equation}
\end{lemma}
\begin{proof}
If the edge $(u,v)$ is in the graph $G$, we apply Theorem~\ref{thm:edge-local-complementation}.
\begin{align}
    \ket{\psi_{XX}}&=\frac{1}{2}
    \pars{(I+X)_u
    +(I-X)_uX_v}\ket{G}\\
    &=\frac{1}{2}
    H_uH_v\pars{(I+Z)_u
    +(I-Z)_uZ_v}H_uH_v\ket{G}\\
    &=\frac{1}{2}
    H_uH_v\pars{(I+Z)_u
    +(I-Z)_uZ_v}Z_uZ_v
    \pars{\prod_{p\in M_u,q\in M_v}CZ_{p,q}}
    \ket{G}\\
    &=
    H_uH_vZ_uZ_v
    \pars{\prod_{p\in M_u,\, q\in M_v}CZ_{p,q}}
    \ket{\psi_{ZZ}}\\
    &=H_uH_vCZ_{u,v}\pars{\prod_{p\in M_u,\, q\in M_v}CZ_{p,q}}\ket{G}
\end{align}
If the edge $(u,v)$ is not in graph $G$, we consider the inner product of $\ket{\psi_{XX}}$ with a computational basis state $\ket s$ corresponding to an arbitrary bit string $s$.
Let $s_i$ be the $i^{\text{th}}$ bit of $s$. For any set $R$, let $s_R$ be the sum of values of $s_i$ for indices $i\in R$.
\begin{align}
    \braket{s}{\psi_{XX}}&=\frac12\braopket{s}{(I+X)_u
    +(I-X)_uX_v}{G}\\
    &=\frac12\braopket{s}{\pars{I+\prod_{w\in N_u}Z_w}
    +\pars{I-\prod_{w\in N_u}Z_w}\pars{\prod_{w\in N_v}Z_w}}{G}\\
    &=\frac12((1+(-1)^{s_{N_u}})
    +(1-(-1)^{s_{N_u}}))(-1)^{s_{N_v}}\braket{s}{G}\\
    &=\frac12(1+(-1)^{s_{N_u}}
    +(-1)^{s_{N_v}}-(-1)^{s_{N_u}+vs_{N_v}})\braket{s}{G}
\end{align}
This expression is equal to $\braket{s}{G}$ if either $s_{N_u}$ or $s_{N_v}$ is even, otherwise this expression is equal to $\braket{s}{G}$.
Similarly, we have that $\braopket{s}{\prod\limits_{p\in N_u,\, q\in N_v}CZ_{p,q}}{G}=(-1)^{s_{N_u}s_{N_v}}\braket{s}{G}$, which is equal to the expression above. This concludes the proof.
\end{proof}
\begin{lemma}
\label{lemma:xy}
If the edge $(u,v)$ is in the graph $G$,
\begin{equation}
    \ket{\psi_{YX}}=\frac{1-i}{\sqrt{2}}\left(\prod_{w\in M_u}S_w\right) H_u \pars{\prod_{w\in M_u\, \Delta\,  M_v}CZ_{u,w}}\ket{L_u(G)}.
\end{equation}
Otherwise,
\begin{equation}
    \ket{\psi_{YX}}=
    \pars{\prod_{w\in M_u\, \Delta\,  N_v}Z_w}
    \pars{\prod_{p,\, q\in M_u\, \Delta\, N_v}CS_{p,q}}
    \pars{\prod_{p,\, q\in M_u}CS_{p,q}}\ket{G}.
\end{equation}
\end{lemma}
\begin{proof}
Let $G'$ refer to the graph resulting from toggling all the edges between vertices in the set $M_u$ in $G$. If the edge $(u,v)$ is in the graph $G$, we have the following.
\begin{align}
\ket{\psi_{YX}}&=
\frac{1}{2}
    \pars{(I+Y)_u
    +(I-Y)_uX_v}\ket{G}\\
&=\frac{1}{2}\left( 
(I-iZ_uX_u)\ket{G}
+X_v(I+iZ_uX_u)\ket{G} \right)\\
&=\frac{1}{2}\left( \underbrace{
\pars{I-i\prod_{w\in M_u}Z_w}\ket{G}}_{A=M_u}
+X_v\underbrace{\pars{I+i\prod_{w\in M_u}Z_w}\ket{G}}_{A=M_u} \right)\\
&=\frac{1-i}{2}
\pars{\prod_{w\in M_u}S_w}\ket{G'}
+\frac{1+i}{2}X_v\pars{\prod_{w\in M_u}S_w^3}\ket{G'}\\
&=\frac{1-i}{2}\pars{\prod_{w\in M_u}S_w}
\pars{I+\pars{\prod_{w\in M_u}Z_w}Z_vX_v}\ket{G'}\\
&=\frac{1-i}{2}\pars{\prod_{w\in M_u}S_w}
\underbrace{\pars{I+\prod_{w\in N_v}Z_w}\ket{G'}}_{A=\set{u},\, B=N_v}\\
&=\frac{1-i}{\sqrt2}\pars{\prod_{w\in M_u}S_w}H_uZ_u
\pars{\prod_{w\in N_v}CZ_{u,w}}\ket{G'}\\
&=\frac{1-i}{\sqrt2}\pars{\prod_{w\in M_u}S_w}H_u
\pars{\prod_{w\in M_u\, \Delta\, N_v}CZ_{u,w}}\ket{L_u(G)}
\end{align}
If the edge $(u,v)$ is not in the graph $G$, we have the following.
\begin{align}
\ket{\psi_{YX}}&=
\frac{1}{2}
    \pars{(I+Y)_u
    +(I-Y)_uX_v}\ket{G}\\
&=\frac{1}{2}\left( 
(I-iZ_uX_u)\ket{G}
+X_v(I+iZ_uX_u)\ket{G} \right)\\
&=\frac{1}{2}\left( \underbrace{
\pars{I-i\prod_{w\in M_u}Z_w}\ket{G}}_{A=M_u}
+X_v\underbrace{\pars{I+i\prod_{w\in M_u}Z_w}\ket{G}}_{A=M_u} \right)\\
&=\frac{1-i}{2}
\pars{\prod_{w\in M_u}S_w}\ket{G'}
+\frac{1+i}{2}X_v\pars{\prod_{w\in M_u}S_w^3}\ket{G'}\\
&=\frac{1-i}{2}\pars{\prod_{w\in M_u}S_w}
\underbrace{
\pars{I+i\pars{\prod_{w\in M_u\, \Delta\, N_v}Z_w}}\ket{G'}
}_{A=M_u\, \Delta\, N_v}\\
&=\pars{\prod_{w\in M_u}S_w}
\pars{\prod_{w\in M_u\, \Delta\, N_v}Z_w}
\pars{\prod_{p,q\in M_u\, \Delta\, N_v}CS_{p,q}}\ket{G'}\\
&=
\pars{\prod_{w\in M_u\, \Delta\, N_v}Z_w}
\pars{\prod_{p,q\in M_u\, \Delta\, N_v}CS_{p,q}}
\pars{\prod_{p,q\in M_u}CS_{p,q}}\ket{G}
\end{align}
\end{proof}

\begin{lemma}
\label{lemma:yy}
If the edge $(u,v)$ is in the graph $G$,
\begin{equation}
    \ket{\psi_{YY}}=-i\pars{\prod_{p,q\in M_u}CS_{p,q}}\pars{\prod_{p,q\in M_v}CS_{p,q}}\ket{G}.
\end{equation}
Otherwise,
\begin{equation}
    \ket{\psi_{YY}}=\frac{1-i}{\sqrt{2}}\pars{\prod_{w\in M_u}S_w} H_u \pars{\prod_{w\in M_v}CZ_{u,w}}\ket{L_u(G)}.
\end{equation}
\end{lemma}
\begin{proof}
Let $G'$ refer to the graph resulting from toggling all the edges between vertices in the set $M_u$ in $G$. If the edge $(u,v)$ is in the graph $G$, we have the following.
\begin{align}
\ket{\psi_{YY}}&=
\frac{1}{2}
    \pars{(I+Y)_u
    +(I-Y)_uY_v}\ket{G}\\
&=\frac{1}{2}\left( 
(I-iZ_uX_u)\ket{G}
+Y_v(I+iZ_uX_u)\ket{G} \right)\\
&=\frac{1}{2}\left( \underbrace{
\pars{I-i\prod_{w\in M_u}Z_w}\ket{G}}_{A=M_u}
+Y_v\underbrace{\pars{I+i\prod_{w\in M_u}Z_w}\ket{G}}_{A=M_u} \right)\\
&=\frac{1-i}{2}
\pars{\prod_{w\in M_u}S_w}\ket{G'}
+\frac{1+i}{2}Y_v\pars{\prod_{w\in M_u}S_w^3}\ket{G'}\\
&=\frac{1-i}{2}\pars{\prod_{w\in M_u}S_w}
\pars{I-i\pars{\prod_{w\in M_u}Z_w}X_v}\ket{G'}\\
&=\frac{1-i}{2}\pars{\prod_{w\in M_u}S_w}
\underbrace{\pars{I-i\prod_{w\in M_v}Z_w}\ket{G'}}_{A=M_v}\\
&=-i\pars{\prod_{w\in M_u}S_w}
\pars{\prod_{p,q\in M_v}CS_{p,q}}\ket{G'}\\
&=-i\pars{\prod_{p,q\in M_u}CS_{p,q}}
\pars{\prod_{p,q\in M_v}CS_{p,q}}\ket{G}
\end{align}
If the edge $(u,v)$ is not in the graph $G$, we have the following.
\begin{align}
\ket{\psi_{YY}}&=
\frac{1}{2}
    \pars{(I+Y)_u
    +(I-Y)_uY_v}\ket{G}\\
&=\frac{1}{2}\left( 
(I-iZ_uX_u)\ket{G}
+Y_v(I+iZ_uX_u)\ket{G} \right)\\
&=\frac{1}{2}\left( \underbrace{
\pars{I-i\prod_{w\in M_u}Z_w}\ket{G}}_{A=M_u}
+Y_v\underbrace{\pars{I+i\prod_{w\in M_u}Z_w}\ket{G}}_{A=M_u} \right)\\
&=\frac{1-i}{2}
\pars{\prod_{w\in M_u}S_w}\ket{G'}
+\frac{1+i}{2}Y_v\pars{\prod_{w\in M_u}S_w^3}\ket{G'}\\
&=\frac{1-i}{2}\pars{\prod_{w\in M_u}S_w}
\underbrace{
\pars{I+\pars{\prod_{w\in M_u\, \Delta\, M_v}Z_w}}\ket{G'}
}_{A=\set{u},\, B=M_u\, \Delta\, M_v}\\
&=\frac{1-i}{\sqrt2}
\pars{\prod_{w\in M_u}S_w}H_uZ_u
\pars{\prod_{w\in M_u\, \Delta\, M_v}CZ_{u,w}}\ket{G'}\\
&=\frac{1-i}{\sqrt2}
\pars{\prod_{w\in M_u}S_w}H_u
\pars{\prod_{w\in M_v}CZ_{u,w}}\ket{L_u(G)}
\end{align}
\end{proof}

This completes the proofs of all of the expressions in Table~\ref{table:psipq-expressions}.
\chapter{Basics of the ZX-calculus}
\label{app:zx-calculus}

The ZX-calculus is a graphical language used for expressing quantum processes.
Specifically, the ZX-calculus consists of \textit{ZX diagrams}, tensor networks made of nodes and edges, with each node representing a specific tensor and each edge acting as an inner product.

Although the ZX-calculus can represent all tensors, we will use it for the representation of Clifford gates.
This appendix contains a basic introduction to the ZX-calculus.
We cover the basics of which nodes can be used in ZX diagrams as well as several \textit{rewrite rules}, diagrammatic transformations which allow us to manipulate and simplify ZX diagrams.
For our diagrams, we use the conventions as described in~\cite{backens2014zx}.

The $Z$ (or green) nodes and $X$ (or red) nodes represent tensors that can be used to represent quantum operations such as qubits, gates, or measurements.
Each node has a phase, with empty nodes representing a phase of 0.
A $Z$ node with $n$ inputs, $m$ outputs, and phase $\alpha$ is equivalent to the tensor $\ket{0}^{\otimes m}\bra{0}^{\otimes n} + e^{i\alpha}\ket{1}^{\otimes m}\bra{1}^{\otimes n}$.
An $X$ node with $n$ inputs, $m$ outputs, and phase $\beta$ is equivalent to the tensor $\ket{+}^{\otimes m}\bra{+}^{\otimes n} + e^{i\beta}\ket{-}^{\otimes m}\bra{-}^{\otimes n}$.
Generally, ZX diagrams ignore normalization and scalar factors.

A ZX diagram represents a tensor network where connections between nodes are inner products on the indices of the corresponding tensors. It is possible to convert efficiently between quantum circuits and ZX diagrams, with examples given in Section~\ref{subsec:encoding_circuit}, Section~\ref{subsec:dodecahedron}, Appendix~\ref{app:compiler}, and in Appendix~\ref{app:convert-zxcf-to-circuit}. Hadamard gates, represented by yellow boxes \zx{\zxH{}}, can be placed on edges between nodes or on free edges.
Another convention for Hadamard gates placed on edges is to colour the edge in blue rather than the default black. For our purposes, edges with Hadamards will usually be represented by blue edges while edges without Hadamards will be represented by black edges.

The fundamentals of the ZX-calculus allow us to express various commonly seen quantum operations.
For instance below, we see the initialization of a $\ket+$ state, a $Z$ gate, a $CX$ gate, and a $CZ$ gate.
In the $CX$ gate, the wire with the $Z$ node is the control and the wire with the $X$ node is the target.
Note that the dashed edge in the $CZ$ gate is blue, meaning it is Hadamarded.
If the colours of the $Z$ nodes in the $\ket+$ state and the $Z$ gate where changed to red, these would be $X$ nodes depicting the $\ket0$ state and the $X$ gate.

\begin{center}
\begin{tikzpicture}[scale=1.5]
	\begin{pgfonlayer}{nodelayer}
		\node [style=z-node-demo] (0) at (-0.5, 1) {};
        \node [draw] at (0, -0.25) {$\ket+$};
		\node [style=none] (1) at (0.5, 1) {};
		\node [style=none] (2) at (1.5, 1) {};
		\node [style=none] (3) at (3.5, 1) {};
		\node [style=z-node-demo] (4) at (2.5, 1) {$\pi$};
        \node [draw] at (2.5, -0.25) {$Z$};
		\node [style=none] (5) at (4.5, 1.5) {};
		\node [style=none] (6) at (4.5, 0.5) {};
		\node [style=none] (7) at (5.5, 1.5) {};
		\node [style=none] (8) at (5.5, 0.5) {};
		\node [style=z-node-demo] (9) at (5, 1.5) {};
		\node [style=x-node-demo] (10) at (5, 0.5) {};
        \node [draw] at (5, -0.25) {$CX$};
		\node [style=none] (11) at (6.5, 1.5) {};
		\node [style=none] (12) at (6.5, 0.5) {};
		\node [style=none] (13) at (7.5, 1.5) {};
		\node [style=none] (14) at (7.5, 0.5) {};
		\node [style=z-node-demo] (15) at (7, 1.5) {};
		\node [style=z-node-demo] (16) at (7, 0.5) {};
        \node [draw] at (7, -0.25) {$CZ$};
	\end{pgfonlayer}
	\begin{pgfonlayer}{edgelayer}
		\draw (0) to (1.center);
		\draw (2.center) to (4);
		\draw (4) to (3.center);
		\draw (15) to (13.center);
		\draw (15) to (11.center);
		\draw (16) to (14.center);
		\draw (12.center) to (16);
		\draw (10) to (8.center);
		\draw (10) to (9);
		\draw (5.center) to (9);
		\draw (7.center) to (9);
		\draw (6.center) to (10);
		\draw [style=h] (15) to (16);
	\end{pgfonlayer}
\end{tikzpicture}
\end{center}

Note that projection operations and measurements can result in a disconnected diagram.
For example, the Pauli projector $\ket+\bra+=\frac{I+X}2$ would be depicted as follows.

\begin{center}
    \begin{tikzpicture}
	\begin{pgfonlayer}{nodelayer}
		\node [style=z-node-demo] (0) at (-0.5, 1) {};
		\node [style=none] (1) at (0.5, 1) {};
		\node [style=z-node-demo] (18) at (-1.5, 1) {};
		\node [style=none] (19) at (-2.5, 1) {};
	\end{pgfonlayer}
	\begin{pgfonlayer}{edgelayer}
		\draw (0) to (1.center);
		\draw (18) to (19.center);
	\end{pgfonlayer}
\end{tikzpicture}
\end{center}

The ZX-calculus also has a set of basic rewrite rules that allow for diagrams to be converted into equivalent forms. Since the ZX-calculus is complete with respect to Clifford codes and states~\cite{backens2014zx}, these basic rewrite rules can be used to derive all other rewrite rules.

\begin{definition}[Basic rewrite rules~\cite{van2020zx}]
\label{def:zx-basic-rewrite-rules}
For each basic rewrite rule below, an analogous rule holds when the colours are interchanged, meaning all $Z$ and $X$ nodes are changed into $X$ and $Z$ nodes, respectively.
Note that all other rewrite rules are derived from these eight, but the ones shown are not a minimal set of rewrite rules~\cite{coecke2018picturing}. 
\begin{enumerate}[(1)]

    \item Merging/unmerging rule: Two $Z$ (or $X$) nodes with phases $\alpha$ and $\beta$ that are connected by edges with no Hadamards may be combined into a single $Z$ (or $X$) node with phase $\alpha + \beta$. The resulting node has all the external edges of the two original nodes. A node may also be unmerged into two nodes, with the partition of the external edges connected to each node done arbitrarily.
    \begin{center}
        \includegraphics[scale=0.30]{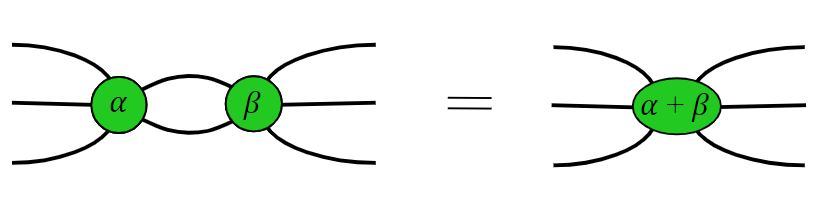}
    \end{center}

    \item Identity removal rule: $Z$ (or $X$) nodes with phase 0 and exactly two edges can be removed. Note that the two edges could both be input or output edges.
    \begin{center}
        \includegraphics[scale=0.40]{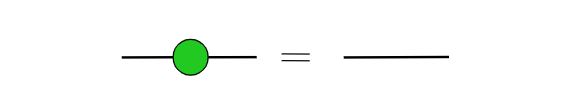}
    \end{center}

    \item Hadamard cancellation rule: Two Hadamard gates on one edge can be cancelled since $HH = I$.

\begin{center}
    \begin{tikzpicture}
	\begin{pgfonlayer}{nodelayer}
		\node [style=none] (0) at (-2, 0) {};
		\node [style=none] (1) at (1, 0) {};
		\node [style=h-box] (2) at (-1, 0) {};
		\node [style=h-box] (3) at (0, 0) {};
	\end{pgfonlayer}
	\begin{pgfonlayer}{edgelayer}
		\draw (0.center) to (2);
		\draw (2) to (3);
		\draw (3) to (1.center);
	\end{pgfonlayer}
\end{tikzpicture}
\bgroup \large = \egroup\raisebox{0.3em}{
\begin{tikzpicture}
	\begin{pgfonlayer}{nodelayer}
		\node [style=none] (0) at (-2, 0) {};
		\node [style=none] (1) at (1, 0) {};
	\end{pgfonlayer}
	\begin{pgfonlayer}{edgelayer}
		\draw (1.center) to (0.center);
	\end{pgfonlayer}
\end{tikzpicture}
}
\end{center}
    
    \item $\pi$-copy rule: A $Z$ (or $X$) node of phase $\pi$ slides through and copies onto all the other edges of an adjacent $X$ (or $Z$) node. The $X$ (or $Z$) node has its phase negated.
    Shown below is an example with a $Z$ node of phase $\pi$, which implements a $Z$ gate.
    \begin{center}
        \includegraphics[scale=0.55]{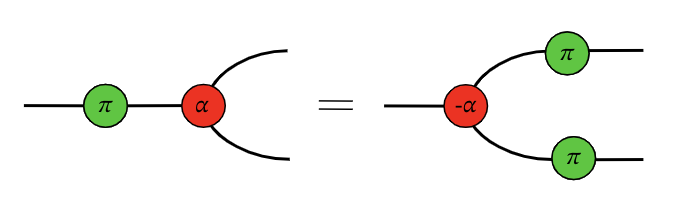}
    \end{center}

    \item State-copy rule: A $Z$ (or $X$ node) with a phase of $a\pi$ with $a \in \{0,1\}$ can be copied through onto all of the other adjacent edges of an $X$ (or $Z$ node) with any phase $\alpha$.
\begin{center}
     \includegraphics[scale=0.25]{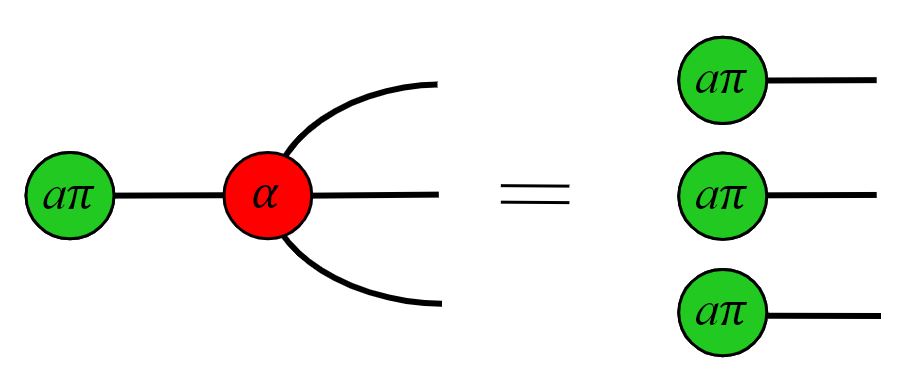}
    \end{center}

    \item Colour change rule: A $Z$ (or $X$) node can be exchanged for an $X$ (or $Z$) node if all edges adjacent to the node have Hadamards added to them.
    \begin{center}
        \includegraphics[scale=0.32]{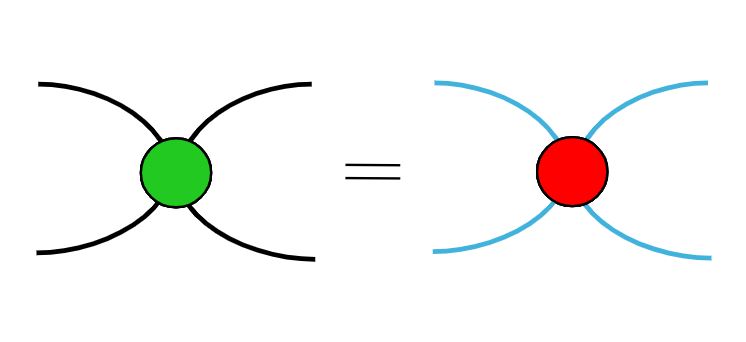}
    \end{center}

    \item Bialgebra rule: By acting on an edge between a $Z$ and an $X$ node, each external edge gets one node, and a complete bipartite graph is formed between these new nodes. An example is shown below. There may be one or more (rather than two) edges coming in from the left side of the graph, and there may be one or more edges exiting on the right side of the graph.
    \begin{center}
        \includegraphics[scale=0.25]{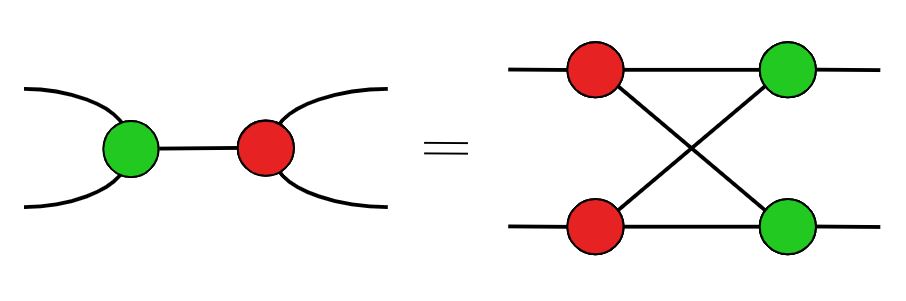}
    \end{center}

    \item Hopf rule: If a $Z$ node and $X$ node share multiple edges that have no Hadamards, two of these shared edges may be removed together.
    \begin{center}
        \includegraphics[scale=0.20]{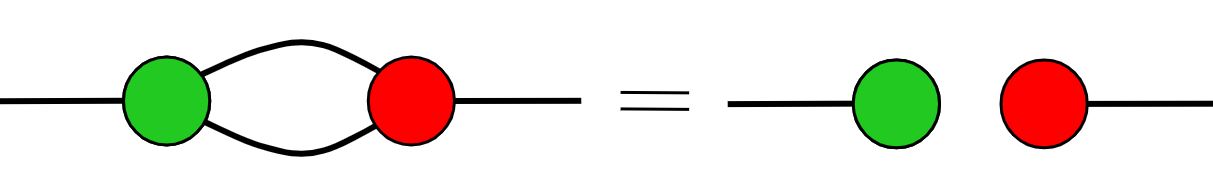}
    \end{center}
    
\end{enumerate}

\end{definition}

Additional rules that will be used later are given below.

\begin{definition}[Derived rewrite rules]
Below are two rules that can be derived from the set of basic rewrite rules above.
\label{def:zx-derived-rewrite-rules}

\begin{enumerate}[(1)]
       \item Loop rule: Self-loops on a node can be removed. If the self-loop has a Hadamard, then removing the loop adds a phase of $\pi$ to the node.
    \begin{center}
        \includegraphics[scale=0.27]{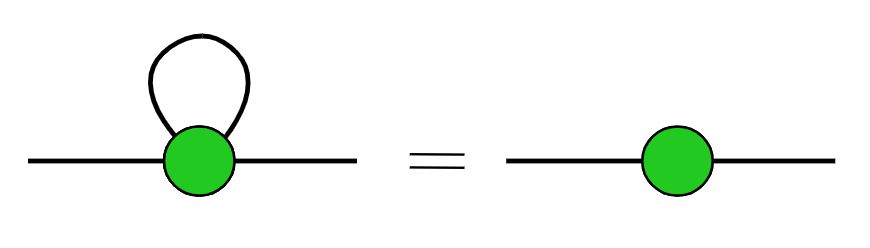}
    \end{center}
       
    \item Hadamard-sliding rule: This rule allows the colours of two adjacent vertices to be swapped while switching their sets of neighbours and toggling the edges between those two sets.
    \begin{center}
        \includegraphics[scale=0.30]{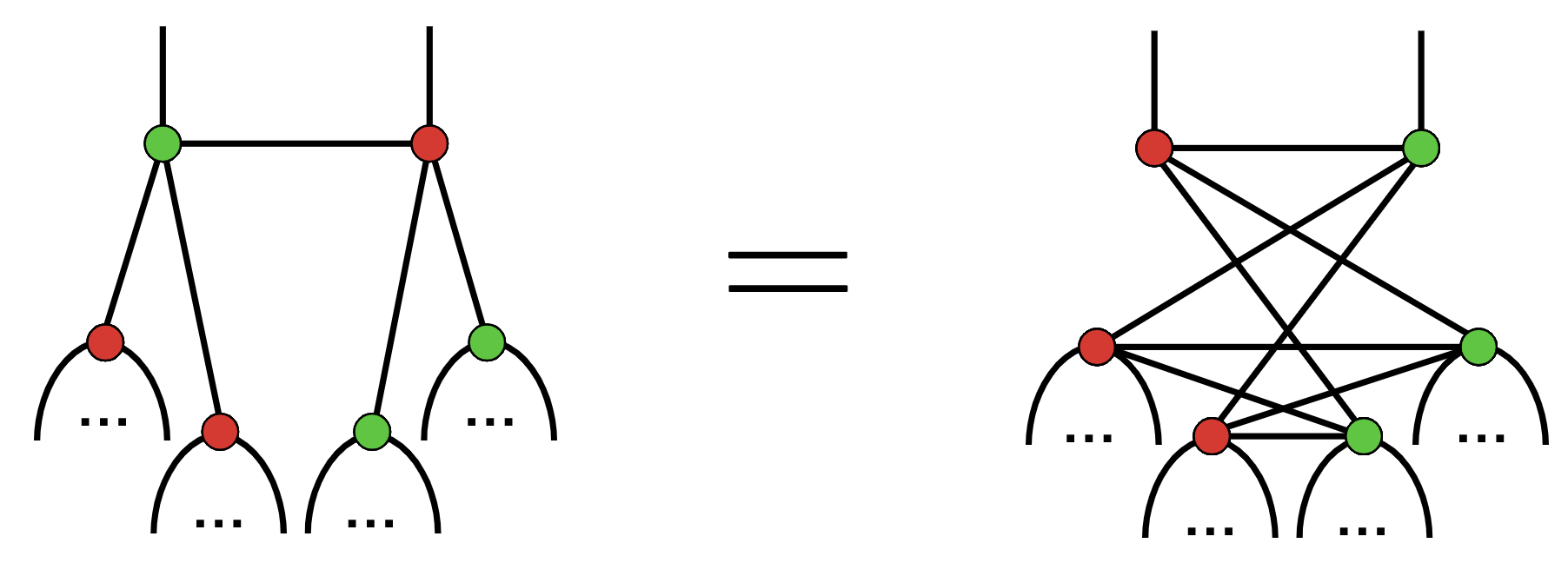}
    \end{center}
\end{enumerate}

We note that the Hadamard-sliding rule is a special case of Equation~\ref{eq:edge-local-complementation}.
The case is special since the neighbouring sets of the two top vertices do not intersect.

\end{definition}
\chapter{Tableau to graph proofs} 
\label{app:compiler}

In Chapter~\ref{chapter:qcodes-graphs} we briefly sketched the algorithm that maps encoders to ZXCFs. Here, we give a detailed explanation of this algorithm.
The first part of our algorithm maps the encoder into any equivalent ZX diagram.
We can then transform the diagram via a sequence of ZX equivalence rules into a ZXCF.
To accomplish a transformation into ZX diagrams, we first recall Theorem~\ref{thm:canonicalform-state} from Chapter~\ref{chapter:qstatesgraphs}.
The theorem proves the existence of a canonical form called HK form for stabilizer states.

Although we did not use ZX calculus to construct the HK form for states, our construction for codes uses graph states with local Clifford operators, which can be mapped directly to ZX calculus.

An HK diagram is constructed by starting with a graph state.
As per Definition~\ref{def:graphstate}, a graph state is a Clifford state corresponding to an undirected graph wherein the nodes become $\ket{+}$ qubits and the edges become controlled-$Z$ gates.
Each node is then endowed with a single \textit{free edge}, which does not connect to any other node.
Local Clifford operations in $\langle H,S,Z\rangle$ are then applied to the qubits of the graph state to create an HK diagram.

Such a construction has the capacity to express any stabilizer state. A direct mapping enables the presentation of a HK diagram into the ZX calculus.
\begin{enumerate}[(1)]
    \item HK vertices become $Z$ nodes, as a vertex in a graph state begins in the state $\ket+$, which is a $Z$ node with a single output.
    \item HK (non-free) edges become (non-free) edges with Hadamard gates, as both of these correspond to $CZ$ gates.
    \item HK local operations of $S$, $Z$, or $SZ$ become phases on the corresponding node of $\frac\pi2$, $\pi$, and $\frac{3\pi}2$, respectively. The local operation $H$ becomes a Hadamard gate on the node's free edge.
\end{enumerate}
Thus, we will effectively consider HK forms to be in the ZX calculus representation and denote them \textit{ZX-HK forms}. We remark that
transformations of similar decorated graph state families into ZX calculus presentations have been considered in other contexts, such as by \textcite{backens2014zx}.

Our constructions in Chapter~\ref{chapter:qstatesgraphs} are limited to stabilizer states. However, by vector space duality we can equivalently transform operators into states and vice versa. The formal statement of this duality in quantum information theory is given by the Choi-Jamio\l{}kowski isomorphism. The composition of such a transformation with the result above gives the following lemma.
\begin{lemma}
\label{lemma:ZXencoder_to_HK}

    There exists an efficient transformation of any Clifford ZX encoder diagram with $n-k$ input edges and $n$ output edges into a corresponding ZX presentation of a stabilizer state in ZX-HK form, with $n + (n-k) = 2n-k$ free edges.

\end{lemma}
\begin{proof}
We proceed by turning the encoder diagram into a Clifford circuit and then mapping the output of that circuit to HK form.
Any $X$ (red) nodes in the ZX encoder diagram can be transformed into $Z$ (green) nodes with the same phase surrounded by Hadamard gates. Now, we can interpret the ZX encoder diagram as a state by treating the input edges as additional output edges.
This means all free edges of the diagram are output edges so the diagram encodes a Clifford state.
This is the ZX version of the Choi-Jamio\l{}kowski isomorphism.

As shown by \textcite{backens2014zx}, any Clifford ZX diagram can be decomposed entirely in terms of the following operations.

First, we could have a $Z$ node with only one output or only one input. In a circuit, this is a $\ket+$ qubit or a post-selected $\bra+$ measurement, respectively.
We can also have a Hadamarded edge or a node with $\frac\pi2$ phase, which are represented as $H$ or $S$ gates in circuits, respectively.
Lastly, we can have a $Z$ node of degree 3, with either 2 inputs and 1 output or 1 input and 2 outputs that acts as a merge or a split.
This is equivalent to applying a $CX$ operation between two qubits where the second qubit is either initialized to $\ket0$ before the $CX$ gate or post-selected by the measurement $\bra0$ after the $CX$ gate, respectively.

If we need to compute the action of a post-selected measurement, that means we have an internal node, since the measuring node would not have any free edges.
To remove an internal node, we first try merging it with a neighbouring node if they share a non-Hadamarded edge.
If this is impossible, we consider the neighbours of the internal node, if those can be merged with other nodes, we do so until the internal node has a neighbour that cannot be merged further.
We then apply the bialgebra rule from Definition~\ref{def:zx-basic-rewrite-rules}
between the internal node an unmergeable neighbour.
Specifically, we apply the version of the bialgebra rule where all $X$ nodes have been replaced with $Z$ nodes and Hadamard gates.
This bialgebra operation creates several new nodes, but since the bialgebra rule changes the colours of the nodes connected to each free edge of the rule, this allows us to merge all new nodes that are not connected to a free edge of the diagram.
This means that we can use the bialgebra rule to remove any internal nodes, simplifying out the action of post-selection measurements in our circuit.

As we apply each of these operations we keep track of the current HK form for our state.
One concern we might have is that one of these operations will require us to implement a two-qubit gate across two-qubits that are spatially separated in the ZX diagram, and thus temporally separated in our circuit.
Since we can always express any diagram using the above operations, we only need to be able to apply the above operations to the current state expressed by the diagram.
Thus, repeatedly updating the HK form for the current state encoded by the circuit is all we have to do.
All the above operations get applied to the free edges of this state.
All the manipulations necessary to keep a state updated in HK form are detailed in Chapter~\ref{chapter:qstatesgraphs}.
In addition, these steps can all be done efficiently by Theorem~\ref{thm:canonical-algorithm}, so this gives us an efficient procedure for turning a ZX diagram into a corresponding stabilizer state in HK form.
\end{proof}

The application of Lemma~\ref{lemma:ZXencoder_to_HK} results in a ZX-HK form. That is, there is one $Z$ node per vertex of the graph in HK form, and any internal edge has a Hadamard gate on it.
If it would be possible to merge two nodes so that a node has multiple free edges, we add two Hadamard gates to one of the free edges and a phase-free $Z$ node in between the Hadamard gates.
These operations do not change the state and they can return a diagram to ZX-HK form.
An example of this is seen in Figure~\ref{fig:nine-qubit-code}(b)-(c).

Free edges can only have Hadamard gates on them if they satisfy both the condition that their phase is a multiple of $\pi$, corresponding to the local Clifford gates $H$ and $HZ$ in HK form, and the condition that the associated node is not connected to any lower-numbered nodes, as indexed by the output edges.
Nodes whose free edges have no Hadamard gate are free to have any multiple of $\frac\pi2$ as a phase, corresponding to local Clifford operations $I$, $S$, $Z$, or $SZ$, respectively.

We are now equipped with a ZX-HK diagram that represents a state. This diagram has only output edges. To return it into an encoder diagram, we partition the $2n-k$ free edges into $n-k$ input edges and $n$ output edges. In the circuit representation, this is equivalent to turning bra's into ket's and vice versa to map between an operator and a state. For example, $\ketbra{00}{1} - \ketbra{11}{0} \leftrightarrow \ket{001} - \ket{110}$. Now, if there are any edges between input nodes (those in $\CI$), we can simply remove them. They correspond to controlled-$Z$ operations, and we can take off any unitary operation on the input by the encoder definition. The same goes for local Cliffords on $\CI$. 

At this stage, the obtained ZX-HK form is in encoder-respecting form. It also satisfies the Edge and Hadamard rules---there are no input edges in ZX-HK form so the analogous rule is enforced only on lower-numbered nodes, but if we number the nodes from the beginning such that the input nodes are lower-numbered than all output nodes, then the transformed ZX-HK form will satisfy our ZXCF Hadamard rule. So, all that remains is to simplify the diagram to obey the RREF and Clifford rules.

\begin{lemma}
    An encoder diagram in ZX-HK form can be efficiently transformed to satisfy the RREF rule, while continuing to satisfy the Edge and Hadamard rules.
    \label{theorem:rref}
\end{lemma}

\begin{proof}
This proof relies on two basic properties of ZX calculus.
The first is that a $CX$ gate can be expressed on a pair of qubits by adding a pair of phase-free nodes connected by an edge, with a $Z$ node on the control and an $X$ node on the target.
The second fact is the bialgebra rewrite rule from Definition~\ref{def:zx-basic-rewrite-rules}, which allows us to rewrite a $Z$ node connected to an $X$ node by the following sequence of steps.
First, remove the two nodes and the edge between them.
Next, add an $X$ node on each edge previously leading into the $Z$ node; do the same for the $X$ node by adding $Z$ nodes.
Lastly, connect each new $X$ node to each new $Z$ node.

Applying a $CX$ gate to a pair of inputs is a unitary operation and does not change our code.
As discussed in Appendix~\ref{app:zx-calculus}, a $CX$ consists of adding a pair of connected nodes to existing qubit wires, a $Z$ node to the control wire and an $X$ to the target wire.
After applying a $CX$ gate this way, we perform the following sequence of steps.
The control $Z$ node of the $CX$ gate fuses with the first input.
The target $X$ node is used to perform the bialgebra rule with the second input.
This results in two $Z$ nodes where the $CX$ target used to be, one of these is on the free input edge and the other is on the edge connecting the target node to the first input.
The latter of these $Z$ nodes is merged with the first input while the $Z$ node on the free edge becomes the new input node for the second input.
Meanwhile, the $X$ nodes on each input-output edge previously connected to the second input now find themselves connected to both the first and second input.
Additionally, we can turn them into $Z$ nodes after commuting them through the Hadamard gate on their input-output edge, merging them with the outputs on the other side.
Any doubled edges are removed by the Hopf rule from Definition~\ref{def:zx-basic-rewrite-rules}.
The result is that we added the row of $M_{\mathcal{D}}$ corresponding to the first input too the row of the second.
This allows to do arbitrary row operations since $M_{\mathcal{D}}$ is an adjacency matrix and is evaluated mod 2.
This concludes the proof of the fact that $M_{\mathcal{D}}$ can be turned into RREF without affecting the Edge and Hadamard rules.
\end{proof}

The only remaining task is to enforce the Clifford rule. For any pivot node $p \in \CP$ with a non-zero phase, let $v\in\CI$ be its associated input node. We apply the effect of Equation~(\ref{eq:simplification-hsh}) $v$, which we can do freely since we are just applying $\sqrt{X}=HSH$, a unitary operation, to $v$, an input node.
Notably, this does not change the entries of $M_{\mathcal{D}}$.
Importantly, this operation increases the phases of each neighbours of $v_{\text{in}}$ by $\frac\pi2$ (due to multiplication by $S$).
We repeat process until the phase of $p$ vanishes.

The last step of the Clifford rule that we need to satisfy is to remove edges between pivots.
If there are any edges between pivots $p_1$ and $p_2$, we can remove them by applying the following procedure.
First, we apply the input unitary operation $H_{p_1}H_{p_2}CZ_{p_1,p_2}$.
This connects the input nodes and applies $H$ vertices to the inputs.
Next, we can now perform the transformation given by Equation~(\ref{eq:edge-local-complementation}).
In particular, one effect of this transformation is toggling the edge $(p_1,p_2)$ since each pivot is connected to its own input.
This operation will apply the input unitary $Z_{p_1}Z_{p_2}$, which can be removed as it does not change the encoder.
Additionally, this operations swaps the neighbours of $p_1$ and $p_2$, which is also an input unitary that we can correct with a SWAP gate.
We repeat until no pivot-pivot edges remain.

With that step, the transformation from Clifford encoder to ZXCF is complete. Although this process may seem lengthy, all of these steps can be done systematically in an efficient manner, without having to go back to fix earlier rules.
Specifically, this algorithm takes worst-case $O(n^3)$ time. Creating the HK diagram for the encoder requires us to apply at most $O(n)$ Pauli projections, each of which can be applied in at most $O(n^2)$ time. After the HK form is created, operations such as row-reducing a matrix can be done in $O(n^3)$ time. We suspect that the worst-case time complexity cannot be improved, but that there is a lot of room for heuristic runtime improvements and optimizations.
\chapter{Derivation of ZXCF counting recursion}
\label{app:recursion}

In this appendix we provide the derivation of the recursive counting formula of ZXCF diagrams, given in Equation~(\ref{eq:zxcf-count}).
To derive $f$, we provide the following visualization.
We start with no assigned nodes and we must sequentially assign the $n$ output nodes to be pivots (case $A$) or non-pivots (case $B$), where pivots must be matched with input nodes.
The current number of assigned pivots is tracked by $p$ and the number of non-pivot outputs is tracked by $o$.
The number of remaining output nodes to be assigned is tracked by $n$.
The number of remaining input nodes to be matched with a pivot is $k$, which is $n-m$, since $m=n-k$.
Thus, $m$ tracks the number of unassigned output nodes that will need to be assigned as non-pivots.

We first describe $A_{n,m,p,o}$.
Suppose we want the next output node to be a pivot (in $\CP$).
Since there is a bijection between pivots and inputs, we can only assign additional pivots only if we have more remaining nodes $n$ than the the number that need to become non-pivots $m$.
Otherwise, we have $n=m$ and we cannot assign a node as a pivot, so in that case we have $A_{n,m,p,o}=0$.
The matching between pivots and input nodes is fully constrained by the RREF rule, which sorts the inputs and pivots together.
The Clifford rule says that no pivot nodes may have local Clifford operations, so we just need to choose the edges connecting the new pivot to nodes we have already assigned. Since there are no pivot-pivot edges and the pivot connects to only one input, we have exactly $2^o$ possibilities, where $o$ is the number of previously assigned non-pivot outputs.
These $o$ nodes are the only nodes to which we can connect our pivot.
Having made an assignment, $p$ increases by 1, and $n$ decreases by 1.
Thus, when $n>m$, we have $A_{n,m,p,o}=2^of(n-1,m,p+1,o)$.

We now describe $B_{n,m,p,o}$.
Suppose instead we want the next output node to be a non-pivot (in $\CO$).
If $m=0$, we must assign all of our remaining nodes to be pivots and we cannot assign a node to be an input, so in this case, $B_{n,m,p,o}=0$.
Otherwise, we have to choose the local Clifford operation as well as the node's edges to previously assigned vertices.
If we choose to connect the node to none of the $p$ assigned inputs, $p$ assigned pivots, or $o$ assigned vertices, then we are allowed to place any of the 6 local Cliffords on the node.
If any of those edges are present, we cannot apply a Hadamard to the output edge due to the Hadamard rule. Hence, 4 choices remain for the local Clifford operation.
This works out to a total of $4(2^{2p+o}-1)+6=2^{2p+o+2}+2$ possibilities.
To finish, we decrement the number of output qubits without changing the number of inputs.
This increases $o$ by 1 while decreasing both $n$ and $m$ by 1.
Hence, when $m>0$, we have $B_{n,m,p,o}=(2^{2p+o+2}+2)f(n-1,m-1,p,o+1)$.
\chapter{Quantum Gilbert-Varshamov bound}
\label{app:QGV}

The quantum Gilbert-Varshamov bound, Theorem~\ref{thm:quantum_gv}, is a well known result, discussed in texts such as \textcite{nielsen2002quantum}.
For completeness, we provide a self-contained proof of this theorem.
\begin{proof}
We proceed via a probabilistic method.
That is, we compute the probability that a random stabilizer code has parameters $\llbracket n, k, d\rrbracket$ where $n$ is large.
To construct a random stabilizer code, we begin with $Z_1, \dots, Z_{n-k}$ as stabilizers, and then choose a random Clifford operator $U$ on $n$ qubits. The random stabilizer code is given by stabilizers $S_1,\, \dots,\, S_{n-k}$, defined as $S_1 \coloneq U Z_1 U^\dagger,\, \dots,\, S_{n-k} \coloneq U Z_{n-k} U^\dag$. Such a code is a uniformly random stabilizer code because (a) all of its generators commute and (b) a random Clifford takes distinct non-identity Paulis to uniformly random non-identity Paulis.

We fix a $n$-qubit non-identity Pauli matrix $P$ and note that $P' \coloneq U P U^\dag$ is a fixed uniformly random non-identity Pauli matrix. Let $p = \Pr[P, S_1] = 0,\, \dots,\, [P, S_{n-k}] = 0]$ be the probability that $P$ commutes with the entire stabilizer tableau.
We note that $p=\Pr[[P', Z_1] = 0,\, \dots,\, [P', Z_{n-k}] = 0]$.
If $P'$ commutes with $Z_v$, then the Pauli on the $v^{\text{th}}$ qubit must be either $I$ or $Z$. There are $2(2^{n-k} 4^k - 1)$ non-identity Paulis which have $I$ or $Z$ on the first $n-k$ qubits, and $2(4^{n} - 1)$ total non-identity Paulis. For large enough $n$ we have the following.
\begin{align}
        p = \frac{2(2^{n-k} 4^k - 1)}{2(4^n - 1)} \approx \frac{2^{n+k}}{4^n} = \frac{1}{2^{n-k}} .
    \end{align}
    On the other hand, the total number of non-identity Paulis with weight at most $d$ is \begin{align}
        N = \sum_{m = 1}^{d} 3^m \binom{n}{m} \leq 3^d d\binom{n}{d} ,
    \end{align}
    using the assumption that $d \leq \frac{n}{2}$ so that the last term is the greatest. By Stirling's approximation, as $n$ increases, we have, \begin{align}
        N \longrightarrow 2^{n H(d/n) + d \log 3 + O(\log n)} .
    \end{align}
    To complete the proof, we apply the union bound. Let $A(S, P)$ be the event that $P$ commutes with $S_1,\, \dots,\, S_{n-k}$, and let $A(S) = \bigcup\limits_{P \neq I} A(S, P)$. Then \begin{align}
    \begin{aligned} \label{eq:app:QGV}
        \Pr[A(S)] & \leq \sum_{P \neq I} \Pr[A(S, P)] \\
        & \leq 2^{n H(d/n) + d \log 3 + O(\log n)} 2^{k - n} .
    \end{aligned}
    \end{align}
    A $\llbracket n, k, d\rrbracket$ code exists if and only if $\Pr[A(S)] < 1$. By setting the right-hand side of this inequality to be at most 1, we obtain the quantum Gilbert-Varshamov bound.
\end{proof}
\chapter{Constructing the 3-by-3 toric code}
\label{app:3-by-3-toric}

In Section~\ref{section:surfacetoric}, we provided the 2 $\times$ 2 and 3 $\times$ 3 toric codes in ZX calculus. Here, we provide a more detailed description of the methodology used to determine the structure of the 3 $\times$ 3 toric code.

\begin{figure}[h]
\centering
\includegraphics[scale=0.45]{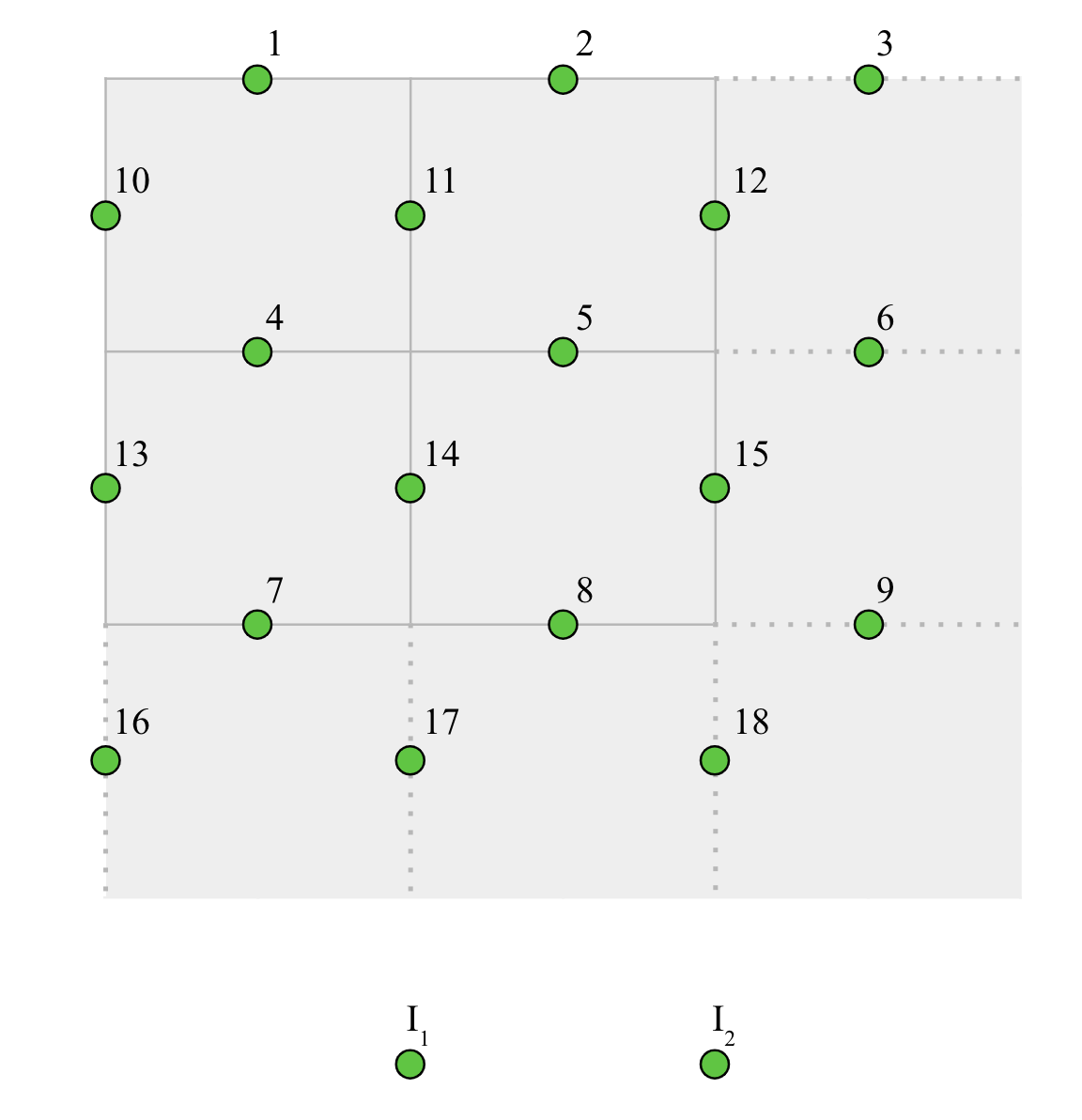}
\caption[Nodes of 3 $\times$ 3 toric code]{The input and output nodes in the 3 $\times$ 3 toric code. The boundary conditions are periodic, and all nodes are initially set to $Z$ nodes.}
\label{fig:3-by-3}
\end{figure}

Among the 18 outputs in Figure~\ref{fig:3-by-3}, we expect some of the free output edges to contain Hadamard gates.
For now, we keep the nodes as $Z$ nodes.
Suppose the output edge onto vertex 1 has a Hadamard gate. This implies that applying the stabilizer $Z_1Z_4Z_{10}Z_{11}$ would result in sliding a $Z$ gate from the end of the output edge, through the Hadamard gate, converting the $Z$ gate to an $X$ gate, then through vertex 1 itself.
By the $\pi$-copy rule from Definition~\ref{def:zx-basic-rewrite-rules}, the $X$ gate, which is an $X$ node with phase $\pi$, copies itself onto all internal edges connected to vertex 1.

Continuing the assumption that output node 1 has a Hadamard gate, if we instead apply the stabilizer $X_1X_3X_{10}X_{16}$, we slide an $X$ gate from the end of the output edge, through the Hadamard gate, converting the $X$ gate to a $Z$ gate, then onto vertex 1. By the merging rule from Definition~\ref{def:zx-basic-rewrite-rules}, the $Z$ gate, which has phase $\pi$, merges with vertex 1, a phase 0 $Z$ node. This results in vertex 1 gaining a phase of $\pi$.

By the above, the behaviour of the $Z$ and $X$ gates on an output node with a Hadamard gate is understood.
Analogous behaviour occurs on an output node without a Hadamard by switching all the $X$ nodes to $Z$ nodes used in the processes above and vice versa.

To determine all of the edges in the 3 $\times$ 3 toric code in Figure~\ref{fig:3-by-3}, we consider the process of applying the stabilizers onto the output nodes. Note that all of the internal edges among nodes in the diagram must be edges with Hadamard gate so that the merging rule cannot be applied to merge multiple nodes into one. To simplify our work, we set the output nodes with Hadamards to be 1, 2, 3, 7, 8, 9, 16, 17, and 18.
Then, by stabilizer $Z_1Z_4Z_{10}Z_{11}$, the nodes 4, 10, and 11 gain a phase of $\pi$ from their $Z$ gates while node 1 will cause a $\pi$-copy rule to move $X$ gates onto the internal edges connected to node 1. Since all internal edges have Hadamard gates, moving the $X$ gates through the Hadamards will result in $Z$ gates. If these $Z$ gates went to any nodes other than nodes 4, 10, and 11, the stabilizer would not have kept the configuration the same. Therefore, the $Z$ gates must arrive at only nodes 4, 10, and 11. This works because the $\pi$'s from these $Z$ gates cancel with the $\pi$s already at the nodes. Thus, the only internal edges to node 1 are from nodes 4, 10, and 11.

Using similar reasoning, we can deduce the rest of the internal edges among the output nodes by establishing which nodes must be connected to others by looking at previously established connections.
Lastly, we need to determine the logical operators to determine the connections of the input nodes.
We look for sets of nodes such that when any stabilizer is applied the input node remains at phase 0. The resulting diagram is shown in Figure~\ref{final 3 by 3}.
\chapter{Converting ZXCF's into quantum circuits}
\label{app:convert-zxcf-to-circuit}

In Section~\ref{subsec:encoding_circuit} and Section~\ref{section:background}, we described a procedure for converting a ZXCF into a quantum circuit, which is repeated here in a slightly modified form.
In this appendix, we present a step-by-step example of applying this procedure to a graph code with $k$ logical qubits and $n$ physical qubits.

\begin{enumerate}[(1)]
\item Start with $k$ open wires representing the inputs of the circuit.
\item Add a $\ket{0}$ state for each of the $n-k$ non-pivot output nodes.
\item Apply an $H$ gate to all $n$ wires.
\item Apply a $CX$ gate between the pairs wires corresponding to the edges between inputs and non-pivot outputs. The wire of the input node is the target qubit, and the wire of the output node is the control qubit.
\item Apply a $CZ$ gate between pairs of wires corresponding to the output-output edges.
\item Apply the local operations attached to the outputs to the corresponding wires.
\end{enumerate}

\begin{figure}
        \begin{subfigure}{0.45\textwidth}
            \includegraphics[width=\linewidth]{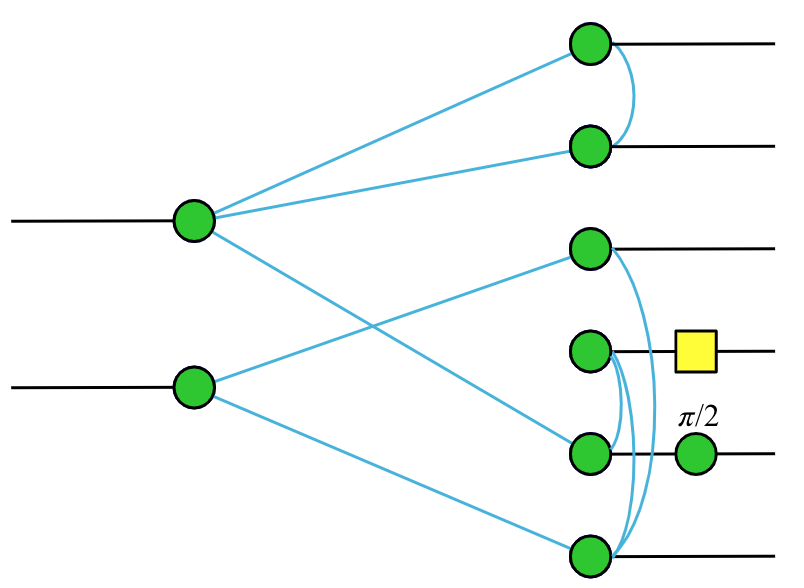}
            \caption{The ZXCF of a graph code. Blue edges represent edges with Hadamards. Note that two of the output edges have local operations. One of the local operations is a Hadamard gate while the other is an $S$ gate.}
            \label{fig:zxcf-graph-code}
        \end{subfigure}\hfill
        \begin{subfigure}{0.45\textwidth}
            \includegraphics[width=\linewidth]{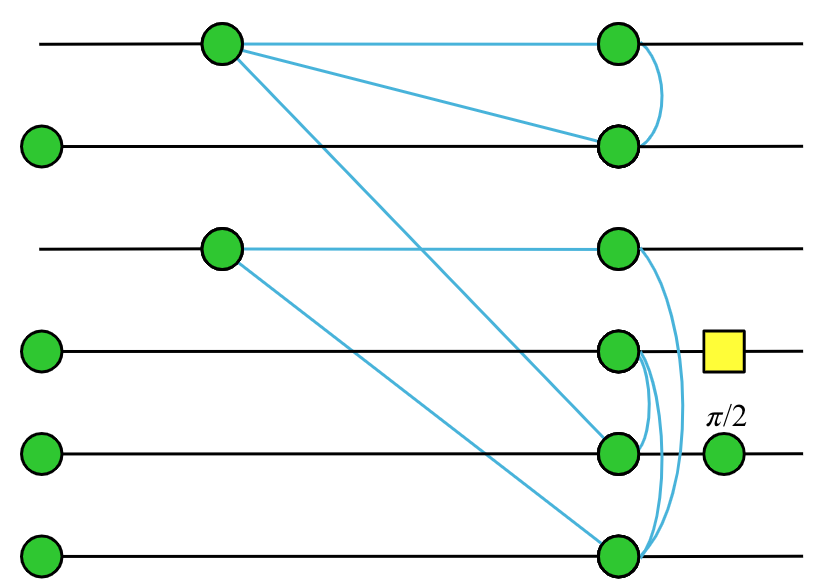}
            \caption{The non-pivot output nodes are unmerged into two $Z$ nodes each, and one of each pair is placed to the left.}
            \label{fig:unmerge-outputs}
        \end{subfigure}
        \begin{subfigure}{0.45\textwidth}
            \includegraphics[width=\linewidth]{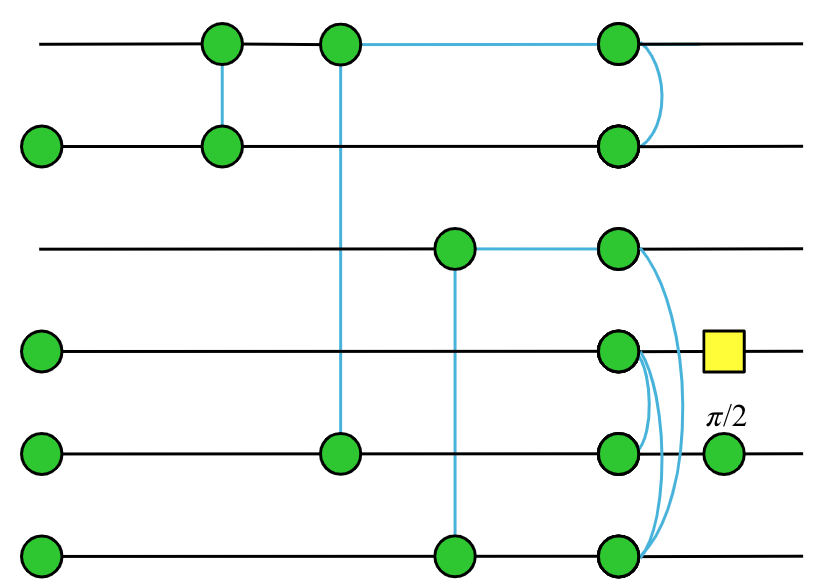}
            \caption{By unmerging the input nodes, the edges between the inputs and non-pivot outputs can be shown separately as $CZ$ gates.}
            \label{fig:unmerge-input-outputs}
        \end{subfigure}\hfill
        \begin{subfigure}{0.45\textwidth}
            \includegraphics[width=\linewidth]{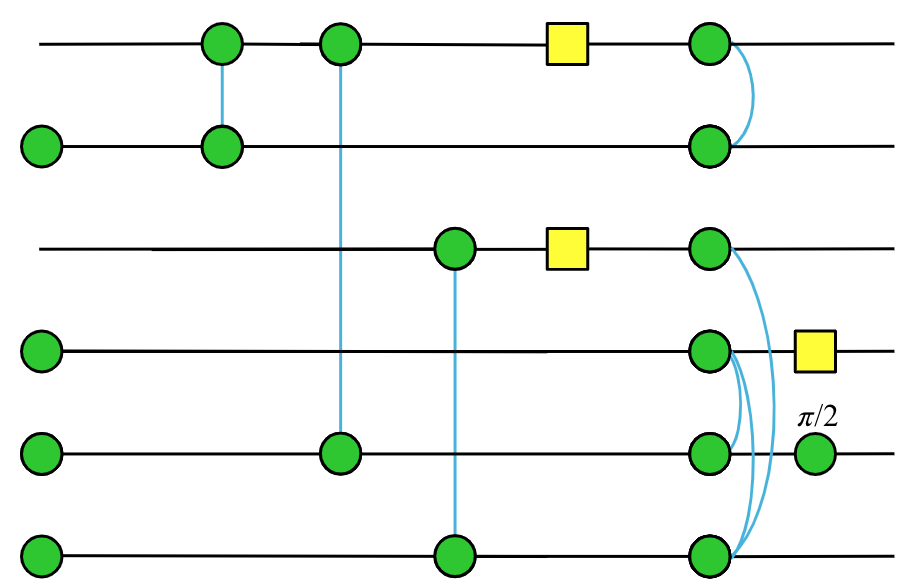}
            \caption{The blue input-pivot edges are displayed with a yellow Hadamard gate on them.}
            \label{fig:change-to-yellow-box}
        \end{subfigure}
        \begin{subfigure}{0.45\textwidth}
            \includegraphics[width=\linewidth]{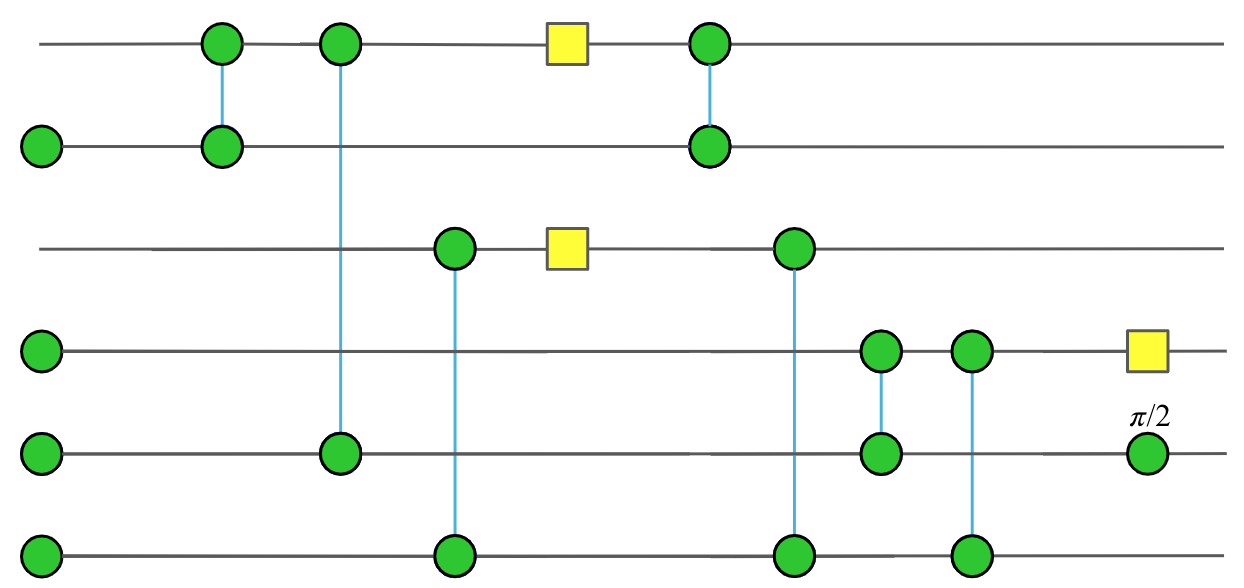}
            \caption{Similar to \ref{fig:unmerge-input-outputs}, the nodes in each output-output edge are unmerged. Note that the local operations are still at the very right-hand side of the diagram.}
            \label{fig:split-output-output-edges}
        \end{subfigure}\hfill
        \begin{subfigure}{0.45\textwidth}
            \includegraphics[width=\linewidth]{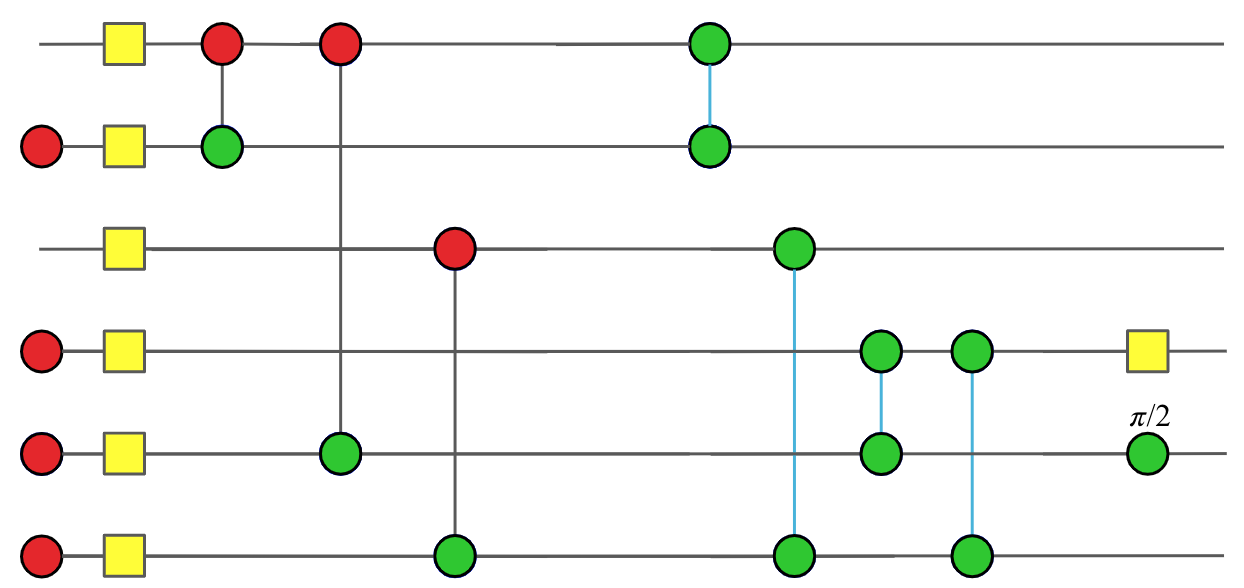}
            \caption{The Hadamards in the middle of \ref{fig:split-output-output-edges} are pushed to the left, and the green $\ket{+}$ states transformed into the form $H\ket{0}$, which is a Hadamard on a red $\ket{0}$ state.}
            \label{fig:push-hadamards}
        \end{subfigure}
    \caption[Conversion of graph code into quantum circuit]{Conversion of a graph code in ZX canonical form into an equivalent quantum circuit diagram.}
    \label{fig:detailed-circuit-transformation-diagram}
\end{figure}

We now show why this works.
Consider the example given in Figure~\ref{fig:zxcf-graph-code}. We will convert this graph code in ZXCF into a circuit.
We can first move the input nodes to be along the same horizontal wire as their pivots nodes. Then, we split the non-pivot nodes by unmerging two zero-phase $Z$ nodes. The resulting diagram is shown in Figure~\ref{fig:unmerge-outputs}.
From here, the edges from inputs to non-pivot outputs can be separated by unmerging nodes and expressing each edge separately, as shown in Figure~\ref{fig:unmerge-input-outputs}.
In Figure~\ref{fig:change-to-yellow-box}, the Hadamards between the inputs and pivots are shown explicitly as yellow boxes.
In Figure~\ref{fig:split-output-output-edges}, we do a similar unmerging of nodes to separately express the edges between nodes.

The steps used in these diagrams hold in general. We can unmerge each node until all the edges are expressed separately (and the non-pivot nodes have an initial state), and, to keep things organized, we can keep the input-output edges on the left side and the output-output edges on the right side.

From Figure~\ref{fig:split-output-output-edges}, note that the edges with Hadamards between the nodes of a ZX diagram are equivalent to the $CZ$ gates between the corresponding wires in a quantum circuit. Additionally, the $Z$ nodes at the start are equivalent to $\ket{+}$.
Note that all of these operations can be implemented in a standard quantum circuit.
We have thus reconstructed the algorithm described in Section~\ref{subsec:encoding_circuit}.

As a potential additional step, consider sliding the $k$ Hadamards on the input wires towards the start of the circuit, through the $Z$ nodes on the input wires.
Since $ZH = HX$, this means each of the $CZ$'s that the $H$'s pass through turns into a $CX$ with the target qubit on the input's wire. Also, we may exchange the $Z$ nodes at the start for a $X$ node and an $H$, since $H\ket{0} = \ket{+}$. The result of this operation is shown in Figure~\ref{fig:push-hadamards}.
From here, we can see why the above procedure produces a quantum encoder from any graph code.

%%% Bibliography (biblatex)  %%%%%%%%%%%%%%%%%%%%%%%%%%%%%%%%%%%%%%%%%%%%%%%%%%%%%%%%%%%%%%%%%%%%%%

\defbibheading{bibintoc}{\chapter*{#1}\addcontentsline{toc}{backmatter}{\refname}} 
% this sets the title of contents name for bibliography to \refname (= References)
% change "backmatter" to "chapter" if you prefer a bold face entry in the table of contents

\printbibliography[title={\refname},heading=bibintoc]

@article{aaronson2004improved,
  title={Improved simulation of stabilizer circuits},
  author={Aaronson, Scott and Gottesman, Daniel},
  journal={Physical Review A},
  volume={70},
  number={5},
  pages={052328},
  year={2004},
  publisher={APS},
  url={https://doi.org/10.1103/PhysRevA.70.052328},
  doi={10.1103/PhysRevA.70.052328}
}

@article{adcock2020mapping,
  doi = {10.22331/q-2020-08-07-305},
  url = {https://doi.org/10.22331/q-2020-08-07-305},
  title = {Mapping graph state orbits under local complementation},
  author = {Adcock, Jeremy C. and Morley-Short, Sam and Dahlberg, Axel and Silverstone, Joshua W.},
  journal = {{Quantum}},
  issn = {2521-327X},
  publisher = {{Verein zur F{\"{o}}rderung des Open Access Publizierens in den Quantenwissenschaften}},
  volume = {4},
  pages = {305},
  month = aug,
  year = {2020}
}

@article{aharonov2003-universality,
      title={A Simple Proof that Toffoli and Hadamard are Quantum Universal}, 
      author={Dorit Aharonov},
      year={2003},
      eprint={quant-ph/0301040},
      archivePrefix={arXiv},
      primaryClass={quant-ph},
      url={https://arxiv.org/abs/quant-ph/0301040}
}

@article{anders2006fast,
  title = {Fast simulation of stabilizer circuits using a graph-state representation},
  author = {Anders, Simon and Briegel, Hans J.},
  journal = {Phys. Rev. A},
  volume = {73},
  issue = {2},
  pages = {022334},
  numpages = {9},
  year = {2006},
  publisher = {American Physical Society},
  doi = {10.1103/PhysRevA.73.022334},
  url = {https://link.aps.org/doi/10.1103/PhysRevA.73.022334}
}

@article{andersen2020repeated,
  title={Repeated quantum error detection in a surface code},
  author={Andersen, Christian Kraglund and Remm, Ants and Lazar, Stefania and Krinner, Sebastian and Lacroix, Nathan and Norris, Graham J and Gabureac, Mihai and Eichler, Christopher and Wallraff, Andreas},
  journal={Nature Physics},
  volume={16},
  number={8},
  pages={875--880},
  year={2020},
  publisher={Nature Publishing Group},
  url={https://doi.org/10.1038/s41567-020-0920-y},
  doi={10.1038/s41567-020-0920-y}
}

@article{arute2019quantum,
  title={Quantum supremacy using a programmable superconducting processor},
  author={Arute, Frank and Arya, Kunal and Babbush, Ryan and Bacon, Dave and Bardin, Joseph C and Barends, Rami and Biswas, Rupak and Boixo, Sergio and Brandao, Fernando GSL and Buell, David A and others},
  journal={Nature},
  volume={574},
  number={7779},
  pages={505--510},
  year={2019},
  publisher={Nature Publishing Group},
  url={https://doi.org/10.1038/s41586-019-1666-5},
  doi={10.1038/s41586-019-1666-5}
}

@article{ashikhmin2001asymptotically,
  title = {Asymptotically good quantum codes},
  author = {Ashikhmin, Alexei and Litsyn, Simon and Tsfasman, Michael A.},
  journal = {Phys. Rev. A},
  volume = {63},
  issue = {3},
  pages = {032311},
  numpages = {5},
  year = {2001},
  publisher = {American Physical Society},
  doi = {10.1103/PhysRevA.63.032311},
  url = {https://link.aps.org/doi/10.1103/PhysRevA.63.032311}
}

@article{backens2014zx,
  title={The ZX-calculus is complete for stabilizer quantum mechanics},
  author={Backens, Miriam},
  journal={New Journal of Physics},
  volume={16},
  number={9},
  pages={093021},
  year={2014},
  publisher={IOP Publishing},
  url={https://doi.org/10.1088/1367-2630/16/9/093021},
  doi={10.1088/1367-2630/16/9/093021}
}

@article{bahramgiri2007enumerating,
      title={Enumerating the Classes of Local Equivalency in Graphs}, 
      author={Mohsen Bahramgiri and Salman Beigi},
      year={2007},
      eprint={math/0702267},
      archivePrefix={arXiv},
      primaryClass={math.CO},
      url={https://arxiv.org/abs/math/0702267}
}

@inproceedings{bergamaschi2024approaching,
author = {Bergamaschi, Thiago and Golowich, Louis and Gunn, Sam},
title = {Approaching the Quantum Singleton Bound with Approximate Error Correction},
year = {2024},
publisher = {Association for Computing Machinery},
address = {New York, NY, USA},
url = {https://doi.org/10.1145/3618260.3649680},
doi = {10.1145/3618260.3649680},
booktitle = {Proceedings of the 56th Annual ACM Symposium on Theory of Computing},
pages = {1507–1516},
numpages = {10},
location = {Vancouver, BC, Canada},
series = {STOC 2024}
}

@article{bouchet1993recognizing,
title={Recognizing locally equivalent graphs},
volume={114},
ISSN={0012-365X},
url={http://dx.doi.org/10.1016/0012-365x(93)90357-y},
DOI={10.1016/0012-365x(93)90357-y},
number={1–3},
journal={Discrete Mathematics},
publisher={Elsevier BV},
author={Bouchet, André},
year={1993},
month=apr,
pages={75–86}
}

@article{brakerski2021cryptographic,
  title={A cryptographic test of quantumness and certifiable randomness from a single quantum device},
  author={Brakerski, Zvika and Christiano, Paul and Mahadev, Urmila and Vazirani, Umesh and Vidick, Thomas},
  journal={Journal of the ACM (JACM)},
  volume={68},
  number={5},
  pages={1--47},
  year={2021},
  publisher={ACM New York, NY},
  url={https://doi.org/10.1109/FOCS.2018.00038},
  doi={10.1109/FOCS.2018.00038}
}

@article{bravyi1998quantum,
  title={Quantum codes on a lattice with boundary},
  author={Bravyi, Sergey B and Kitaev, A Yu},
  journal={arXiv preprint quant-ph/9811052},
  year={1998},
  url={https://doi.org/10.48550/arXiv.quant-ph/9811052},
  doi={10.48550/arXiv.quant-ph/9811052}
}

@article{bravyi2005-magicstate,
  title = {Universal quantum computation with ideal Clifford gates and noisy ancillas},
  author = {Bravyi, Sergey and Kitaev, Alexei},
  journal = {Phys. Rev. A},
  volume = {71},
  issue = {2},
  pages = {022316},
  numpages = {14},
  year = {2005},
  publisher = {American Physical Society},
  doi = {10.1103/PhysRevA.71.022316},
  url = {https://link.aps.org/doi/10.1103/PhysRevA.71.022316}
}

@article{bravyi2016-stabrank,
  title = {Trading Classical and Quantum Computational Resources},
  author = {Bravyi, Sergey and Smith, Graeme and Smolin, John A.},
  journal = {Phys. Rev. X},
  volume = {6},
  issue = {2},
  pages = {021043},
  numpages = {14},
  year = {2016},
  publisher = {American Physical Society},
  doi = {10.1103/PhysRevX.6.021043},
  url = {https://link.aps.org/doi/10.1103/PhysRevX.6.021043}
}

@article{bravyi2019simulation,
  doi = {10.22331/q-2019-09-02-181},
  url = {https://doi.org/10.22331/q-2019-09-02-181},
  title = {Simulation of quantum circuits by low-rank stabilizer decompositions},
  author = {Bravyi, Sergey and Browne, Dan and Calpin, Padraic and Campbell, Earl and Gosset, David and Howard, Mark},
  journal = {{Quantum}},
  issn = {2521-327X},
  publisher = {{Verein zur F{\"{o}}rderung des Open Access Publizierens in den Quantenwissenschaften}},
  volume = {3},
  pages = {181},
  month = sep,
  year = {2019}
}

@article{bravyi2024high,
title={High-threshold and low-overhead fault-tolerant quantum memory},
volume={627},
ISSN={1476-4687},
url={http://dx.doi.org/10.1038/s41586-024-07107-7},
DOI={10.1038/s41586-024-07107-7},
number={8005},
journal={Nature},
publisher={Springer Science and Business Media LLC},
author={Bravyi, Sergey and Cross, Andrew W. and Gambetta, Jay M. and Maslov, Dmitri and Rall, Patrick and Yoder, Theodore J.},
year={2024},
month=mar,
pages={778–782}
}

@article{breuckmann2021balanced,
author = {Breuckmann, Nikolas P. and Eberhardt, Jens N.},
title = {Balanced Product Quantum Codes},
year = {2021},
issue_date = {Oct. 2021},
publisher = {IEEE Press},
volume = {67},
number = {10},
issn = {0018-9448},
url = {https://doi.org/10.1109/TIT.2021.3097347},
doi = {10.1109/TIT.2021.3097347},
journal = {IEEE Trans. Inf. Theor.},
month = oct,
pages = {6653–6674},
numpages = {22}
}

@article{calderbank1996good,
  title = {Good quantum error-correcting codes exist},
  author = {Calderbank, A. R. and Shor, Peter W.},
  journal = {Phys. Rev. A},
  volume = {54},
  issue = {2},
  pages = {1098--1105},
  numpages = {0},
  year = {1996},
  publisher = {American Physical Society},
  doi = {10.1103/PhysRevA.54.1098},
  url = {https://link.aps.org/doi/10.1103/PhysRevA.54.1098}
}

@article{chancellor2016graphical,
  title={Graphical structures for design and verification of quantum error correction},
  author={Chancellor, Nicholas and Kissinger, Aleks and Zohren, Stefan and Roffe, Joschka and Horsman, Dominic},
  journal={Quantum Science and Technology},
  year={2016},
  url={https://doi.org/10.1088/2058-9565/acf157},
  doi={10.1088/2058-9565/acf157}
}

@article{choi1975completely,
  title={Completely positive linear maps on complex matrices},
  author={Choi, Man-Duen},
  journal={\emph{Linear algebra and its applications}},
  volume={10},
  number={3},
  pages={285--290},
  year={1975},
  publisher={Elsevier},
  url={http://dx.doi.org/10.1016/0024-3795(75)90075-0},
  DOI={10.1016/0024-3795(75)90075-0}
}

@article{cleve2015near,
  title={Near-linear constructions of exact unitary 2-designs},
  author={Cleve, Richard and Leung, Debbie and Liu, Li and Wang, Chunhao},
  journal={arXiv preprint arXiv:1501.04592},
  year={2015},
  url={http://dx.doi.org/10.26421/qic16.9-10-1},
  DOI={10.26421/qic16.9-10-1}
}

@inproceedings{coecke2008interacting,
  title={Interacting quantum observables},
  author={Coecke, Bob and Duncan, Ross},
  booktitle={International Colloquium on Automata, Languages, and Programming},
  pages={298--310},
  year={2008},
  organization={Springer},
  url={https://doi.org/10.1007/978-3-540-70583-3_25},
  doi={10.1007/978-3-540-70583-3_25}
}

@article{coecke2011interacting,
    title = {Interacting quantum observables: categorical algebra and diagrammatics},
    author = {Coecke, Bob and Duncan, Ross},
    journal = {\emph{New Journal of Physics}},
    volume = {13(4)},
    pages = {043016},
    year = {2011},
    url={http://dx.doi.org/10.1088/1367-2630/13/4/043016},
    DOI={10.1088/1367-2630/13/4/043016}
}

@inproceedings{coecke2018picturing,
  title={Picturing quantum processes: A first course on quantum theory and diagrammatic reasoning},
  author={Coecke, Bob and Kissinger, Aleks},
  booktitle={Diagrammatic Representation and Inference: 10th International Conference, Diagrams 2018, Edinburgh, UK, June 18-22, 2018, Proceedings 10},
  pages={28--31},
  year={2018},
  organization={Springer},
  url={https://doi.org/10.1007/978-3-319-91376-6_6},
  doi={10.1007/978-3-319-91376-6_6}
}

@article{couvreur2013construction,
  title={A construction of quantum LDPC codes from Cayley graphs},
  author={Couvreur, Alain and Delfosse, Nicolas and Z{\'e}mor, Gilles},
  journal={IEEE transactions on information theory},
  volume={59},
  number={9},
  pages={6087--6098},
  year={2013},
  publisher={IEEE},
  url={http://dx.doi.org/10.1109/isit.2011.6034209},
  DOI={10.1109/isit.2011.6034209}
}

@article{cowtan2022quantum,
  title={Quantum double aspects of surface code models},
  author={Cowtan, Alexander and Majid, Shahn},
  journal={Journal of Mathematical Physics},
  volume={63},
  number={4},
  pages={042202},
  year={2022},
  publisher={AIP Publishing LLC},
  url={https://doi.org/10.1063/5.0063768},
  doi={10.1063/5.0063768}
}

@article{de2020zx,
  title={The ZX calculus is a language for surface code lattice surgery},
  author={de Beaudrap, Niel and Horsman, Dominic},
  journal={Quantum},
  volume={4},
  pages={218},
  year={2020},
  publisher={Verein zur F{\"o}rderung des Open Access Publizierens in den Quantenwissenschaften},
  url={https://doi.org/10.22331/q-2020-01-09-218},
  doi={10.22331/q-2020-01-09-218}
}

@article{duncan2013verifying,
  title={Verifying the Steane code with Quantomatic},
  author={Duncan, Ross and Lucas, Maxime},
  journal={arXiv preprint arXiv:1306.4532},
  year={2013},
  url={https://doi.org/10.48550/arXiv.1306.4532},
  doi={10.48550/arXiv.1306.4532}
}

@article{duncan2020graph,
  title={Graph-theoretic Simplification of Quantum Circuits with the ZX-calculus},
  author={Duncan, Ross and Kissinger, Aleks and Perdrix, Simon and Van De Wetering, John},
  journal={Quantum},
  volume={4},
  pages={279},
  year={2020},
  publisher={Verein zur F{\"o}rderung des Open Access Publizierens in den Quantenwissenschaften},
  url={https://doi.org/10.22331/q-2020-06-04-279},
  doi={10.22331/q-2020-06-04-279}
}

@article{east2022aklt,
  title={AKLT-states as ZX-diagrams: diagrammatic reasoning for quantum states},
  author={East, Richard DP and van de Wetering, John and Chancellor, Nicholas and Grushin, Adolfo G},
  journal={PRX Quantum},
  volume={3},
  number={1},
  pages={010302},
  year={2022},
  publisher={APS},
  url={https://doi.org/10.1103/PRXQuantum.3.010302},
  doi={10.1103/PRXQuantum.3.010302}
}

@article{edmonds1965maximum,
  title={Maximum matching and a polyhedron with 0, 1-vertices},
  author={Edmonds, Jack},
  journal={Journal of research of the National Bureau of Standards B},
  volume={69},
  number={125-130},
  pages={55--56},
  year={1965},
  url={http://dx.doi.org/10.6028/jres.069b.013},
  DOI={10.6028/jres.069b.013}
}

@article{edmonds1965paths,
  title={Paths, trees, and flowers},
  author={Edmonds, Jack},
  journal={Canadian Journal of mathematics},
  volume={17},
  pages={449--467},
  year={1965},
  publisher={Cambridge University Press},
  url={http://dx.doi.org/10.4153/cjm-1965-045-4},
  DOI={10.4153/cjm-1965-045-4}
}

@article{elliott2008graphical,
  title = {Graphical description of the action of Clifford operators on stabilizer states},
  author = {Elliott, Matthew B. and Eastin, Bryan and Caves, Carlton M.},
  journal = {Phys. Rev. A},
  volume = {77},
  issue = {4},
  pages = {042307},
  numpages = {13},
  year = {2008},
  publisher = {American Physical Society},
  doi = {10.1103/PhysRevA.77.042307},
  url = {https://link.aps.org/doi/10.1103/PhysRevA.77.042307}
}

@book{elliott2008stabilizer,
      title={Stabilizer states and local realism}, 
      author={Matthew B. Elliott},
      year={2008},
      eprint={0807.2876},
      archivePrefix={arXiv},
      primaryClass={quant-ph},
      url={https://doi.org/10.48550/arXiv.0807.2876},
      doi={10.48550/arXiv.0807.2876}
}

@article{elliott2009graphical,
  title={Graphical description of Pauli measurements on stabilizer states},
  author={Elliott, Matthew B and Eastin, Bryan and Caves, Carlton M},
  journal={Journal of Physics A: Mathematical and Theoretical},
  volume={43},
  number={2},
  pages={025301},
  year={2009},
  publisher={IOP Publishing},
  url={http://dx.doi.org/10.1088/1751-8113/43/2/025301},
  DOI={10.1088/1751-8113/43/2/025301}
}

@article{englbrecht2022transformations,
  doi = {10.22331/q-2022-10-25-846},
  url = {https://doi.org/10.22331/q-2022-10-25-846},
  title = {Transformations of {S}tabilizer {S}tates in {Q}uantum {N}etworks},
  author = {Englbrecht, Matthias and Kraft, Tristan and Kraus, Barbara},
  journal = {{Quantum}},
  issn = {2521-327X},
  publisher = {{Verein zur F{\"{o}}rderung des Open Access Publizierens in den Quantenwissenschaften}},
  volume = {6},
  pages = {846},
  month = oct,
  year = {2022}
}

@inproceedings{farhi2012quantum,
  title={Quantum money from knots},
  author={Farhi, Edward and Gosset, David and Hassidim, Avinatan and Lutomirski, Andrew and Shor, Peter},
  booktitle={Proceedings of the 3rd Innovations in Theoretical Computer Science Conference},
  pages={276--289},
  year={2012},
  url={https://doi.org/10.48550/arXiv.1004.5127},
  doi={10.48550/arXiv.1004.5127}
}

@article{fowler2012surface,
  title={Surface codes: Towards practical large-scale quantum computation},
  author={Fowler, Austin G and Mariantoni, Matteo and Martinis, John M and Cleland, Andrew N},
  journal={Physical Review A},
  volume={86},
  number={3},
  pages={032324},
  year={2012},
  publisher={APS},
  url={https://doi.org/10.1103/PhysRevA.86.032324},
  doi={10.1103/PhysRevA.86.032324}
}

@article{garcia2012efficient,
      title={Efficient Inner-product Algorithm for Stabilizer States}, 
      author={Hector J. Garcia and Igor L. Markov and Andrew W. Cross},
      year={2013},
      eprint={1210.6646},
      archivePrefix={arXiv},
      primaryClass={cs.ET},
      url={https://doi.org/10.48550/arXiv.1210.6646},
      doi={10.48550/arXiv.1210.6646}
}

@article{garcia2017geometry,
  title={On the geometry of stabilizer states},
  author={Garc{\'\i}a, H{\'e}ctor J and Markov, Igor L and Cross, Andrew W},
  journal={arXiv preprint arXiv:1711.07848},
  year={2017},
  url={http://dx.doi.org/10.26421/qic14.7-8-9},
  DOI={10.26421/qic14.7-8-9}
}

@article{georgescu2014quantum,
  title={Quantum simulation},
  author={Georgescu, Iulia M and Ashhab, Sahel and Nori, Franco},
  journal={Reviews of Modern Physics},
  volume={86},
  number={1},
  pages={153},
  year={2014},
  publisher={APS},
  url={https://doi.org/10.1103/RevModPhys.86.153},
  doi={10.1103/RevModPhys.86.153}
}

@article{gidney2021stim,
  doi = {10.22331/q-2021-07-06-497},
  url = {https://doi.org/10.22331/q-2021-07-06-497},
  title = {Stim: a fast stabilizer circuit simulator},
  author = {Gidney, Craig},
  journal = {{Quantum}},
  issn = {2521-327X},
  publisher = {{Verein zur F{\"{o}}rderung des Open Access Publizierens in den Quantenwissenschaften}},
  volume = {5},
  pages = {497},
  month = jul,
  year = {2021}
}

@article{gottesman1996class,
  title = {Class of quantum error-correcting codes saturating the quantum Hamming bound},
  author = {Gottesman, Daniel},
  journal = {Phys. Rev. A},
  volume = {54},
  issue = {3},
  pages = {1862--1868},
  numpages = {0},
  year = {1996},
  publisher = {American Physical Society},
  doi = {10.1103/PhysRevA.54.1862},
  url = {https://link.aps.org/doi/10.1103/PhysRevA.54.1862}
}

@book{gottesman1997stabilizer,
  title={Stabilizer codes and quantum error correction},
  author={Gottesman, Daniel},
  year={1997},
  publisher={California Institute of Technology},
  url={https://doi.org/10.48550/arXiv.quant-ph/9705052},
  doi={10.48550/arXiv.quant-ph/9705052}
}

@article{gottesman1998-stabilizers,
  title = {Theory of fault-tolerant quantum computation},
  author = {Gottesman, Daniel},
  journal = {Phys. Rev. A},
  volume = {57},
  issue = {1},
  pages = {127--137},
  numpages = {0},
  year = {1998},
  publisher = {American Physical Society},
  doi = {10.1103/PhysRevA.57.127},
  url = {https://link.aps.org/doi/10.1103/PhysRevA.57.127}
}

@article{gottesman2009introduction,
  title={An Introduction to Quantum Error Correction and Fault-Tolerant Quantum Computation},
  author={Gottesman, Daniel},
  journal={arXiv preprint arXiv:0904.2557},
  year={2009},
  url={https://doi.org/10.48550/arXiv.0904.2557},
  doi={10.48550/arXiv.0904.2557}
}

@inproceedings{grover1996fast,
  title={A fast quantum mechanical algorithm for database search},
  author={Grover, Lov K},
  booktitle={Proceedings of the twenty-eighth annual ACM symposium on Theory of computing},
  pages={212--219},
  year={1996},
  url={https://doi.org/10.1145/237814.237866},
  doi={10.1145/237814.237866}
}

@article{hamming1950error,
  title={Error detecting and error correcting codes},
  author={Hamming, Richard W},
  journal={The Bell system technical journal},
  volume={29},
  number={2},
  pages={147--160},
  year={1950},
  publisher={Nokia Bell Labs},
  url={http://dx.doi.org/10.1002/j.1538-7305.1950.tb00463.x},
  DOI={10.1002/j.1538-7305.1950.tb00463.x}
}

@article{harris2018calderbank,
  title = {Calderbank-Shor-Steane holographic quantum error-correcting codes},
  author = {Harris, Robert J. and McMahon, Nathan A. and Brennen, Gavin K. and Stace, Thomas M.},
  journal = {Phys. Rev. A},
  volume = {98},
  issue = {5},
  pages = {052301},
  numpages = {6},
  year = {2018},
  publisher = {American Physical Society},
  doi = {10.1103/PhysRevA.98.052301},
  url = {https://link.aps.org/doi/10.1103/PhysRevA.98.052301}
}

@article{harrow2004superdense,
  title={Superdense coding of quantum states},
  author={Harrow, Aram and Hayden, Patrick and Leung, Debbie},
  journal={Physical review letters},
  volume={92},
  number={18},
  pages={187901},
  year={2004},
  publisher={APS},
  url={https://doi.org/10.1103/PhysRevLett.92.187901},
  doi={10.1103/PhysRevLett.92.187901}
}

@article{hein2004-graphpauli,
  title = {Multiparty entanglement in graph states},
  author = {Hein, M. and Eisert, J. and Briegel, H. J.},
  journal = {Phys. Rev. A},
  volume = {69},
  issue = {6},
  pages = {062311},
  numpages = {20},
  year = {2004},
  publisher = {American Physical Society},
  doi = {10.1103/PhysRevA.69.062311},
  url = {https://link.aps.org/doi/10.1103/PhysRevA.69.062311}
}

@incollection{hein2006entanglement,
  title={Entanglement in graph states and its applications},
  author={Hein, Marc and D{\"u}r, Wolfgang and Eisert, Jens and Raussendorf, Robert and Van den Nest, Maarten and Briegel, H-J},
  booktitle={Quantum computers, algorithms and chaos},
  pages={115--218},
  year={2006},
  publisher={IOS Press},
  url={https://arxiv.org/abs/quant-ph/0602096}
}

@article{hu2022improved,
  title={Improved graph formalism for quantum circuit simulation},
  author={Hu, Alexander Tianlin and Khesin, Andrey Boris},
  journal={Physical Review A},
  volume={105},
  number={2},
  pages={022432},
  year={2022},
  publisher={APS},
  url={https://doi.org/10.1103/PhysRevA.105.022432},
  doi={PhysRevA.105.022432}
}

@article{huang2020predicting,
  title={Predicting many properties of a quantum system from very few measurements},
  author={Huang, Hsin-Yuan and Kueng, Richard and Preskill, John},
  journal={Nature Physics},
  volume={16},
  number={10},
  pages={1050--1057},
  year={2020},
  publisher={Nature Publishing Group},
  url={https://doi.org/10.1038/s41567-020-0932-7},
  doi={10.1038/s41567-020-0932-7}
}

@article{huang2023graphical,
  title={Graphical {CSS} Code Transformation Using {ZX} Calculus},
  author={Huang, Jiaxin and Li, Sarah Meng and Yeh, Lia and Kissinger, Aleks and Mosca, Michele and Vasmer, Michael},
  journal={arXiv preprint arXiv:2307.02437},
  year={2023},
  url={http://dx.doi.org/10.4204/eptcs.384.1},
  DOI={10.4204/eptcs.384.1}
}

@article{jamiolkowski1972linear,
  title={Linear transformations which preserve trace and positive semidefiniteness of operators},
  author={Jamio{\l}kowski, Andrzej},
  journal={Reports on Mathematical Physics},
  volume={3},
  number={4},
  pages={275--278},
  year={1972},
  publisher={Elsevier},
  url={http://dx.doi.org/10.1016/0034-4877(72)90011-0},
  DOI={10.1016/0034-4877(72)90011-0}
}

@article{kandala2019error,
  title={Error mitigation extends the computational reach of a noisy quantum processor},
  author={Kandala, Abhinav and Temme, Kristan and C{\'o}rcoles, Antonio D and Mezzacapo, Antonio and Chow, Jerry M and Gambetta, Jay M},
  journal={Nature},
  volume={567},
  number={7749},
  pages={491--495},
  year={2019},
  publisher={Nature Publishing Group},
  url={https://doi.org/10.1038/s41586-019-1040-7},
  doi={10.1038/s41586-019-1040-7}
}

@article{kapshikar2023hardness,
  title={On the hardness of the minimum distance problem of quantum codes},
  author={Kapshikar, Upendra and Kundu, Srijita},
  journal={IEEE Transactions on Information Theory},
  volume={69},
  number={10},
  pages={6293--6302},
  year={2023},
  publisher={IEEE},
  url={http://dx.doi.org/10.1109/tit.2023.3286870},
  DOI={10.1109/tit.2023.3286870}
}

@mastersthesis{kerzner2021clifford,
  title={Clifford simulation: Techniques and applications},
  author={Kerzner, Alexander},
  year={2021},
  school={University of Waterloo},
  url = {http://hdl.handle.net/10012/17038}
}

@software{khesinlu_graphcodes,
author = {Khesin, Andrey Boris and Lu, Jonathan Z.},
title = {{graphcodes (GitHub Repository)}},
url = {https://github.com/jz-lu/graphcodes}
}

@article{khesin2021extending,
title={Extending the graph formalism to higher-order gates}, volume={23},
ISSN={1533-7146},
url={http://dx.doi.org/10.26421/qic23.13-14-5},
DOI={10.26421/qic23.13-14-5},
number={13--14},
journal={Quantum Information and Computation},
publisher={Rinton Press},
author={Khesin, A. and Ren, K.},
year={2023},
month=nov,
pages={1128–1141}
}

@article{khesin2023graphical,
      title={Graphical quantum Clifford-encoder compilers from the ZX calculus}, 
      author={Andrey Boris Khesin and Jonathan Z. Lu and Peter W. Shor},
      year={2024},
      eprint={2301.02356},
      archivePrefix={arXiv},
      primaryClass={quant-ph},
      url={https://doi.org/10.48550/arXiv.2301.02356},
      doi={10.48550/arXiv.2301.02356}
}

@article{khesin2024equivalence,
      title={Equivalence Classes of Quantum Error-Correcting Codes}, 
      author={Andrey Boris Khesin and Alexander Li},
      year={2024},
      eprint={2406.12083},
      archivePrefix={arXiv},
      primaryClass={quant-ph},
      url={https://doi.org/10.48550/arXiv.2406.12083},
      doi={10.48550/arXiv.2406.12083}
}

@article{khesin2024universal,
      title={Universal graph representation of stabilizer codes}, 
      author={Andrey Boris Khesin and Jonathan Z. Lu and Peter W. Shor},
      year={2024},
      eprint={2411.14448},
      archivePrefix={arXiv},
      primaryClass={quant-ph},
      url={https://doi.org/10.48550/arXiv.2411.14448},
      doi={10.48550/arXiv.2411.14448}
}

@article{kim2023evidence,
  title={Evidence for the utility of quantum computing before fault tolerance},
  author={Kim, Youngseok and Eddins, Andrew and Anand, Sajant and Wei, Ken Xuan and Van Den Berg, Ewout and Rosenblatt, Sami and Nayfeh, Hasan and Wu, Yantao and Zaletel, Michael and Temme, Kristan and others},
  journal={\emph{Nature}},
  volume={618},
  number={7965},
  pages={500--505},
  year={2023},
  publisher={Nature Publishing Group UK London},
  url={http://dx.doi.org/10.1038/s41586-023-06096-3},
  DOI={10.1038/s41586-023-06096-3}
}

@article{kissinger2020reducing,
  title={Reducing the number of non-Clifford gates in quantum circuits},
  author={Kissinger, Aleks and van de Wetering, John},
  journal={Physical Review A},
  volume={102},
  number={2},
  pages={022406},
  year={2020},
  publisher={APS},
  url={https://doi.org/10.1103/PhysRevA.102.022406},
  doi={10.1103/PhysRevA.102.022406}
}

@article{kissinger2022phase,
      title={Phase-free ZX diagrams are CSS codes (...or how to graphically grok the surface code)}, 
      author={Aleks Kissinger},
      year={2022},
      eprint={2204.14038},
      archivePrefix={arXiv},
      primaryClass={quant-ph},
      url={https://doi.org/10.48550/arXiv.2204.14038},
      doi={10.48550/arXiv.2204.14038}
}

@incollection{kitaev1997quantum,
  title={Quantum error correction with imperfect gates},
  author={Kitaev, A Yu},
  booktitle={Quantum communication, computing, and measurement},
  pages={181--188},
  year={1997},
  publisher={Springer},
  url={https://doi.org/10.1007/978-1-4615-5923-8_19},
  doi={10.1007/978-1-4615-5923-8_19}
}

@article{kitaev2003fault,
  title={Fault-tolerant quantum computation by anyons},
  author={Kitaev, A Yu},
  journal={Annals of physics},
  volume={303},
  number={1},
  pages={2--30},
  year={2003},
  publisher={Elsevier},
  url={http://dx.doi.org/10.1016/s0003-4916(02)00018-0},
  DOI={10.1016/s0003-4916(02)00018-0}
}

@article{knill2001benchmarking,
  title={Benchmarking quantum computers: The five-qubit error correcting code},
  author={Knill, Emanuel and Laflamme, Raymond and Martinez, Rudy and Negrevergne, Camille},
  journal={Physical Review Letters},
  volume={86},
  number={25},
  pages={5811},
  year={2001},
  publisher={APS},
  url={https://doi.org/10.1103/PhysRevLett.86.5811},
  doi={10.1103/PhysRevLett.86.5811}
}

@article{larsen2019deterministic,
  title={Deterministic generation of a two-dimensional cluster state},
  author={Larsen, Mikkel V and Guo, Xueshi and Breum, Casper R and Neergaard-Nielsen, Jonas S and Andersen, Ulrik L},
  journal={Science},
  volume={366},
  number={6463},
  pages={369--372},
  year={2019},
  publisher={American Association for the Advancement of Science},
  url={https://doi.org/10.1126/science.aay4354},
  doi={10.1126/science.aay4354}
}

@mastersthesis{li2023graphical,
  title={Graphical CSS Code Transformation Using ZX Calculus},
  author={Li, Sarah Meng},
  year={2023},
  school={University of Waterloo},
  url = {http://hdl.handle.net/10012/20193}
}

@article{mcelvanney2022complete,
   title={Complete Flow-Preserving Rewrite Rules for MBQC Patterns with Pauli Measurements},
   volume={394},
   ISSN={2075-2180},
   url={http://dx.doi.org/10.4204/EPTCS.394.5},
   DOI={10.4204/eptcs.394.5},
   journal={Electronic Proceedings in Theoretical Computer Science},
   publisher={Open Publishing Association},
   author={McElvanney, Tommy and Backens, Miriam},
   year={2023},
   month=nov, pages={66–82}
}

@article{misra1992constructive,
  title={A constructive proof of Vizing's theorem},
  author={Misra, Jayadev and Gries, David},
  journal={Information Processing Letters},
  volume={41},
  number={3},
  pages={131--133},
  year={1992},
  publisher={Elsevier},
  url={http://dx.doi.org/10.1016/0020-0190(92)90041-s},
  DOI={10.1016/0020-0190(92)90041-s}
}

@book{nielsen2002quantum,
  place={Cambridge},
  title={Quantum Computation and Quantum Information: 10th Anniversary Edition},
  publisher={Cambridge University Press},
  author={Nielsen, Michael A. and Chuang, Isaac L.},
  year={2010},
  url={https://doi.org/10.1017/CBO9780511976667},
  doi={10.1017/CBO9780511976667}
}

@article{panteleev2021degenerate,
  doi = {10.22331/q-2021-11-22-585},
  url = {https://doi.org/10.22331/q-2021-11-22-585},
  title = {Degenerate {Q}uantum {LDPC} {C}odes {W}ith {G}ood {F}inite {L}ength {P}erformance},
  author = {Panteleev, Pavel and Kalachev, Gleb},
  journal = {{Quantum}},
  issn = {2521-327X},
  publisher = {{Verein zur F{\"{o}}rderung des Open Access Publizierens in den Quantenwissenschaften}},
  volume = {5},
  pages = {585},
  month = nov,
  year = {2021}
}

@inproceedings{panteleev2022asymptotically,
  title={Asymptotically good quantum and locally testable classical LDPC codes},
  author={Panteleev, Pavel and Kalachev, Gleb},
  booktitle={Proceedings of the 54th Annual ACM SIGACT Symposium on Theory of Computing},
  pages={375--388},
  year={2022},
  url={http://dx.doi.org/10.1145/3519935.3520017},
  DOI={10.1145/3519935.3520017}
}

@article{pastawski2015holographic,
  title={Holographic quantum error-correcting codes: Toy models for the bulk/boundary correspondence},
  author={Pastawski, Fernando and Yoshida, Beni and Harlow, Daniel and Preskill, John},
  journal={Journal of High Energy Physics},
  volume={2015},
  number={6},
  pages={1--55},
  year={2015},
  publisher={Springer},
  url={http://dx.doi.org/10.1007/jhep06(2015)149},
  DOI={10.1007/jhep06(2015)149}
}

@article{peham2022equivalence,
  title={Equivalence Checking of Quantum Circuits With the ZX-Calculus},
  author={Peham, Tom and Burgholzer, Lukas and Wille, Robert},
  journal={IEEE Journal on Emerging and Selected Topics in Circuits and Systems},
  volume={12},
  number={3},
  pages={662--675},
  year={2022},
  publisher={IEEE},
  url={https://doi.org/10.1109/JETCAS.2022.3202204},
  doi={10.1109/JETCAS.2022.3202204}
}

@article{peleg2022lower,
  doi = {10.22331/q-2022-02-15-652},
  url = {https://doi.org/10.22331/q-2022-02-15-652},
  title = {Lower {B}ounds on {S}tabilizer {R}ank},
  author = {Peleg, Shir and Shpilka, Amir and Volk, Ben Lee},
  journal = {{Quantum}},
  issn = {2521-327X},
  publisher = {{Verein zur F{\"{o}}rderung des Open Access Publizierens in den Quantenwissenschaften}},
  volume = {6},
  pages = {652},
  month = feb,
  year = {2022}
}

@article{poremba2024learning,
      title={The Learning Stabilizers with Noise problem}, 
      author={Alexander Poremba and Yihui Quek and Peter Shor},
      year={2024},
      eprint={2410.18953},
      archivePrefix={arXiv},
      primaryClass={quant-ph},
      url={https://doi.org/10.48550/arXiv.2410.18953},
      doi={10.48550/arXiv.2410.18953}
}

@article{qassim2021improved,
  doi = {10.22331/q-2021-12-20-606},
  url = {https://doi.org/10.22331/q-2021-12-20-606},
  title = {Improved upper bounds on the stabilizer rank of magic states},
  author = {Qassim, Hammam and Pashayan, Hakop and Gosset, David},
  journal = {{Quantum}},
  issn = {2521-327X},
  publisher = {{Verein zur F{\"{o}}rderung des Open Access Publizierens in den Quantenwissenschaften}},
  volume = {5},
  pages = {606},
  month = dec,
  year = {2021}
}

@software{quantomatic,
  title = {{Quantomatic}},
  howpublished = {\url{https://quantomatic.github.io/}},
  note = {Accessed: August 5, 2023},
}

@article{rakovszky2024physics,
      title={The Physics of (good) LDPC Codes II. Product constructions}, 
      author={Tibor Rakovszky and Vedika Khemani},
      year={2024},
      eprint={2402.16831},
      archivePrefix={arXiv},
      primaryClass={quant-ph},
      url={https://doi.org/10.48550/arXiv.2402.16831},
      doi={10.48550/arXiv.2402.16831}
}

@article{raussendorf2003measurement,
  title={Measurement-based quantum computation on cluster states},
  author={Raussendorf, Robert and Browne, Daniel E and Briegel, Hans J},
  journal={Physical review A},
  volume={68},
  number={2},
  pages={022312},
  year={2003},
  publisher={APS},
  url={http://dx.doi.org/10.1142/s0219749909005699},
  DOI={10.1142/s0219749909005699}
}

@article{sarvepalli2009sharing,
  title = {Sharing classical secrets with Calderbank-Shor-Steane codes},
  author = {Sarvepalli, Pradeep Kiran and Klappenecker, Andreas},
  journal = {Phys. Rev. A},
  volume = {80},
  issue = {2},
  pages = {022321},
  numpages = {4},
  year = {2009},
  publisher = {American Physical Society},
  doi = {10.1103/PhysRevA.80.022321},
  url = {https://link.aps.org/doi/10.1103/PhysRevA.80.022321}
}

@article{shor1995scheme,
  title={Scheme for reducing decoherence in quantum computer memory},
  author={Shor, Peter W},
  journal={Physical review A},
  volume={52},
  number={4},
  pages={R2493},
  year={1995},
  publisher={APS},
  url={https://doi.org/10.1103/PhysRevA.52.R2493},
  doi={10.1103/PhysRevA.52.R2493}
}

@article{shor1999polynomial,
  title={Polynomial-time algorithms for prime factorization and discrete logarithms on a quantum computer},
  author={Shor, Peter W},
  journal={SIAM review},
  volume={41},
  number={2},
  pages={303--332},
  year={1999},
  publisher={SIAM},
  url={https://doi.org/10.1137/S0097539795293172},
  doi={10.1137/S0097539795293172}
}

@article{sipser1996expander,
  title={Expander codes},
  author={Sipser, Michael and Spielman, Daniel A},
  journal={IEEE transactions on Information Theory},
  volume={42},
  number={6},
  pages={1710--1722},
  year={1996},
  publisher={IEEE},
  url={http://dx.doi.org/10.1109/sfcs.1994.365734},
  DOI={10.1109/sfcs.1994.365734}
}

@article{steane1996multiple,
  title={Multiple-particle interference \& quantum error correction},
  author={Steane, Andrew},
  journal={Proceedings of the Royal Society of London. Series A: Mathematical, Physical and Engineering Sciences},
  volume={452},
  number={1954},
  pages={2551--2577},
  year={1996},
  publisher={The Royal Society London},
  url={https://doi.org/10.1098/rspa.1996.0136},
  doi={10.1098/rspa.1996.0136}
}

@article{steane1999enlargement,
  title={Enlargement of {C}alderbank-{S}hor-{S}teane quantum codes},
  author={Steane, Andrew M},
  journal={\emph{IEEE Transactions on Information Theory}},
  volume={45},
  number={7},
  pages={2492--2495},
  year={1999},
  publisher={IEEE},
  url={http://dx.doi.org/10.1109/18.796388},
  DOI={10.1109/18.796388}
}

@article{tanner1981recursive,
  title={A recursive approach to low complexity codes},
  author={Tanner, R},
  journal={IEEE Transactions on information theory},
  volume={27},
  number={5},
  pages={533--547},
  year={1981},
  publisher={IEEE},
  url={http://dx.doi.org/10.1109/tit.1981.1056404},
  DOI={10.1109/tit.1981.1056404}
}

@article{tillich2013quantum,
  title={Quantum LDPC codes with positive rate and minimum distance proportional to the square root of the blocklength},
  author={Tillich, Jean-Pierre and Z{\'e}mor, Gilles},
  journal={IEEE Transactions on Information Theory},
  volume={60},
  number={2},
  pages={1193--1202},
  year={2013},
  publisher={IEEE},
  url={http://dx.doi.org/10.1109/tit.2013.2292061},
  DOI={10.1109/tit.2013.2292061}
}

@article{vandennest2004graphical,
  title = {Graphical description of the action of local Clifford transformations on graph states},
  author = {Van den Nest, Maarten and Dehaene, Jeroen and De Moor, Bart},
  journal = {Phys. Rev. A},
  volume = {69},
  issue = {2},
  pages = {022316},
  numpages = {7},
  year = {2004},
  publisher = {American Physical Society},
  doi = {10.1103/PhysRevA.69.022316},
  url = {https://link.aps.org/doi/10.1103/PhysRevA.69.022316}
}

@article{van2020zx,
      title={ZX-calculus for the working quantum computer scientist}, 
      author={John van de Wetering},
      year={2020},
      eprint={2012.13966},
      archivePrefix={arXiv},
      primaryClass={quant-ph},
      url={https://doi.org/10.48550/arXiv.2012.13966},
      doi={10.48550/arXiv.2012.13966}
}

@article{van2021constructing,
  title={Constructing quantum circuits with global gates},
  author={van de Wetering, John},
  journal={New Journal of Physics},
  volume={23},
  number={4},
  pages={043015},
  year={2021},
  publisher={IOP Publishing},
  url={https://doi.org/10.1088/1367-2630/abf1b3},
  doi={10.1088/1367-2630/abf1b3}
}

@article{vizing,
    author = {Vizing, Vadim Georgievich},
    title = {On an estimate of the chromatic class of a p-graph},
    journal = {Diskret. Analiz.},
    volume = {3},
    pages = {25--30},
    year = {1964},
    url={http://dx.doi.org/10.1007/bf01885700},
    DOI={10.1007/bf01885700}
}

@article{wang2019boson,
  title={Boson sampling with 20 input photons and a 60-mode interferometer in a 1 0 14-dimensional hilbert space},
  author={Wang, Hui and Qin, Jian and Ding, Xing and Chen, Ming-Cheng and Chen, Si and You, Xiang and He, Yu-Ming and Jiang, Xiao and You, L and Wang, Z and others},
  journal={Physical review letters},
  volume={123},
  number={25},
  pages={250503},
  year={2019},
  publisher={APS},
  url={https://doi.org/10.1103/PhysRevLett.123.250503},
  doi={10.1103/PhysRevLett.123.250503}
}

@article{wu2023zx,
      title={A ZX-Calculus Approach for the Construction of Graph Codes}, 
      author={Zipeng Wu and Song Cheng and Bei Zeng},
      year={2024},
      eprint={2304.08363},
      archivePrefix={arXiv},
      primaryClass={quant-ph},
      url={https://doi.org/10.48550/arXiv.2304.08363},
      doi={10.48550/arXiv.2304.08363}
}

@article{yu2007graphical,
      title={Graphical Quantum Error-Correcting Codes}, 
      author={Sixia Yu and Qing Chen and C. H. Oh},
      year={2007},
      eprint={0709.1780},
      archivePrefix={arXiv},
      primaryClass={quant-ph},
      url={https://doi.org/10.48550/arXiv.0709.1780},
      doi={10.48550/arXiv.0709.1780}
}

@article{zarei2017strong,
  title = {Strong-weak coupling duality between two perturbed quantum many-body systems: Calderbank-Shor-Steane codes and Ising-like systems},
  author = {Zarei, Mohammad Hossein},
  journal = {Phys. Rev. B},
  volume = {96},
  issue = {16},
  pages = {165146},
  numpages = {8},
  year = {2017},
  publisher = {American Physical Society},
  doi = {10.1103/PhysRevB.96.165146},
  url = {https://link.aps.org/doi/10.1103/PhysRevB.96.165146}
}

% biblatex also supports chapter-by-chapter bibliography, https://tex.stackexchange.com/a/296502/119566
% see the biblatex manual, section 3.14.3

%%%% Option for natbib %%%%%%%%%%%%%

%%   use an appropriate style (.bst) and your own .bib file[s]

%\bibliographystyle{plainnat}
%\bibliography{mitthesis-sample.bib}

\end{document}